\documentclass[
  aps,
  prx,
  reprint,
  superscriptaddress,
  eqsecnum,
  amsfonts,amssymb,amsmath
]{revtex4-2}

\usepackage{bm}
\usepackage{latexsym}
\usepackage{mathrsfs}
\usepackage{physics}
\usepackage{graphicx}
\usepackage{xcolor}
\usepackage{hyperref}
\usepackage{comment}
\hypersetup{
  setpagesize  = false,
  colorlinks   = true,
  urlcolor     = blue,
  linkcolor    = blue,
  citecolor    = blue
}

\newcommand{\sket}[1]{\lvert{#1}\rangle\!\rangle}
\newcommand{\sbra}[1]{\langle\!\langle{#1}\rvert}
\newcommand{\sbk}[2]{\langle\!\langle{#1}\vert{#2}\rangle\!\rangle}

\newcommand{\gtensor}{\mathbin{\widehat\otimes}}

\DeclareMathOperator{\Real}{Re}

\usepackage{amsthm}
\theoremstyle{plain}
\newtheorem{theorem}{Theorem}
\newtheorem{lemma}[theorem]{Lemma}
\newtheorem{proposition}[theorem]{Proposition}

\theoremstyle{definition}

\begin{document}

\title{Quantum Computational Resources and Conformal Field Theory: \\ Unifying Spins, Bosons, and Fermions}

\author{Ryota Matsuda}
\email{r-matsuda918@g.ecc.u-tokyo.ac.jp}
\affiliation{Department of Physics, The University of Tokyo, 7-3-1 Hongo, Bunkyo-ku, Tokyo 113-0033, Japan}

\author{Masahiro Hoshino}
\email{hoshino-masahiro921@g.ecc.u-tokyo.ac.jp}
\affiliation{Department of Physics, The University of Tokyo, 7-3-1 Hongo, Bunkyo-ku, Tokyo 113-0033, Japan}

\author{Masaki Oshikawa}
\email{oshikawa.1@osu.edu}
\affiliation{Department of Physics, The Ohio State University, Columbus Ohio 43210, USA}
\affiliation{Institute for Solid State Physics, The University of Tokyo, Kashiwa, Chiba 277-8581, Japan}

\author{Yuto Ashida}
\email{ashida@phys.s.u-tokyo.ac.jp}
\affiliation{Department of Physics, The University of Tokyo, 7-3-1 Hongo, Bunkyo-ku, Tokyo 113-0033, Japan}
\affiliation{Institute for Physics of Intelligence, The University of Tokyo, 7-3-1 Hongo, Bunkyo-ku, Tokyo 113-0033, Japan}

\date{\today}

\begin{abstract}
  Characterizing a quantum state through the lens of quantum resources provides an information-theoretic perspective on many-body systems. While quantum entanglement serves as the paradigmatic example of a quantum resource, recent studies have shown that quantum magic, a resource for universal quantum computation, can capture aspects of many-body states complementary to those described by entanglement.
  In spin systems, conformal field theory (CFT) analysis of the stabilizer R\'enyi entropy has revealed universal features of nonstabilizerness distinct from entanglement, but a comparable formulation for bosonic and fermionic non-Gaussianity has yet to be established. In this work, we introduce a unified measure, the magic R\'enyi entropy (MRE), to quantify computational resources in spins, bosons, and fermions on an equal footing.
  The MRE is shown to be a resource monotone under stabilizer and Gaussian protocols that can involve measurements and feedforward operations.
  In particular, the fermionic MRE gives a robust exponential lower bound on the Gaussian rank, providing an accessible diagnostic for classical simulability based on Gaussian decompositions.
  The MRE allows us to reveal common universal aspects of nonstabilizerness and non-Gaussianity in critical many-body states. In particular, our CFT analysis shows that the universal contribution to the MRE appears as the size-independent term determined by the Affleck-Ludwig boundary entropy.
  We find that bulk-induced boundary renormalization-group flows can drive a transition in this universal contribution within a single bulk critical phase.
  As a concrete demonstration, we present a detailed CFT analysis of non-Gaussianity in interacting spinless fermions and nonstabilizerness in the XXZ chain, which are both described by the Tomonaga-Luttinger liquid.
  These two resources are encoded in different boundary states and undergo transitions at different Luttinger parameters, $K=3$ and $K=1$, respectively.
  Our field-theoretical predictions are confirmed numerically by exact diagonalization and tensor-network methods.
  These results provide a unified understanding of many-body magic across spins, bosons, and fermions.

\end{abstract}

\maketitle


\section{Introduction}

\subsection{\label{subsec:background}Background and motivation}

Quantum resource theory~\cite{chitambar2019quantum} provides a way of characterizing information-theoretic properties of many-body states that cannot be captured by conventional local observables.
Entanglement is the paradigmatic example of such a resource~\cite{amico2008entanglement,eisert2010colloquium}, and it has been used to characterize topological phases of matter~\cite{chen2010local,zeng2019quantum}, exotic nonequilibrium phases~\cite{skinner2019measurementinduced}, and even black-hole thermodynamics through holography~\cite{ryu2006holographic}.
In particular, it has been known that a one-dimensional critical state exhibits entanglement entropy that scales logarithmically with subsystem size.
This universal scaling is governed by the central charge of the underlying conformal field theory (CFT)~\cite{osterloh2002scaling,osborne2002entanglement,vidal2003entanglement,holzhey1994geometric,pasqualecalabrese2004entanglement,calabrese2009entanglement,cardy2008form,castroalvaredo2009bipartite}.
Importantly, entanglement entropy is defined without reference to a particular choice of microscopic degrees of freedom and has been applied to spins~\cite{osborne2002entanglement,vidal2003entanglement}, bosons~\cite{frerot2016entanglement,kim2024entanglement}, and fermions~\cite{wolf2006violation,gioev2006entanglement} on an equal footing.

Besides entanglement, quantum magic has recently attracted much attention as a resource associated with universal quantum computation and also as a complementary probe of many-body phenomena~\cite{liu2022manybody}.
Magic quantifies the departure of a given state from classically simulable states, e.g., stabilizer states for spins~\cite{gottesman1998heisenberg} and Gaussian states for bosons~\cite{bartlett2002efficient} and fermions~\cite{valiant2002quantum,terhal2002classical,bravyi2004lagrangian}.
Thus, nonstabilizerness in spin systems and non-Gaussianity in bosonic or fermionic systems are necessary resources for achieving quantum computational advantage beyond the corresponding efficiently simulable models~\cite{bravyi2005universal,lloyd1999quantum,hebenstreit2019all}.
In studies of nonstabilizerness in many-body systems, the stabilizer R\'enyi entropy (SRE)~\cite{leone2022stabilizer} has become a useful measure because it is amenable to both numerical computation~\cite{tarabunga2023manybody,tarabunga2024nonstabilizerness,lami2023nonstabilizerness} and analytical treatment~\cite{iannotti2025entanglement,montanalopez2024exact,trino2026stabilizershannon}.
For instance, the SRE has been used to characterize phase transitions~\cite{niroula2024phase,moca2025nonstabilizerness}, chaotic systems including Sachdev-Ye-Kitaev (SYK) models~\cite{goto2022probing,bera2025nonstabilizerness,jasser2025stabilizer,zhang2026stabilizer}, and dynamical phenomena such as magic spreading and anticoncentration~\cite{turkeshi2025magic,tirrito2025anticoncentration}. Moreover, its critical behavior has been studied using CFT methods~\cite{white2021conformal,hoshino2026stabilizer,hoshino2026stabilizera}.
Recent work has also investigated non-Gaussianity in many-body systems by introducing measures such as fermionic antiflatness~\cite{sierant2026fermionic} and extending the analysis to hybrid systems~\cite{crew2026magic,sarkis2025magic}, lattice gauge theories~\cite{santra2025quantum}, random states~\cite{ares2026nongaussianity}, and disordered systems~\cite{falcao2026fermionic}.

While the specific definition of magic depends on a choice of the microscopic degrees of freedom, nonstabilizerness and non-Gaussianity share closely related mathematical formulations.
A first indication comes from Hudson's theorem and its discrete analogue: in suitable quasiprobability representations, both bosonic Gaussian pure states and stabilizer states are characterized by the same absence of negativity~\cite{hudson1974when,soto1983when,gross2006hudsons}.
Thus, stabilizer states can be regarded as discrete counterparts of bosonic Gaussian states.
Recent frameworks have made this connection more concrete from different perspectives.
For example, quadratic tensors can describe stabilizer and Gaussian states within a common algebraic formalism~\cite{bauer2026quadratic}, and a diagrammatic framework for Clifford and matchgate circuits has been proposed~\cite{kang20252d}.
The group-Fourier decomposition also provides common resource witnesses for nonstabilizerness and fermionic non-Gaussianity~\cite{bermejo2025characterizing}.

However, a unified framework for quantifying magic in spins, bosons, and fermions, particularly in many-body settings, is still lacking.
For instance, quadratic tensors and the diagrammatic framework concern the resources of tensor-network or diagrammatic representations, rather than those of the states themselves.
The group-Fourier decomposition requires a compact group of free unitaries on a finite-dimensional Hilbert space and thus does not apply to bosons; for fermions, its resource witness depends only on the covariance matrix and cannot capture higher-order correlations.
In contrast, we recall that CFT describes entanglement entropy at criticality largely independently of whether the microscopic system consists of spins, bosons, or fermions. A unified understanding of many-body magic thus requires a quantitative and tractable common language for nonstabilizerness and non-Gaussianity. These considerations lead to two questions:

\begin{itemize}
  \item[(A)]{Is there a computable magic monotone that applies to spins, bosons, and fermions in a unified manner?}
  \item[(B)]{How does such a measure behave in quantum many-body systems? Does it contain universal data at criticality, and how can it be analyzed by a field theory?}
\end{itemize}

To answer question (A), we introduce the magic R\'enyi entropy (MRE) as a resource measure of nonstabilizerness and non-Gaussianity defined by a common principle, namely the convolution-multiplication duality.
Our construction uses quantum convolution~\cite{bu2025stabilizer,bu2025efficient,lyu2024fermionic,coffman2025measuring}, which naturally arises in free-state testing.
Just as classical convolution drives independent random variables toward a Gaussian distribution in the central limit theorem, quantum convolution drives quantum states toward stabilizer or Gaussian states~\cite{cushen1971quantummechanical,bu2023quantum,bu2023discrete,mehrabi2026central}. Since pure stabilizer and Gaussian states remain pure under convolution while the other pure states become mixed, the purity loss caused by convolution quantifies magic.
Using this duality, we establish monotonicity of the MRE under stabilizer and Gaussian protocols, including measurements and feedforward, for deterministic transformations between pure states.
We further prove that, within the convolution families considered, monotonicity for bosons and qudits singles out the convolutions satisfying the duality.
This requirement uniquely fixes the resulting MRE.
These results establish the MRE as a common quantitative basis for studying many-body magic across different microscopic degrees of freedom.

Furthermore, we formulate the MRE within a field-theoretical framework and show that its universal contribution at criticality is determined by the infrared data of the convolution-induced boundary condition.
This addresses question (B) by rewriting the MRE as a partition function of the replicated theory in the Euclidean path-integral formalism.
In a one-dimensional critical system described by a (1$+$1)-dimensional CFT, this boundary condition flows to an infrared conformal boundary condition under boundary renormalization-group (RG) flows.
Consequently, the MRE acquires a universal size-independent term characterized by the Affleck-Ludwig boundary entropy, or equivalently, the $g$ factor of the infrared boundary state~\cite{affleck1991universal}.

To validate these general results, we apply the framework to studying different computational resources through both analytical and numerical calculations.
Specifically, we compare the fermionic MRE of an interacting spinless-fermion chain and the SRE of the XXZ chain.
These chains are related by the Jordan-Wigner transformation and share the same low-energy description known as the Tomonaga-Luttinger liquid (TLL)~\cite{haldane1981luttinger,giamarchi2003quantum}.
The two computational resources, however, are encoded in distinct boundary states.
We show that bulk interactions induce boundary RG flows that drive transitions in the universal contributions within a single bulk critical phase.
These transitions occur at different Luttinger parameters: $K=3$ for the fermionic MRE and $K=1$ for the SRE.
We confirm these field-theoretical predictions numerically by exact-diagonalization and perfect-sampling calculations~\cite{lami2023nonstabilizerness,matsuda2026efficient}.
Thus, systems described by the same TLL can exhibit different universal behavior of many-body magic, reflecting the distinct boundary states associated with the computational resources of their microscopic degrees of freedom.

\subsection{\label{subsec:overview}Overview of the key results}

\begin{figure*}[t]
  \centering
  \includegraphics[width=\linewidth,clip]{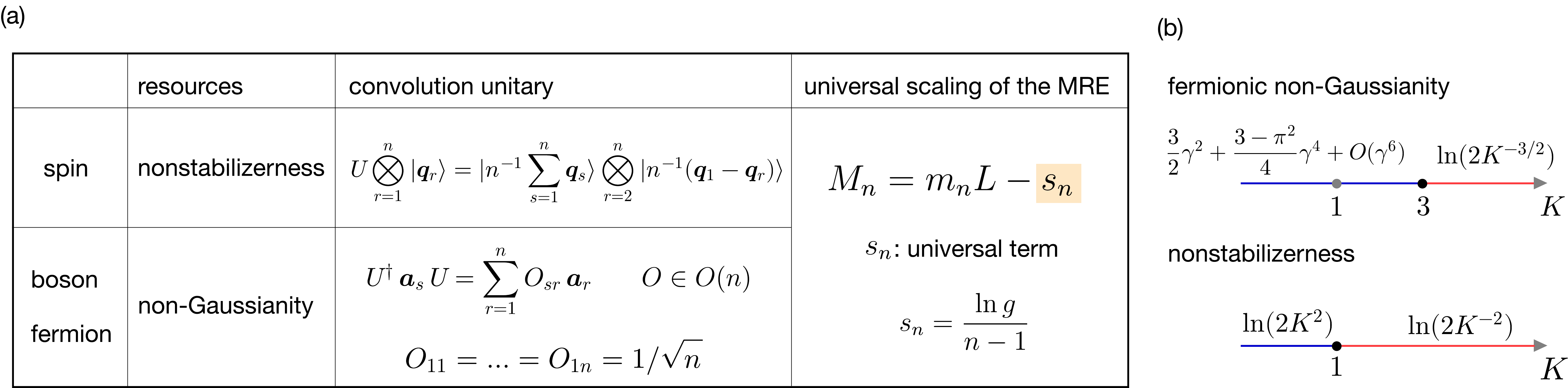}
  \caption{Summary of the main results.
    (a) Unified description of quantum computational resources for spins, bosons, and fermions.
The order-$n$ MRE $M_n$ is defined through a convolution unitary in all three cases.
In spins with local Hilbert-space dimension $d$, it reduces to the SRE and quantifies nonstabilizerness; for bosons and fermions, it quantifies non-Gaussianity.
In one-dimensional critical systems, the MRE can exhibit a common scaling form with a size-independent term $s_n$ determined by the $g$ factor of the infrared boundary condition in the replicated theory.
(b) Distinct universal behavior of fermionic non-Gaussianity and spin nonstabilizerness in systems described by the same TLL at a given Luttinger parameter $K$.
The upper and lower diagrams show the universal contribution $s_2$ to the second-order fermionic MRE of an interacting spinless-fermion chain and to the second-order SRE of the XXZ chain, respectively.
Here, $\gamma=(K-1)/(K+1)$, and $K=1$ is the free-fermion point.
Boundary CFT predicts both the boundary entropy $s_2$ and its transitions at $K=3$ for the fermionic MRE and $K=1$ for the SRE.
The fermionic result for $K<3$ is obtained perturbatively near $K=1$; the other expressions are nonperturbative within their respective ranges of $K$.
  }
  \label{fig:overview}
\end{figure*}

Before going into details, we present a nontechnical summary of the main results.
Our central goal is to construct a measure of many-body magic that treats spins, bosons, and fermions in a common language, and to identify the universal behavior of this measure through a field-theoretical analysis (see Fig.~\ref{fig:overview}(a)).
Here, stabilizer states are the free states for spins, whereas Gaussian states are the free states for bosons and fermions.
Our construction keeps this distinction explicit, but packages it into a common framework so that the subsequent field-theoretical analysis can be carried out in a unified manner.

The first main result, which addresses question (A), is the construction of the MRE as a convolution-based measure of magic.
Specifically, given $n$ copies of a pure state $\rho$, we apply the convolution unitary $U$ that mixes the $n$ copies, keep one output replica, and trace out the other $n-1$ output replicas:
\begin{align}
  \mathcal{C}_{n}(\rho)
  =\tr_{1^{\rm c}}\!\left[U\rho^{\otimes n}U^{\dag}\right],
\end{align}
where $1^{\rm c}$ represents the remaining replicas $\{2,3,\ldots,n\}$.
The MRE is defined as the logarithm of the purity of this output state,
\begin{align}
  M_{n}(\rho)
  =\frac{1}{1-n}\ln\tr\!\left[\mathcal{C}_{n}(\rho)^{2}\right].
\end{align}
The convolution unitary is chosen so that a free input state remains pure after convolution, whereas the output of a resourceful state becomes mixed.

Crucially, we show that the following convolution-multiplication duality is central to establishing monotonicity, 
\begin{align}
  \chi_{\mathcal{C}_{n}(\rho_{1},\ldots,\rho_{n})}(\bm{u})\propto\prod_{r=1}^{n}\chi_{\rho_{r}}(\bm{u}'),
  \label{eq:duality_overview}
\end{align}
where $\mathcal{C}_{n}(\rho_{1},\ldots,\rho_{n})$ denotes the convolution of possibly distinct inputs $\rho_r$, $\chi_{\rho}$ is the characteristic function of $\rho$, and $\bm{u}'$ is linear in $\bm{u}$.
As in the Fourier convolution theorem, the output characteristic function is the product of the input characteristic functions up to a phase factor.
All factors are evaluated at the same phase-space point $\bm{u}'$, and the convolution treats all replicas equally. 

We construct explicit convolution unitaries satisfying Eq.~\eqref{eq:duality_overview} as follows. For spins with local Hilbert-space dimension $d$, a unitary satisfying the duality exists if and only if $n$ and $d$ are coprime.
One example is the permutation of the computational basis
\begin{align}
  U\bigotimes_{r=1}^{n}\ket{\bm{q}_r}
  =\ket{n^{-1}\sum_{s=1}^{n}\bm{q}_s}\bigotimes_{r=2}^{n}\ket{n^{-1}(\bm{q}_1-\bm{q}_r)},
\end{align}
where $\bm{q}_r\in\mathbb{Z}_d^{L}$ labels the computational basis of the $L$ sites in replica $r$, and $n^{-1}$ is the multiplicative inverse of $n$ modulo $d$.
For $\bm{u}=(\bm{q},\bm{p})$ with $\bm{q},\bm{p}\in\mathbb{Z}_d^{L}$, this gives $\bm{u}'=(\bm{q},n^{-1}\bm{p})$.
In bosons and fermions, a natural choice is a Gaussian unitary implementing the Helmert rotation.
The retained output modes are then the normalized center-of-mass combinations of the $n$ input replicas,
\begin{align}
  U^{\dag}\bm{a}_{1}U=\frac{1}{\sqrt{n}}\sum_{r=1}^{n}\bm{a}_{r},
\end{align}
where $\bm{a}_{r}$ collects the annihilation operators of the $L$ modes in replica $r$, leading to $\bm{u}'=\bm{u}/\sqrt{n}$.
All convolution unitaries satisfying the duality yield the same MRE with the following closed forms
\begin{align}
  M_{n}(\rho)=\frac{1}{1-n}
  \begin{cases}
    \displaystyle\ln\qty[\frac{1}{d^{L}}\sum_{\bm{u}}\qty|\chi_{\rho}(\bm{u})|^{2n}]              & \text{spins},    \\[12pt]
    \displaystyle\ln\qty[n^{L}\int\!\!\frac{d\bm{u}}{(2\pi)^{L}}\qty|\chi_{\rho}(\bm{u})|^{2n}]   & \text{bosons},   \\[12pt]
    \displaystyle\ln\qty[\frac{1}{2^{L}}\norm{\chi_{\rho}\qty(\frac{2\bm{u}}{\sqrt{n}})^{n}}^{2}] & \text{fermions},
  \end{cases}
\end{align}
where $\bm{u}$ is a vector of Grassmann variables for fermions, and $\norm{\cdot}$ is the coefficient norm of a Grassmann polynomial defined in Sec.~\ref{subsec:MRE_BF}.
The spin expression coincides with the SRE~(cf.~Eq.~\eqref{eq:qudit_SRE}), while the bosonic and fermionic expressions are its direct analogues.

We prove monotonicity of this MRE under deterministic transformations between pure states implemented by stabilizer or Gaussian protocols. The result holds for spins and bosons whenever the MRE is defined, and for fermions at $n=2$. Importantly, these protocols include measurements and feedforward, which are essential for implementing a universal gate set using magic states~\cite{bravyi2005universal,gottesman2001encoding,hebenstreit2019all}.
We also point out that the monotonicity significantly restricts the choice of the convolution unitary; for spins and bosons, the only allowed choices are those satisfying the duality~\eqref{eq:duality_overview}, up to the identification of convolution unitaries that give the same MRE.
These results establish the MRE as a unified resource monotone for spins, bosons, and fermions.

Having established monotonicity, we next relate the MRE to the cost of representing many-body states as superpositions of free states.
In fermions, this cost is characterized by the fermionic Gaussian rank~\cite{dias2024classical}, defined as the minimum number of pure Gaussian states whose superposition reproduces a given state.
For a pure state $\ket{\psi}$ of $L$ fermionic modes, we prove that the approximate Gaussian rank $\chi_{\delta}(\psi)$ at fixed error $\delta<1$ satisfies
\begin{align}
  \chi_{\delta}(\psi)\geq\exp\qty[\frac{M_{2}(\psi)}{4}-c_{*}L+O(\ln L)],
\end{align}
where $c_{*}=0.0615\ldots$ is a constant.
Thus, if the MRE density $M_{2}(\psi)/L$ remains above $4c_{*}=0.246\ldots$ by a finite margin as $L\to\infty$, even approximate representations require exponentially many Gaussian components.
This property also distinguishes the MRE from other computable resource measures constructed from the covariance matrix~\cite{sierant2026fermionic,tarabunga2026computable,leone2026fermionic,turkeshi2026williamson,bermejo2025characterizing}, as there exists a pair of states with identical covariance matrices but with exponentially different fermionic Gaussian rank.
These results connect convolution-based resource measures to the cost of classical simulation methods based on Gaussian superpositions.

The second key result, which addresses question (B), is a theoretical understanding of universal aspects of the MRE in many-body systems.
The crucial observation is that, when the convolution is written in replicated Liouville space, the MRE can be represented as the Euclidean path integral with a defect localized at an imaginary-time slice.
This defect encodes the convolution unitary, the trace over the discarded replicas, and the final purity measurement, which together define a single unnormalized projection operator $P_n$ acting on the replicas.
The path integral with this defect insertion leads to
\begin{align}
  M_n(\rho)=\frac{1}{1-n}\ln\frac{Z_n}{Z^{2n}},
  \label{eq:MRE_Z_overview}
\end{align}
where $Z_n$ is the replicated partition function with the insertion $P_n$, and $Z^{2n}$ is the reference normalization by decoupled replicas with $Z$ being the partition function of a single replica.
After folding the path integral across this defect, the geometry becomes a cylinder with a boundary condition imposed by the convolution operation~\cite{oshikawa1996defect,oshikawa1997boundary}.
When the given system is described by a CFT in 1$+$1 dimensions, the boundary CFT techniques show that the MRE for a large system size $L$ takes the form
\begin{align}
  M_{n}=m_{n}L-s_{n}+o(1),
  \qquad
  s_{n}=\frac{\ln g}{n-1}.
  \label{eq:Mn_overview}
\end{align}
Here, $m_n$ is a nonuniversal line contribution from short-distance physics, and the size-independent term $s_n$ is a universal contribution determined by the Affleck-Ludwig boundary entropy, or equivalently, the logarithm of the $g$ factor of the infrared boundary state.

As an example, let us consider the second-order ($n=2$) MRE of the half-filled interacting spinless-fermion chain [Fig.~\ref{fig:overview}(b)].
The low-energy theory of the system is a TLL with parameter $K\geq1/2$, which encodes the interaction strength.
The free-fermion point is $K=1$, and deviations from this value indicate non-Gaussianity.
Therefore, at the free-fermion point $K=1$, the MRE vanishes, and in particular $s^{\text{fermion}}_2=0$.

Introducing the interaction between fermions changes the Luttinger parameter $K$ away from $1$, resulting in a nontrivial MRE.
Since the TLL is a free boson field theory for any value of $K$, one might expect that the universal contribution can be straightforwardly calculated.
However, the difficulty arises because the boundary condition imposed by the convolution unitary is not simple (such as Dirichlet or Neumann) in terms of bosons.
In general, conformally invariant boundary conditions of multiple TLLs can be beyond the free boson description~\cite{affleck2001quantum}.
Nevertheless, we can treat some of the conformally invariant boundary conditions relevant for the MRE exactly, or in terms of perturbation theory.

We develop a formalism in which the convolution unitary acts on the bulk Hamiltonian instead of on the boundary.
In this formalism, the boundary condition is imposed only by the trace and purity operations, while the difference between the rotated and original bulk Hamiltonians acts as a bulk perturbation.
Near the boundary, this perturbation generates boundary operators through the bulk-boundary operator product expansion (OPE)~\cite{fredenhagen2007bulkinduced}.
If the OPE generates only the identity and irrelevant operators, the boundary remains stable, while the renormalized $g$ factor and hence the universal term vary continuously.
A nonidentity relevant or marginal operator can instead drive the boundary to a different infrared fixed point.

We calculate the universal Affleck-Ludwig boundary entropy for the conformally invariant boundary condition smoothly connected to that at the free-fermion point $K=1$, to find
\begin{align}
  s^{\text{fermion}}_{2} & =\frac{3}{2}\gamma^{2}+\frac{3-\pi^{2}}{4}\gamma^{4}+O(\gamma^{6}),
  \label{eq:sfermi2_gamma}
  \\
  \gamma                 & =\frac{K-1}{K+1}.
\end{align}
The stability of this boundary condition is determined by the scaling dimension of the leading induced nonidentity boundary operator,
\begin{align}
  h_{\mathrm{min}}(K)=\frac{4\min(K,1)}{K+1}.
\end{align}
This boundary operator is irrelevant ($h_{\mathrm{min}}(K)>1$) for $1/2\leq K<3$; the above boundary condition thus remains stable and Eq.~\eqref{eq:sfermi2_gamma} is valid in this range of $K$.
However, for $K>3$, the boundary operator becomes relevant ($h_{\mathrm{min}}<1$) and drives the boundary to a different infrared fixed point.
We also identify this fixed point and obtain its boundary entropy as
\begin{align}
  s^{\text{fermion}}_{2}=\ln\qty(2K^{-3/2}).
\end{align}
Thus, the universal contribution undergoes a transition at $K=3$ within a single bulk critical phase.

Furthermore, our boundary RG analysis singles out the balanced mixing of the Helmert rotation among two-replica rotations.
Let $\zeta$ be the mixing angle of such a rotation, so that $\zeta=\pi/4$ gives the Helmert rotation and $\zeta=0$ gives the identity.
Then the balanced mixing at $\zeta=\pi/4$ is the only choice which gives a nonzero universal contribution in the infrared limit.
Specifically, using current algebra, we find that a small imbalance $\delta\zeta=\zeta-\pi/4$ couples to a boundary operator with scaling dimension $1-\gamma^{2}$ at leading order, giving the boundary RG flow
\begin{align}
  \frac{d\,\delta\zeta}{dl}=\gamma^{2}\delta\zeta,
\end{align}
where $l$ is the logarithmic RG scale, and even a small imbalance grows at long distances.
In fact, for $0\leq\zeta<\pi/4$, the boundary flows to that of the trivial convolution at $\zeta=0$, in which the convolution unitary is the identity and the MRE itself vanishes.

We also demonstrate how a different computational resource of the same TLL can be analyzed by our unified framework.
To this end, we consider the second-order SRE of the $S=1/2$ XXZ chain, which is related to the spinless-fermion chain by the Jordan-Wigner transformation.
The SRE has the same form of representation in terms of $P_2$ as the MRE, allowing us to apply the same boundary framework.
We treat the bulk interaction as a perturbation about the $K=1$ boundary state, constructed from two copies of the known boundary state for the critical Ising chain~\cite{hoshino2026stabilizer}.
The perturbation generates a marginal boundary operator at first order, and its coupling grows under the boundary RG flow with a sign set by $K-1$.
The resulting flows reach different infrared fixed points for $K>1$ and $K<1$, yielding
\begin{align}
  s^{\text{spin}}_{2}=
  \begin{cases}
    \ln\qty(2K^{2}),  & K\leq1, \\
    \ln\qty(2K^{-2}), & K\geq1.
  \end{cases}
\end{align}
The spin SRE thus undergoes a transition at $K=1$, making a contrast to $K=3$ for the fermionic MRE.
This difference originates in the boundary states associated with the two resources, which determine the scaling dimensions of the induced boundary operators.

Altogether, our convolution construction and boundary formulation together provide a common framework for resource monotonicity and the critical scaling in Eq.~\eqref{eq:Mn_overview} across spins, bosons, and fermions.
At the same time, different resource definitions lead to different boundary states, explaining why systems sharing the same low-energy description can nevertheless support distinct universal behavior of many-body magic.

\subsection{Organization of this paper}

The remainder of our paper is organized as follows.
Section~\ref{sec:preliminaries} reviews the notions of magic for spins, bosons, and fermions in a parallel manner and sets up the free-state structures used throughout the paper.
Section~\ref{sec:MRE} defines the MRE for each microscopic degree of freedom, and discusses its resource-theoretic properties, including monotonicity under stabilizer and Gaussian protocols.
It also proves the robust exponential lower bound on the fermionic Gaussian rank by the MRE.
Section~\ref{sec:field_theory_MRE} develops the field-theoretical formulation of the MRE by rewriting the MRE as a replicated Liouville-space path integral.
Section~\ref{sec:conformal_invariance} introduces the framework of bulk-induced boundary RG flow, and derives the boundary RG equation, which serves as a tool to analyze boundary transitions.
Section~\ref{sec:fermion_TLL} applies the framework to the MRE of the TLL realized in an interacting spinless-fermion chain.
We derive the universal contribution near the free-fermion point, and examine boundary stability along the mixing angle and the TLL parameter.
We then derive the universal contribution for $K>3$ and compare the predictions with numerical calculations.
Section~\ref{sec:spin_XXZ} discusses the SRE of the XXZ chain to show how the universal contribution depends on the computational resource even for the same TLL.
Finally, Sec.~\ref{sec:discussions} summarizes the results and discusses open problems.

\section{\label{sec:preliminaries}Preliminaries}

In this section, we first summarize the key concepts necessary to formulate our results. Section~\ref{subsec:qubit} focuses on spins, while Sec.~\ref{subsec:boson-fermion} treats bosons and fermions using a unified notation. Crucially, the magic structures in these three systems can be discussed in a parallel manner, and this structural similarity naturally leads to our unified field-theoretical description in the subsequent sections.

\subsection{\label{subsec:qubit}Spins}

For simplicity, we first describe qubits and then extend the definitions to qudits.
We define the single-qubit Pauli operators by
\begin{align}
  \sigma(q,p)
  = i^{qp}X^{q}Z^{p},
  \label{eq:Pauli_single}
\end{align}
where $q,p\in\{0,1\}$, and $X$ and $Z$ are the usual Pauli operators.
Writing out the four Pauli operators explicitly gives
\begin{align}
  \sigma(0,0)=I,\;
  \sigma(0,1)=Z,\;
  \sigma(1,0)=X,\;
  \sigma(1,1)=Y.
  \label{eq:Pauli_explicit}
\end{align}
For an $L$-qubit system, we denote the indices as $\bm{u}=(\bm{q},\bm{p})\in\{0,1\}^{2L}$ with $\bm{q}=(q_{1},q_{2},\ldots,q_{L})$ and $\bm{p}=(p_{1},p_{2},\ldots,p_{L})$, and define the corresponding Pauli string by
\begin{align}
  \sigma(\bm{u}):=\bigotimes_{i=1}^{L}\sigma(q_{i},p_{i}) = i^{\bm{q}\cdot\bm{p}}\bigotimes_{i=1}^{L}X_i^{q_i}Z_i^{p_i}.
  \label{eq:Pauli_string}
\end{align}
Here, $\bm{q}\cdot\bm{p}=\sum_{i=1}^{L}q_{i}p_{i}$ is the inner product of bit strings.
The characteristic function of a state $\rho$ is given by
\begin{align}
  \chi_{\rho}(\bm{u})
  = \tr[\sigma(\bm{u})\rho].
  \label{eq:chi_qubit}
\end{align}
The Pauli strings form an orthogonal operator basis with respect to the Hilbert-Schmidt inner product:
\begin{align}
  \tr[\sigma(\bm{u})\sigma(\bm{v})]
  = 2^{L}\delta_{\bm{u},\bm{v}}.
  \label{eq:Pauli_orthogonality}
\end{align}
Using this orthogonality and the characteristic function, any state can be expanded in the Pauli basis as
\begin{equation}\label{eq:rho_expansion_qubit}
  \rho = \frac{1}{2^{L}}\sum_{\bm{u}}\chi_{\rho}(\bm{u})\sigma(\bm{u}).
\end{equation}

The free states of the magic resource theory for qubits are the stabilizer states, which are defined as the Clifford orbit of the computational basis states.
We first define the $L$-qubit Pauli group as
\begin{equation}\label{eq:Pauli_group_qubit}
  \mathcal{P}_{L} = \qty{c\,\sigma(\bm{u})\mid c\in\{\pm 1,\pm i\},\,\bm{u}\in\{0,1\}^{2L}},
\end{equation}
and the Clifford group as the set of unitaries that leave $\mathcal{P}_{L}$ invariant under conjugation,
\begin{equation}\label{eq:Clifford_group_qubit}
  \mathrm{Cl}_{L} = \qty{U\mid UPU^{\dag}\in\mathcal{P}_{L} \text{ for all } P\in\mathcal{P}_{L}}.
\end{equation}
The set of stabilizer states is then given by the Clifford orbit of the computational basis,
\begin{equation}\label{eq:Stab_states_qubit}
  \mathrm{Stab}_{L} = \qty{U_{C}\ket{0}^{\otimes L}\mid U_{C}\in\mathrm{Cl}_{L}}.
\end{equation}
According to the Gottesman-Knill theorem, any quantum computation that starts with a stabilizer state and consists of Clifford unitaries and computational-basis measurements can be efficiently simulated on a classical computer.
Thus, in order to achieve quantum computational advantage, one needs to use nonstabilizer states, which are often called magic states.

The stabilizer R\'enyi entropy (SRE) is a widely used measure of magic for many-body spin systems.
For qubits, it is defined as the classical R\'enyi entropy of the probability distribution defined by the squared characteristic function normalized by $2^{L}$:
\begin{equation}
  \mathrm{SRE}_{n}(\rho)
  = \frac{1}{1-n}\ln\left[\frac{1}{2^{L}}\sum_{\bm{u}}\left[\chi_{\rho}(\bm{u})\right]^{2n}\right].
  \label{eq:SRE}
\end{equation}
Here, $n=2,3,\ldots$ is the R\'enyi index.
For a pure state $\rho$, the SRE satisfies the following basic properties:
\begin{enumerate}
  \item Faithfulness:
        $\mathrm{SRE}_{n}(\rho)\ge 0$, with equality holding if and only if $\rho\in\mathrm{Stab}_{L}$.
  \item Invariance under Clifford unitaries:
        $\mathrm{SRE}_{n}(U_{C}\rho U_{C}^{\dag})=\mathrm{SRE}_{n}(\rho)$ for any $U_{C}\in\mathrm{Cl}_{L}$.
  \item Additivity:
        $\mathrm{SRE}_{n}(\rho_{1}\otimes\rho_{2})=\mathrm{SRE}_{n}(\rho_{1})+\mathrm{SRE}_{n}(\rho_{2})$.
\end{enumerate}
Moreover, $\mathrm{SRE}_{n}$ with $n\geq2$ satisfies monotonicity under stabilizer protocols~\cite{leone2024stabilizer}.

Another related quantity is the linear stabilizer entropy,
\begin{equation}
  \mathrm{SRE}_{n}^{\rm lin}(\rho)
  = 1 - \frac{1}{2^{L}}\sum_{\bm{u}}\left[\chi_{\rho}(\bm{u})\right]^{2n},
\end{equation}
which can be also expressed as $\mathrm{SRE}_{n}^{\rm lin}(\rho)=1-\exp\qty[(1-n)\mathrm{SRE}_{n}(\rho)]$.
While the linear SRE does not satisfy the additivity, it is a strong monotone under stabilizer protocols for $n\geq 2$~\cite{leone2024stabilizer}.

We next extend these definitions to magic resource theories of qudits with local Hilbert-space dimension $d\geq 2$.
We define $\omega_{d}=e^{2\pi i/d}$ as a primitive $d$-th root of unity.
The single-qudit shift and clock operators are defined by
\begin{align}
  X=\sum_{q=0}^{d-1}\ket{q+1}\bra{q},\qquad
  Z=\sum_{q=0}^{d-1}\omega_{d}^{\,q}\ket{q}\bra{q},
  \label{eq:qudit_XZ}
\end{align}
where the labels are understood modulo $d$.
The Heisenberg-Weyl operators generalize the Pauli operators~\eqref{eq:Pauli_single} and are defined by
\begin{align}
  D(q,p)=\tau^{\,qp}X^{q}Z^{p},\qquad
  \tau=\begin{cases} -e^{\pi i/d} & (d\ \text{odd})\\ e^{\pi i/d} & (d\ \text{even})\end{cases}
  \label{eq:qudit_HW}
\end{align}
with $q,p\in\mathbb{Z}_{d}$.
For an $L$-qudit system we write $\bm{u}=(\bm{q},\bm{p})$ with $\bm{q},\bm{p}\in\mathbb{Z}_{d}^{L}$ and define the $L$-qudit Heisenberg-Weyl operator
\begin{align}
  D(\bm{u})=\bigotimes_{i=1}^{L}D(q_{i},p_{i}).
  \label{eq:qudit_HW_string}
\end{align}
The characteristic function of a state $\rho$ is
\begin{align}
  \chi_{\rho}(\bm{u})=\tr[D(\bm{u})\rho],
  \label{eq:qudit_chi}
\end{align}
and the strings $D(\bm{u})$ form an orthogonal operator basis, $\tr[D(\bm{u})^{\dag}D(\bm{v})]=d^{L}\delta_{\bm{u},\bm{v}}$.
The $L$-qudit Heisenberg-Weyl group is defined as
\begin{align}
  \mathcal{D}_{L}=\qty{\tau^{\,j}\,D(\bm{u})\mid
                    j\in\mathbb{Z}_{2d},\,\bm{u}\in\mathbb{Z}_{d}^{2L}},
  \label{eq:qudit_HW_group}
\end{align}
which generalizes the Pauli group~\eqref{eq:Pauli_group_qubit}.
The Heisenberg-Weyl operators obey the commutation relation
\begin{equation}
  D(\bm{u})D(\bm{v})
  =\omega_{d}^{-\langle\bm{u},\bm{v}\rangle}\,D(\bm{v})D(\bm{u}),
  \qquad
  \langle\bm{u},\bm{v}\rangle\equiv\bm{u}^{\rm T}\Omega\,\bm{v},
  \label{eq:HW_commutation_qudit}
\end{equation}
where
\begin{equation}
  \Omega=\begin{pmatrix}0 & I_{L}\\ -I_{L} & 0\end{pmatrix}
\end{equation}
is the symplectic matrix and $\langle\cdot,\cdot\rangle$ is the induced symplectic form on $\mathbb{Z}_{d}^{2L}$.

The free states of the qudit magic resource theory are again the stabilizer states.
The Clifford group is the set of unitaries leaving $\mathcal{D}_{L}$ invariant under conjugation,
\begin{align}
  \mathrm{Cl}_{L}=\qty{U\mid
                    UDU^{\dag}\in\mathcal{D}_{L}
                    \text{ for all }D\in\mathcal{D}_{L}}.
  \label{eq:qudit_Clifford}
\end{align}
A pure state is a stabilizer state if it is the common eigenvector with eigenvalue one of an abelian subgroup of $\mathcal{D}_{L}$ with $d^{L}$ elements that contains no nontrivial multiple of the identity~\cite{gross2006hudsons,hahn2025bridging}.
We denote the set of these states by $\mathrm{Stab}_{L}$.
For prime $d$, the stabilizer states are the Clifford orbit of the computational basis~\cite{gross2006hudsons},
\begin{align}
  \mathrm{Stab}_{L}=\qty{U_{C}\ket{0}^{\otimes L}\mid
                      U_{C}\in\mathrm{Cl}_{L}},
  \label{eq:qudit_Stab}
\end{align}
which generalizes Eq.~\eqref{eq:Stab_states_qubit}.
For composite $d$, $\mathrm{Stab}_{L}$ also contains states outside this orbit.
For example, the state $(\ket{0}+\ket{3}+\ket{6})/\sqrt{3}$ for $d=9$ is stabilized by the group generated by $X^{3}$ and $Z^{3}$, whose elements all cube to the identity.
It is not a Clifford image of $\ket{0}$, which is stabilized by $Z$ with $Z^{3}\neq I$.

The stabilizer R\'enyi entropy has been generalized to qudits as~\cite{wang2023stabilizer}
\begin{align}
  \mathrm{SRE}_{n}(\rho)
  =\frac{1}{1-n}\ln\left[\frac{1}{d^{L}}\sum_{\bm{u}}
  \lvert\chi_{\rho}(\bm{u})\rvert^{2n}\right].
  \label{eq:qudit_SRE}
\end{align}
For pure states, in parallel with the qubit properties listed above, one can prove faithfulness, invariance under Clifford unitaries, and additivity for any $d$~\cite{wang2023stabilizer}.
Moreover, monotonicity under stabilizer protocols has been proved for odd prime $d$~\cite{turkeshi2025magic}.
The stabilizer protocols and the precise notions of the monotone and strong monotone are defined in Sec.~\ref{sec:mono}, where spins, bosons, and fermions are treated in a unified manner.

\subsection{\label{subsec:boson-fermion}Bosons and fermions}

We next describe bosonic and fermionic systems on an equal footing.
Let $\eta=+1$ for bosons and $\eta=-1$ for fermions, and write the commutation relations as $[A,B]_{\eta}\equiv AB-\eta BA$.
The canonical creation and annihilation operators obey
\begin{align}
  [a_{i},a_{j}^{\dag}]_{\eta}=\delta_{ij},\qquad [a_{i},a_{j}]_{\eta}
  =[a_{i}^{\dag},a_{j}^{\dag}]_{\eta}=0.
  \label{eq:CCR_CAR_prelim}
\end{align}
We can define the displacement operators labeled by a phase-space variable $\bm{u}$, which play the role of Pauli strings in the bosonic and fermionic settings.

For $L$-mode bosons, we write the continuous variables as $\bm{u}=(\bm{q},\bm{p})\in\mathbb{R}^{2L}$, with $\bm{q}=(q_{1},q_2,\ldots,q_{L})$ and $\bm{p}=(p_{1},p_2,\ldots,p_{L})$.
We introduce the quadrature operators
\begin{align}
  r_{i} = \frac{a_{i}+a_{i}^{\dag}}{\sqrt{2}},\qquad
  r_{L+i} = \frac{a_{i}-a_{i}^{\dag}}{i\sqrt{2}},\qquad i=1,2,\ldots,L,
  \label{eq:boson_ops}
\end{align}
which satisfy $[r_{j},r_{k}]=i\Omega_{jk}$ with the symplectic matrix $\Omega$ introduced in Sec.~\ref{subsec:qubit}.
We define the bosonic displacement operator by~\cite{cahill1969ordered}
\begin{align}
  D_{\mathrm{b}}(\bm{u})
  = \exp(-i\bm{u}^{\rm T}\Omega\bm{r}),
  \label{eq:disp_boson}
\end{align}
and the characteristic function of a state $\rho$ is written as
\begin{align}
  \chi_{\rho}(\bm{u})
  = \tr\!\qty[D_{\mathrm{b}}(\bm{u})\rho].
  \label{eq:chi_boson}
\end{align}
Since $\bm{u}^{\rm T}\Omega\bm{r}$ is Hermitian for real phase-space labels $\bm{u}$, the displacement operator satisfies
\begin{align}
  D_{\mathrm{b}}(\bm{u})^{\dag}
  = \exp(i\bm{u}^{\rm T}\Omega\bm{r})
  = D_{\mathrm{b}}(-\bm{u}).
  \label{eq:disp_dagger_boson}
\end{align}
The bosonic displacement operators form a self-dual continuous operator basis:
\begin{align}
  \tr\!\qty[D_{\mathrm{b}}(\bm{u})^{\dag}D_{\mathrm{b}}(\bm{v})]
  = (2\pi)^{L}\delta(\bm{u}-\bm{v}).
  \label{eq:disp_orth_boson}
\end{align}
Consequently, any state can be expressed as
\begin{equation}\label{boson_cha}
  \rho=\int \!\!\frac{d\bm{u}}{(2\pi)^{L}}\, \chi_{\rho}(-\bm{u})D_{\mathrm{b}}(\bm{u}).
\end{equation}

For $L$-mode fermions, we use the same block notation $\bm{u}=(\bm{q},\bm{p})$, but now $q_i$ and $p_i$ are real Grassmann variables.
They anticommute with one another and with the fermionic operators.
We introduce Majorana operators in direct analogy with the bosonic quadratures,
\begin{align}
  \gamma_{i}
   & = a_{i}+a_{i}^{\dag},\qquad
  \bar{\gamma}_{i}
  = i\qty(a_{i}-a_{i}^{\dag}),
  \nonumber                                   \\
  \bm{r}
   & =(\gamma_{1},\gamma_2,\ldots,\gamma_{L},
  \bar{\gamma}_{1},\bar{\gamma}_2,\ldots,\bar{\gamma}_{L}),
  \label{eq:Majorana}
\end{align}
which satisfy $[r_{j},r_{k}]_{-}=2\delta_{jk}$.
The fermionic displacement operator is~\cite{cahill1999density}
\begin{align}
  D_{\mathrm{f}}(\bm{u})
  = \exp(\tfrac{1}{2}\bm{r}^{\rm T}\bm{u}).
  \label{eq:disp_fermion}
\end{align}
We define the characteristic function by
\begin{align}
  \chi_{X}(\bm{u})
  = \tr\!\qty[X D_{\mathrm{f}}(\bm{u})],
  \label{eq:chi_fermion}
\end{align}
where Grassmann variables are placed to the right of an operator $X$ before taking the trace.
Throughout the paper, we impose fermion-parity superselection: every pure state has definite fermion parity, and every density operator is parity even. 

The free states of the magic resource theory for bosons and fermions are Gaussian states.
Including the mixed-state case, a bosonic Gaussian state with mean value $\bar{\bm{u}}$ and covariance matrix $\sigma$ is expressed as
\begin{equation}
  \chi_{\rho_G}(\bm{u}) = \exp(-\frac{1}{4} \bm{u}^{\rm T}\Omega^{\rm T} \sigma \Omega \bm{u}+i\bar{\bm{u}}^{\rm T}\Omega \bm{u}).
  \label{eq:chi_gaussian_boson}
\end{equation}
Similarly, a fermionic Gaussian state with covariance matrix $\sigma$ is given by
\begin{equation}
  \chi_{\rho_G}(\bm{u}) = \exp(\frac{i}{8}\bm{u}^{\rm T}\sigma \bm{u}).
  \label{eq:chi_gaussian_fermion}
\end{equation}
We denote by $\mathrm{Gauss}_{L}$ the set of pure Gaussian states, namely the pure states whose characteristic functions take the form of Eq.~\eqref{eq:chi_gaussian_boson} or \eqref{eq:chi_gaussian_fermion}.

For bosons, a Gaussian unitary is generated by an at-most-quadratic Hermitian operator,
\begin{align}
  U_{G} = \exp\left[i\qty(\frac{1}{2}\bm{r}^{\rm T} H\bm{r}
                +\bar{\bm{r}}^{\rm T}\bm{r})\right],
  \label{eq:BGU}
\end{align}
with $H\in\mathbb{R}^{2L\times 2L}$, $H=H^{\rm T}$, and $\bar{\bm{r}}\in\mathbb{R}^{2L}$.
It acts on the quadratures by an affine symplectic transformation
\begin{align}
  U_{G}^\dagger \bm{r}U_{G}
         & = S\bm{r}+\bm{d},\nonumber                                 \\
  S = e^{-\Omega H},\quad
  \bm{d} & = -\qty(\int_{0}^{1}\!e^{-s\Omega H}ds)\Omega\bar{\bm{r}},
  \label{eq:symplectic}
\end{align}
and therefore forms a group which we denote by $\mathcal{G}_{L}$.
For fermions, a Gaussian unitary is generated by a quadratic Majorana operator,
\begin{align}
  U_{G}
  = \exp\!\left(\frac{1}{4}\bm{r}^{\rm T} A\bm{r}\right),
  \quad
  A\in\mathbb{R}^{2L\times 2L},\quad A^{\rm T}=-A,
  \label{eq:FGU}
\end{align}
and acts linearly on the Majorana operators by a special orthogonal transformation,
\begin{align}
  U_{G}^\dagger \bm{r}U_{G}
  = O\bm{r},\qquad
  O = e^{A}.
  \label{eq:FGU_action}
\end{align}
The fermionic Gaussian unitaries also form a group, which we denote by $\mathcal{G}_{L}$ as well.
Importantly, for bosons and even-parity fermions, the pure Gaussian states are the states generated by Gaussian unitaries from the vacuum, i.e., $\qty{U_{G}\ket{0}^{\otimes L} \mid U_{G}\in\mathcal{G}_{L}}$, analogous to Eq.~\eqref{eq:Stab_states_qubit}.
For fermionic Gaussian states with odd parity, a similar expression holds by choosing an odd-parity Gaussian state as the reference state.

Gaussian states play the role of free states because Gaussian quantum processes are efficiently classically simulable.
For bosonic continuous-variable systems, the analogue of the Gottesman-Knill theorem states that processes initialized in Gaussian states and composed of Gaussian unitaries, Gaussian measurements, and Gaussian feedforward can be simulated efficiently on a classical computer.
Similarly, fermionic Gaussian dynamics, also known as fermionic linear optics, admits an efficient classical simulation in terms of covariance matrices or equivalent Grassmann Gaussian representations.
Therefore, non-Gaussian states or operations constitute the necessary resource for quantum computational advantage in bosonic and fermionic systems.

\section{\label{sec:MRE}Magic R\'enyi entropy}

Here we introduce the magic R\'enyi entropy (MRE) to quantify the amount of quantum computational resources in spins, bosons, and fermions. As illustrated in Fig.~\ref{fig:convolution}, the MRE is defined through a convolution operation of replicas, namely a unitary mixing $n$ copies followed by a partial trace leaving a single copy.
The construction is designed so that free pure states remain pure under the convolution, whereas resourceful pure states become mixed.
This purity loss gives a common measure of nonstabilizerness and non-Gaussianity.

\begin{figure}[t]
  \centering
  \includegraphics[width=\columnwidth]{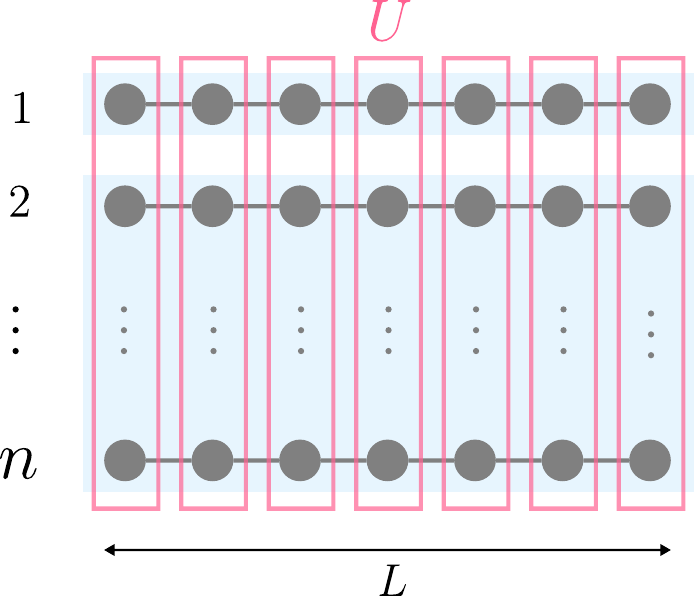}
  \caption{Schematic illustration of quantum convolution. The replica-mixing unitary $U$ is applied to $n$ copies of the many-body state, after which replicas $2,3,\ldots,n$ are traced out.
    The unitary $U$ generates correlations among the replicas only for resourceful states, so that the output replica left after the partial trace becomes mixed.
    The MRE detects the magic by measuring the purity of the resulting convolved state.}
  \label{fig:convolution}
\end{figure}

\subsection{\label{subsec:MRE_definition}Definition}

Let $U$ be an $n$-replica unitary acting on $\mathscr{H}^{\otimes n}$, where $\mathscr{H}$ is the Hilbert space of the original system of $L$ qudits or $L$ modes. At this stage, we keep $U$ abstract for the sake of generality, leaving its explicit form to be specified below.
For $n$ input states $\rho_{1},\rho_2,\ldots,\rho_{n}$ on $\mathscr{H}$, we define the convolution $\mathcal{C}_n$ by
\begin{align}
  \mathcal{C}_{n}(\rho_{1},\rho_2,\ldots,\rho_{n})
  = \tr_{1^{\rm c}}\qty[
                     U\qty(\rho_{1}\otimes\rho_2\otimes\cdots\otimes\rho_{n})U^{\dag}],
  \label{eq:convolution_general}
\end{align}
where $\tr_{1^{\rm c}}$ denotes the partial trace over the remaining output replicas $\{2,3,\ldots,n\}$.

When all inputs are identical, we write the self-convolution as
\begin{align}
  \mathcal{C}_{n}(\rho)
  \equiv \mathcal{C}_{n}(\underbrace{\rho,\ldots,\rho}_n)
  = \tr_{1^{\rm c}}\qty[
                     U\,\rho^{\otimes n}U^{\dag}].
  \label{eq:self_convolution}
\end{align}
To quantify the purity loss, we introduce the convolution purity
\begin{align}
  \nu_n(\rho_{1},\rho_2,\ldots,\rho_{n}) = \tr\!\qty[\mathcal{C}_{n}(\rho_{1},\rho_2,\ldots,\rho_{n})^{2}],
  \label{eq:convolution_purity}
\end{align}
and, for identical inputs, we abbreviate it as
\begin{equation}
  \nu_n(\rho)\equiv\nu_n(\underbrace{\rho,\ldots,\rho}_n).
\end{equation}
Then, the MRE associated with $U$ is defined as
\begin{align}
  M_{n}(\rho)
  = \frac{1}{1-n}\ln\nu_n(\rho).
  \label{eq:MRE_def}
\end{align}
The prefactor $1/(1-n)$ is chosen so that $M_{n}$ agrees with the SRE as discussed below.

Said differently, the MRE uses the second-order R\'enyi entanglement entropy to quantify the correlation between replica 1 and the other replicas, which is generated by the convolution unitary $U$.
While we focus on pure states in this work, we remark that the R\'enyi entanglement entropy is no longer a correlation measure for mixed states. Nevertheless, the convolution-based magic measure can still be defined for mixed states provided that an appropriate correlation measure, such as the mutual information, is used \cite{hahn2025measuring,bu2025efficient,lyu2024fermionic}.

We first consider Clifford convolution unitaries for spins and Gaussian convolution unitaries implementing real orthogonal rotations of the replicas for bosons and fermions.
In each case, the same unitary acts on the $n$ replicas of every site.
We require the self-convolution of any pure free state to remain pure.
For these unitaries, the characteristic function of the convolved state takes the form
\begin{equation}
  \chi_{\mathcal{C}_{n}(\rho_{1},\ldots,\rho_{n})}(\bm{u})\propto\prod_{r=1}^{n}\chi_{\rho_{r}}(t_{r}\bm{u}'),
  \label{eq:weighted_duality}
\end{equation}
where the proportionality factor is a phase, $\bm{u}'$ is an invertible linear function of $\bm{u}$, and $t_r$ is the weight of input replica $r$; see Appendix~\ref{app:properties_qudit} for spins and Sec.~\ref{subsec:MRE_BF} for bosons and fermions.
We say that $U$ genuinely mixes the replicas if at least two weights are nonzero for bosons and fermions, or at least two are coprime to $d$ for spins.
This excludes the trivial case in which the retained replica remains uncorrelated with the other replicas for every input state.

We formulate the basic properties of the MRE as the following theorem:
\begin{theorem}
  \label{thm:MRE-properties}
  Let $U$ be a convolution unitary that genuinely mixes the replicas.
  For a pure state $\rho=\ketbra*{\psi}{\psi}$, the MRE associated with $U$ satisfies the following properties:
  \begin{enumerate}
    \item Faithfulness: $M_{n}(\rho)\ge 0$, with equality holding if and only if $\ket{\psi}\in\mathrm{Stab}_{L}$ (for spins) or $\ket{\psi}\in\mathrm{Gauss}_{L}$ (for bosons and fermions).
    \item Invariance under free unitaries: $M_{n}(U_{F}\rho U_{F}^{\dag})=M_{n}(\rho)$ for any Clifford unitary $U_{F}\in\mathrm{Cl}_{L}$ (for spins) or any Gaussian unitary $U_{F}\in\mathcal{G}_{L}$ (for bosons and fermions).
    \item Additivity: $M_{n}(\rho_{1}\otimes\rho_{2})=M_{n}(\rho_{1})+M_{n}(\rho_{2})$.
  \end{enumerate}
\end{theorem}
The proofs are presented in Appendix~\ref{app:MRE_properties}.

The preceding theorem admits many convolution unitaries satisfying the same basic properties.
A resource measure must additionally be monotone under free protocols that include free measurements and feedforward.
We thus focus on convolution unitaries satisfying the convolution-multiplication duality
\begin{equation}
  \chi_{\mathcal{C}_{n}(\rho_{1},\ldots,\rho_{n})}(\bm{u})\propto\prod_{r=1}^{n}\chi_{\rho_{r}}(\bm{u}'),
  \label{eq:duality_general}
\end{equation}
where all input characteristic functions are evaluated at the same point $\bm{u}'$; in other words,
equal replica weights reduce Eq.~\eqref{eq:weighted_duality} to Eq.~\eqref{eq:duality_general} with the common weight absorbed into $\bm{u}'$.
We regard convolution unitaries whose weights differ only in sign as equivalent, since they yield the same MRE as inferred from $|\chi_{\rho}(-\bm{u})|=|\chi_{\rho}(\bm{u})|$ for spins and bosons, and from $\chi_{\rho}(-\bm{u})=\chi_{\rho}(\bm{u})$ for fermions.
In Sec.~\ref{sec:mono}, we prove that the associated MRE is monotone for spins and bosons whenever defined, and for fermions at $n=2$.
For bosons and qudits, we further prove that this choice of convolution unitaries is essentially the only one that yields a monotone.
Unless otherwise stated, the MRE henceforth denotes the measure associated with the above duality.

The construction~\eqref{eq:MRE_def} unifies resource quantifiers previously studied separately.
For spins, Ref.~\cite{bu2025stabilizer} introduced such measures for qubits at odd $n$ and for qudits of odd prime dimension at $n=2$.
We extend this construction to arbitrary local dimension $d$ and all replica numbers $n$ coprime to $d$ in Sec.~\ref{subsec:MRE_qubit}.
For bosons and fermions, related quantities have been studied using a beam-splitter unitary on two replicas~\cite{hahn2025measuring,bu2025efficient,lyu2024fermionic,coffman2025measuring} and its iterations on more than two replicas~\cite{bu2025efficient,lyu2024fermionic,coffman2025measuring}.
In Sec.~\ref{subsec:MRE_BF}, we extend this construction to arbitrary replica number $n$ using convolution unitaries satisfying the duality.

\subsection{\label{subsec:MRE_qubit}Spins}
We begin with spins by constructing a convolution unitary for qudits that satisfies the convolution-multiplication duality~\eqref{eq:duality_general}.
For a replica number \(n\geq 2\) coprime to $d$, we introduce $ n^{-1}\in \mathbb Z_d$ to represent the multiplicative inverse of \(n\) modulo \(d\), i.e.,
\begin{equation}
  n n^{-1} \equiv 1 \pmod{d} .
\end{equation}
For \(n\) replicas of an \(L\)-qudit state, we define a convolution unitary \(U\) by its action on the computational basis as follows:
\begin{align}
  U
  \bigotimes_{r=1}^n |\bm{q}_r\rangle
   & =
  \left|
  n^{-1}\sum_{s=1}^n \bm{q}_s
  \right\rangle
  \bigotimes_{r=2}^n
  \left|
  n^{-1}(\bm{ q}_1-\bm{ q}_r)
  \right\rangle,\nonumber \\
  U^\dagger
  \bigotimes_{r=1}^n |\bm{q}_r\rangle
   & =
  \left|
  \sum_{s=1}^n \bm{q}_s
  \right\rangle
  \bigotimes_{r=2}^n
  \left|
  \sum_{s=1}^n \bm{q}_s-n\bm{q}_r
  \right\rangle,
  \label{eq:qudit_nd_convolution_unitary}
\end{align}
where all additions, subtractions, and scalar multiplications are understood to be performed componentwise in \(\mathbb Z_d^L\).
This map is a permutation of the computational basis and hence defines a unitary.
Equation~\eqref{eq:qudit_nd_convolution_unitary} gives an extension of the convolution unitaries in Ref.~\cite{bu2025stabilizer} and reproduces them for \(n=2\) with an odd prime $d$ and for qubits with an odd $n$, where $n^{-1}\equiv1$ and $-\bm{q}_r\equiv\bm{q}_r$ on $\mathbb{Z}_2$.

To derive the convolution-multiplication duality of this convolution unitary, we compute the induced transformation on the Heisenberg-Weyl operators of the retained output replica.
Introducing
\begin{equation}
  X_r(\bm{ q})=\bigotimes_{j=1}^L X_{r,j}^{q_j},
  \qquad
  Z_r(\bm{ p})=\bigotimes_{j=1}^L Z_{r,j}^{p_j},
\end{equation}
where the subscript \(r\) labels the replica, and using Eq.~\eqref{eq:qudit_nd_convolution_unitary}, one obtains
\begin{align}
  U^\dagger
  \left[
    X_1(\bm{ q})\otimes I^{\otimes(n-1)}
    \right]
  U
   & =
  \bigotimes_{r=1}^n X_r(\bm{ q}),
  \label{eq:X_conjugation_qudit_nd}
  \\
  U^\dagger
  \left[
    Z_1(\bm{ p})\otimes I^{\otimes(n-1)}
    \right]
  U
   & =
  \bigotimes_{r=1}^n Z_r( n^{-1}\,\bm{ p}).
  \label{eq:Z_conjugation_qudit_nd}
\end{align}
Consequently, we have
\begin{align}
   & U^\dagger
  \left[
    D_1(\bm{u})\otimes I^{\otimes(n-1)}
    \right]
  U\nonumber                                                                           \\
   & \;\;\;\;\;\;\;\;\;\;\;\;= \tau^{{\bm q}\cdot{\bm p}-n{\bm q}\cdot(n^{-1}{\bm p})}
  \bigotimes_{r=1}^n
  D_r(\bm{u}'),
  \label{eq:D_conjugation_qudit_nd}
\end{align}
where $D_r$ denotes the Heisenberg-Weyl operator acting on replica $r$, and $\bm{u}'\equiv(\bm{q}, n^{-1} \bm{p})$.
Taking the expectation value of Eq.~\eqref{eq:D_conjugation_qudit_nd} in $\rho_1\otimes\cdots\otimes\rho_n$ and using the definition~\eqref{eq:convolution_general} of the convolution, we obtain the convolution-multiplication duality
\begin{align}
  \chi_{{\cal C}_n(\rho_1,\ldots,\rho_n)}({\bm u})= \tau^{{\bm q}\cdot{\bm p}-n{\bm q}\cdot(n^{-1}{\bm p})} \prod_{r=1}^n
  \chi_{\rho_r}({\bm u}').
  \label{eq:duality_qudit_general}
\end{align}

The duality shows that the MRE coincides with the SRE.
For the self-convolution, Eq.~\eqref{eq:duality_qudit_general} gives
\begin{equation}
  \left|
  \chi_{\mathcal{C}_n(\rho)}(\bm{u})
  \right|^2
  =
  \left|
  \chi_\rho({\bm u}')
  \right|^{2n}.
  \label{eq:self_convolution_abs_qudit_nd}
\end{equation}
Since the orthogonality of the Heisenberg-Weyl operators gives $\tr\rho^{2}=d^{-L}\sum_{\bm{u}}|\chi_{\rho}(\bm{u})|^{2}$ and multiplication by \( n^{-1}\) is a bijection of \(\mathbb Z_d^L\), we can relabel \( n^{-1}\,\bm{p}\mapsto \bm{ p}\) and obtain
\begin{equation}
  \nu_n(\rho)
  =
  \frac{1}{d^L}
  \sum_{{\bm u}\in\mathbb Z_d^{2L}}
  \left|
  \chi_\rho({\bm u})
  \right|^{2n}.
  \label{eq:purity_qudit_nd_final}
\end{equation}
Comparing this with the qudit SRE~\eqref{eq:qudit_SRE}, we obtain $M_{n}=\text{SRE}_{n}$ for qudits.
Any convolution unitary satisfying the duality~\eqref{eq:duality_general} gives the same expression because $\bm u\mapsto\bm u'$ is invertible.

We remark that the SRE~\eqref{eq:qudit_SRE} itself is defined also when $n$ and $d$ are not coprime, while the convolution-multiplication duality~\eqref{eq:duality_general} cannot be satisfied in this case.
To see this, let us assume that a unitary maps $D_1(\bm{u})$ to $\bigotimes_{r}D_r(\bm{u}')$ up to a phase factor, with $\bm{u}'$ common to all replicas as in Eq.~\eqref{eq:D_conjugation_qudit_nd}.
For simplicity, we set $L=1$.
Setting $D_1(\bm{u})$ to the shift and clock operators $X$ and $Z$ yields
\begin{align}
  U^{\dag}(X_{1}\otimes I^{\otimes(n-1)})U & =c_xX'^{\otimes n},\nonumber \\
  U^{\dag}(Z_{1}\otimes I^{\otimes(n-1)})U & =c_zZ'^{\otimes n},
  \label{eq:unitary_x_z}
\end{align}
where $X'$ and $Z'$ are single-qudit Heisenberg-Weyl operators and $c_x$ and $c_z$ are phase factors.
Since the commutation relation $X_{1}Z_{1}=\omega_{d}^{-1}Z_{1}X_{1}$ is preserved under unitary transformation, we have
\begin{equation}
  c_xc_zX'^{\otimes n}Z'^{\otimes n}=\omega_{d}^{-1}c_xc_zZ'^{\otimes n}X'^{\otimes n}.
\end{equation}
Meanwhile, Eq.~\eqref{eq:HW_commutation_qudit} gives $X'Z'=\omega_{d}^{-k}Z'X'$ with some $k\in\mathbb{Z}_d$, and applying it in each replica gives $X'^{\otimes n}Z'^{\otimes n}=\omega_{d}^{-nk}Z'^{\otimes n}X'^{\otimes n}$.
Since these two relations are consistent only if $nk\equiv1\pmod{d}$, there is no unitary satisfying Eq.~\eqref{eq:unitary_x_z} when $n$ and $d$ are not coprime.
For qubits, this excludes even $n$.

\subsection{\label{subsec:MRE_BF}Bosons and fermions}

In bosonic and fermionic cases, the convolution unitary $U$ rotates the replicas by an $n\times n$ real orthogonal matrix $O$.
Labeling the replicas by $r\in\{1,2,\ldots,n\}$, we collect the bosonic or fermionic annihilation operators with modes $i=1,2,\ldots,L$ into the vector
\begin{equation}
  \bm{a}_{r}\equiv(a_{r,1},a_{r,2},\ldots,a_{r,L})^{\rm T},\label{avecdef}
\end{equation}
on which $U$ acts as
\begin{align}
  U^\dagger \,\bm{a}_{s}\,U= \sum_{r=1}^{n}O_{sr}\,\bm{a}_{r}.
  \label{eq:U_BF}
\end{align}

We first show that the weights of this convolution unitary are the entries of the first row of $O$.
The convolution unitary transforms the characteristic function as
\begin{equation}
  \chi_{U(\rho_{1}\otimes\cdots\otimes\rho_{n})U^{\dag}}(\bm{u}_{1},\ldots,\bm{u}_{n})=\prod_{r=1}^{n}\chi_{\rho_{r}}\!\left(\sum_{s=1}^{n}O_{sr}\bm{u}_{s}\right),
  \label{eq:conv_char_transform}
\end{equation}
where $\bm{u}_{s}$ is either a $2L$-dimensional real-valued vector in the bosonic case or a $2L$-dimensional real Grassmann vector in the fermionic case (see Eqs.~\eqref{eq:boson_ops} and \eqref{eq:Majorana}).
Tracing out replicas $2,3,\ldots,n$ by setting $\bm{u}_{2}=\bm{u}_3=\cdots=\bm{u}_{n}={\bm 0}$ yields the characteristic function of the convolved state as
\begin{align}
  \chi_{\mathcal{C}_{n}(\rho_{1},\ldots,\rho_{n})}(\bm{u})
  = \prod_{r=1}^{n}\chi_{\rho_{r}}(O_{1r}\bm{u}).
  \label{eq:duality_BF_weights}
\end{align}
Equation~\eqref{eq:duality_BF_weights} is the form~\eqref{eq:weighted_duality} with $\bm{u}'=\bm{u}$ and the weights $t_r=O_{1r}$.

We choose $O$ to satisfy the duality~\eqref{eq:duality_general}, for which all the weights are equal to $1/\sqrt{n}$ because the first row of $O$ is a unit vector.
In this case, the retained output mode is the normalized center-of-mass combination of the $n$ input replicas,
\begin{equation}
  U^{\dag}\bm{a}_{1}U=\frac{1}{\sqrt{n}}\sum_{r=1}^{n}\bm{a}_{r},
\end{equation}
and Eq.~\eqref{eq:duality_BF_weights} becomes
\begin{align}
  \chi_{\mathcal{C}_{n}(\rho_{1},\ldots,\rho_{n})}(\bm{u})
  = \prod_{r=1}^{n}
  \chi_{\rho_{r}}\!\qty(\frac{\bm{u}}{\sqrt{n}}).
  \label{eq:duality_BF}
\end{align}
Every $O$ satisfying the duality thus gives the same MRE.
An example of such $O$ is the Helmert matrix
\begin{align}
  O_{1r} = \frac{1}{\sqrt{n}},\quad
  O_{sr} = \begin{cases}
             \dfrac{1}{\sqrt{s(s-1)}} & r<s, \\[6pt]
             -\sqrt{\dfrac{s-1}{s}}   & r=s, \\[8pt]
             0                        & r>s,
           \end{cases}\quad 2\le s\le n,
  \label{eq:O_case2}
\end{align}
where $r$ and $s$ are the input and output replica labels, respectively.

The duality~\eqref{eq:duality_BF} gives the MRE in closed forms, which are useful for proving the monotonicity of the MRE under Gaussian protocols.
For bosons, the self-duality relation~\eqref{eq:disp_orth_boson} together with Eq.~\eqref{eq:duality_BF} yields the convolution purity in a closed form:
\begin{align}
  \nu_{n}(\rho_{1},\ldots,\rho_{n})
  = \int\!\!\frac{d\bm{u}}{(2\pi)^{L}}\,
  \prod_{r=1}^{n}\qty|\chi_{\rho_{r}}\!\qty(\frac{\bm{u}}{\sqrt{n}})|^{2}.
  \label{eq:purity_boson_chi}
\end{align}
Setting all inputs to $\rho$ and rescaling $\bm{u}$, the MRE is written as a higher moment of the characteristic function:
\begin{align}
  M_{n}(\rho)
  = \frac{1}{1-n}\ln \left[n^L \int\!\!\frac{d\bm{u}}{(2\pi)^{L}}\,
                       \qty|\chi_{\rho}(\bm{u})|^{2n}\right].
  \label{eq:MRE_boson_chi}
\end{align}
This expression makes explicit the parallel with the spin MRE obtained from Eq.~\eqref{eq:purity_qudit_nd_final}, where the discrete sum over Heisenberg-Weyl labels is replaced by a continuous phase-space integral.
The prefactor $n^{L}$ ensures that pure Gaussian states have zero MRE.
We remark that the same quantity has appeared in Refs.~\cite{crew2026magic,hahn2025bridging} in the context of magic in hybrid spin-boson models and of the Gottesman-Kitaev-Preskill encoding, respectively.

For fermions, because the displacement operator $D_{\mathrm{f}}(\bm{u})$ is not self-dual under the trace, the MRE does not admit the same reduction to a simple higher moment of $\chi_{\rho}$.
Instead, the closed form involves the Grassmann-monomial coefficients of the characteristic function. Specifically, for an increasingly ordered subset $A\subset\{1,2,\ldots,2L\}$, we define $\Gamma_A=\prod_{j\in A}r_j$ and $\bm{u}_A=\prod_{j\in A}u_j$.
Using this ordered monomial and the trace order~\eqref{eq:chi_fermion}, we have
\begin{equation}
  \chi_{X}(2\bm{u})=\sum_{A}\kappa_{X}(A)\,\bm{u}_{A},\qquad
  \kappa_{X}(A)=\tr[\Gamma_{A}^{\dagger}X].
  \label{eq:chi_fermion_expansion}
\end{equation}
For a Grassmann polynomial $f=\sum_{A}f_{A}\bm{u}_{A}$, we define the coefficient norm by $\norm{f}^{2}=\sum_{A}|f_{A}|^{2}$.
Since the Majorana monomials are orthogonal under the trace, $\tr[\Gamma_{A}^{\dagger}\Gamma_{B}]=2^{L}\delta_{A,B}$, the coefficient norm yields the squared Hilbert-Schmidt norm of any operator as $\tr[X^{\dagger}X]=2^{-L}\norm{\chi_{X}(2\bm{u})}^{2}$.
Combining this identity with Eq.~\eqref{eq:duality_BF}, we obtain the closed form of the fermionic convolution purity,
\begin{equation}
  \nu_{n}(\rho_{1},\ldots,\rho_{n})=\frac{1}{2^L}\norm{\prod_{r=1}^{n}\chi_{\rho_{r}}\!\qty(\frac{2\bm{u}}{\sqrt{n}})}^{2}.
  \label{eq:purity_fermion_chi}
\end{equation}
Therefore, the closed-form expression of the fermionic MRE reads
\begin{equation}
  M_{n}(\rho)
  = \frac{1}{1-n}\ln \left[\frac{1}{2^L}\left\|\chi_{\rho}\!\left(\frac{2\bm{u}}{\sqrt{n}}\right)^n\right\|^{2}\right],
\end{equation}
analogous to Eqs.~\eqref{eq:qudit_SRE} and \eqref{eq:MRE_boson_chi}.

For fermions, Sec.~\ref{sec:fermion_TLL} also analyzes the universal contribution to the MRE for $n=2$ with convolution unitaries that do not satisfy the duality.
When $n=2$, it suffices to consider the rotation matrix,
\begin{align}
  O(\zeta) = \begin{pmatrix}
               \cos\zeta & -\sin\zeta \\
               \sin\zeta & \cos\zeta
             \end{pmatrix},
  \label{eq:O_case1}
\end{align}
with $0<\zeta\leq\pi/4$.
Exchanging the two weights does not change Eq.~\eqref{eq:duality_BF_weights}, because the two inputs are identical.
Thus, with the identification of weights that differ only in sign, any other orthogonal matrix $O$ yields the same MRE as $O(\zeta)$ for some $\zeta$ in this range.
The action of the corresponding unitary $U(\zeta)$ is that of a beam splitter with reflectivity $\sin^{2}\zeta$ that mixes the two replicas.
The balanced beam splitter with $\zeta=\pi/4$ gives the MRE with $n=2$ defined above.

\subsection{\label{sec:mono}Monotonicity theorems}

We now turn to monotonicity of the MRE under free protocols, treating spins, bosons, and fermions in a unified manner.
In all cases, a free protocol is constructed from any composition of the following six elements:
(i) free unitaries,
(ii) partial traces,
(iii) free measurements,
(iv) additions of free ancillas,
(v) any of (i)--(iv) conditioned on measurement outcomes,
and
(vi) any of (i)--(iv) conditioned on classical random variables.
For spins, the free unitaries, measurements, and ancillas are Clifford unitaries, computational-basis measurements, and $\ket{0}$ ancillas, respectively, and the resulting protocols are called stabilizer protocols.
For bosons, they are Gaussian unitaries, homodyne measurements, and pure Gaussian ancillas, and the resulting protocols are called bosonic Gaussian protocols.
A fermionic Gaussian protocol is defined in the same way as the bosonic one, except that the homodyne measurements are replaced by occupation-number measurements.

We say that a resource quantifier is a {\it monotone} if it does not increase under any free protocol that maps a pure input \(|\psi\rangle\) deterministically to a pure output \(|\varphi\rangle\).
Similarly, a resource quantifier is called a {\it strong monotone} if it does not increase on average under any free protocol that maps a pure input \(|\psi\rangle\) randomly to pure outputs \(|\varphi_\lambda\rangle\) with certain probabilities $p_\lambda$.
Importantly, the intermediate states in a free protocol do not need to be pure states, as long as the final output is pure.

Similar to the linear SRE, we also define the linear MRE by
\begin{align}
  M_{n}^{\rm lin}(\rho) = 1 - \nu_n(\rho),
\end{align}
which is related to $M_{n}$ via $M_{n}^{\rm lin}=1-\exp\qty[(1-n)M_{n}]$.
Since the linear MRE is a monotonically increasing function of the MRE, it also satisfies the faithfulness and the invariance in Theorem~\ref{thm:MRE-properties}, while it is not additive, as in the case of the linear SRE.

For spins, the MRE with any replica number $n\geq 2$ coprime to the local Hilbert-space dimension $d$ coincides with the SRE, and its linear version coincides with the linear SRE (see Sec.~\ref{subsec:MRE_qubit}).
Monotonicity of the SRE has been proved for qubits~\cite{leone2024stabilizer} and for odd prime $d$~\cite{turkeshi2025magic}, and strong monotonicity of the linear SRE has been proved for qubits~\cite{leone2024stabilizer}.
We extend these results to every $d$:

\begin{theorem}
  \label{thm:qudit-protocol-monotonicity}
  For every local Hilbert-space dimension $d$ and every integer $n\geq2$ coprime to $d$, the spin MRE $M_n$ is a monotone, while its linear version $M_n^{\rm lin}$ is a strong monotone.
\end{theorem}

For bosons and fermions, we prove the following theorems:

\begin{theorem}
  \label{thm:boson-protocol-monotonicity}
  For every integer \(n\geq2\), the bosonic MRE \(M_n\) is a monotone, while its linear version \(M_n^{\rm lin}\) is a strong monotone.
\end{theorem}

\begin{theorem}
  \label{thm:fermion-protocol-monotonicity}
  The second-order fermionic MRE \(M_2\) is a monotone, while its linear version \(M_2^{\rm lin}\) is a strong monotone.
\end{theorem}
Proofs of the above statements are given in Appendix~\ref{app:monotone}.

We next determine which convolution unitaries yield monotones.
For qudits and bosons, this question is governed by the replica weights in Eq.~\eqref{eq:weighted_duality}.
We assume all weights are nonzero since a replica with zero weight does not affect the convolved state.
For qudits, Eq.~\eqref{eq:weighted_duality} gives (see Appendix~\ref{app:properties_qudit})
\begin{equation}
  \nu_n(\rho_1,\ldots,\rho_n)=\frac{1}{d^{L}}\sum_{\bm{u}}\prod_{r=1}^{n}\qty|\chi_{\rho_r}(t_r\bm{u})|^{2}.
  \label{eq:nu_general_qudit}
\end{equation}
For bosons, applying Eq.~\eqref{eq:duality_BF_weights} as in the derivation of Eq.~\eqref{eq:purity_boson_chi} yields
\begin{equation}
  \nu_n(\rho_1,\ldots,\rho_n)=\int\!\!\frac{d\bm{u}}{(2\pi)^{L}}\prod_{r=1}^{n}\qty|\chi_{\rho_r}(O_{1r}\bm{u})|^{2}.
  \label{eq:nu_general_boson}
\end{equation}
The following theorems show that the monotonicity singles out the convolutions satisfying the duality.

\begin{theorem}
  \label{thm:qudit-convolution-uniqueness}
  Let $d$ be a prime local Hilbert-space dimension, and consider a Clifford convolution unitary acting uniformly on all sites.
  The resulting MRE satisfies $M_n=\mathrm{SRE}_n$ and is a monotone if and only if the weights $t_r$ are equal up to sign.
  The linear MRE $M_n^{\rm lin}$ is a strong monotone if and only if the same weight condition holds.
\end{theorem}

\begin{theorem}
  \label{thm:boson-convolution-uniqueness}
  Let $O$ be an orthogonal matrix defining a bosonic convolution.
  The resulting MRE $M_n$ is a monotone if and only if the weights $O_{1r}$ are equal up to sign.
  The linear MRE $M_n^{\rm lin}$ is a strong monotone if and only if the same weight condition holds.
\end{theorem}
Proofs of the above statements are given in Appendix~\ref{app:convolution_uniqueness}.
For example, the non-Gaussian entropy $\text{NGE}_{2,k}$ proposed by Ref.~\cite{bu2025efficient} is not a monotone for $k\geq2$, because its replica weights are not equal up to sign; for $k=2$, these weights are $(1/2,1/2,1/\sqrt{2})$.

We here explain why the convolution satisfying the duality is special.
The monotonicity of the MRE reduces to the fact that a free measurement does not decrease the convolution purity on average, because free unitaries and the addition of free ancillas preserve $\nu_n$.
We write the state before the measurement as $\ket{\Psi}=\sum_{x}\sqrt{p(x)}\ket{\eta_{x}}_{R}\otimes\ket{x}_{Q}$, where $\ket{x}_{Q}$ is the measurement basis of the measured register $Q$, $p(x)$ is the probability of the outcome $x$, and $\ket{\eta_{x}}_{R}$ is the normalized state of the other register $R$ conditioned on $x$.
Since the convolution unitary maps the measurement basis of the $n$ replicas to itself up to phases, the convexity of the purity bounds $\nu_n(\Psi)$ by the average over the outcomes $x_1,\ldots,x_n$ of the replicas,
\begin{equation}
  \nu_n(\Psi)\leq\sum_{x_1,\ldots,x_n}\prod_{r=1}^{n}p(x_r)\,\nu_n(\eta_{x_1},\ldots,\eta_{x_n}).
\end{equation}
With the duality, the characteristic functions of the different inputs $\eta_{x_1},\ldots,\eta_{x_n}$ are evaluated at the same point, and H\"older's inequality gives
\begin{equation}
  \nu_n(\eta_{x_1},\ldots,\eta_{x_n})\leq\prod_{r=1}^{n}\nu_n(\eta_{x_r})^{1/n}.
\end{equation}
Substituting this bound, the sum factorizes over the replicas, and Jensen's inequality gives
\begin{equation}
  \nu_n(\Psi)\leq\qty[\sum_{x}p(x)\,\nu_n(\eta_{x})^{1/n}]^{n}\leq\sum_{x}p(x)\,\nu_n(\eta_{x}).
\end{equation}
For bosons, the sums over the outcomes are replaced by integrals over the homodyne variables.

If the conditions in Theorems~\ref{thm:qudit-convolution-uniqueness} and \ref{thm:boson-convolution-uniqueness} are violated, there is a deterministic free protocol that increases $M_n$, consisting of a measurement of an ancilla and a free unitary conditioned on the outcome.
The protocol relies on the fact that, without the duality, free unitaries acting differently on the replicas can increase the convolution purity.
Namely, there are a pure state $\ket{\varphi}$ and free unitaries $U_{F,1},\ldots,U_{F,n}$ such that
\begin{equation}
  \nu_n(U_{F,1}\varphi,\ldots,U_{F,n}\varphi)>\nu_n(\varphi).
  \label{eq:purity_increase}
\end{equation}
With the duality, H\"older's inequality prevents this increase, because it bounds the left-hand side by $\prod_{r}\nu_n(U_{F,r}\varphi)^{1/n}=\nu_n(\varphi)$.
Without the duality, the characteristic functions of the inputs are evaluated at different points, and free unitaries acting differently on the replicas can compensate this difference.
In Appendix~\ref{app:convolution_uniqueness}, we construct such $\ket{\varphi}$ and $U_{F,r}$, where $U_{F,r}$ rescales the computational-basis labels for qudits and squeezes the wave function for bosons.
For an integer $\ell\geq1$, the input of the protocol is
\begin{equation}
  \ket{\Psi}=\frac{1}{\sqrt{n}}\sum_{r=1}^{n}\qty(U_{F,r}\ket{\varphi})^{\otimes\ell}\otimes\ket{r}_Q,
  \label{eq:protocol_input}
\end{equation}
where $\ket{r}_Q$ are states of the ancilla $Q$ such that a free measurement of $Q$ determines $r$.
Measuring $Q$, applying $(U_{F,r}^{\dagger})^{\otimes\ell}$ for the outcome $r$, and discarding $Q$ map $\ket{\Psi}$ to $\ket{\varphi}^{\otimes\ell}$ deterministically.
Since the $n$ replicas of $\ket{\Psi}$ contain the term in which replica $r$ carries $(U_{F,r}\ket{\varphi})^{\otimes\ell}$, its convolution purity is bounded as
\begin{align}
  \nu_n(\Psi)
   & \geq C\,\nu_n(U_{F,1}\varphi,\ldots,U_{F,n}\varphi)^{\ell}\nonumber \\
   & >\nu_n(\varphi)^{\ell}=\nu_n(\varphi^{\otimes\ell})
  \label{eq:protocol_chain}
\end{align}
for large $\ell$, where $C>0$ does not depend on $\ell$ (Lemma~\ref{lem:feedforward}).
The protocol therefore increases $M_n$.
By contrast, the MRE associated with any $O$ does not increase under a pure-to-pure Gaussian channel, which consists of adding vacuum ancillas, Gaussian unitaries, and tracing out a subsystem (see Appendix~\ref{app:MRE_properties}).

For fermions, a similar argument for the monotonicity does not apply even with the duality for two reasons.
First, the parity superselection restricts free measurements to occupation numbers, which are quadratic in the creation and annihilation operators, so that the convolution unitary does not map their eigenbasis to itself.
Second, since the characteristic function is a Grassmann polynomial, the convolution purity~\eqref{eq:purity_fermion_chi} is the coefficient norm of a product rather than a sum of products of nonnegative functions, and H\"older's inequality does not apply.
We thus prove the corresponding inequality for $n=2$ by a direct evaluation in Appendix~\ref{app:monotone}.
Whether the MRE with $n>2$ or the MRE associated with other convolution unitaries is a monotone remains open.

\subsection{\label{subsec:rank}Robust exponential lower bound on the fermionic Gaussian rank}

In this subsection, we show that the fermionic MRE bounds from below the cost of classically simulating a state as a superposition of Gaussian states.
Building on the efficient simulation of fermionic Gaussian circuits, a superposition of $N$ pure Gaussian states can be simulated at a cost that grows polynomially with $N$~\cite{bravyi2019simulation,dias2024classical}.
The cost of this simulation is thus governed by the Gaussian rank
\begin{equation}
  \chi(\psi)=\min\qty{N\ \middle|\ \ket{\psi}=\sum_{j=1}^{N}c_{j}\ket{\varphi_{j}}},
  \label{eq:Gaussian_rank}
\end{equation}
where the minimum is taken over pure Gaussian states $\ket{\varphi_{j}}$ and complex coefficients $c_{j}$.
We note that this definition does not require the Gaussian states $\ket{\varphi_{j}}$ to be orthogonal to each other.

In practical situations, it is typical to allow an approximation error $\delta$ in the state.
The simulation cost is then governed by the $\delta$-approximate Gaussian rank~\cite{dias2024classical}
\begin{equation}
  \chi_{\delta}(\psi)=\min\qty{\chi(\psi')\ \middle|\ \norm{\ket{\psi}-\ket{\psi'}}\leq\delta},
  \label{eq:approx_Gaussian_rank}
\end{equation}
where the minimum is taken over all vectors $\ket{\psi'}$, so that $\chi_{\delta}\leq\chi_{0}=\chi$.
The gap in this inequality can be exponentially large, since an $L$-fold tensor product of non-Gaussian states can have $\chi$ growing exponentially in $L$ while $\chi_{\delta}=1$ at a fixed $\delta>0$~\cite{reardon-smith2026exponential}.
Our result is then given as follows:
\begin{theorem}
  \label{thm:gaussian-rank}
  For any sequence of pure states $\ket{\psi_{L}}$ of $L$ fermionic modes and any fixed $0\leq\delta<1$,
  \begin{equation}
    \liminf_{L\to\infty}\frac{1}{L}\ln\chi_{\delta}(\psi_{L})\geq\frac{m_{2}}{4}-c_{*},
    \label{eq:rank_bound}
  \end{equation}
  where $m_{2}=\liminf_{L\to\infty}M_{2}(\psi_{L})/L$ is the density of the second-order fermionic MRE, and $c_{*}=0.0615\ldots$ is a constant.
\end{theorem}
The proof and the closed form of $c_{*}$ are given in Appendix~\ref{app:rank_proof}.
In short, this theorem states that the Gaussian rank grows exponentially in $L$ whenever $m_{2}>4c_{*}$.
The threshold $4c_{*}=0.246\ldots$ is below the maximum MRE density $2\ln(4/3)=0.575\ldots$ per mode~\cite{hoshino2026universal}.

Exponential lower bounds on the stabilizer or Gaussian rank are important for understanding quantum computational advantage, as they constrain classical simulation based on explicit superpositions of free states.
In general, such bounds are difficult to obtain because they must rule out every decomposition with fewer terms.
For the magic state $\ket{T}=(\ket{0}+e^{i\pi/4}\ket{1})/\sqrt{2}$, exponential growth of the stabilizer rank with the number of copies was conjectured when the rank was introduced a decade ago~\cite{bravyi2016trading}.
Yet the best known lower bound remains nearly quadratic~\cite{mehraban2024quadratic}.
Recent exponential lower bounds on fermionic and bosonic Gaussian ranks apply to tensor products of non-Gaussian states~\cite{reardon-smith2026exponential,wei2026robust}.
To our knowledge, Theorem~\ref{thm:gaussian-rank} provides the first state-dependent bound that can certify exponential growth of the fermionic Gaussian rank without assuming a tensor-product structure.
The bound constrains all Gaussian decompositions through $M_{2}$, which is computed directly from the state without optimization over decompositions.

We remark that such an exponential lower bound can not be achieved from the covariance-matrix-based measures of fermionic non-Gaussianity~\cite{sierant2026fermionic,tarabunga2026computable,leone2026fermionic,turkeshi2026williamson,bermejo2025characterizing}.
For example, the GHZ state $(\ket{0\cdots0}+\ket{1\cdots1})/\sqrt{2}$ and $\ket{M}^{\otimes L/4}$ with $\ket{M}=(\ket{0000}+\ket{1111})/\sqrt{2}$, the standard magic state for matchgate computation~\cite{hebenstreit2019all}, both of $L$ modes, have vanishing Majorana covariance matrices.
However, the GHZ state is a superposition of only two Gaussian states, whereas the Gaussian rank of $\ket{M}^{\otimes L/4}$ grows exponentially in $L$~\cite{reardon-smith2026exponential,wei2026robust}.
The MRE distinguishes them, as it remains bounded for the GHZ state~\cite{hoshino2026universal} but grows linearly in $L$ for $\ket{M}^{\otimes L/4}$:
\begin{align}
  M_2(\mathrm{GHZ})          & =\ln\frac{2^{L+1}}{2^{L-2}+7}\leq\ln8, \\
  M_2(\ket{M}^{\otimes L/4}) & =\frac{L}{4}\ln\frac{32}{11}.
\end{align}
Since the MRE density of $\ket{M}^{\otimes L/4}$ exceeds the threshold $4c_*$, our Theorem~\ref{thm:gaussian-rank} certifies the exponential growth of its Gaussian rank from the MRE alone.

This bound allows the MRE to diagnose the cost of Gaussian decompositions in many-body systems.
For example, chaotic evolution under the Sachdev-Ye-Kitaev model produces states with nearly maximal MRE density~\cite{hoshino2026universal}.
Our theorem thus rules out their approximation by a superposition of polynomially many Gaussian states. At the same time, however, not every interacting many-body state is costly. In quantum impurity problems, for example, a superposition of $\exp[O(b^3)]$ Gaussian states can achieve an energy within $2^{-b}$ of the ground-state energy~\cite{bravyi2017complexity}. 
The MRE itself also exhibits distinct scaling behaviors depending on many-body states; it grows linearly in critical states while it grows logarithmically or saturates in the open XXZ chain of spinful fermions, depending on the anisotropy~\cite{matsuda2026efficient}.

\section{\label{sec:field_theory_MRE}Field theory of the magic R\'enyi entropy}

To formulate a field-theoretical approach to the MRE, we recast the replica construction in Liouville space.
The definition of the MRE contains a sequence of operations: the convolution unitary first mixes the replicas, the output replicas are traced out, and the purity of the remaining state is evaluated.
Vectorizing operators as superkets converts traces and partial traces into overlaps with maximally entangled states (MESs).
The partition function, the purity, and the MRE are then represented by the same doubled Euclidean path integral, distinguished only by the boundary state that sews the replicas.
In this formulation, the convolution unitary $U$ is localized as a line defect.
After folding, this insertion can be assigned either to a nontrivial boundary state or to inter-replica couplings in the bulk.
As a result, we identify the universal constant term of the MRE by the Affleck-Ludwig boundary entropy of the corresponding boundary conformal field theory (BCFT).

\subsection{\label{subsec:liouville}Liouville-space formulation}

We use the Liouville-space conventions of Refs.~\cite{gyamfi2020fundamentals,schmutz1978realtime}.
The construction begins by fixing a basis of the original Hilbert space $\mathscr{H}$.
To be concrete, below we use the language of bosons and fermions, where the Fock basis $\ket{\bm{m}}=\ket{m_{1},m_2,\ldots,m_{L}}$ associated with the operators of Sec.~\ref{subsec:boson-fermion} is used; however, the formulation can be extended to qubits and qudits by interpreting $\ket*{\bm{m}}$ as the computational basis.

Once the basis is fixed, each operator on $\mathscr{H}$ can be viewed as a vector.
The resulting Liouville space, denoted by $\bar{\mathscr{H}}$, is equipped with the Hilbert-Schmidt inner product
\begin{align}
  \sbk{A}{B} \equiv \tr\!\qty[A^{\dag}B].
  \label{eq:HS_inner}
\end{align}
We denote the vectorized form of $A$ by the superket $\sket{A}$, and call linear operators acting on $\bar{\mathscr{H}}$ superoperators.
We define the elementary superkets as follows:
\begin{align}
  \sket{\bm{m},\bm{n}}
  \equiv \sket{\,\ket{\bm{m}}\!\bra{\bm{n}}\,},
  \label{eq:Liouville_basis}
\end{align}
which form an orthonormal basis of $\bar{\mathscr{H}}$,
\begin{align}
  \sbk{\bm{m},\bm{n}}{\bm{m}',\bm{n}'}
  = \delta_{\bm{m}\bm{m}'}\delta_{\bm{n}\bm{n}'}.
  \label{eq:Liouville_orth}
\end{align}
Consequently, any superket admits the following expansion:
\begin{align}
  \sket{A}
   & = \sum_{\bm{m},\bm{n}}\bra{\bm{m}}A\ket{\bm{n}}\,
  \sket{\bm{m},\bm{n}},
  \nonumber                                            \\
  \sbk{\bm{m},\bm{n}}{A}
   & = \bra{\bm{m}}A\ket{\bm{n}}.
  \label{eq:superket_expand}
\end{align}

We define the superoperators $a_{i}$ and $\tilde{a}_{i}$ by their action on the basis~\eqref{eq:Liouville_basis}:
\begin{align}
  a_{i}\sket{\bm{m},\bm{n}}
   & \equiv \sket{a_{i}\ket{\bm{m}}\!\bra{\bm{n}}},
  \label{eq:superop_a}                              \\
  \tilde{a}_{i}\sket{\bm{m},\bm{n}}
   & \equiv \eta^{\mu+1}
  \sket{\,\ket{\bm{m}}\!\bra{\bm{n}}a_{i}^{\dag}\,},
  \label{eq:superop_atilde}
\end{align}
where $\mu=\sum_{i=1}^{L}(m_{i}-n_{i})$ is the particle-number difference between the left and right sectors of the basis vector, and the superoperator $a_{i}$ acting on the Liouville space is denoted by the same symbol as the original operator for notational simplicity.
The adjoint superoperators $a_{i}^{\dag}$ and $\tilde{a}_{i}^{\dag}$ are determined through the inner product~\eqref{eq:HS_inner}.
These superoperators satisfy the following algebra:
\begin{align}
  [a_{i},a_{k}^{\dag}]_{\eta}
   & = [\tilde{a}_{i},\tilde{a}_{k}^{\dag}]_{\eta}
  = \delta_{ik},
  \label{eq:superalg_diag}                         \\
  [a_{i},\tilde{a}_{k}]_{\eta}
   & = [a_{i}^{\dag},\tilde{a}_{k}]_{\eta}
  = [a_{i},\tilde{a}_{k}^{\dag}]_{\eta}
  = [a_{i}^{\dag},\tilde{a}_{k}^{\dag}]_{\eta}
  = 0,
  \label{eq:superalg_cross}                        \\
  [a_{i},a_{k}]_{\eta}
   & = [a_{i}^{\dag},a_{k}^{\dag}]_{\eta}
  = [\tilde{a}_{i},\tilde{a}_{k}]_{\eta}
  = [\tilde{a}_{i}^{\dag},\tilde{a}_{k}^{\dag}]_{\eta}
  = 0.
  \label{eq:superalg_same}
\end{align}
Thus $\bar{\mathscr{H}}$ can be regarded as the Hilbert space of $2L$ degrees of freedom with the same statistics as the original system.
The normalized Fock states are $\ket{\bm{m}}=\prod_{i=1}^{L}{(a_i^{\dag})^{m_i}}/{\sqrt{m_i!}}\ket{0}$ with the product ordered by increasing $i$.
The Liouville basis is generated from the supervacuum $\sket{0,0}\equiv\sket{\,\ket{0}\!\bra{0}\,}$ as
\begin{align}
  \sket{\bm{m},\bm{n}}
  & = \frac{\eta^{N_{\bm{n}}(N_{\bm{n}}-1)/2}}{\sqrt{\prod_i m_i!n_i!}}
  \left[\prod_{i=1}^{L}(a_i^{\dag})^{m_i}\right]
  \nonumber \\
  & \quad\times\left[\prod_{i=1}^{L}(\tilde a_i^{\dag})^{n_i}\right]\sket{0,0},
  \label{eq:Fock_Liouville}
\end{align}
where $N_{\bm{n}}=\sum_i n_i$ and both products are ordered increasingly.
The identity operator is represented formally by the unnormalized MES,
\begin{align}
  \sket{I}
  \equiv \sum_{\bm{m}}\sket{\bm{m},\bm{m}}
  = \exp\!\left(\sum_{i=1}^{L}a_{i}^{\dag}\tilde{a}_{i}^{\dag}\right)
  \!\sket{0,0}.
  \label{eq:MES}
\end{align}
Strictly speaking, $\sket{I}$ is not a Hilbert-Schmidt vector for bosons since $\sbk{I}{I}$ diverges;
we interpret it as a distribution under a suitable regularization. 
With this prescription, the trace of any operator $A$ is given by
\begin{align}
  \tr[A] = \sbk{I}{A}.
  \label{eq:trace_MES}
\end{align}

More generally, the superoperator that acts by left multiplication with an operator $A$ is denoted by the same symbol $A$.
The auxiliary-sector superoperator $\tilde{A}$ is obtained by applying the same construction to $A^{*}$ in the tilde sector, where $A^{*}$ denotes the complex conjugation.
For particle-number-conserving operators $A$ and $B$ the following identities hold:
\begin{align}
  \sket{A}
   & = A\sket{I}
  = \tilde{A}^{\dag}\sket{I},
  \label{eq:id_A}       \\
  \sket{AB}
   & = A\sket{B},\qquad
  \sket{BA}
  = \tilde{A}^{\dag}\sket{B}.
  \label{eq:id_AB}
\end{align}
We refer to Ref.~\cite{schmutz1978realtime} for a non-particle-conserving case.
Equations~\eqref{eq:id_A} and~\eqref{eq:id_AB} immediately imply the following conjugation identity for any particle-number-conserving unitary $U$
\begin{align}
  \sket{UAU^{\dag}}
  = U\tilde{U}\sket{A}.
  \label{eq:id_conjugation}
\end{align}

When several replicas are present, the creation and annihilation superoperators within each replica satisfy the same algebra as the single-replica case, while operators from different replicas $r$ and $s$ ($r\ne s$) commute or anticommute for bosons or fermions, respectively:
\begin{align}
  [a_{r,i},a_{s,j}]_{\eta}
  = [a_{r,i},\tilde{a}_{s,j}]_{\eta}
  = [\tilde{a}_{r,i},\tilde{a}_{s,j}]_{\eta}
  = 0.
  \label{eq:inter_replica}
\end{align}
Using the two-replica Liouville space $\bar{\mathscr{H}}^{\otimes 2}$, the purity of $\rho$ can be expressed as
\begin{align}
  \tr\!\qty[\rho^{2}]
  = \sbk{\rho}{\rho}
  = \sbra{I}_{1\tilde{2}}\sbra{I}_{\tilde{1}2}\,\qty(\sket{\rho}_{1}\otimes\sket{\rho}_{2}).
  \label{eq:purity_replica}
\end{align}
Here $\sket{I}_{r\tilde{s}}$ and $\sket{I}_{\tilde{r}s}$ denote the MESs pairing the original sector of one replica with the auxiliary sector of the other:
\begin{align}
  \sket{I}_{r\tilde{s}}
   & \equiv \exp\!\left(\sum_{i=1}^{L}
  a_{r,i}^{\dag}\tilde{a}_{s,i}^{\dag}\right)
  \!\sket{0,0}_{r\tilde{s}},
  \nonumber                            \\
  \sket{I}_{\tilde{r}s}
   & \equiv \exp\!\left(\sum_{i=1}^{L}
  \tilde{a}_{r,i}^{\dag}a_{s,i}^{\dag}\right)
  \!\sket{0,0}_{\tilde{r}s}.
  \label{eq:cross_MES}
\end{align}
This cross-replica MES will play a central role in the folded picture of the path-integral formulation below.

\subsection{\label{subsec:path_integral}Path integral in Liouville space}

We next translate the Liouville-space identities into the Euclidean path-integral language.
This step fixes the reference boundary condition associated with the MES before the convolution unitary is inserted.
We consider the Gibbs state
\begin{align}
  \rho
  = \frac{1}{Z}\,e^{-\beta H},\qquad
  Z = \tr\!\qty[e^{-\beta H}],
  \label{eq:Gibbs}
\end{align}
which reduces to the ground-state density operator in the limit $\beta\to\infty$. While we will ultimately use this limit to extract the universal constant term, keeping finite $\beta$ makes the folding structure transparent.

\subsubsection{\label{subsubsec:Z}Partition function}

The Liouville-space representation of Eq.~\eqref{eq:Gibbs} follows from Eq.~\eqref{eq:id_A}:
\begin{align}
  \sket{\rho}
  = \frac{1}{Z}\,
  e^{-(\beta/2)H}\otimes e^{-(\beta/2)\tilde{H}}
  \sket{I}.
  \label{eq:rho_Liouville}
\end{align}
The trace condition $1=\tr[\rho]=\sbk{I}{\rho}$ gives
\begin{align}
  Z
  = \sbra{I}\,
  e^{-(\beta/2)H}\otimes e^{-(\beta/2)\tilde{H}}\,
  \sket{I}.
  \label{eq:Z_Liouville}
\end{align}
Equation~\eqref{eq:Z_Liouville} is a path integral on the strip $\tau\in[0,\beta/2]$ with the MES boundary state placed at both ends.
Compared with the ordinary thermal path integral on $\tau\in[0,\beta]$, the time interval is halved and the field content is doubled.
The two copies are governed by the additive Liouville-space bulk Hamiltonian $H\otimes I+I\otimes \tilde{H}$.
Equivalently, this is the standard folding of the identity defect of the original theory.
The state $\sket{I}$ should therefore be understood as the boundary state that sews the original and auxiliary fields along the folded defect.
The explicit sewing conditions are provided in Appendix~\ref{app:bc_MES}.
The upper panel of Fig.~\ref{fig:Liouville_W_folded} illustrates this folded-strip representation of Eq.~\eqref{eq:Z_Liouville}.

\subsubsection{\label{subsubsec:purity}Purity}

The same folded picture also gives a useful representation of the purity.
For the Gibbs state~\eqref{eq:Gibbs}, one has
\begin{align}
  \tr\!\qty[\rho^{2}]
  = \frac{1}{Z^{2}}\,W,\qquad
  W = \tr\!\qty[e^{-2\beta H}],
  \label{eq:purity_Z}
\end{align}
and the nontrivial path integral corresponds to $W$.
In Liouville space, this numerator is
\begin{align}
  W
  = \sbra{I}\,
  e^{-\beta H}\otimes e^{-\beta \tilde{H}}\,
  \sket{I},
  \label{eq:W_Liouville}
\end{align}
which is the same folded path integral as Eq.~\eqref{eq:Z_Liouville}, but now on the longer interval $\tau\in[0,\beta]$.
To calculate the MRE, it is useful to fold this expression once more so that the purity itself is represented by a sewing between two Liouville replicas.
Specifically, using Eq.~\eqref{eq:purity_replica}, we obtain
\begin{align}
  W
  = \sbra{I}_{1\tilde{2}}\sbra{I}_{\tilde{1}2}\,
  \qty(
  e^{-(\beta/2)H}\otimes e^{-(\beta/2)\tilde{H}}
  \sket{I})^{\otimes 2},
  \label{eq:W_folded}
\end{align}
which is a path integral on $\tau\in[0,\beta/2]$ in the two-replica Liouville space.
The boundary at one end is the product $\sket{I}^{\otimes 2}=\sket{I}_{1\tilde{1}}\sket{I}_{2\tilde{2}}$ of the boundary states within the same replicas, while the other end is sewn by the cross-replica MESs $\sket{I}_{1\tilde{2}}\sket{I}_{\tilde{1}2}$ (cf. Eq.~\eqref{eq:cross_MES}).
Thus, taking a purity amounts to replacing the ordinary MES boundary by a boundary state that exchanges the original sector of one replica with the auxiliary sector of the other.
The lower panel of Fig.~\ref{fig:Liouville_W_folded} shows the folding structure that gives Eq.~\eqref{eq:W_folded}.

\begin{figure}[b]
  \centering
  \includegraphics[width=\columnwidth]{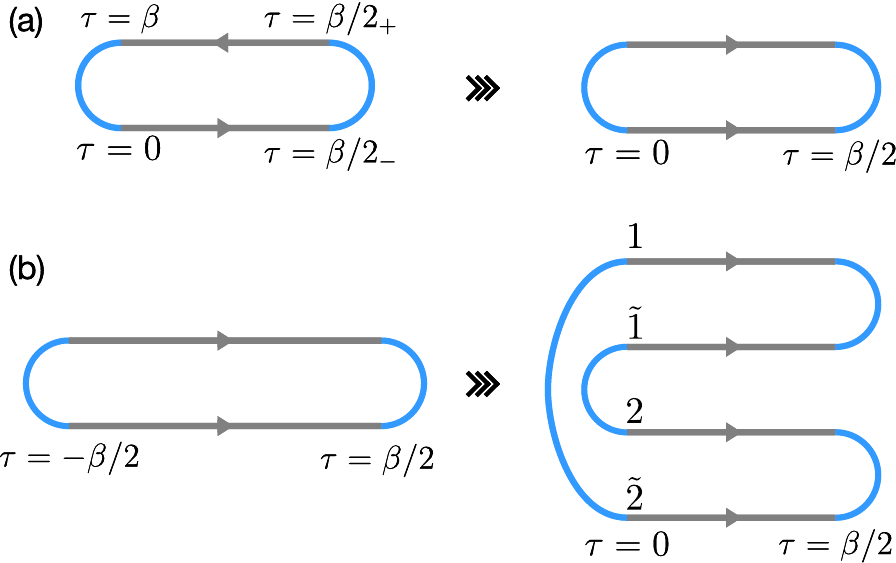}
  \caption{Folding representations of the thermal partition function and the purity numerator. The upper panel (a) shows the ordinary thermal path integral (left) and its Liouville-space folded-strip representation (right) in Eq.~\eqref{eq:Z_Liouville}. The lower panel (b) shows the corresponding purity geometry and its second folding into the two-replica Liouville-space expression in Eq.~\eqref{eq:W_folded}.}
  \label{fig:Liouville_W_folded}
\end{figure}

\subsection{Field theory of the MRE}
\label{subsec:MRE_field_theory}

We now present the path-integral formulation of the self-convolution~\eqref{eq:self_convolution} that defines the MRE~\eqref{eq:MRE_def}.
The new ingredient relative to the purity is the unitary $U$ that mixes the input replicas before replicas $2,3,\ldots,n$ are traced out.
Specifically, the convolved state~\eqref{eq:self_convolution} is represented in Liouville space as
\begin{align}
  \sket{\mathcal{C}_{n}(\rho)}
  = \sbra{I}_{1^{\rm c}}\,
  U\otimes\tilde{U}\,\sket{\rho}^{\otimes n},
  \label{eq:Crho_Liouville}
\end{align}
where $\sbra{I}_{1^{\rm c}}\equiv\bigotimes_{r=2}^{n}\sbra{I}_{r\tilde{r}}$ implements the partial trace over the discarded output replicas.
Substituting Eq.~\eqref{eq:Crho_Liouville} into Eq.~\eqref{eq:MRE_def} gives
\begin{align}
  M_{n}(\rho)
   & = \frac{1}{1-n}\ln\sbk{\mathcal{C}_{n}(\rho)}{\mathcal{C}_{n}(\rho)}\notag \\
   & = \frac{1}{1-n}\ln\sbra{\rho}^{\otimes n}P_{n}\sket{\rho}^{\otimes n},
  \label{eq:MRE_projector}
\end{align}
where the trace-and-purity operation is encoded in the replica sewing insertion
\begin{align}
  P_{n}
  = \qty(U\otimes\tilde{U})^{\!\dag}
  \qty(
  I_{1\tilde 1}\otimes
  \sket{I}_{1^{\rm c}}\sbra{I}_{1^{\rm c}}
  )
  \qty(U\otimes\tilde{U}).
  \label{eq:Pn_def}
\end{align}
Here, $I_{1\tilde 1} = I_1\otimes \tilde I_{\tilde 1}$ denotes the identity superoperator on the retained output replica 1, and $\sket{I}_{1^{\rm c}}$ is the product of the MESs within each replica $2,3,\ldots,n$, which is used to trace out the discarded replicas.
Since $\sbk{I}{I}=\tr I$ for each MES, $P_{n}$ is a projector up to normalization, $P_{n}^{2}=(\tr I)^{n-1}P_{n}$. For a thermal state~\eqref{eq:Gibbs}, the matrix element in Eq.~\eqref{eq:MRE_projector} becomes a ratio of Euclidean partition functions,
\begin{align}
  \sbra{\rho}^{\otimes n}P_{n}
  \sket{\rho}^{\otimes n}
  = \frac{Z_{n}}{Z^{2n}},
  \label{eq:MRE_partition_ratio}
\end{align}
where $Z=\tr[e^{-\beta H}]$ is the ordinary partition function while
\begin{align}
  Z_{n}
   & = \sbra{I}^{\otimes n}\,
  \qty(e^{-(\beta/2)H}\otimes
  e^{-(\beta/2)\tilde{H}})^{\!\otimes n}\,
  \,
  \nonumber                   \\
   & \quad\times\, P_{n}\,
  \qty(e^{-(\beta/2)H}\otimes
  e^{-(\beta/2)\tilde{H}})^{\!\otimes n}\,
  \sket{I}^{\otimes n}
  \label{eq:Zn_def}
\end{align}
is a path integral on $[-\beta/2,\beta/2]$ in the $n$-replica Liouville space, with the operator $P_n$ inserted at $\tau=0$ and the reference MES boundary $\sket{I}^{\otimes n}$ placed at the two outer ends [Fig.~\ref{fig:MRE_BCFT}(a)].
The inserted operator $P_n$ can be seen as a line defect that first rotates the replica fields by $U\otimes\tilde{U}$ and then imposes the MES sewing associated with the partial trace.
For $n=1$, $P_1$ is the identity, and $Z_1$ reduces to the numerator $W$ of the purity in Eq.~\eqref{eq:purity_Z}.

The partition function $Z_n$ can be analyzed by folding the theory at the line defect $P_n$, i.e., by folding the two sides of $\tau=0$ against each other in the same manner as Eq.~\eqref{eq:W_folded}.
We write the replicas on the two sides of the insertion as $\mathsf{R}=(1,2,\ldots,n)$ and $\mathsf{R}'=(1',2',\ldots,n')$, respectively, and similarly denote their tilde copies by $\tilde{\mathsf{R}}=(\tilde{1},\tilde{2},\ldots,\tilde{n})$ and $\tilde{\mathsf{R}}'=(\tilde{1}',\tilde{2}',\ldots,\tilde{n}')$.
After this folding, the MRE can be written as
\begin{align}
  M_{n}(\rho)
  = \frac{1}{1-n}
  \ln\sbra{\mathcal{J}}\,\mathcal{U}\,
  \sket{\rho}^{\otimes 2n},
  \label{eq:MRE_folded}
\end{align}
where the unitary acting on the 4-folded replica space is
\begin{align}
  \mathcal{U}
   & = \qty(U_{\mathsf{R}}\otimes
  \tilde{U}_{\tilde{\mathsf{R}}})
  \otimes
  \qty(U_{\mathsf{R}'}\otimes
  \tilde{U}_{\tilde{\mathsf{R}}'}),
  \label{eq:U2n_def}
\end{align}
and the inter-replica boundary state is
\begin{align}
  \sket{\mathcal{J}}
  = \sket{I}_{1\tilde{1}'}\!\otimes\!
  \sket{I}_{\tilde{1}1'}\!\otimes\!
  \qty[\bigotimes_{r=2}^{n}
    \qty(\sket{I}_{r\tilde{r}}\!\otimes\!\sket{I}_{r'\tilde{r}'})].
  \label{eq:J2n_def}
\end{align}
The first two MESs in Eq.~\eqref{eq:J2n_def} sew the retained output replicas that enter the final expression of the purity, while the remaining MESs implement the partial traces over the other replicas.
The path-integral representation of Eq.~\eqref{eq:MRE_folded}, shown in Fig.~\ref{fig:MRE_BCFT}(b), is a strip of imaginary-time length $\beta/2$, where the bulk can be viewed as a $4n$-sheet folded theory.
Here, one end of the strip carries the reference boundary state $\sket{\mathcal{I}}\equiv\sket{I}^{\otimes 2n}$, which is the product of $2n$ elementary MESs.
At the other end, the cross-replica boundary state $\sket{\mathcal{J}}$ in Eq.~\eqref{eq:J2n_def} implements the purity and trace sewings, and the adjacent line defect ${\cal U}$ encodes the convolution rotation.

\begin{figure}[t]
  \centering
  \includegraphics[width=\columnwidth]{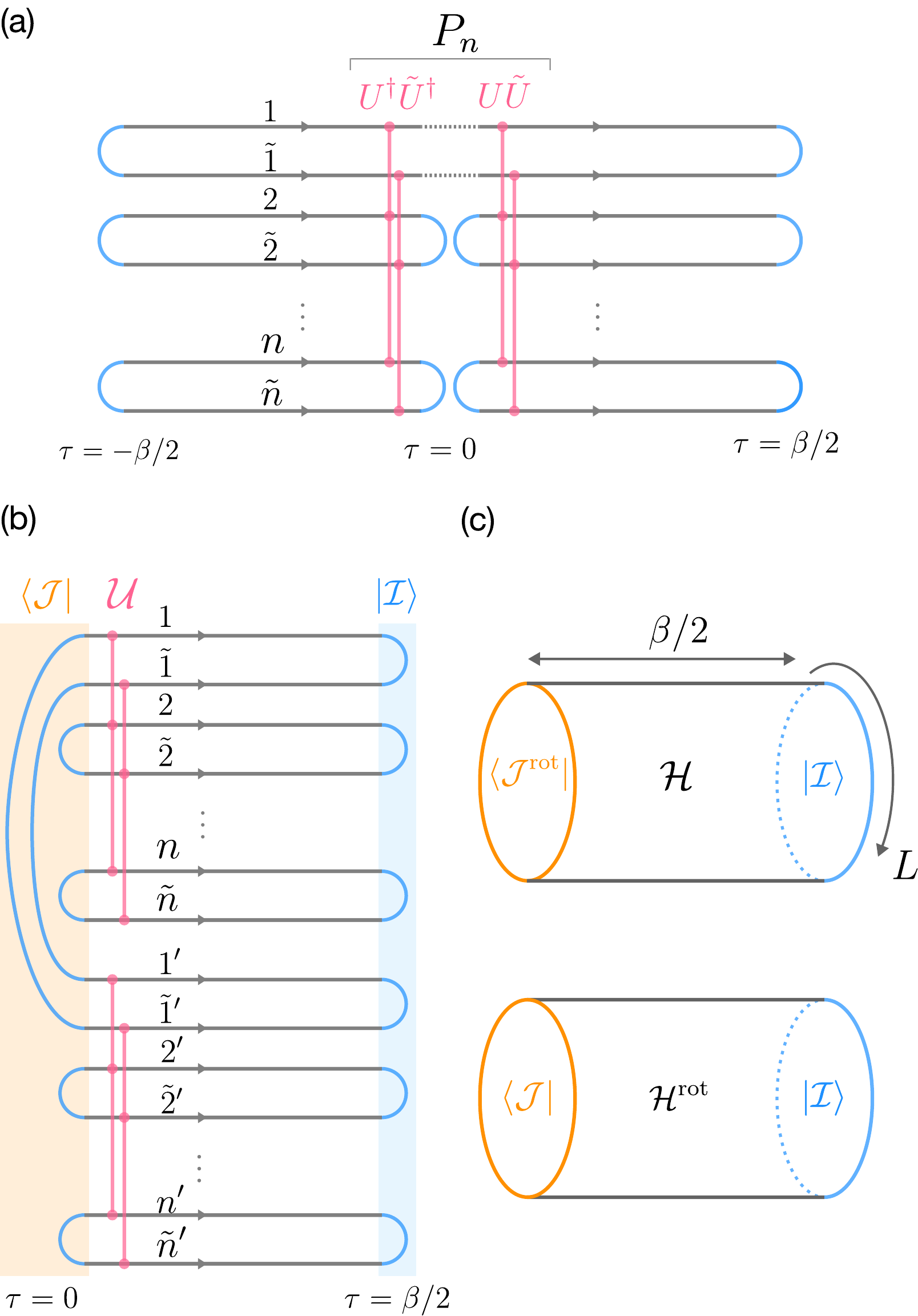}
  \caption{Euclidean path-integral representation of the MRE in Eq.~\eqref{eq:MRE_folded}.
    (a) Path integral in the $n$-replica Liouville space with $2n$ sheets of size $L\times \beta$. (b) $4n$-sheet folded theory obtained by folding (a) at $\tau=0$.
    In both panels, the spatial direction is abbreviated for simplicity.
    (c) Two interpretations of the MRE.
    The upper panel describes the rotated-boundary picture.
    While the bulk theory $\cal H$ consists of the product of each decoupled replica, an application of the convolution unitary to the edge leads to the rotated boundary state $\sket{\mathcal{J}^{\text{rot}}}=\mathcal{U}^{\dag}\sket{\mathcal{J}}$.
    The lower panel describes the rotated-bulk picture.
    The convolution unitary is applied to the bulk and leads to the rotated bulk Hamiltonian $\mathcal{H}^{\text{rot}}={\cal U}{\cal H}{\cal U}^\dagger$ with inter-replica couplings, while the boundary is kept as the unrotated state $\sket{\mathcal{J}}$.}
  \label{fig:MRE_BCFT}
\end{figure}

\subsubsection*{Rotated-boundary picture}

There are two equivalent ways to evaluate Eq.~\eqref{eq:MRE_folded}.
The first method is to keep the bulk theory as a tensor product of $4n$ sheets and absorb the convolution unitary into the boundary at $\tau=0$ (see the upper panel of Fig.~\ref{fig:MRE_BCFT}(c)).
Defining
\begin{align}
  \sket{\mathcal{J}^{\text{rot}}}
  = \mathcal{U}^{\dag}\sket{\mathcal{J}},
  \label{eq:Gamma2n_def}
\end{align}
the replica partition function becomes
\begin{align}
  Z_{n}
  = \sbra{\mathcal{J}^{\text{rot}}}\,
  e^{-(\beta/2)\mathcal{H}}\,
  \sket{\mathcal{I}},
  \label{eq:Zn_boundary}
\end{align}
where $\mathcal{H} = \sum_{r\in\mathsf{R}\cup\mathsf{R}'}\qty(H_{r}\otimes I+I\otimes \tilde{H}_{\tilde{r}})$ denotes the total Hamiltonian of the $4n$-folded tensor-product theory, in which $H_{r}$ is the Hamiltonian $H$ acting on replica $r$ and the bulk is decoupled across all sheets.
Equation~\eqref{eq:Zn_boundary} is a path integral on a strip between the rotated boundary state $\sket{\mathcal{J}^{\text{rot}}}$ and the product MES state $\sket{\mathcal{I}}$.
The former, $\sket{\mathcal{J}^{\text{rot}}}$, incorporates the action of the convolution unitary, and its conformal invariance is not \textit{a priori} guaranteed because $\cal U$ does not generally act as a simple linear transformation in the IR limit.
We address this question of conformal invariance in Sec.~\ref{sec:conformal_invariance} for a class of density-density interacting models.
We note that the latter, $\sket{\mathcal{I}}=\sket{I}^{\otimes 2n}$, corresponds to the folded identity-defect boundary and defines a conformal boundary state in the IR limit.

When both ends define conformal boundary states, the low-temperature limit $\beta\gg L$ is controlled by the closed-channel ground state of the replicated CFT.
The logarithm of $Z_{n}$ then contains the bulk contributions, the nonuniversal boundary line energies, and the Affleck-Ludwig $g$ factor~\cite{affleck1991universal}.
After including the reference factor $Z^{2n}$, the bulk term cancels and the remaining large-$L$ behavior takes the form
\begin{align}
  -\ln\frac{Z_{n}}{Z^{2n}}
  = b L - \ln g + o(1),
  \label{eq:free_energy_BCFT}
\end{align}
where $bL$ is a nonuniversal line energy, and $\ln g$ is the universal boundary entropy.
Note that both $b$ and $g$ depend on $n$ through the convolution unitary $U$.
In a closed-channel normalization, this entropy factor can be written as
\begin{align}
  \ln g
   & = \ln g_{\mathcal{J}}-\ln g_{\mathcal{I}},
  \label{eq:rel_g}                                         \\
  g_{\mathcal{J}}
   & = \sbk{\mathcal{J}^{\text{rot}}}{\mathrm{CFT}},\qquad
  g_{\mathcal{I}}
  = \sbk{\mathcal{I}}{\mathrm{CFT}},
  \label{eq:lng_def}
\end{align}
where $\sket{\mathrm{CFT}}$ is the CFT ground state of the $4n$-folded tensor-product bulk theory.
The subtraction by $g_{\mathcal{I}}$ accounts for the fact that the MRE measures the convolution boundary against the same MES boundary that normalizes the ordinary partition functions $Z^{2n}$.
Throughout this paper, we normalize the artificial boundary state $\sket{\mathcal{I}}$ created by the folding so that $g_{\mathcal{I}}=1$, and hence the relative $g$ factor in Eq.~\eqref{eq:rel_g} is denoted simply by $g_{\mathcal{J}}$.
Combining Eqs.~\eqref{eq:MRE_partition_ratio} and~\eqref{eq:free_energy_BCFT} with Eq.~\eqref{eq:MRE_folded}, we obtain the ground-state scaling of the MRE as
\begin{align}
  M_{n}(\rho)
  = m_{n} L - s_{n} + o(1),\qquad
  s_{n} = \frac{\ln g_{\mathcal{J}}}{n-1},
  \label{eq:Mn_BCFT}
\end{align}
where the size-independent term $s_n$ is determined by the universal factor $g_{\mathcal{J}}$, while the slope $m_{n}=b/(n-1)$ is nonuniversal.

The rotated-boundary picture explained here is useful when one can identify the IR field configuration of the rotated boundary state $\sket{\mathcal{J}^{\text{rot}}}$.
A concrete example is the Ising critical state, for which the MRE reduces to the SRE, and the $g$ factor of $\sket{\mathcal{J}^{\text{rot}}}$ has been derived as \cite{hoshino2026stabilizer,hoshino2026stabilizera}
\begin{equation}
  g_{{\cal J},\text{Ising}}=\sqrt{n}.
\end{equation}
More generally, for qubits, one can use the convolution-multiplication duality~\eqref{eq:duality_qudit_general} with $d=2$ to express $P_n$ as
\begin{equation}
  P_n = \frac{1}{2^L} \sum_{\bm{u}} (\sket{\sigma(\bm{u})}\sbra{\sigma(\bm{u})})^{\otimes n}.\label{eq:bell_measurement}
\end{equation}
Since the Liouville-space expressions of the Pauli strings are Bell states between the original and tilde sectors, $P_n$ imposes the field configuration specified by these Bell measurements.
Once this field configuration is determined, the conformal boundary condition imposed by $\sket{\mathcal{J}^{\text{rot}}}$ may be inferred from the corresponding replicated theory.

Nonetheless, several points often remain unclear within this point of view.
First, it is typically not obvious how one can determine the IR field configuration specified by $P_n$, even when we know its microscopic expression such as Eq.~\eqref{eq:bell_measurement}.
Second, such a microscopic expression does not allow us to discuss the RG flow of the boundary state. In a field-theoretic analysis of systems under quantum measurements, one can treat the measurement strength as the coupling constant of the boundary action.
Equation~\eqref{eq:bell_measurement}, for instance, corresponds to a projective measurement, and this coupling constant becomes infinitely large, which in general obstructs discussion of RG properties~\cite{ashida2024systemenvironment, hoshino2025entanglement}.
Third, it is hard to find an analogous microscopic expression of $P_n$ for fermions.
These issues motivate an alternative interpretation introduced below.

\subsubsection*{Rotated-bulk picture}

Another way to interpret Eq.~\eqref{eq:MRE_folded} is to keep the boundary at $\tau=0$ in the elementary MES form $\sket{\mathcal{J}}$ and instead rotate the bulk evolution (see the lower panel of Fig.~\ref{fig:MRE_BCFT}(c)).
This follows by inserting $\mathcal{U}^{\dag}\mathcal{U}=1$ next to $\sket{\mathcal{I}}$ and using the identity $\mathcal{U}^{\dag}\sket{\mathcal{I}}=\sket{\mathcal{I}}$, which is the replicated version of the MES invariance under $U\otimes\tilde{U}$ (cf. Eq.~\eqref{eq:U2n_def}).

To write down the rotated theory, we recall that the bulk is the folded $4n$-replica theory that consists of the four replica sectors $\mathsf{R},\tilde{\mathsf{R}},\mathsf{R}',\tilde{\mathsf{R}}'$.
With this notation, we represent the replicated Hamiltonian in each sector by
\begin{align}
  H_{\mathsf{R}}
   & \equiv \sum_{r=1}^{n} H_{r},\qquad
  \tilde{H}_{\tilde{\mathsf{R}}}
  \equiv \sum_{r=1}^{n} \tilde{H}_{\tilde{r}},
  \nonumber                              \\
  H_{\mathsf{R}'}
   & \equiv \sum_{r=1}^{n} H_{r'},\qquad
  \tilde{H}_{\tilde{\mathsf{R}}'}
  \equiv \sum_{r=1}^{n} \tilde{H}_{\tilde{r}'}.
  \label{eq:replica_H_def}
\end{align}
The rotated Hamiltonian of the full $4n$-folded theory, $\mathcal{H}^{\mathrm{rot}}=\mathcal{U}\,\mathcal{H}\,\mathcal{U}^{\dag}$, then decomposes across the four sectors as
\begin{align}
  \mathcal{H}^{\mathrm{rot}}
   & =
  U_{\mathsf{R}}H_{\mathsf{R}}U_{\mathsf{R}}^{\dag}
  +\tilde{U}_{\tilde{\mathsf{R}}}\tilde{H}_{\tilde{\mathsf{R}}}\tilde{U}_{\tilde{\mathsf{R}}}^{\dag}
  \nonumber \\
   & \quad
  +U_{\mathsf{R}'}H_{\mathsf{R}'}U_{\mathsf{R}'}^{\dag}
  +\tilde{U}_{\tilde{\mathsf{R}}'}\tilde{H}_{\tilde{\mathsf{R}}'}\tilde{U}_{\tilde{\mathsf{R}}'}^{\dag}.
  \label{eq:Hrot_2n}
\end{align}
The partition function can then be written as
\begin{align}
  Z_{n}
  = \sbra{\mathcal{J}}\,
  e^{-(\beta/2)\mathcal{H}^{\mathrm{rot}}}\,
  \sket{\mathcal{I}}.
  \label{eq:Zn_bulk}
\end{align}
In the zero-temperature limit, we have the same scaling as Eq.~\eqref{eq:Mn_BCFT}, for which the $g$ factor can be expressed as
\begin{equation}
  g_{\mathcal{J}}=\sbk{\mathcal{J}}{\mathrm{CFT}^{\rm rot}}.
  \label{eq:g_rot_bulk}
\end{equation}
Here, $\sket{{\mathrm{CFT}^{\rm rot}}}$ is the ground state of the rotated bulk theory~\eqref{eq:Hrot_2n}.
Since the unitary rotation leaves the spectrum unchanged, the rotated bulk has the same universal finite-size spectrum as the original tensor-product CFT.
When the microscopic Hamiltonian is not quadratic, Eq.~\eqref{eq:Hrot_2n} in general generates inter-replica interactions within each of the replica sectors $\mathsf{R},\tilde{\mathsf{R}},\mathsf{R}',\tilde{\mathsf{R}}'$, and the resulting ground state $\sket{{\mathrm{CFT}^{\rm rot}}}$ is no longer a simple product of the replicas. The advantage of this picture, however, is that the boundary state $\sket{\mathcal{J}}$ retains the simple sewing form as in Eq.~\eqref{eq:J2n_def}; this makes the corresponding conformal boundary condition transparent and provides a practical starting point for perturbative calculations of the $g$ factor, as illustrated for the TLL in Sec.~\ref{sec:fermion_TLL}.

\section{\label{sec:conformal_invariance}Conformal invariance of the boundary state}
As shown above, the universal contribution to the MRE is controlled by a boundary condition in the folded theory. To apply the BCFT formalism, we need to examine whether or not this boundary condition preserves conformal invariance even after applying the convolution unitary.
In the rotated-bulk picture, absorbing the convolution unitary into the bulk evolution leaves the boundary state in the elementary sewing form $\sket{\mathcal{J}}$ in Eq.~\eqref{eq:J2n_def}, while the unitary replaces the tensor-product bulk Hamiltonian by the coupled-replica Hamiltonian~\eqref{eq:Hrot_2n}.
Consequently, the conformal invariance of $\sket{\mathcal{J}}$ cannot be assumed \textit{a priori}, but must instead be carefully examined as a boundary condition within the resulting coupled-replica theory.
While addressing this issue for arbitrary theories is not straightforward, we here provide a useful criterion for a class of models with density-density interactions, which sets up the BCFT analysis of the TLL in Sec.~\ref{sec:fermion_TLL}.

\subsection{Interacting bosonic and fermionic Hamiltonians}

We consider a number-conserving lattice Hamiltonian with a general quadratic part and a density-density interaction,
\begin{align}
  H(\bm{a}^{\dag},\bm{a})
  = \sum_{ij}a_{i}^{\dag}T_{ij}a_{j}
  + \frac{1}{2}\sum_{ij}V_{ij}n_{i}n_{j},
  \quad
  n_{i}=a_{i}^{\dag}a_{i},
  \label{eq:density_density_H}
\end{align}
where $T_{ij}$ and $V_{ij}$ are real symmetric matrices, and the operators $a_{i}$ can be either bosonic or fermionic. Since the convolution unitary acts identically and independently on each of the four folded sectors $\mathsf{R},\tilde{\mathsf{R}},\mathsf{R}',\tilde{\mathsf{R}}'$, it suffices to consider a single sector. To be concrete, we focus on the sector $\mathsf{R}$ including $n$ replicas, whose unrotated Hamiltonian is
\begin{align}
  H_{\mathsf{R}}
  = \sum_{r=1}^{n}H\qty(\bm{a}_{r}^{\dag},\bm{a}_{r}).
  \label{eq:Htot_density}
\end{align}
Here we consider the most general case where the convolution unitary induces the orthogonal transformation among the $n$ replicas as in Eq.~\eqref{eq:U_BF}.
The density interaction can be split into a part that is invariant under the convolution rotation and a part that is not.
It is then useful to introduce the total (uniform) component of the replica-density vector,
\begin{align}
  N_{i}
   & = \frac{1}{\sqrt{n}}\sum_{r=1}^{n}n_{r,i},
  \label{eq:NDelta_density}
\end{align}
where $n_{r,i}=a_{r,i}^{\dag}a_{r,i}$, and $n-1$ relative densities $\Delta_{m,i}$ ($m=2,3,\ldots,n$) form an arbitrary orthonormal basis of the subspace orthogonal to $N_{i}$.
We note that the particular choice of relative densities is not important; only the split between the total-density direction and its orthogonal complement enters below.
Taken together, the single-sector Hamiltonian can be written as
\begin{align}
  H_{\mathsf{R}}
   & = H_{0}+H_{\Delta},
  \nonumber                                               \\
  H_{0}
   & = \sum_{r=1}^{n}\sum_{ij}a_{r,i}^{\dag}T_{ij}a_{r,j}
  +\frac{1}{2}\sum_{ij}V_{ij}N_{i}N_{j},
  \nonumber                                               \\
  H_{\Delta}
   & = \frac{1}{2}\sum_{m=2}^{n}\sum_{ij}
  V_{ij}\Delta_{m,i}\Delta_{m,j}.
  \label{eq:H0_HDelta_density}
\end{align}
The total-density part $H_{0}$ is invariant under the convolution unitary while the relative-density part $H_{\Delta}$ is not.
We thus write the rotated Hamiltonian as
\begin{align}
  H_{\mathsf{R}}^{\mathrm{rot}}
  = UH_{\mathsf{R}}U^{\dag}
  = H_{0}+H_{\Delta}^{\mathrm{rot}},
  \quad
  H_{\Delta}^{\mathrm{rot}}
  = UH_{\Delta}U^{\dag},
  \label{eq:Htilde_density}
\end{align}
where $H_{\Delta}^{\mathrm{rot}}$ is the rotated relative-density interaction.
Summing the four rotated sectors $\mathsf{R},\tilde{\mathsf{R}},\mathsf{R}',\tilde{\mathsf{R}}'$ gives the full $\mathcal{H}^{\mathrm{rot}}$ in Eq.~\eqref{eq:Hrot_2n}.

\subsection{Conformal invariance}

To examine the conformal invariance of $\sket{\mathcal{J}}$ within this coupled-replica theory, we must determine how the rotated interaction affects the bulk CFT and its boundary condition \cite{fredenhagen2007symmetries}. Since the rotated Hamiltonian is unitarily equivalent to the original one, its spectrum and bulk universality class are fixed by the original CFT.
To be concrete, we thus focus on the situations where $H_{0}$, $H_{\mathsf{R}}$, and $H_{\mathsf{R}}^{\mathrm{rot}}$ lie on the same conformal manifold \footnote{In general, if $H_{0}$ and $H_{\mathsf{R}}^{\mathrm{rot}}$ are not known to lie on the same conformal manifold, the infrared bulk fixed point should first be identified before the boundary CFT analysis can be applied.}; indeed, this condition is satisfied in the examples considered later.
The universal part of $H_{\Delta}^{\mathrm{rot}}$ is then represented by an exactly marginal bulk deformation of the fixed point associated with $H_{0}$, possibly accompanied by irrelevant corrections.
The resulting Euclidean action for the rotated bulk Hamiltonian $ H_{\mathsf{R}}^{\mathrm{rot}}$ may be written as \cite{komatsu2026continuous}
\begin{align}
  \mathcal{S}
  = \mathcal{S}_{0}+\mathcal{S}_{\lambda},
  \qquad
  \mathcal{S}_{\lambda}
  = -\lambda\int d^{2}z\,\mathcal{O}(z,\bar{z}),
  \label{eq:S_lambda_boundary}
\end{align}
where $\mathcal{S}_{0}$ is the fixed-point action associated with $H_{0}$, $\mathcal{O}$ is an exactly marginal bulk operator resulting from $H_{\Delta}^{\mathrm{rot}}$, and $\lambda$ is an effective coupling strength that vanishes in the Gaussian limit of the microscopic Hamiltonian~\eqref{eq:density_density_H}.

As detailed in Appendix~\ref{app:bulk_boundary}, when $\mathcal{O}$ is brought close to the boundary, it can induce boundary perturbations. A nontrivial question is whether or not this bulk deformation induces relevant or marginal boundary operators through the bulk-boundary operator product expansion (OPE)~\cite{fredenhagen2007bulkinduced}.
More specifically, a possible boundary flow can be characterized by the following bulk-induced boundary action:
\begin{align}
  \delta\mathcal{S}_{\partial}
  = \int dx\sum_{a}\mu_{a}(l)\,
  \Lambda_{\text{c}}^{1-h_{a}}\psi_{a}(x),
  \label{eq:boundary_action_general}
\end{align}
where $\psi_{a}$ is a boundary operator of scaling dimension $h_{a}$, $\Lambda_{\text{c}}$ is a running momentum cutoff, and $l=\ln(\Lambda_{\mathrm{UV}}/\Lambda_{\text{c}})$ is the logarithmic RG scale.
To leading order in the boundary couplings $\mu_a$, the boundary RG equation is given by (see Appendix~\ref{app:bulk_boundary})
\begin{align}
  \frac{d\mu_{a}}{dl}
  = (1-h_{a})\mu_{a}+C_{a}(\lambda)
  +O(\lambda\mu,\mu^{2}),
  \label{eq:boundary_RG_general}
\end{align}
where $C_a(\lambda)$ is the source term determined by the bulk-boundary OPE of $\mathcal{O}$.
We can infer that there are three main possible scenarios:
\begin{enumerate}
  \item All generated boundary operators $\psi_a$ are irrelevant, i.e., $h_{a}>1$ $\forall a$. This is a rather trivial case where the boundary condition remains the same, and the $g$ factor retains the value associated with the unperturbed theory.
  \item
        The identity channel having $h_{\boldsymbol{1}}=0$ is generated, while all the other channels are irrelevant. In this case, the boundary condition itself still remains conformally invariant and unchanged, but the $g$ factor and hence the universal contribution $s_n$ to the MRE can change from the unperturbed value.
  \item
        A nonidentity relevant or marginal operator with $h_{a}\le1$ is generated. This is a case where the bulk deformation can induce a boundary RG flow to a different boundary fixed point. For $h_{a}=1$, the flow can be driven by the source term $C_{a}(\lambda)$ in Eq.~\eqref{eq:boundary_RG_general}. According to the $g$ theorem \cite{friedan2004boundary}, the $g$ factor in the IR limit should have a value less than the
        UV value, indicating the monotonic decrease of $s_n$ along the boundary RG flows.
\end{enumerate}
Consequently, the boundary condition imposed by $\sket{\mathcal{J}}$ in Eq.~\eqref{eq:J2n_def} remains stable and conformally invariant provided that the bulk-boundary OPE does not generate any nonidentity boundary operators with $h_{a}\le1$. When such an operator is generated, it can drive a boundary transition toward a distinct infrared fixed point.
Below we demonstrate the appearance of both regimes of boundary stability and instability depending on the microscopic interaction strength in the TLL.

\section{\label{sec:fermion_TLL}Magic R\'enyi entropy of the spinless-fermion chain}

We now apply the BCFT formulation developed above to analyze the second-order MRE, $M=-\ln\nu_{2}(\rho)$, of the TLL realized in an interacting spinless-fermion chain. For the sake of notational simplicity, throughout this section we set $n=2$ and drop the index $n$, writing $M$ and $s$ for $M_2$ and $s_2$. The purpose of this section is to analytically evaluate Eq.~\eqref{eq:g_rot_bulk} and determine the universal contribution $s=\ln g_{\mathcal{J}}$. A current-algebra calculation further gives the quartic correction and the instability against an angular displacement from the balanced case at $\zeta=\pi/4$. We also examine the scaling dimensions of the bulk-induced boundary operators, which predict the occurrence of the boundary transition at the TLL parameter $K=3$. Finally, we confirm these results by numerical calculations.

\subsection{\label{subsec:fermion_TLL_model}Bosonized rotated Hamiltonian}

We consider the half-filled interacting spinless-fermion chain
\begin{align}
  H
  = -t\sum_{i=1}^{L}
  \left(a_{i}^{\dag}a_{i+1}+\mathrm{H.c.}\right)
  + V\sum_{i=1}^{L}n_{i}n_{i+1}
  \label{eq:spinless_H}
\end{align}
with $n_{i}=a_{i}^{\dag}a_{i}$.
In the critical regime $-2t<V\leq2t$, its low-energy theory is governed by the TLL with the following velocity and TLL parameter
\begin{align}
  v
  = \pi t\,
  \frac{\sqrt{1-(V/2t)^{2}}}{\arccos(V/2t)},
  \qquad
  K
  = \frac{\pi}{2[\pi-\arccos(V/2t)]}.
  \label{eq:TLL_vK}
\end{align}
The free-fermion point is $K=1$, and repulsive (attractive) interactions give $K<1$ ($K>1$).

To be specific, we consider the balanced rotation acting on the two replicas $r=1,2$,
\begin{equation}
  U^\dag \begin{pmatrix}
    \bm{a}_1 \\ \bm{a}_2
  \end{pmatrix}U = \frac{1}{\sqrt{2}}\begin{pmatrix}
    1 & -1 \\
    1 & 1
  \end{pmatrix} \begin{pmatrix}
    \bm{a}_1 \\ \bm{a}_2
  \end{pmatrix},
  \label{eq:balanced_U_fermion}
\end{equation}
though we will generalize the analysis to an arbitrary rotation angle later. We recall that the vector $\bm{a}_r$ encapsulates all the modes $i=1,2,\ldots,L$ of replica $r$ (see Eq.~\eqref{avecdef}).
Applying the general decomposition~\eqref{eq:H0_HDelta_density} to this model, the Hamiltonian splits as $H_{\mathsf{R}}=H_{0}+H_{\Delta}$ with
\begin{align}
  H_{0}
   & = -t\sum_{i=1}^{L}\sum_{r=1}^{2}
  \left[a_{r,i}^{\dag}a_{r,i+1}+\mathrm{H.c.}\right]
  +V\sum_{i=1}^{L}N_{i}N_{i+1},\label{eq:H0micro} \\
  H_{\Delta}
   & = V\sum_{i=1}^{L}\Delta_{i}\Delta_{i+1},
  \label{eq:lattice_H_Delta_n2}
\end{align}
where $N_{i} = (n_{1,i}+n_{2,i})/\sqrt{2}$ and $\Delta_{i} = (n_{1,i}-n_{2,i})/\sqrt{2}$ are the uniform and relative densities.
The convolution unitary leaves $H_{0}$ invariant and rotates only the relative density $H_{\Delta}$ to $H_{\Delta}^{\mathrm{rot}}$, namely
\begin{equation}
  H_{\Delta}^{\mathrm{rot}}
  = V\sum_{i=1}^{L}\Delta_{i}^{\mathrm{rot}}\Delta_{i+1}^{\mathrm{rot}},
  \qquad
  \Delta_{i}^{\mathrm{rot}} \equiv U\Delta_{i}U^{\dag}.
  \label{eq:lattice_rotated_H_n2}
\end{equation}
For the rotation~\eqref{eq:balanced_U_fermion}, the relative operator becomes an inter-replica hopping $\Delta_{i}^{\mathrm{rot}} = [a_{1,i}^{\dag}a_{2,i}+a_{2,i}^{\dag}a_{1,i}]/\sqrt{2}$.
This means that the convolution converts the relative-density interaction $H_{\Delta}$ into the inter-replica current-current interaction $H_{\Delta}^{\mathrm{rot}}$.

At low energies, the continuum fermion field can be decomposed into right and left movers as
\begin{align}
  a_{r}(x)
  \simeq e^{ip_{F}x}a_{rR}(x)
  +e^{-ip_{F}x}a_{rL}(x),
  \label{eq:fermion_chiral_decomp}
\end{align}
where we denote the continuum limit of ${a}_{r,i}$ by $a_{r}(x)$ with $r=1,2$ being the replica index.
We bosonize the chiral fields by
\begin{align}
  a_{rR}(x)
   & = \frac{\eta_{rR}}{\sqrt{L}}
  :e^{i[\theta_{r}(x)-\phi_{r}(x)-\pi x/L]}:,
  \nonumber                       \\
  a_{rL}(x)
   & = \frac{\eta_{rL}}{\sqrt{L}}
  :e^{i[\theta_{r}(x)+\phi_{r}(x)+\pi x/L]}:,
  \label{eq:fermion_bosonization}
\end{align}
where $\eta_{r\chi}$ are Klein factors, and all colons denote the normal ordering. The bosonic fields obey $\phi_{r}(x)=\phi_{r}(x+L)$ and $\theta_{r}(x)=\theta_{r}(x+L)$.
We choose the Klein factors so that both the replica currents and the MES sewings have the bosonized forms given below.
Throughout this paper, we set the lattice constant $a=1$.
The background shifts are included to implement the Neveu-Schwarz sector for the chiral fermions, i.e., $a_{r\chi}(x)=-a_{r\chi}(x+L)$, and the compactification is
\begin{align}
  \phi_{r}
   & \sim\phi_{r}+\pi Q_{r},
  \qquad
  \theta_{r}\sim\theta_{r}+\pi J_{r}. \label{eq:compact_cond}
\end{align}
Here, the winding numbers $Q_r,J_r\in\mathbb{Z}$ satisfy the selection rule
\begin{equation}
  Q_{r}
  \equiv J_{r}\pmod{2}.
  \label{eq:TLL_compactification}
\end{equation}
It is useful to introduce symmetric and antisymmetric bosonic fields by
\begin{align}
  \theta_{\pm}
  = \frac{\theta_{1}\pm\theta_{2}}{\sqrt{2}},
  \qquad
  \phi_{\pm}
  = \frac{\phi_{1}\pm\phi_{2}}{\sqrt{2}}.
  \label{eq:TLL_pm_fields}
\end{align}
Before the convolution, the replicated Hamiltonian consists of two decoupled TLLs in the symmetric ($+$) and antisymmetric ($-$) sectors, $H_{+}+H_{-}$ with
\begin{align}
  H_{\pm}\simeq\frac{v}{2\pi}\int dx\left[K\left(\partial_{x}\theta_{\pm}\right)^{2}+\frac{1}{K}\left(\partial_{x}\phi_{\pm}\right)^{2}\right].
  \label{eq:TLL_unrotated_pm}
\end{align}
The convolution unitary acts only on the antisymmetric sector, as a rotation of the replica pseudospin.
To see this, we introduce the pseudospin current densities
\begin{align}
  \bm{J}_{\chi}(x)
  = \frac{1}{2}
  \begin{pmatrix}
    a_{1\chi}^{\dag}(x) & a_{2\chi}^{\dag}(x)
  \end{pmatrix}
  \bm{\sigma}
  \begin{pmatrix}
    a_{1\chi}(x) \\
    a_{2\chi}(x)
  \end{pmatrix}
  \label{eq:fermion_currents}
\end{align}
with $\chi=R,L$, whose bosonized expressions are
\begin{align}
  J_{R/L}^{x}
   & = \frac{1}{L}
  :\cos[\sqrt{2}(\phi_{-}\mp\theta_{-})]:,
  \nonumber                    \\
  J_{R/L}^{y}
   & = \pm\frac{1}{L}
  :\sin[\sqrt{2}(\phi_{-}\mp\theta_{-})]:,
  \nonumber                    \\
  J_{R/L}^{z}
   & = -\frac{1}{2\sqrt{2}\pi}
  (\partial_{x}\phi_{-}\mp\partial_{x}\theta_{-}).
  \label{eq:current_bosonization}
\end{align}
The balanced rotation~\eqref{eq:balanced_U_fermion} is a $\pi/2$ rotation of the pseudospin around the $y$ axis, and its bosonized form is
\begin{equation}
  U=e^{-i\frac{\pi}{2}\int dx\left(J_{R}^{y}+J_{L}^{y}\right)}.
  \label{eq:current_rotation_unitary}
\end{equation}
Since $U$ involves only the antisymmetric fields, it leaves $H_{+}$ invariant.
In terms of the currents, the antisymmetric sector reads
\begin{align}
  H_{-}
  ={} & \frac{\pi v}{3}\left(K+\frac{1}{K}\right)\int dx\left[:\bm{J}_{R}\cdot\bm{J}_{R}:+:\bm{J}_{L}\cdot\bm{J}_{L}:\right]
  \nonumber                                                                                                                  \\
      & +2\pi v\left(\frac{1}{K}-K\right)\int dx\,J_{R}^{z}J_{L}^{z}.
  \label{eq:TLL_Hminus_currents}
\end{align}
The unitary $U$ leaves the first term invariant and rotates $J_{R}^{z}J_{L}^{z}$ into $J_{R}^{x}J_{L}^{x}$, which is the continuum counterpart of the inter-replica hopping $\Delta_{i}^{\mathrm{rot}}$ on the lattice.
The replica-singlet part $H_{0}$ in Eq.~\eqref{eq:H0micro} consists of $H_{+}$ and the first term of Eq.~\eqref{eq:TLL_Hminus_currents}, which can be bosonized as \footnote{Precisely speaking, there can appear non-Gaussian potential terms when bosonizing Eq.~\eqref{eq:H0micro} that, however, cancel out with those arising from the bosonization of Eq.~\eqref{eq:lattice_H_Delta_n2}. }
\begin{align}
  H_{0}
  \simeq \sum_{s=\pm}
  \frac{v_{s}}{2\pi}\int dx
  \left[
    K_{s}(\partial_{x}\theta_{s})^{2}
    +\frac{1}{K_{s}}(\partial_{x}\phi_{s})^{2}
    \right],
  \label{eq:TLL_H0_bosonized}
\end{align}
which is the Gaussian boson CFT with the following parameters for each sector
\begin{align}
   & v_{+}=v,\qquad
  v_{-}=\frac{K+K^{-1}}{2}v,
  \nonumber         \\
   & K_{+}=K,\qquad
  K_{-}=1.
  \label{eq:TLL_parameters_rotated}
\end{align}
In contrast, the rotated interaction gives the non-Gaussian contribution that depends only on the antisymmetric sector:
\begin{align}
  H_{\Delta}^{\mathrm{rot}}
   & \simeq -\frac{4\pi^2\lambda}{L^{2}}\int dx
  \left[
    :\cos(2\sqrt{2}\phi_{-}):
    +:\cos(2\sqrt{2}\theta_{-}):
    \right],
  \nonumber                               \\
  \lambda
   & = \frac{v(K-K^{-1})}{4\pi},
  \label{eq:TLL_HDelta_bosonized}
\end{align}
which can be equivalently written as $H_{\Delta}^{\mathrm{rot}}=-8\pi^2\lambda\int dx\,J_R^xJ_L^x$ in currents. This perturbation is exactly marginal because the total Hamiltonian  $H_{0}+H_{\Delta}^{\mathrm{rot}}$, known as the self-dual sine-Gordon theory \cite{lecheminant2002criticality}, is unitarily equivalent to the Gaussian theory $H_++H_-$ (cf.~Eq.~\eqref{eq:TLL_unrotated_pm}).

\subsection{\label{subsec:fermion_TLL_boundary}Boundary state of the folded theory}

After folding, the replicated theory for calculating $M$ involves $8$ bosonic sheets in total (corresponding to the $n=2$ case in Fig.~\ref{fig:MRE_BCFT}).
Following the labeling of Eq.~\eqref{eq:J2n_def}, we organize these sheets into the four replica sectors $\mathsf{R}=(1,2)$, $\mathsf{R}'=(1',2')$, $\tilde{\mathsf{R}}=(\tilde{1},\tilde{2})$, and $\tilde{\mathsf{R}}'=(\tilde{1}',\tilde{2}')$, where the tilde (prime) denotes the auxiliary (purity) partner.
The boundary state $\sket{\mathcal{J}}$ in Eq.~\eqref{eq:J2n_def} can be written as
\begin{equation}
  \sket{\mathcal{J}}
  = \sket{I}_{1\tilde{1}'}
  \otimes\sket{I}_{\tilde{1}1'}
  \otimes\sket{I}_{2\tilde{2}}
  \otimes\sket{I}_{2'\tilde{2}'}.
  \label{eq:J_TLL_MES}
\end{equation}
Each MES imposes the sewing condition on the fermionic field (cf. Eq.~\eqref{eq:sewing_fermion}).
After bosonization, the boundary state $\sket{\mathcal{J}}$ thus imposes the sewing conditions on the bosonic fields as
\begin{align}
  \phi_{1}   & =\phi_{\tilde{1}'},\;\;\;\;
  \phi_{\tilde{1}}=\phi_{1'},\;\;\;\;
  \phi_{2}=\phi_{\tilde{2}},\;\;\;\;
  \phi_{2'}=\phi_{\tilde{2}'},
  \label{eq:J_phi_TLL}                     \\
  \theta_{1} & =-\theta_{\tilde{1}'},\;\;
  \theta_{\tilde{1}}=-\theta_{1'},\;\;
  \theta_{2}=-\theta_{\tilde{2}},\;\;
  \theta_{2'}=-\theta_{\tilde{2}'},
  \label{eq:J_theta_TLL}
\end{align}
which define the so-called mixed Dirichlet-Neumann conditions for the compact bosons. Here, the replica retained after the convolution is sewn across the two purity branches while the discarded replica is sewn to its auxiliary copy. This imposes a nontrivial boundary condition when the bulk becomes the coupled-replica theory, and the universal constant in the MRE is controlled by the boundary entropy of this gluing condition.

The zero-mode lattice is easiest to describe after grouping the $8$ sheets into the $4$ replica sectors and separating each sector into symmetric and antisymmetric fields.
Relabeling the sectors $\mathsf{R},\mathsf{R}',\tilde{\mathsf{R}},\tilde{\mathsf{R}}'$ by the index $m=1,2,3,4$ in this order, we define the linear combinations $\phi_{m,\pm}$ by
\begin{align}
  \phi_{1,\pm}
   & = \frac{\phi_{1}\pm\phi_{2}}{\sqrt{2}},
   &
  \phi_{2,\pm}
   & = \frac{\phi_{1'}\pm\phi_{2'}}{\sqrt{2}},
  \nonumber                                                    \\
  \phi_{3,\pm}
   & = \frac{\phi_{\tilde{1}}\pm\phi_{\tilde{2}}}{\sqrt{2}},
   &
  \phi_{4,\pm}
   & = \frac{\phi_{\tilde{1}'}\pm\phi_{\tilde{2}'}}{\sqrt{2}},
  \label{eq:four_sector_pm_phi}
\end{align}
and also introduce $\theta_{m,\pm}$ in the same manner.
The canonical compact-boson variables are obtained by absorbing the TLL parameters $K_{\pm}$ into the compactification radii.
This rescaling moves all dependence on the compactification radii into the zero-mode lattice, which allows us to extract the boundary entropy from the lattice volume~\cite{oshikawa2010boundary,furukawa2011entanglement,ashida2024systemenvironment}.
We thus use the following rescaled $8$-component fields:
\begin{align}
  \bm{\Phi}
   & =
  (\Phi_{1,+},\Phi_{1,-},\ldots,\Phi_{4,+},\Phi_{4,-})^{\rm T},
  \nonumber \\
  \bm{\Theta}
   & =
  (\Theta_{1,+},\Theta_{1,-},\ldots,\Theta_{4,+},\Theta_{4,-})^{\rm T},
  \nonumber \\
  \Phi_{m,\pm}
   & =
  \frac{\phi_{m,\pm}}{\sqrt{K_{\pm}}},
  \qquad
  \Theta_{m,\pm}
  =\sqrt{K_{\pm}}\theta_{m,\pm}.
  \label{eq:PhiTheta_TLL}
\end{align}
For the boundary-state construction, we abbreviate the velocity factors $v_\pm$ because they do not affect the gluing condition; the velocities will be restored only when the open-channel zero-point energy is evaluated later.
With this convention, the Gaussian part of the replicated Hamiltonian is (see Eq.~\eqref{eq:TLL_H0_bosonized})
\begin{align}
  \mathcal{H}_{0}
  = \int\frac{dx}{2\pi}
  \left[
    (\partial_{x}\bm{\Phi})^{2}
    +(\partial_{x}\bm{\Theta})^{2}
    \right].
  \label{eq:eight_boson_H0}
\end{align}

We can expand the bosonic fields in terms of the oscillator modes and the zero modes generated by the windings along the spatial direction.
The oscillatory part is written as
\begin{align}
  \bm{A}_{\ell}^{\pm}(x)
  =
  \bm{a}_{\ell,L}e^{-ik_{\ell}x}
  \pm\bm{a}_{\ell,R}e^{ik_{\ell}x}
  +\mathrm{H.c.},
  \label{eq:oscillator_modes_TLL}
\end{align}
where $\bm{a}_{\ell,L(R)}$ is a vector of annihilation operators of left- (right-) moving oscillator modes having quantum number $\ell$ and $ k_{\ell}=2\pi\ell/L$.
The fields are then expanded as
\begin{align}
  \bm{\Phi}(x)
   & = \bm{\Phi}_{0}
  +\frac{2\pi}{L}\bm{T}x
  +\sum_{\ell>0}\frac{1}{\sqrt{4\ell}}
  \bm{A}_{\ell}^{+}(x),
  \nonumber            \\
  \bm{\Theta}(x)
   & = \bm{\Theta}_{0}
  +\frac{2\pi}{L}\bm{W}x
  +\sum_{\ell>0}\frac{1}{\sqrt{4\ell}}
  \bm{A}_{\ell}^{-}(x),
  \label{eq:PhiTheta_modes_TLL}
\end{align}
where $\bm{\Phi}_{0}$ and $\bm{\Theta}_{0}$ are the zero-mode variables, and $\bm{W}$ and $\bm{T}$ are their conjugates. These operators satisfy the commutation relations
\begin{align}
  [(\boldsymbol{\Phi}_{0})_{i},(\boldsymbol{W})_{j}]                       & =\frac{i}{2}\delta_{ij},\;\;[(\boldsymbol{\Theta}_{0})_{i},(\boldsymbol{T})_{j}]=\frac{i}{2}\delta_{ij}, \\
  [(\boldsymbol{a}_{\ell,\chi})_{i},(\boldsymbol{a}_{\ell',\chi'}^{\dagger})_{j}] & =\delta_{\ell\ell'}\delta_{\chi\chi'}\delta_{ij},
\end{align}
where $\chi,\chi'\in\{ L,R\}$, and $\ell,\ell'\in\mathbb{N}$.
Using these mode expansions, the Hamiltonian becomes
\begin{align}
  \mathcal{H}_{0}
  = \frac{2\pi}{L}
  \left[
    \bm{T}^{2}
    +\bm{W}^{2}
    +\sum_{\ell=1}^{\infty}\sum_{\chi=L,R}
    \ell\,\bm{a}_{\ell,\chi}^{\dag}\cdot\bm{a}_{\ell,\chi}
    -\frac{2}{3}
    \right],
  \label{eq:eight_boson_H0_modes}
\end{align}
where the constant $-2/3$ comes from the zero-point energy $-c_{\rm tot}/12$ of the eight compact bosons having the total central charge $c_{\rm tot}=8$.

The MES boundary conditions become an orthogonal gluing condition for the chiral compact bosons.
Specifically, using $\bm{\Phi}_{R/L}=\bm{\Phi}\mp\bm{\Theta}$,
the boundary conditions~\eqref{eq:J_phi_TLL} and~\eqref{eq:J_theta_TLL} can be summarized as
\begin{align}
  \partial_{x}\left(\bm{\Phi}_{L}-G\bm{\Phi}_{R}\right)\sket{\mathcal{J}}=0
  \qquad
  \text{for all }x,
  \label{eq:J_gluing_O}
\end{align}
where $G$ is a symmetric orthogonal matrix fixed by the gluing conditions and compactification data.
The matrix $G$ tells us which linear combinations obey Neumann or Dirichlet conditions.
The compactification lattice of the zero modes, denoted by $\Lambda$, must be compatible with the gluing condition~\eqref{eq:J_gluing_O} and the selection rule~\eqref{eq:TLL_compactification}, which constrains the allowed winding numbers.
We can now explicitly construct the boundary state as
\begin{align}
  \sket{\mathcal{J}}
  = g_{0,\mathcal{J}}\,
  \exp\!\left[
          \sum_{\ell=1}^{\infty}
          \bm{a}_{\ell,L}^{\dag}\cdot G\bm{a}_{\ell,R}^{\dag}
          \right]
  \sum_{\bm{\lambda}\in\Lambda}\sket{\bm{\lambda}},
  \label{eq:J_boundary_state_TLL}
\end{align}
where $\bm{\lambda}=(\bm{T},\bm{W})^{\rm T}$ is a composite winding vector. The coefficient $g_{0,\mathcal{J}}=\langle\!\langle 0|\cal{J}\rangle\!\rangle$ is the overlap with the Gaussian vacuum $\sket{0}$ of $\mathcal H_0$. As shown below, this is different from the physical $g$-factor $g_{\cal{J}}$ appearing as the MRE universal term, which is given by the overlap $g_{\mathcal{J}}={}_\text{rot}\langle\!\langle 0|\cal{J}\rangle\!\rangle$ with the ground state $\sket{0}_{\text{rot}}$ of the full rotated Hamiltonian ${\cal H}^{\text{rot}}$.  
Explicit forms of $G$ and $\Lambda$ for Eqs.~\eqref{eq:J_phi_TLL} and \eqref{eq:J_theta_TLL} are given in Appendix~\ref{app:compact}.

In general, a class of states obeying Eq.~\eqref{eq:J_gluing_O} defines conformal boundary states of free-boson CFTs. Indeed, it is easy to explicitly check their conformal invariance by
\begin{align}
  \left[T(x)-\bar{T}(x)\right]\sket{\mathcal{J}}
   & =0,
  \label{eq:J_conformal_free}
\end{align}
where $T
    =(\partial_{x}\bm{\Phi}_{R})^{2}
$ and
$ \bar{T}=(\partial_{x}\bm{\Phi}_{L})^{2}$ are the stress tensors.
Consequently, the Gaussian part $\mathcal{H}_0$ of the rotated theory reduces to a standard BCFT problem, for which $\sket{\mathcal{J}}$ remains conformally invariant, and its universal $g$ factor can be evaluated exactly.
This boundary state remains stable in the presence of the non-Gaussian bulk perturbation as long as the bulk-induced boundary OPE generates no nonidentity channels with $h_{a}\le1$. Nevertheless, the non-Gaussian perturbation in Eq.~\eqref{eq:TLL_HDelta_bosonized} can modify the coefficients $g_{0,\mathcal{J}}$, $g_{\mathcal{J}}$ through the identity channel in a nontrivial manner as discussed below.

\subsection{\label{subsec:fermion_TLL_gfactor}Universal contribution near the free-fermion point}
To analytically evaluate the value of the $g$ factors, we here consider a cylinder geometry of the 8-component replicated theory with the same boundary condition $\mathcal{J}$ at both ends.
This auxiliary partition function is not the MRE partition function itself, but it gives a useful way to normalize the boundary state because the closed-channel amplitude naturally contains the $g$ factor.
After a modular transformation, this $\mathcal{J}\mathcal{J}$ cylinder amplitude must be interpretable as an open-channel trace with the identity operator as the leading state.
Equating the two descriptions fixes the normalization of $\sket{\mathcal{J}}$, which gives the universal $g$ factor entering the MRE~\eqref{eq:Mn_BCFT}.

\subsubsection*{Duality and perturbative analysis}

The partition function of the cylinder theory is
\begin{align}
  Z_{\mathcal{J}\mathcal{J}}
  = \sbra{\mathcal{J}}
  e^{-(\beta/2)\mathcal{H}^{\mathrm{rot}}}
  \sket{\mathcal{J}},
  \qquad
  \mathcal{H}^{\mathrm{rot}}
  = \mathcal{H}_{0}+\mathcal{H}_{\Delta}^{\mathrm{rot}},
  \label{eq:ZJJ_TLL}
\end{align}
where $\mathcal{H}_{\Delta}^{\mathrm{rot}}$ is the replicated version of Eq.~\eqref{eq:TLL_HDelta_bosonized}, for which we shall label each sector of the replicas by $(m,-)$ with $m=1,2,3,4$ as consistent with the notation used in Eqs.~\eqref{eq:four_sector_pm_phi} and \eqref{eq:PhiTheta_TLL}.
Our strategy is to first evaluate the partition function exactly for the Gaussian part $ \mathcal{H}_{0}$ and then include the corrections due to the non-Gaussian part $\mathcal{H}_{\Delta}^{\mathrm{rot}}$ perturbatively. Throughout this calculation, we use unit $v=1$.

For this purpose, it is noteworthy that the $\mathcal{J}\mathcal{J}$ cylinder geometry satisfies the duality under $K\leftrightarrow K^{-1}$, which constrains the perturbative expansion around $K=1$. To show this, it suffices to see that the boundary conditions~\eqref{eq:J_phi_TLL} and~\eqref{eq:J_theta_TLL} are invariant under $K\leftrightarrow K^{-1}$ and $\phi\leftrightarrow\theta$, up to the particle-hole symmetry transformation $(\phi,\theta)\to(-\phi,-\theta)$ on four of the sheets.
This observation shows that $Z_{\mathcal{J}\mathcal{J}}$ and $g_{\mathcal{J}}$ are even functions of the following parameter
\begin{align}
  \gamma
  \equiv \frac{K-1}{K+1}.
  \label{eq:gamma_TLL}
\end{align}
This duality thus forbids odd terms in $\gamma$ in the boundary entropy, which, in particular, indicates that the leading universal change of $g_{\mathcal{J}}$ is $O(\gamma^{2})$.

More specifically, we can expand the partition function as
\begin{equation}
  Z_{\mathcal{J}\mathcal{J}}=Z_{\mathcal{J}\mathcal{J}}^{(0)}+\delta Z_{\mathcal{J}\mathcal{J}}+O(\gamma^{4}),\label{eq:zggp}
\end{equation}
where the Gaussian part is given by
\begin{equation}
  Z_{\mathcal{J}\mathcal{J}}^{(0)}= \sbra{\mathcal{J}}e^{-\frac{\beta}{2}\mathcal{H}_{0}} \sket{\mathcal{J}},
\end{equation}
and the leading $O(\gamma^2)$ correction due to the non-Gaussian part is
\begin{align}
  \delta Z_{\mathcal{J}\mathcal{J}} & =\frac{1}{2}\int_{0}^{\beta/2}d\tau_{1}d\tau_{2} \nonumber \\
  \times                            &
  \sbra{\mathcal{J}}e^{-\frac{\beta}{2}\mathcal{H}_{0}}{\cal T}_{\tau}\left[\mathcal{H}_{\Delta}^{\mathrm{rot}}(\tau_{1})\mathcal{H}_{\Delta}^{\mathrm{rot}}(\tau_{2})\right] \sket{\mathcal{J}}.
  \label{eq:delta_Z}
\end{align}
Here, the imaginary-time Heisenberg representation is defined by
\begin{equation}
  A(\tau)\equiv e^{\tau\mathcal{H}_{0}}Ae^{-\tau\mathcal{H}_{0}},
\end{equation}
and ${\cal T}_\tau$ represents the imaginary-time ordering.

\subsubsection*{Gaussian part}
We can exactly evaluate the Gaussian part of the partition function by
\begin{align}
  Z_{\mathcal{J}\mathcal{J}}^{(0)}
  = \frac{g_{0,\mathcal{J}}^{2}}{\eta(q)^{8}}
  \sum_{\bm{\lambda}\in\Lambda}
  q^{\bm{\lambda}^{2}/2},
  \label{eq:ZJJ_gaussian}
\end{align}
where $q=e^{-2\pi\beta/L}$, and $\eta(q)=q^{1/24}\prod_{n=1}^{\infty}\left(1-q^{n}\right)$ is the Dedekind eta function.
The factor $\eta(q)^{-8}$ is the oscillator determinant of the 8 compact bosons, while the sum over $\Lambda$ corresponds to the zero-mode contributions.

This Gaussian contribution already leads to a $K$-dependent normalization factor. To show this, we use the Poisson resummation to obtain
\begin{align}
  Z_{\mathcal{J}\mathcal{J}}^{(0)}
  = \frac{g_{0,\mathcal{J}}^{2}}{\mathcal{V}_{\Lambda}\eta(\tilde q)^{8}}
  \sum_{\bm{\lambda}^{*}\in\Lambda^{*}}
  \tilde q^{\bm{\lambda}^{*2}/2}
  \label{eq:ZJJ_poisson}
\end{align}
with $ \tilde q=e^{-2\pi L/\beta}.$
Here, the unit-cell volume of the compactification lattice is (see Appendix~\ref{app:compact})
\begin{align}
  \mathcal{V}_{\Lambda}
  = \frac{(K+1)^{2}}{4K}
  = \frac{1}{1-\gamma^{2}},
  \label{eq:VLambda_TLL}
\end{align}
and the factor $\mathcal{V}_{\Lambda}^{-1}$ appears because the Poisson resummation converts the closed-channel zero-mode sum into an open-channel sum over the dual lattice.

The long-cylinder limit separates the open-channel energy from the boundary entropy.
In the modular-transformed regime $1\ll\beta\ll L$, and after restoring the velocities in the zero-point energy, Eq.~\eqref{eq:ZJJ_poisson} becomes
\begin{align}
  Z_{\mathcal{J}\mathcal{J}}^{(0)}
  \simeq
  \frac{g_{0,\mathcal{J}}^{2}}{\mathcal{V}_{\Lambda}}
  \exp\!\left[
          \frac{\pi L}{3\beta}
          \left(\frac{1}{v_{+}}+\frac{1}{v_{-}}\right)
          \right].
  \label{eq:ZJJ_gaussian_long}
\end{align}
The multiplicative constant contributes to the boundary entropy, while the exponential in Eq.~\eqref{eq:ZJJ_gaussian_long} is the open-channel Casimir contribution.

\subsubsection*{Non-Gaussian correction}

The leading correction from the non-Gaussian term $\mathcal{H}_{\Delta}^{\mathrm{rot}}$ is second order in $\gamma$.
Since the coupling coefficient $\lambda$ in Eq.~\eqref{eq:TLL_HDelta_bosonized} is already proportional to $K-K^{-1}$, it suffices to evaluate the second-order correlators in $\delta Z_{\mathcal{J}\mathcal{J}}$ at $K=1$.
Performing the perturbative calculations within the $K=1$ Gaussian theory, we get (see Appendix~\ref{app:perturbation})
\begin{equation}
  \frac{\delta Z_{\mathcal{J}\mathcal{J}}}{Z_{\mathcal{J}\mathcal{J}}^{(0)}}
  \simeq \gamma^{2}
  \left[
    2-\frac{2\pi\beta}{3L}E_{2}(q)
    +\frac{8\pi^{2}\beta^{2}}{L^{2}}
    \left(2q\frac{d}{dq}\ln\vartheta_{3}(q)\right)^{2}
    \right],
  \label{eq:deltaZ_second_order_TLL}
\end{equation}
where
\begin{align}
  E_{2}(q)
  = 1-24\sum_{\ell=1}^{\infty}\frac{\ell q^{\ell}}{1-q^{\ell}},
  \qquad
  \vartheta_{3}(q)
  = \sum_{m\in\mathbb{Z}}q^{m^{2}/2}.
  \label{eq:E2_theta_TLL}
\end{align}
The constant $2$ and the Eisenstein series $E_{2}$ in Eq.~\eqref{eq:deltaZ_second_order_TLL} come from oscillator modes, while the last term corresponds to the zero-mode contribution expressed by the theta function $\vartheta_{3}$.
We perform a modular transformation on Eq.~\eqref{eq:deltaZ_second_order_TLL}. To this end,
we use
\begin{align}
  E_{2}(\tilde q)
  = -\frac{\beta^{2}}{L^{2}}E_{2}(q)
  +\frac{6}{\pi}\frac{\beta}{L},
  \qquad
  \vartheta_{3}(\tilde q)
  = \sqrt{\frac{\beta}{L}}\,\vartheta_{3}(q),
  \label{eq:modular_E2_theta_TLL}
\end{align}
and obtain the leading contribution as
\begin{align}
  \frac{\delta Z_{\mathcal{J}\mathcal{J}}}{Z_{\mathcal{J}\mathcal{J}}^{(0)}}
  \simeq
\frac{2\pi L}{3\beta}\gamma^2
  \label{eq:deltaZ_long_TLL}
\end{align}
in the long-cylinder limit $1\ll\beta\ll L$.
The finite terms thus cancel while the term proportional to $L/\beta$ remains as an open-channel energy correction.

\subsubsection*{Universal constant term}

From Eqs.~\eqref{eq:zggp}, \eqref{eq:ZJJ_gaussian_long}, and \eqref{eq:deltaZ_long_TLL}, we obtain
\begin{align}
  Z_{\mathcal{J}\mathcal{J}}
  \simeq
  g_{0,\mathcal{J}}^{2}(1-\gamma^{2})
  \exp\!\left(\frac{2\pi L}{3\beta}\right),
  \label{eq:ZJJ_combined_TLL}
\end{align}
where we use the expansions $\mathcal{V}_{\Lambda}^{-1}\simeq1-\gamma^{2}$ and $v_{-}\simeq 1+2\gamma^{2}$. We note that all $\gamma^{2}L/\beta$ terms precisely cancel, which is a useful consistency check that the final result is a genuine boundary quantity without a leftover Casimir-energy correction. Meanwhile, the open-channel picture gives
\begin{equation}
  Z_{\mathcal{J}\mathcal{J}}={\rm tr}\left[e^{-L\mathcal{H}_{\mathcal{J}\mathcal{J}}}\right]\stackrel{1\ll\beta\ll L}{\simeq} \exp\!\left(\frac{2\pi L}{3\beta}\right),
  \label{eq:open_channel_identity_TLL}
\end{equation}
where $\mathcal{H}_{\mathcal{J}\mathcal{J}}$ is the 8-component theory of length $\beta/2$ with the boundary $\mathcal{J}$ at both ends, and we used the fact that the zero-point energy of the segment of length $\beta/2$ is $-c_{\rm tot}\pi/(12\beta)$ with $c_{\rm tot}=8$.
Comparing Eqs.~\eqref{eq:ZJJ_combined_TLL} and \eqref{eq:open_channel_identity_TLL}, we have $  g_{0,\mathcal{J}}\simeq 1+\frac{1}{2}\gamma^{2}$. To obtain the physical $g$-factor, we consider the closed-channel limit $\beta/L\to\infty$. The comparison of the partition functions in the similar manner as above leads to $g_{\cal{J}}^2\simeq g_{0,\cal{J}}^2(1+2\gamma^2)$, where all $\gamma^2\beta/L$ terms again cancel. Thus, we obtain
\begin{align}
  g_{\mathcal{J}}
  = 1+\frac{3}{2}\gamma^{2}+O(\gamma^{4}),
  \label{eq:gGamma_TLL}
\end{align}
leading to the universal constant term of the second-order MRE as
\begin{align}
  s = \frac{3}{2}\gamma^{2}+O(\gamma^{4}).
  \label{eq:s2_TLL_universal}
\end{align}
It is now evident that the size-independent term is solely determined by the TLL parameter $K$.

\subsubsection*{Convolution with arbitrary rotation angle}

The above perturbative analysis can be extended to a general two-replica rotation angle.
Specifically, we consider the two-replica beam splitter $U(\zeta)$ introduced in Eq.~\eqref{eq:O_case1}, with the mixing angle $\zeta$ now kept arbitrary,
\begin{align}
  U^\dag(\zeta)
  \begin{pmatrix}
    a_{1\chi} \\
    a_{2\chi}
  \end{pmatrix}
  U(\zeta)
  =
  \begin{pmatrix}
    \cos\zeta & -\sin\zeta \\
    \sin\zeta & \cos\zeta
  \end{pmatrix}
  \begin{pmatrix}
    a_{1\chi} \\
    a_{2\chi}
  \end{pmatrix}.
  \label{eq:general_angle_U_TLL}
\end{align}
At the leading order, the resulting identity-channel contribution to the $g$ factor is obtained as (see Appendix~\ref{app:general_angle})
\begin{align}
  g_{\text{UV},\mathcal{J}}(\zeta)
  = 1+\gamma^{2}
  \left[
    2\sin^{2}(2\zeta)
    -\frac{1}{2}\sin^{4}(2\zeta)
    \right]
  +O(\gamma^{4}),
  \label{eq:gGamma_angle_TLL}
\end{align}
which leads to
\begin{equation}
  s_{\text{UV}}(\zeta) = \gamma^{2}
  \left[
    2\sin^{2}(2\zeta)
    -\frac{1}{2}\sin^{4}(2\zeta)
    \right]
  +O(\gamma^{4}).\label{eq:s2_TLL_universal_angle}
\end{equation}
The balanced convolution corresponds to $\zeta=\pi/4$, for which Eq.~\eqref{eq:s2_TLL_universal_angle} reduces to Eq.~\eqref{eq:s2_TLL_universal}. As shown below, however, when the angle $\zeta$ is displaced from this balanced value, the bulk induces a nonidentity relevant boundary perturbation, which leads to nontrivial boundary RG flows. Together with the $g$ theorem, the value of the $g$ function at $\zeta\neq \pi/4$ thus should monotonically decrease from the above UV value to an infrared value $g_{\text{IR}}$.

\subsection{Fourth-order correction and angular instability}
\label{subsec:fermion_TLL_current}

The fourth-order correction to the boundary entropy and the boundary RG flow at a general mixing angle can be systematically obtained within a current-algebra framework. The relevant algebra is that of fermion bilinears, including the pseudospin currents~\eqref{eq:fermion_currents}.
Its commutator term contributes to the connected four-current amplitude and generates the boundary current associated with deviations from $\zeta=\pi/4$.

At the balanced angle $\zeta=\pi/4$, we obtain
\begin{equation}
  s=\frac{3}{2}\gamma^2+\frac{3-\pi^2}{4}\gamma^4+O(\gamma^6),
  \qquad \zeta=\frac{\pi}{4},
  \label{eq:s2_TLL_quartic}
\end{equation}
which extends the perturbative result~\eqref{eq:s2_TLL_universal} from the rotated-bulk picture.
The derivation is given in Appendix~\ref{app:current_corrections}.

The current-algebra analysis also reveals an instability to angular displacements $\delta\zeta=\zeta-\pi/4$.
To leading order near the free-fermion point $K=1$, the linearized boundary RG equation is
\begin{equation}
  \frac{d\,\delta\zeta}{dl}
  =\gamma^2\delta\zeta,
  \label{eq:TLL_angle_RG}
\end{equation}
where $l$ is the logarithmic RG scale defined below Eq.~\eqref{eq:boundary_action_general}.
The corresponding boundary operator has scaling dimension $1-\gamma^2$ to this order, so the angular displacement is relevant for sufficiently small nonzero $|\gamma|$.

Our analysis indicates that the $O(\gamma^2)$ flow on $0\leq\zeta\leq\pi/4$ has only two fixed points, $\zeta=0$ and $\pi/4$ (Appendix~\ref{app:current_corrections}).
All intermediate angles flow to $\zeta=0$, where the convolution unitary is the identity and $s=0$, and only balanced mixing retains a nonzero universal contribution in the IR limit.
Thus, away from the balanced angle, Eq.~\eqref{eq:s2_TLL_universal_angle} should be understood as a fixed-order perturbative result at UV scale before the contributions associated with this boundary RG flow become appreciable.

The boundary CFT formulation also uncovers an amusing connection of the results above to quantum-wire junctions of TLLs, a problem studied in condensed matter physics over several decades.
Exchanging space and imaginary time maps the sewing boundary to an auxiliary junction of eight semi-infinite TLL wires, within the boundary-CFT framework of such junctions~\cite{oshikawa2006junctions,shi2016impurity}.
At a conformal fixed point, $s$ is the junction impurity entropy relative to the paired sewing of uninterrupted liquids with normalization fixed by Eqs.~\eqref{eq:rel_g} and~\eqref{eq:lng_def}.
The universal contribution to bulk ground-state non-Gaussianity thus has a direct interpretation in the boundary thermodynamics of an interacting junction.
While the $\gamma^2$ term of $s$ can be obtained by adapting the calculation of Ref.~\cite{shi2016impurity} to the present junction, we further derive the $\gamma^4$ term in Eq.~\eqref{eq:s2_TLL_quartic}.

In this language, the mixing angle $\zeta$ at $K=1$ parametrizes a continuous family of boundary conditions for coherent fermion scattering.
Weak interactions make the mixing amplitudes run under boundary RG~\cite{lal2002junction}.
The present sewing boundary is reflectionless, because each MES in Eq.~\eqref{eq:J2n_def} sews a replica in $\mathsf{R}$ or $\mathsf{R}'$ to a tilde copy in $\tilde{\mathsf{R}}$ or $\tilde{\mathsf{R}}'$, and $\mathcal{U}$ in Eq.~\eqref{eq:U2n_def} does not mix the replicas with the tilde copies.
The source for the boundary current thus vanishes at first order in the interaction, as in junctions~\cite{lal2002junction,shi2016impurity}.
Once $K\neq1$, only a much more limited set of these boundary conditions can exist~\cite{oshikawa2006junctions,shi2016impurity}, which is consistent with the two fixed points $\zeta=0$ and $\pi/4$ of the flow.

\subsection{\label{subsec:fermion_TLL_stability}Boundary stability}
At the balanced angle $\zeta=\pi/4$, we now use the criterion discussed in Sec.~\ref{sec:conformal_invariance} to examine whether the boundary condition $\mathcal{J}$ remains the infrared boundary fixed point of the perturbed theory.
At the boundary, the perturbing operator
\begin{equation}
  \mathcal{O}\propto\sum_{m=1}^{4}\qty[\cos(2\sqrt{2}\phi_{m,-})+\cos(2\sqrt{2}\theta_{m,-})]
  \label{eq:TLL_O}
\end{equation}
carries a zero-mode charge as follows (see Appendix~\ref{app:perturbation}):
\begin{equation}
  {\cal O}\propto\sum_{m=1}^{4}\sum_{\sigma,\tilde{\sigma}=\pm}e^{i\sqrt{2}\boldsymbol{q}_{m}^{\sigma\tilde{\sigma}}\cdot\boldsymbol{\Phi}_{R}},
\end{equation}
where we represent the vertex operator as a sum of chiral vertex operators with charges
\begin{align}
  \bm{q}_{m}^{\sigma\tilde{\sigma}}
  = \sigma\bm{e}_{m}+\tilde{\sigma}\tilde{\bm{e}}_{m},
  \qquad
  \tilde{\bm{e}}_{m}=G^{\rm T}\bm{e}_{m},
  \qquad
  \sigma,\tilde{\sigma}=\pm1.
  \label{eq:boundary_charge_element_TLL}
\end{align}
Here, $\bm{e}_{m}$ is the unit vector for the $(m,-)$ component, and the explicit form of the gluing orthogonal matrix $G$ is given in Appendix~\ref{app:compact}.
In general, a total charge of a boundary operator $\psi_a$ generated at any perturbative order takes the form
\begin{align}
  \bm{Q}_{a}
   & = \sum_{m=1}^{4}
  \left(t_{m}\bm{e}_{m}+u_{m}\tilde{\bm{e}}_{m}\right),
\end{align}
where we introduce the integers
\begin{equation}
  t_{m},u_{m}\in\mathbb{Z},\qquad
  t_{m}\equiv u_{m}\pmod{2}.
  \label{eq:boundary_charge_TLL}
\end{equation}
The boundary scaling dimension of $\psi_a$ is then given by
\begin{align}
  h_{a}=\bm{Q}_{a}^{2}.
  \label{eq:boundary_dimension_TLL}
\end{align}
At even perturbative orders, charge-neutral combinations
$Q_{\boldsymbol{1}}=0$ can appear, which gives rise to the identity channel with $h_{\boldsymbol{1}}=0$ and can renormalize the boundary entropy. This indicates that the $g$ factor must be an even function of $\lambda$, which is consistent with the duality $K\leftrightarrow K^{-1}$ discussed above. Note that the current channel $h_{\partial\boldsymbol{\Phi}}=1$ is forbidden by the symmetry under $\boldsymbol{\Phi}_{\chi}\to-\boldsymbol{\Phi}_{\chi}$ of the Gaussian bulk theory, the sewing boundary, and the perturbing operator. Thus, the neutral sector contributes only to the identity or the operators with higher-order derivatives having $h_{a}\geq 2$.

The stability criterion is then controlled by the smallest nonzero scaling dimension in the charged sectors, which can be given by
\begin{align}
  h_{\mathrm{min}}(K)
  = \frac{4\min(K,1)}{K+1}
  = 2(1-|\gamma|).
  \label{eq:hmin_TLL}
\end{align}
Thus, in the critical regime $K\geq1/2$, all nonidentity boundary operators induced by the bulk perturbation are irrelevant when
\begin{align}
  h_{\mathrm{min}}(K)>1
  \quad\Longleftrightarrow\quad
  K<3.
  \label{eq:TLL_stability_window}
\end{align}
Within this window, the identity channel can continuously change $g_{\mathcal{J}}$ as in Eq.~\eqref{eq:gGamma_TLL} while preserving the boundary condition $\mathcal{J}$.
In contrast, when $K>3$, a relevant boundary operator can be generated, which can destabilize the sewing boundary $\mathcal{J}$, leading to a boundary RG flow to a different boundary fixed point. We thus expect the boundary transition occurring at $K=3$, which is reachable on the attractive side of the microscopic model~\eqref{eq:spinless_H}.
The marginal endpoint $h_{\mathrm{min}}=1$ requires a separate analysis, but the symmetry argument above shows that the potentially dangerous current channel is not generated at $\zeta=\pi/4$.

\subsection{\label{subsec:fermion_TLL_largeK}Universal contribution for \texorpdfstring{$K>3$}{K>3}}

When $K>3$, the boundary $\mathcal{J}$ becomes unstable, and $s$ is determined by the infrared fixed point of the bulk-induced boundary RG flow.
We follow this flow in the rotated-bulk picture by evaluating the scaling dimensions within the Gaussian part $\mathcal{H}_{0}$ and the generated operators at the lowest order in $\lambda$ as in Sec.~\ref{subsec:fermion_TLL_stability}. We then evaluate the corresponding $g$ factor.

The flow proceeds in two steps.
In the first step, the bulk perturbation generates at second order in $\lambda$ the boundary action (see Appendix~\ref{app:largeK})
\begin{equation}
  \delta\mathcal{S}_{\partial}
  \propto-\lambda^{2}\int dx\left[\cos2\left(\theta_{1}-\theta_{1'}\right)+\cos2\left(\theta_{2}-\theta_{2'}\right)\right]
  \label{eq:TLL_largeK_step1}
\end{equation}
with a positive proportionality constant.
Each cosine transfers a fermion pair between a replica and its purity partner and has the minimum dimension $h_{\min}=4/(K+1)$ in Eq.~\eqref{eq:hmin_TLL}.
As in the boundary sine-Gordon model~\cite{fendley1994exact}, we assume that their couplings grow to strong coupling for $K>3$ and fix the boundary values at $\theta_{1}=\theta_{1'}$ and $\theta_{2}=\theta_{2'}$.
Together with Eq.~\eqref{eq:J_theta_TLL}, these conditions give
\begin{equation}
  \theta_{1}=\theta_{1'}=-\theta_{\tilde{1}}=-\theta_{\tilde{1}'},
  \quad
  \theta_{2}=\theta_{2'}=-\theta_{\tilde{2}}=-\theta_{\tilde{2}'}
  \label{eq:TLL_largeK_theta_step1}
\end{equation}
at the boundary, so that only $\theta_{1}$ and $\theta_{2}$ remain free.

In the second step, the bulk perturbation further generates an exactly marginal boundary operator. Specifically,  Eqs.~\eqref{eq:four_sector_pm_phi} and~\eqref{eq:TLL_largeK_theta_step1} give $2\sqrt{2}\theta_{m,-}=\pm2(\theta_{1}-\theta_{2})$ at the boundary for every sector $m$, and
each term $\cos(2\sqrt{2}\theta_{m,-})$ of $\mathcal{O}$ generates at first order in $\lambda$
\begin{equation}
  \psi=\cos2\left(\theta_{1}-\theta_{2}\right).
  \label{eq:TLL_largeK_psi}
\end{equation}
This operator is an $\mathrm{SU}(2)_{1}$ current with $h=1$ for any $K$, and no other nonidentity operator with $h\leq1$ is generated (see Appendix~\ref{app:largeK}).

To follow the flow generated by $\psi$, we first specify the boundary states labeled by its coupling $\mu_{\psi}$ in the boundary action~\eqref{eq:boundary_action_general},
\begin{equation}
  \delta\mathcal{S}_{\partial}=\mu_{\psi}\int dx\,\psi.
  \label{eq:TLL_largeK_action}
\end{equation}
This is the boundary sine-Gordon model with an exactly marginal cosine, and its exact solution~\cite{callan1994exact,polchinski1994free} determines these states.
We define $\mu_{\psi}$ in the regularization of Ref.~\cite{callan1994exact}, in which the boundary leaves $\theta_{1}-\theta_{2}$ free at $\mu_{\psi}=0$ and fixes it at $\mu_{\psi}=\pm1/2$.
In the latter case, the value of $2(\theta_{1}-\theta_{2})$ is fixed at the minimum of $\mu_{\psi}\psi$, namely at $\pi$ for $\mu_{\psi}=1/2$ and at $0$ for $\mu_{\psi}=-1/2$.
The corresponding leading-order boundary RG equation~\eqref{eq:boundary_RG_general} is (see Appendix~\ref{app:largeK})
\begin{equation}
  \frac{d\mu_{\psi}}{dl}=C_{\psi}\cos\qty(\pi\mu_{\psi}),
  \label{eq:TLL_largeK_flow}
\end{equation}
where $C_{\psi}<0$ is the term $C_{a}(\lambda)$ for $\psi$ at $\mu_{\psi}=0$ for $K>1$.
This sign comes from the negative coefficient $-4\pi^2\lambda/L^2$ of $\cos(2\sqrt{2}\theta_{m,-})$ in the rotated interaction~\eqref{eq:TLL_HDelta_bosonized}.
Thus, the coupling $\mu_\psi$ decreases monotonically until $\mu_{\psi}=-1/2$.
At this endpoint, the condition $\theta_{1}=\theta_{2}$ is added to Eq.~\eqref{eq:TLL_largeK_theta_step1}, so that only $\theta_{1}$ remains free, i.e., $\theta_{m,-}=0$ for every sector $m$.  This glues the pseudospin currents~\eqref{eq:current_bosonization} of each sector $m$ as
\begin{equation}
  J^{x}_{mL}=J^{x}_{mR}.
\end{equation}
The universal contribution for $K>3$ is determined by the $g$ factor of this fixed point.

To compute the $g$ factor, we bring this boundary condition to the rotated-boundary picture and use the compactification-lattice normalization of this Gaussian boundary state.
The convolution unitary $U$ in Eq.~\eqref{eq:current_rotation_unitary} rotates the pseudospin by $\pi/2$ around the $y$ axis and exchanges the $x$ and $z$ components, so that
\begin{equation}
  J^{z}_{mL}=J^{z}_{mR}.
\end{equation}
By Eq.~\eqref{eq:current_bosonization}, this condition again fixes $\theta_{m,-}$ at the boundary.
Since $U$ involves only $\theta_{m,-}$ and $\phi_{m,-}$, the boundary still fixes the seven combinations of $\theta$ that it fixes before $U$.
For the density field $\phi$, the numbers of the Neumann and Dirichlet directions are therefore $(n_{\mathrm{N}},n_{\mathrm{D}})=(7,1)$.
The $g$ factor of this boundary state is derived as (see Appendix~\ref{app:largeK})
\begin{equation}
  g=2K^{-3/2},
  \label{eq:TLL_largeK_g}
\end{equation}
leading to
\begin{equation}
  s=\ln\left(2K^{-3/2}\right).
  \label{eq:s2_TLL_largeK}
\end{equation}
Unlike Eq.~\eqref{eq:s2_TLL_quartic}, this result is not a perturbative expansion in $\gamma$, and only the identification of the flow relies on the lowest-order analysis.

\subsection{\label{subsec:fermion_TLL_numerical}Numerical results}

To numerically confirm the above analytical predictions, we compute the second-order MRE for the lattice model~\eqref{eq:spinless_H}.
We use two methods, namely exact diagonalization and the perfect-sampling method~\cite{matsuda2026efficient}.
In both methods, we take the boundary condition to be antiperiodic (periodic) when the particle number $N=L/2$ at half filling is even (odd).
We determine the universal contribution $s$ by fitting $M$ to a scaling form
\begin{equation}
  M = m L - s.
  \label{eq:M2_fitting_TLL}
\end{equation}

We first verify our analytical prediction for the $g$ factor near the free-fermion point $K=1$.
In this regime, the finite-size effect is small, and $s$ can be determined with high accuracy.
We thus compute $M$ for even sizes up to $L=18$ by exact diagonalization.
We first determine the ground state by the Lanczos method.
We then construct the convolved state by expressing two replicas as a single chain of length $2L$ and applying the convolution unitary on sites $i$ and $i+L$ with $i=1,2,\dots,L$.
Finally, we compute $M$ as the second-order R\'enyi entanglement entropy between the two halves of the chain partitioned at its center.
We fit $M$ at three consecutive even sizes $L$ centered at $L_0$ to the scaling form~\eqref{eq:M2_fitting_TLL} and denote the resulting universal constant by $s(K,L_0)$.

We test our perturbative prediction~\eqref{eq:s2_TLL_quartic} at the mixing angle $\zeta=\pi/4$, where $\sket{\mathcal{J}}$ remains a stable fixed point in the rotated-bulk picture.
The numerical result, shown in Fig.~\ref{fig:mre_perturbation}, is well converged near $K=1$ and in excellent agreement with the analytical prediction, confirming Eq.~\eqref{eq:s2_TLL_quartic}.
Since Eq.~\eqref{eq:s2_TLL_quartic} is even in $\gamma$, this agreement on both sides of $K=1$ also confirms the invariance of $s$ under the duality $K\leftrightarrow K^{-1}$.
According to this duality, the large deviation near $K=1/2$ is expected to be caused by a substantial finite-size effect originating from the bulk BKT transition.

\begin{figure}[t]
  \centering
  \includegraphics[width=0.45\textwidth]{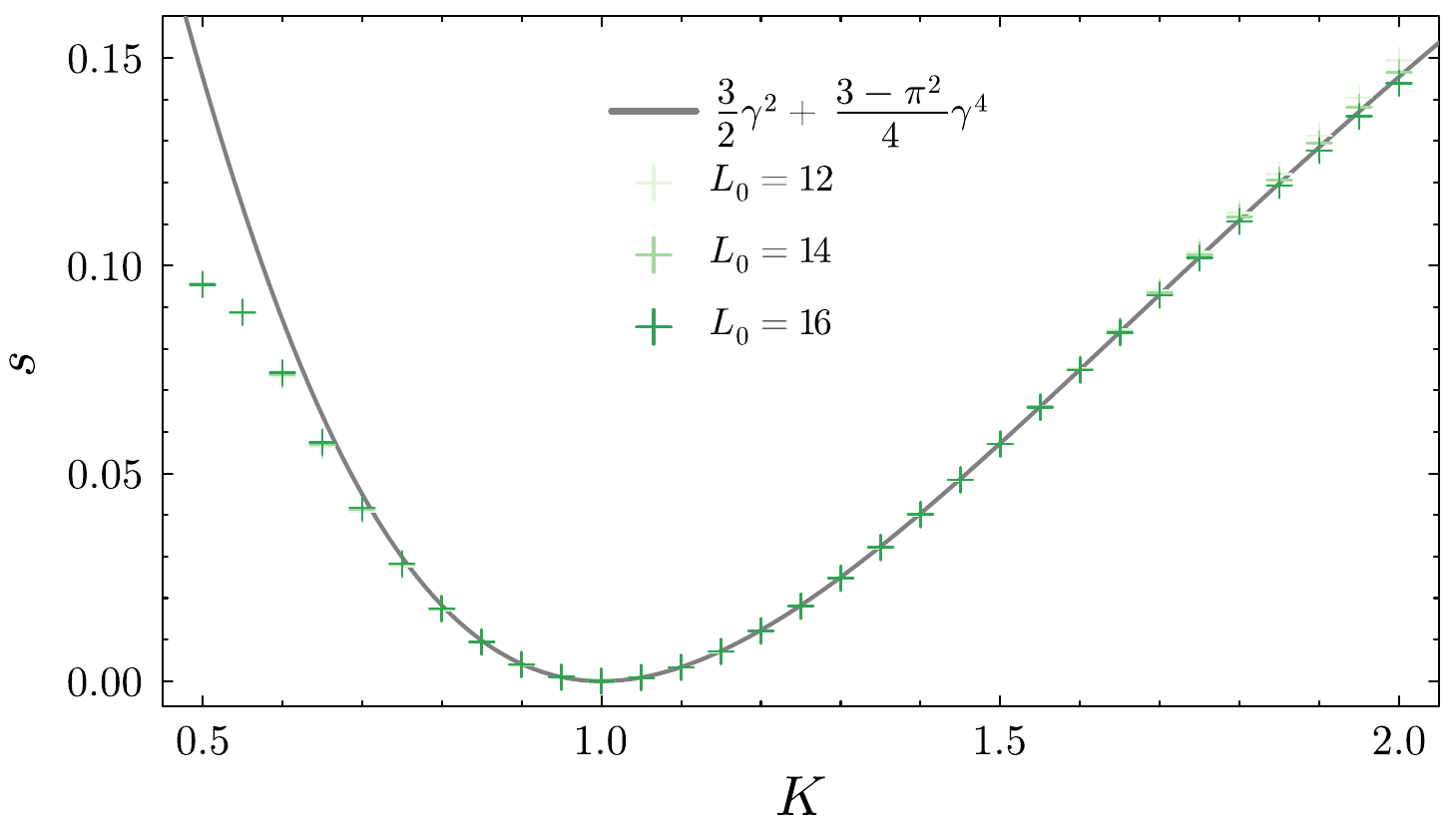}
  \caption{The universal constant term $s$ against $K$ near the free-fermion point $K=1$.
    The solid curve represents the perturbative result~\eqref{eq:s2_TLL_quartic}.
    The data points are obtained by fitting the numerical data of $M$ to the scaling form~\eqref{eq:M2_fitting_TLL} with $L=L_0-2,L_0,L_0+2$ with $L_0=12,14,16$.
  }
  \label{fig:mre_perturbation}
\end{figure}

We next turn to how this result changes along the mixing angle $\zeta$ of the convolution unitary.
Figure~\ref{fig:mre_angle}(a) shows the dependence of the universal contribution $s_{\text{UV}}$ at UV scale on $\zeta$.
Again, in the vicinity of $K=1$, our prediction~\eqref{eq:s2_TLL_universal_angle} shows excellent agreement with the numerical data at small system sizes, confirming the perturbative discussion in Sec.~\ref{subsec:fermion_TLL_current}.

We further test the boundary RG flow along the mixing angle predicted by the RG equation~\eqref{eq:TLL_angle_RG} near $\zeta=\pi/4$.
Let $s(\zeta,L_0)$ denote the universal constant at the angle $\zeta$ obtained in the same way as above, where we omit its dependence on $K$ for notational simplicity.
Since $s(\zeta,L_0)$ contains a finite-size effect common to all angles, we normalize it by its value at the fixed point $\zeta=\pi/4$.
The solution of Eq.~\eqref{eq:TLL_angle_RG} is $\delta\zeta(l)=(\zeta-\pi/4)e^{\gamma^2l}$, subject to the initial condition $\delta\zeta=\zeta-\pi/4$ at the UV scale $l=0$.
Suppose that $s(\zeta,L_0)$ at size $L_0\sim e^{l}$ is given by the UV value at the angle $\pi/4+\delta\zeta(l)$, i.e., $s(\zeta,L_0)\simeq s_{\text{UV}}(\pi/4+\delta\zeta(l))$.
We then obtain at leading order in $\gamma^2$
\begin{equation}
  \frac{1}{\gamma^2}\dv{\ln L_0}\frac{s(\zeta,L_0)}{s(\pi/4,L_0)}
  =(\zeta-\pi/4)\frac{s_{\text{UV}}'(\zeta)}{s_{\text{UV}}(\pi/4)},
  \label{eq:TLL_angle_flow_ratio}
\end{equation}
where $s_{\text{UV}}(\zeta)$ is given by Eq.~\eqref{eq:s2_TLL_universal_angle} and $s_{\text{UV}}'(\zeta)$ is its derivative by $\zeta$.

To check this relation, in Fig.~\ref{fig:mre_angle}(b), we compare the numerical difference
\begin{equation}
  \frac{1}{\gamma^2\ln[(L_0+2)/L_0]}\left[\frac{s(\zeta,L_0+2)}{s(\pi/4,L_0+2)}-\frac{s(\zeta,L_0)}{s(\pi/4,L_0)}\right]
  \label{eq:TLL_angle_flow_difference}
\end{equation}
with the theoretical prediction $(\zeta-\pi/4)s_{\text{UV}}'(\zeta)/{s_{\text{UV}}(\pi/4)}$.
The data for different $K$ collapse onto a single curve that agrees with the theoretical prediction~\eqref{eq:TLL_angle_flow_ratio} near $\zeta=\pi/4$, confirming the RG equation for the mixing angle~\eqref{eq:TLL_angle_RG}.
We note that the flow to $\zeta=0$ itself lies beyond the accessible sizes, because $\delta\zeta(l)$ becomes of order unity only at system sizes exponentially large in $1/\gamma^2$.

\begin{figure}[t]
  \centering
  \includegraphics[width=0.45\textwidth]{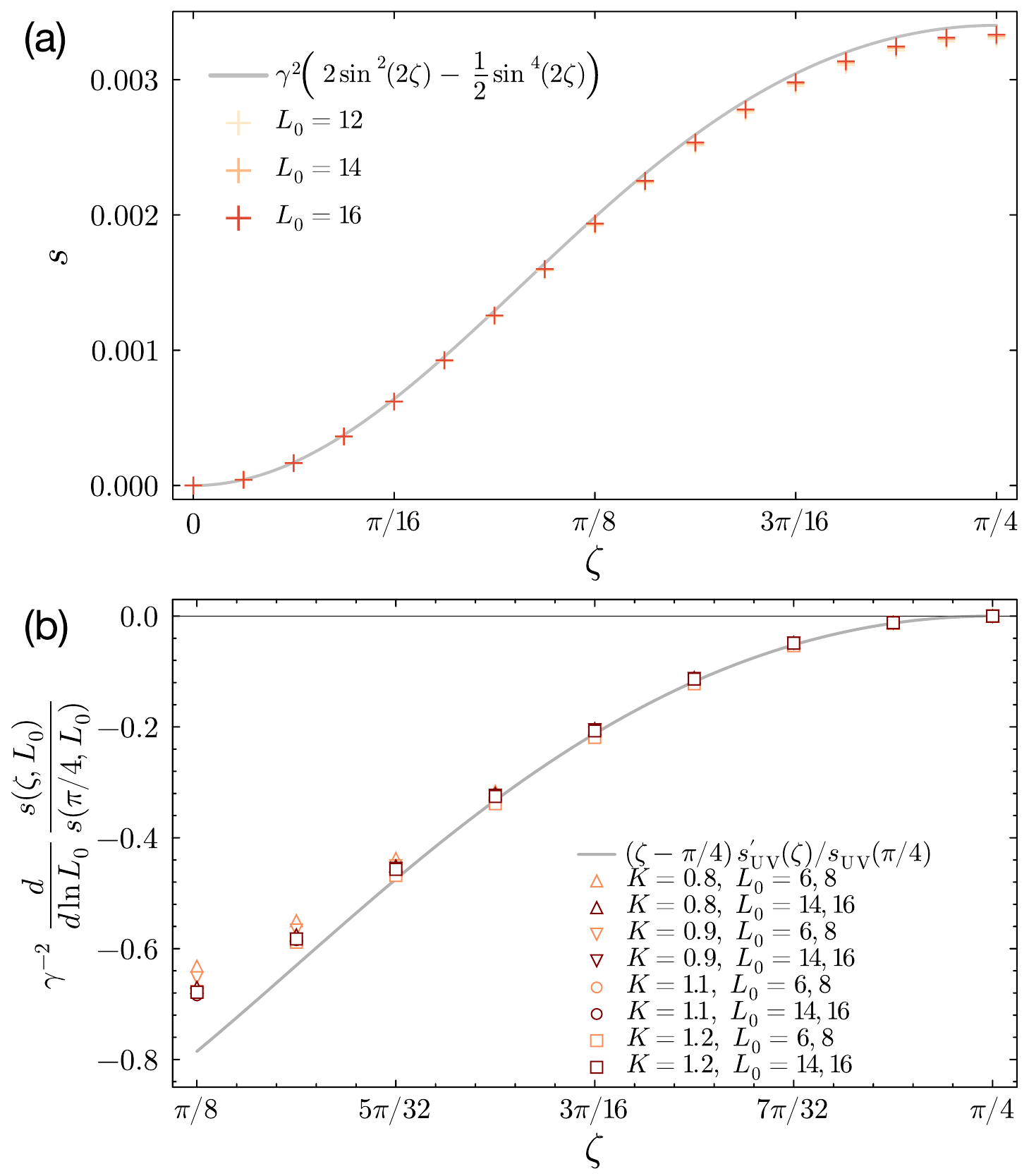}
  \caption{(a) The universal constant term $s$ against the mixing angle $\zeta$ at $K=1.1$.
    The solid curve represents the perturbative result~\eqref{eq:s2_TLL_universal_angle}.
    The data points are obtained by fitting the numerical data of $M$ to the scaling form~\eqref{eq:M2_fitting_TLL} with $L=L_0-2,L_0,L_0+2$ with $L_0=12,14,16$.
    (b) The numerical difference~\eqref{eq:TLL_angle_flow_difference} against $\zeta$ for $K=0.8,0.9,1.1,1.2$ with $L_0=6,14$.
    The solid curve represents the theoretical prediction~\eqref{eq:TLL_angle_flow_ratio}.
  }
  \label{fig:mre_angle}
\end{figure}

Finally, we test the predicted boundary transition by computing $s$ over a wide range of $K$.
At larger $K$, boundary RG flows slow the convergence with system size.
We thus compute $M$ for sizes up to $L=96$ using our perfect-sampling method~\cite{matsuda2026efficient} on ground states obtained by the density-matrix renormalization group~\cite{white1992density}.
We extract $m$ and $s$ by solving Eq.~\eqref{eq:M2_fitting_TLL} at two sizes centered at $L_0$, and again denote the resulting estimate by $s(K,L_0)$.
We estimate the infrared limit by extrapolating $s(K,L_0)$ linearly in $1/L_0$ to $L_0\to\infty$.

Figure~\ref{fig:mre_largeK} shows that $s(K,L_0)$ follows the perturbative result~\eqref{eq:s2_TLL_quartic} near $K=1$.
At larger $K$, the estimates decrease as $L_0$ increases, consistent with the boundary RG flow.
Near $K=3$ from below, this drift persists over the accessible sizes.
The leading boundary operator, with scaling dimension $h_{\mathrm{min}}$ in Eq.~\eqref{eq:hmin_TLL}, approaches marginality, and the linearized flow gives the estimate $L\gtrsim\exp[(K+1)/(3-K)]$ for appreciable relaxation of its coupling.
For $K>3$, the extrapolated values are close to the analytical prediction~\eqref{eq:s2_TLL_largeK}, supporting our identification of the infrared boundary fixed point.

\begin{figure}[t]
  \centering
  \includegraphics[width=0.45\textwidth]{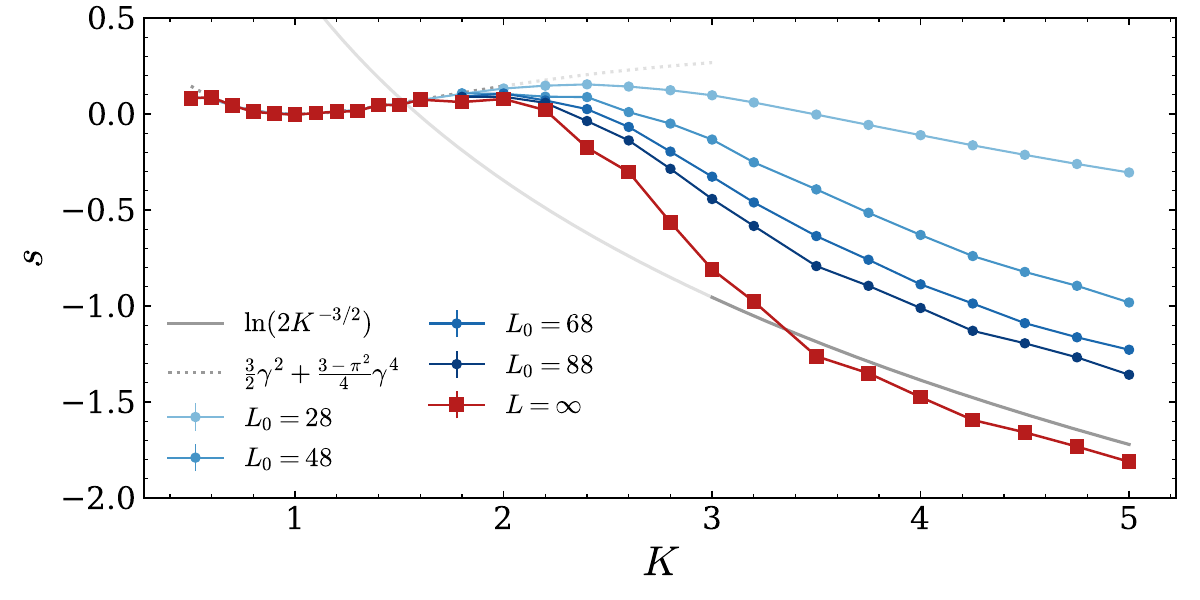}
  \caption{Estimates of the universal constant term $s$ as a function of $K$ obtained by perfect sampling for sizes up to $L=96$.
    For $K\ge1.8$, $s(K,L_0)$ is extracted from Eq.~\eqref{eq:M2_fitting_TLL} using $L=L_0\pm8$.
    The corresponding points labeled $L=\infty$ are obtained by linear extrapolation in $1/L_0$ using $L_0\ge48$.
    For $K<1.8$, the $L=\infty$ values are extrapolated in the same way from estimates using $L=L_0\pm2$ at $L_0=14,22,30$.
    The solid gray curve shows the prediction~\eqref{eq:s2_TLL_largeK} for $K>3$, and the dotted gray curve shows the perturbative result~\eqref{eq:s2_TLL_quartic} near $K=1$.
  }
  \label{fig:mre_largeK}
\end{figure}

\section{\label{sec:spin_XXZ}Stabilizer R\'enyi entropy of the XXZ chain}

We next demonstrate that the universal contribution depends on the computational resource even for the same TLL.
To this end, we consider the XXZ chain
\begin{equation}
  H=\frac{1}{4}\sum_{i=1}^{L}\qty(X_{i}X_{i+1}+Y_{i}Y_{i+1}+\Delta Z_{i}Z_{i+1}),
  \label{eq:XXZ_H}
\end{equation}
which is related to the spinless-fermion chain~\eqref{eq:spinless_H} via the Jordan-Wigner transformation with $\Delta=V/(2t)$.
The two chains therefore have the same ground state up to this mapping and are described by the same TLL with the TLL parameter $K=\pi/[2(\pi-\arccos\Delta)]$.
For the second-order SRE of this state, $\mathrm{SRE}_{2}=mL-s+o(1)$, we find
\begin{equation}
  s=
  \begin{cases}
    \ln\qty(2K^{2}),  & K\leq1, \\
    \ln\qty(2K^{-2}), & K\geq1.
  \end{cases}
  \label{eq:s_XXZ}
\end{equation}
This result is clearly distinct from the universal contribution of the fermionic MRE in Eqs.~\eqref{eq:s2_TLL_quartic} and~\eqref{eq:s2_TLL_largeK}.

To derive this result, we treat the bulk interaction as a perturbation to the boundary state at $K=1$ known from the critical Ising chain.
This is in the same spirit as the rotated-bulk picture, although the second-order SRE is not a convolution-based quantity.
We first use a Clifford unitary to map the XXZ chain onto the Ashkin-Teller chain, in which two critical Ising chains are coupled through their energy operators, without changing the SRE.
At $K=1$, where the two Ising chains decouple, the boundary state is the product of two copies of the boundary state of the critical Ising chain~\cite{hoshino2026stabilizer}.
For $K\neq1$, the coupling between the two Ising chains generates a marginal boundary operator, which drives the boundary RG flow to the state whose $g$ factor is given by Eq.~\eqref{eq:s_XXZ}.
Finally, we confirm this result by numerical calculations of the second-order SRE of the XXZ chain.

\subsection{\label{subsec:XXZ_mapping}Mapping to the Ashkin-Teller chain}

We first map the XXZ chain onto the Ashkin-Teller chain, reducing the problem to that of two coupled Ising chains.
To this end, we introduce the Pauli operators of two Ising chains A and B of length $L/2$ by
\begin{align}
  X^{\mathrm{A}}_{j} & =-X_{2j-1}X_{2j},
                     &
  Z^{\mathrm{A}}_{j} & =(-1)^{j}\prod_{k=1}^{2j-1}Y_{k},
  \nonumber                                              \\
  X^{\mathrm{B}}_{j} & =-Y_{2j-1}Y_{2j},
                     &
  Z^{\mathrm{B}}_{j} & =(-1)^{j}\prod_{k=2j}^{L}X_{k}.
  \label{eq:XXZ_Ising_Pauli}
\end{align}
These operators satisfy $\{X^{F}_{j},Z^{F}_{j}\}=0$ for $F=\mathrm{A},\mathrm{B}$ and commute otherwise, so that this transformation is a Clifford unitary.
In terms of them, the Hamiltonian~\eqref{eq:XXZ_H} becomes that of the quantum Ashkin-Teller chain~\cite{kohmoto1981hamiltonian},
\begin{align}
  H_{\mathrm{AT}}={} & -\frac{1}{4}\sum_{j=1}^{L/2}\Bigl[X^{\mathrm{A}}_{j}+X^{\mathrm{B}}_{j}+Z^{\mathrm{A}}_{j}Z^{\mathrm{A}}_{j+1}+Z^{\mathrm{B}}_{j}Z^{\mathrm{B}}_{j+1}
                                                     \nonumber\\
                                                     &+\Delta\qty(X^{\mathrm{A}}_{j}X^{\mathrm{B}}_{j}+Z^{\mathrm{A}}_{j}Z^{\mathrm{A}}_{j+1}Z^{\mathrm{B}}_{j}Z^{\mathrm{B}}_{j+1})\Bigr],
  \label{eq:XXZ_HAT}
\end{align}
where we take $L$ to be a multiple of $4$ and use $\prod_{k}X_{k}=\prod_{k}Y_{k}=1$ in the ground state for the periodic terms.
The two critical Ising chains decouple at $\Delta=0$, i.e., at $K=1$.

In the continuum limit, each Ising chain is described by the critical Ising CFT with the spin operator $\sigma_{F}$ and the energy operator $\varepsilon_{F}$, where~\cite{zou2020conformal}
\begin{align}
  Z^{F}_{j} & \simeq\sigma_{F},
  \nonumber                                         \\
  X^{F}_{j} & \simeq\ev{X^{F}_{j}}+\varepsilon_{F},
  \qquad
  Z^{F}_{j}Z^{F}_{j+1}\simeq\ev{X^{F}_{j}}-\varepsilon_{F}.
\end{align}
The terms linear in $\varepsilon_{F}$ cancel in the $\Delta$ term of Eq.~\eqref{eq:XXZ_HAT}, which leaves an exactly marginal coupling between the energy operators.
In the folded theory, the eight sheets $r$ and $\tilde{r}$ with $r=1,2,1',2'$ carry 16 Ising CFTs, and the bulk perturbation takes the form of the action~\eqref{eq:S_lambda_boundary} with
\begin{equation}
  \mathcal{O}=\sum_{r}\qty(\varepsilon_{\mathrm{A},r}\varepsilon_{\mathrm{B},r}+\varepsilon_{\mathrm{A},\tilde{r}}\varepsilon_{\mathrm{B},\tilde{r}}),
  \label{eq:XXZ_O}
\end{equation}
where $\lambda\propto 1-K$ near $K=1$.

We next express the second-order SRE in terms of these two chains.
Although a convolution unitary does not exist for $n=2$ (see Sec.~\ref{subsec:MRE_qubit}), the second-order SRE takes the same form as the MRE~\eqref{eq:MRE_projector},
\begin{equation}
  \mathrm{SRE}_{2}=-\ln\sbra{\rho}^{\otimes 2}P_{2}\sket{\rho}^{\otimes 2},
  \label{eq:XXZ_SRE_P2}
\end{equation}
with $P_{2}$ given by Eq.~\eqref{eq:bell_measurement} with $n=2$.
The universal contribution is thus given by $s=\ln g$ as in the scaling form~\eqref{eq:Mn_BCFT}, where $g$ is the $g$ factor of the boundary state to which $P_{2}$ flows in the IR.
Since a Clifford unitary only permutes the superkets $\sket{\sigma(\bm{u})}$ up to signs, it leaves $P_{2}$ invariant, and the second-order SRE of the XXZ chain equals that of the Ashkin-Teller chain.
Writing the superkets as $\sket{\sigma(\bm{u})}=\sket{\sigma(\bm{u}_{\mathrm{A}})}\otimes\sket{\sigma(\bm{u}_{\mathrm{B}})}$, we obtain
\begin{align}
  P_{2} & =P_{2}^{\mathrm{A}}\otimes P_{2}^{\mathrm{B}}
\end{align}
with
\begin{align}
  P_{2}^{F} & =\frac{1}{2^{L/2}}\sum_{\bm{u}_{F}}\qty(\sket{\sigma(\bm{u}_{F})}\sbra{\sigma(\bm{u}_{F})})^{\otimes 2}.
  \label{eq:XXZ_factorize}
\end{align}
At $K=1$, the bulk thus consists of two decoupled Ising CFTs, and the boundary state of each $P_{2}^{F}$ is known from Ref.~\cite{hoshino2026stabilizer}.
For $K\neq1$, $\mathcal{O}$ couples the two Ising CFTs in the bulk and induces a boundary RG flow of this boundary state.

\subsection{\label{subsec:XXZ_flow}Boundary RG flow}
The boundary state at $K=1$ is the product of two copies of the boundary state of the critical Ising chain derived in Ref.~\cite{hoshino2026stabilizer}.
To specify this state $\sket{\Gamma}$, we combine the sheets $r$ and $\tilde{r}$ of each Ising chain into the $\mathbb{Z}_{2}$ orbifold boson $\varphi^{F}_{r}$ with its dual field $\vartheta^{F}_{r}$, so that
\begin{align}
  \varepsilon_{F,r}+\varepsilon_{F,\tilde{r}} & =\cos2\varphi^{F}_{r},
  \qquad
  \varepsilon_{F,r}-\varepsilon_{F,\tilde{r}}\propto\cos2\vartheta^{F}_{r}.
  \label{eq:XXZ_dictionary}
\end{align}
The state $\sket{\Gamma}$ imposes the Dirichlet condition with the value $0$ on the differences of the four $\varphi^{F}_{r}$ and the Neumann condition on $\bm{c}\cdot\bm{\varphi}^{F}$, where $\bm{c}=\frac{1}{2}(1,1,1,1)$ and $\bm{\varphi}^{F}=(\varphi^{F}_{1},\varphi^{F}_{2},\varphi^{F}_{1'},\varphi^{F}_{2'})$.
With $\bm{\varphi}^{F}_{R/L}=\bm{\varphi}^{F}\mp\bm{\vartheta}^{F}$, this condition takes the form of the gluing condition~\eqref{eq:J_gluing_O},
\begin{equation}
  \partial_{x}\qty(\bm{\varphi}^{F}_{L}-G\bm{\varphi}^{F}_{R})\sket{\Gamma}=0,
  \qquad
  G=2\bm{c}\bm{c}^{\rm T}-I_{4}.
  \label{eq:XXZ_gluing}
\end{equation}
Since $\sket{\Gamma}$ has $g=\sqrt{2}$, the boundary state at $K=1$ is
\begin{equation}
  \sket{\Gamma}_{\mathrm{A}}\otimes\sket{\Gamma}_{\mathrm{B}},
  \qquad
  g=\sqrt{2}\times\sqrt{2}=2.
  \label{eq:XXZ_K1}
\end{equation}
The resulting value $s=\ln 2$ agrees with the exact result for the XX chain~\cite{trino2026stabilizershannon}.

We next examine the boundary operators generated by the bulk perturbation for $K\neq1$.
As in Sec.~\ref{subsec:fermion_TLL_stability}, we evaluate them at the lowest order in $\lambda$.
In terms of the above fields, the bulk operator~\eqref{eq:XXZ_O} becomes
\begin{equation}
  \mathcal{O}\propto\sum_{r}\qty(\cos2\varphi^{\mathrm{A}}_{r}\cos2\varphi^{\mathrm{B}}_{r}+\cos2\vartheta^{\mathrm{A}}_{r}\cos2\vartheta^{\mathrm{B}}_{r}).
  \label{eq:XXZ_O_bosons}
\end{equation}
As in Eq.~\eqref{eq:boundary_charge_element_TLL}, these operators generate boundary operators with the charges
\begin{align}
  e^{2i\varphi^{F}_{r}}:   & \quad\sqrt{2}\bm{q}=\bm{e}_{r}+G\bm{e}_{r}=\bm{c},
                           & h                                                               & =\frac{1}{2},
  \nonumber                                                                                                  \\
  e^{2i\vartheta^{F}_{r}}: & \quad\sqrt{2}\bm{q}=-\bm{e}_{r}+G\bm{e}_{r}=\bm{c}-2\bm{e}_{r},
                           & h                                                               & =\frac{3}{2},
  \label{eq:XXZ_charges}
\end{align}
where $\bm{e}_{r}$ is the unit vector of the component $r$ and the dimension is $h=\bm{q}^{2}$.
Since the two Ising chains are independent at $K=1$, the second term of Eq.~\eqref{eq:XXZ_O_bosons} generates only operators with $h=3$, while the first term generates
\begin{equation}
  \psi=\cos\qty(\bm{c}\cdot\bm{\varphi}^{\mathrm{A}})\cos\qty(\bm{c}\cdot\bm{\varphi}^{\mathrm{B}})
  \label{eq:XXZ_psi}
\end{equation}
with $h=1$.
Moreover, $\psi$ is the only nonidentity boundary operator with $h\leq1$ generated at any order (see Appendix~\ref{app:XXZ}).

To follow the flow generated by $\psi$, we rewrite it as the sum of two $\mathrm{SU}(2)_{1}$ currents,
\begin{equation}
  \psi=\frac{1}{2}\qty[\cos\bm{c}\cdot\qty(\bm{\varphi}^{\mathrm{A}}+\bm{\varphi}^{\mathrm{B}})+\cos\bm{c}\cdot\qty(\bm{\varphi}^{\mathrm{A}}-\bm{\varphi}^{\mathrm{B}})].
  \label{eq:XXZ_currents}
\end{equation}
At $K=1$, $\psi$ is exactly marginal and deforms the boundary state continuously without changing its $g$ factor~\cite{callan1994exact,polchinski1994free}.
In terms of its coupling $\mu_{\psi}$ defined as in the boundary action~\eqref{eq:TLL_largeK_action}, each cosine in Eq.~\eqref{eq:XXZ_currents} has the coupling $\mu_{\psi}/2$.
Each argument thus obeys the Neumann condition at $\mu_{\psi}=0$ and the Dirichlet condition at $\mu_{\psi}=\pm1$, as in Sec.~\ref{subsec:fermion_TLL_largeK}.
For $K\neq1$, the bulk perturbation generates $\psi$ on the state with coupling $\mu_{\psi}$ with a coefficient proportional to $\cos(\pi\mu_{\psi}/2)$, so that the boundary RG equation~\eqref{eq:boundary_RG_general} becomes (see Appendix~\ref{app:XXZ})
\begin{equation}
  \frac{d\mu_{\psi}}{dl}=C_{\psi}\cos\frac{\pi\mu_{\psi}}{2},
  \label{eq:XXZ_flow}
\end{equation}
which has the same form as the flow~\eqref{eq:TLL_largeK_flow} of the fermionic MRE for $K>3$ with $\mu_{\psi}$ replaced by $\mu_{\psi}/2$.
Here, $C_{\psi}$ is the term $C_{a}(\lambda)$ for $\psi$ at $\mu_{\psi}=0$, which is of first order in $\lambda$.
Its sign is that of the coefficient $-\lambda$ of $\mathcal{O}$ in the action~\eqref{eq:S_lambda_boundary}, because $2\varphi^{F}_{r}=\bm{c}\cdot\bm{\varphi}^{F}$ at the boundary reduces the first term of Eq.~\eqref{eq:XXZ_O_bosons} to $\psi$ with a positive prefactor (see Appendix~\ref{app:XXZ}).
Since $\lambda\propto1-K$, $C_{\psi}$ has the sign of $K-1$.
The coupling thus changes monotonically until $\mu_{\psi}=1$ for $K>1$ and $\mu_{\psi}=-1$ for $K<1$.

To identify the two endpoints, we determine the boundary values of $\varphi^{F}_{r}$ at $\mu_{\psi}=\pm1$.
The Dirichlet condition fixes the arguments of the two cosines in Eq.~\eqref{eq:XXZ_currents} at the minimum of $\mu_{\psi}\psi$ as in Sec.~\ref{subsec:fermion_TLL_largeK}.
Since $\psi$ contains only $\bm{c}\cdot\bm{\varphi}^{F}$, the differences of $\varphi^{F}_{r}$ keep the value $0$, and all $\varphi^{F}_{r}$ take a common value $\varphi^{F}_{0}$ with $\bm{c}\cdot\bm{\varphi}^{F}=2\varphi^{F}_{0}$.
Equation~\eqref{eq:XXZ_psi} then gives
\begin{align}
  \mu_{\psi}=-1 & :\quad\psi=\cos2\varphi^{\mathrm{A}}_{0}\cos2\varphi^{\mathrm{B}}_{0}=1,
  \nonumber                                                                                 \\
  \mu_{\psi}=1  & :\quad\psi=\cos2\varphi^{\mathrm{A}}_{0}\cos2\varphi^{\mathrm{B}}_{0}=-1,
\end{align}
so that $\cos2\varphi^{F}_{0}=\pm1$.
The Dirichlet condition with $\cos2\varphi=1$ ($-1$) of an orbifold boson corresponds to the fixed (free) condition on the two combined Ising CFTs~\cite{oshikawa1997boundary}.
Therefore, each Ising CFT obeys the free or fixed condition at the endpoints, and chains A and B have the same condition on every sheet at $\mu_{\psi}=-1$ and different conditions on every sheet at $\mu_{\psi}=1$.

\subsection{\label{subsec:XXZ_g}Universal contribution}

To evaluate the $g$ factors of the two endpoints, it is convenient to combine instead the Ising chains A and B on the same sheet, because the bulk perturbation then changes only a TLL parameter.
This combination gives an orbifold boson $\phi_{r}$ with
\begin{equation}
  \sigma_{\mathrm{A},r}\sigma_{\mathrm{B},r}=\cos\phi_{r},
  \qquad
  \varepsilon_{\mathrm{A},r}\varepsilon_{\mathrm{B},r}\propto-(\partial\phi_{r})^{2},
  \label{eq:XXZ_dictionary_TLL}
\end{equation}
and likewise for $\tilde{r}$.
The action of these bosons is
\begin{equation}
  \mathcal{S}=\frac{1}{2\pi K}\sum_{r}\int d^{2}z\,\qty[(\partial\phi_{r})^{2}+(\partial\phi_{\tilde{r}})^{2}],
  \label{eq:XXZ_action}
\end{equation}
where the coefficient follows from $Z^{\mathrm{A}}_{1}Z^{\mathrm{B}}_{1}Z^{\mathrm{A}}_{j}Z^{\mathrm{B}}_{j}=\pm\prod_{k=2}^{2j-1}Z_{k}$, which equates the correlation function of $\cos\phi_{r}$ with the expectation value $\propto j^{-K/2}$ of this string operator in the TLL.

At the endpoints, $\phi_{r}$ obeys the Dirichlet condition when the two combined Ising CFTs have the same condition and the Neumann condition when they have different conditions~\cite{oshikawa1997boundary}, so that $(n_{\mathrm{N}},n_{\mathrm{D}})=(0,8)$ at $\mu_{\psi}=-1$ and $(8,0)$ at $\mu_{\psi}=1$.
For the action~\eqref{eq:XXZ_action}, the $g$ factors of these boundary states are proportional to $K^{2}$ and $K^{-2}$, respectively~\cite{hoshino2026stabilizer}.
With $g=2$ at $K=1$ in Eq.~\eqref{eq:XXZ_K1}, we obtain
\begin{equation}
  g=
  \begin{cases}
    2K^{2},  & \mu_{\psi}=-1, \\
    2K^{-2}, & \mu_{\psi}=1.
  \end{cases}
  \label{eq:XXZ_g_endpoints}
\end{equation}
By the RG equation~\eqref{eq:XXZ_flow}, the boundary state reaches the Neumann condition for $K>1$ and the Dirichlet condition for $K<1$.
We confirm that the endpoint reached on each side is stable (see Appendix~\ref{app:XXZ}), and thus the $g$ factors~\eqref{eq:XXZ_g_endpoints} give the universal contribution~\eqref{eq:s_XXZ}.
Again, this result is not a perturbative expansion, and only the identification of the flow relies on the lowest-order analysis, as in the fermionic case with $K>3$.

\subsection{\label{subsec:XXZ_numerical}Numerical results}

To numerically test the universal contribution~\eqref{eq:s_XXZ}, we compute $\mathrm{SRE}_{2}$ of the XXZ chain~\eqref{eq:XXZ_H} with the periodic boundary condition for sizes up to $L=96$ by applying perfect Pauli sampling~\cite{lami2023nonstabilizerness} to the ground state obtained by the density-matrix renormalization group~\cite{white1992density}.
As in Sec.~\ref{subsec:fermion_TLL_numerical}, we determine $s$ by solving $\mathrm{SRE}_{2}=mL-s$ at two sizes $L=L_{0}\pm8$ for $m$ and $s$, and denote the resulting universal constant by $s(K,L_{0})$.

The result is shown in Fig.~\ref{fig:XXZ_s}.
The numerical data show significant finite-size effects at all $K\neq1$, which indicates the boundary RG flow.
For $K\ge1.8$, $s(K,L_{0})$ extrapolated to $L_{0}\to\infty$ linearly in $1/L_0$ perfectly agrees with our analytical result $\ln(2K^{-2})$.
We remark that the confirmation of the analytical result near $K=1$ and $K=1/2$ is difficult due to the slow convergence.
This is caused by the bulk BKT transition at $K=1/2$ and the boundary RG flow driven by a coupling proportional to $K-1$, which reaches its end only for $L\gtrsim \exp(1/\lvert K-1\rvert)$.

\begin{figure}[t]
  \centering
  \includegraphics[width=0.45\textwidth]{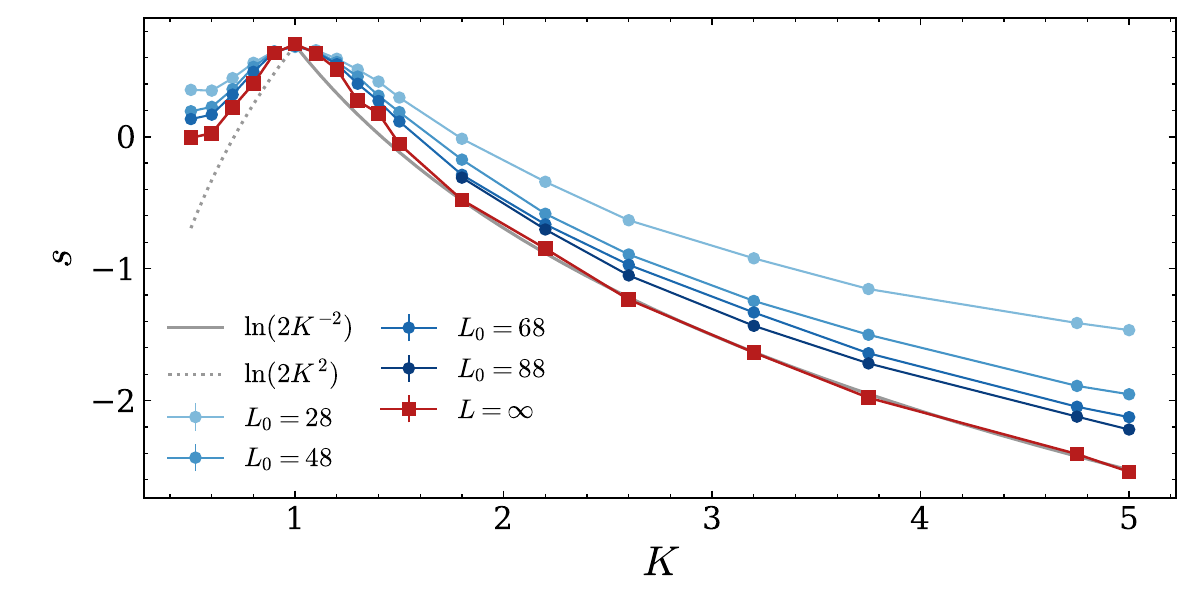}
  \caption{The universal constant term $s$ of the second-order SRE of the XXZ chain against $K$ obtained by sampling for sizes up to $L=96$.
    The data points for each $L_{0}$ are obtained by the scaling form $\mathrm{SRE}_{2}=mL-s$ with $L=L_{0}\pm8$, and those for $L=\infty$ are obtained by extrapolating the data points with $L_{0}\ge48$ linearly in $1/L_{0}$.
    The solid and dotted curves represent $\ln(2K^{-2})$ and $\ln(2K^{2})$ in Eq.~\eqref{eq:s_XXZ}, respectively.
  }
  \label{fig:XXZ_s}
\end{figure}

\section{Summary and Outlook}\label{sec:discussions}

We have shown that quantum computational resources can be discussed on an equal footing across spins, bosons, and fermions by constructing a magic measure together with its field-theoretical framework.
Specifically, we have defined the magic R\'enyi entropy (MRE) by extending previously introduced convolutions, and established it as a valid resource monotone of nonstabilizerness and non-Gaussianity.
We have also related the MRE to the simulation complexity by bounding the fermionic Gaussian rank from below.
By formulating the MRE within the Euclidean path-integral formalism, we have shown that its universal contribution at criticality is determined by the Affleck-Ludwig boundary entropy and can undergo a transition within a single bulk critical phase through bulk-induced boundary RG flows.
As a concrete application, we have analyzed the fermionic MRE of the spinless-fermion chain and the SRE of the XXZ chain, both described by the TLL, and found that the two resources undergo boundary transitions at different points, $K=3$ and $K=1$, respectively.
These analytical results have been confirmed numerically.

While we have proved several key properties of the MRE as a magic measure, its complete resource-theoretic understanding remains open.
For fermions, monotonicity under free protocols remains open for higher orders $n>2$ and for convolution unitaries that do not satisfy the convolution-multiplication duality~\eqref{eq:duality_general}.
The operational interpretation of the MRE also remains incomplete.
For qubits, the SRE quantifies distinguishability from stabilizer states and indistinguishability from Haar-random states, depending on the R\'enyi index $n$~\cite{bittel2026operational}.
For bosons and fermions, interpretations based on Gaussian testing are known only for the two-replica case ($n=2$)~\cite{bu2025efficient,hahn2025measuring,lyu2024fermionic,coffman2025measuring}.

Also, our lower bound on the fermionic Gaussian rank raises the question of whether analogous bounds hold for the stabilizer rank of spins and the Gaussian rank of bosons.
Obtaining comparable exponential bounds for spins remains an interesting open problem (Sec.~\ref{subsec:rank}).
We note that the SRE provides a lower bound on the robustness of magic~\cite{leone2022stabilizer}, which quantifies the overhead of a specific classical simulation algorithm~\cite{howard2017application}; whether the SRE also bounds the stabilizer rank remains open.

Efficient numerical evaluation of the MRE is also needed to investigate many-body magic.
For the SRE, existing methods enable efficient calculations in dynamics and higher-dimensional systems~\cite{tarabunga2023manybody,tarabunga2024nonstabilizerness,huang2026fast,xiao2026exponentially,ding2025evaluating,lami2023nonstabilizerness,sticlet2025nonstabilizerness,liu2025nonequilibrium}.
Developing analogous methods for the bosonic MRE and extending fermionic calculations beyond the perfect-sampling method for matrix-product states~\cite{matsuda2026efficient}, for example through other tensor-network or Monte Carlo approaches, would expand the range of bosonic and fermionic many-body systems accessible to numerical study.

While we have focused on ground states, we expect that the common structure across spins, bosons, and fermions can also manifest in dynamics.
For instance, in one-dimensional systems, both the nonstabilizerness of spin systems and the non-Gaussianity of fermionic systems are numerically found to saturate on a timescale linear in system size under Hamiltonian dynamics and proportional to $\ln L$ under Haar-random dynamics~\cite{tirrito2025anticoncentration,turkeshi2025magic,sierant2026fermionic}.
Such common behavior may be captured by formulating the dynamics of the MRE using a real-time path integral.
Moreover, it is natural to ask whether non-Gaussianity also undergoes the phase transitions reported for the SRE in dynamical settings~\cite{bejan2024dynamical,leone2024phase,niroula2024phase,sierant2026theory,tirrito2024magic,tirrito2025magic,wang2025magic}.

From a broader perspective, studies of many-body magic may also shed light on quantum chaos and, through the holographic principle, on the simulation complexity of quantum gravity; indeed, such connections have been explored for nonstabilizerness~\cite{leone2022stabilizer,leone2021quantum,goto2022probing,white2021conformal,malvimat2026multipartite,cao2024nontrivial,cao2025gravitational}.
The SYK model provides a useful testbed, as it exhibits quantum chaos, admits a holographic gravitational dual, and is solvable in the large-$N$ limit~\cite{sachdev1993gapless,jha2025introduction,kitaev2014hidden,chowdhury2022sachdevyekitaev}.
Viewed as a qubit model, its nonstabilizerness has already been investigated in several settings~\cite{malvimat2026multipartite,zhang2026stabilizer,jasser2025stabilizer,russomanno2025nonstabilizerness,sun2026connecting,bera2025nonstabilizerness,iannotti2026nonstabilizerness}, unveiling a novel phase transition intrinsic to the SRE~\cite{zhang2026stabilizer} and characterizing nonstabilizerness in chaotic systems~\cite{iannotti2026nonstabilizerness}.
Since the SYK model is intrinsically fermionic, we have studied its non-Gaussianity with the MRE in Ref.~\cite{hoshino2026universal} and found a first-order transition in the MRE density that is not reflected in the thermal free energy.
However, the implication of these transitions for the complexity of simulating the gravitational dual is still not well understood.

We hope that our work stimulates further studies toward a systematic understanding of quantum magic in many-body systems.


\begin{acknowledgments}
  We thank Shunsuke Furukawa, Yuxuan Guo, Oliver Hahn, Shinsei Ryu, and Ryuji Takagi for helpful discussions.
  Y.A. and M.O. gratefully acknowledge support from the Simons Center for Geometry and Physics, Stony Brook University at which some of the research for this paper was performed.
  R.M. and M.H. were supported by FoPM, WINGS Program, the University of Tokyo.
  The work of M.H. was also supported by JSPS KAKENHI Grant No.~JP26KJ0910.
  The work of M.O. was supported by JSPS KAKENHI Grant No.~JP23K25791.
  Y.A. acknowledges support from JST FOREST Program (Grant No.~JPMJFR222U), JST CREST (Grant No.~JPMJCR23I2), and JST [Moonshot R\&D] (Grant No.~JPMJMS256J).
  The authors used ChatGPT (OpenAI), Codex (GPT-5.6-Sol and GPT-6-Astra, OpenAI), and Claude Code (Claude Opus 4.8, Opus 5, Opus 5.5, Fable 5, and Fable 5.1, Anthropic) for research and manuscript preparation.
  The authors checked all derivations and revised AI-generated text.
  The authors take full responsibility for the paper.
\end{acknowledgments}

\appendix

\section{Basic properties of the MRE}\label{app:MRE_properties}

In this appendix, we prove Theorem~\ref{thm:MRE-properties}, first for qudits and then for bosons and fermions.

\subsection{Qudits}\label{app:properties_qudit}

We first derive the form~\eqref{eq:weighted_duality} for qudits from the two requirements on $U$ in Sec.~\ref{subsec:MRE_definition}: $U$ is a Clifford unitary acting uniformly on all sites, and it keeps the convolved state of a stabilizer state pure.
The first requirement gives
\begin{equation}
  U^{\dagger}\qty[D_1(\bm{u})\otimes I^{\otimes(n-1)}]U=c(\bm{u})\bigotimes_{r=1}^{n}D_r(T_r\bm{u}),
  \label{eq:T_def}
\end{equation}
where $T_r$ is a $2\times2$ matrix over $\mathbb{Z}_d$ acting on the components $(q_i,p_i)$ of each site $i$, and $c(\bm{u})$ is a phase factor.
Since $\langle T\bm{u},T\bm{v}\rangle=(\det T)\langle\bm{u},\bm{v}\rangle$ for such matrices $T$, conjugating the commutation relation~\eqref{eq:HW_commutation_qudit} by $U$ gives
\begin{equation}
  \sum_{r=1}^{n}\det T_r\equiv1\pmod{d}.
  \label{eq:det_sum}
\end{equation}
Equation~\eqref{eq:T_def} and the orthogonality of the Heisenberg-Weyl operators yield
\begin{equation}
  \nu_n(\rho_1,\ldots,\rho_n)=\frac{1}{d^{L}}\sum_{\bm{u}}\prod_{r=1}^{n}\qty|\chi_{\rho_r}(T_r\bm{u})|^{2}.
  \label{eq:nu_T}
\end{equation}

The second requirement, that $\mathcal{C}_n(\psi)$ be pure for stabilizer states $\psi$, forces all the $T_r$ to be multiples of a single invertible matrix.
For a single-qudit state $\psi$ in the Clifford orbit of $\ket{0}$, $|\chi_\psi(\bm{u})|$ equals one on $\mathbb{Z}_d\bm{v}$ and vanishes elsewhere.
Here, $\bm{v}\in\mathbb{Z}_d^2$ satisfies the condition that no prime factor of $d$ divides both components, and every such vector is realized by a state in this orbit.
Setting $L=1$ in Eq.~\eqref{eq:nu_T}, the purity requirement becomes
\begin{equation}
  1=\nu_n(\psi)=\frac{1}{d}\,\#\qty{\bm{u}\in\mathbb{Z}_d^{2}\ \middle|\ T_r\bm{u}\in\mathbb{Z}_d\bm{v}\text{ for all }r}.
  \label{eq:stab_count}
\end{equation}
By the Chinese remainder theorem, $\mathbb{Z}_d$ decomposes into the rings $\mathbb{Z}_{p^k}$ for the prime powers appearing in the factorization of $d$, and the cardinality in Eq.~\eqref{eq:stab_count} factorizes accordingly.
We show that each factor is at most $p^k$; since their product must equal $d$, every factor must attain this bound.
For $d=p^k$, all noninvertible elements of $\mathbb{Z}_{p^k}$ are divisible by $p$.
Equation~\eqref{eq:det_sum} thus requires at least one determinant $\det T_s$ to be invertible, so $T_s$ is invertible.
The set in Eq.~\eqref{eq:stab_count} is contained in $T_s^{-1}\mathbb{Z}_d\bm{v}$, which has $p^k$ elements.
This proves the bound, and equality forces the two sets to coincide.
Consequently, $T_rT_s^{-1}$ maps every such line into itself, including those generated by $(1,0)$, $(0,1)$, and $(1,1)$.
A $2\times2$ matrix with this property is a scalar multiple of the identity, so $T_r=t_rT_s$ with $t_r\in\mathbb{Z}_{p^k}$.
For general $d$, the Chinese remainder theorem combines these local representations into $T_r=t_rT$, where $t_r\in\mathbb{Z}_d$ and $T$ is invertible.
Equation~\eqref{eq:det_sum} then becomes
\begin{equation}
  \qty(\sum_{r=1}^{n}t_r^{2})\det T\equiv1\pmod{d}.
  \label{eq:t_sum}
\end{equation}
The conjugation relation~\eqref{eq:T_def} gives Eq.~\eqref{eq:weighted_duality} with $\bm{u}'=T\bm{u}$.
The invertible change of summation variable $\bm{u}\mapsto T\bm{u}$ in Eq.~\eqref{eq:nu_T} then yields Eq.~\eqref{eq:nu_general_qudit}.

We now prove Theorem~\ref{thm:MRE-properties} for qudits.
Because $U$ acts uniformly on all sites, self-convolution factorizes over tensor products $\rho_1\otimes\rho_2$.
Thus, $\nu_n$ is multiplicative and $M_n$ is additive.
For invariance, a Clifford unitary $U_C$ satisfies $U_C^{\dagger}D(\bm{u})U_C\propto D(S\bm{u})$, where $S$ is an invertible linear map on $\mathbb{Z}_d^{2L}$.
Hence $|\chi_{U_C\rho U_C^{\dagger}}(\bm{u})|=|\chi_\rho(S\bm{u})|$.
Since $S(t_r\bm{u})=t_rS\bm{u}$, changing the summation variable to $S\bm{u}$ in Eq.~\eqref{eq:nu_general_qudit} leaves $\nu_n$ unchanged.

For faithfulness, $\nu_n\leq1$ gives $M_n\geq0$.
Although Eq.~\eqref{eq:t_sum} was derived using only the Clifford orbit of $\ket{0}$, it ensures the vanishing of $M_n$ required for all states in $\mathrm{Stab}_L$.
Indeed, for $\psi\in\mathrm{Stab}_L$, $|\chi_\psi|$ equals one on a subgroup $G\subset\mathbb{Z}_d^{2L}$ with $d^L$ elements and vanishes elsewhere.
Equation~\eqref{eq:t_sum} gives $\bm{u}=\det T\sum_r t_r(t_r\bm{u})$, so $t_r\bm{u}\in G$ for all $r$ if and only if $\bm{u}\in G$.
The convolution purity~\eqref{eq:nu_general_qudit} thus equals one.
Conversely, suppose that $\nu_n(\psi)=1$ for a pure state $\psi$.
Since $U$ genuinely mixes the replicas, we may relabel them so that $t_1$ and $t_2$ are coprime to $d$.
Absorbing $t_1$ into $T$ and rescaling the weights then sets $t_1=1$.
Bounding the factors with $r\geq3$ in Eq.~\eqref{eq:nu_general_qudit} by one gives
\begin{align}
  1 & =\nu_n(\psi)\leq\frac{1}{d^{L}}\sum_{\bm{u}}|\chi_\psi(\bm{u})|^{2}|\chi_\psi(t_2\bm{u})|^{2}\notag \\
    & \leq\frac{1}{d^{L}}\sum_{\bm{u}}|\chi_\psi(\bm{u})|^{2}=1.
\end{align}
Equality requires $|\chi_\psi(t_2\bm{u})|=1$ whenever $\chi_\psi(\bm{u})\neq0$.
Since $t_2$ is coprime to $d$, multiplication by $t_2$ maps the support of $\chi_\psi$ injectively into the subset where $|\chi_\psi|=1$.
This subset must therefore equal the full support, so $|\chi_\psi|$ takes only the values $0$ and $1$.
The identity $\sum_{\bm{u}}|\chi_\psi(\bm{u})|^{2}=d^{L}$ then implies that exactly $d^L$ vectors satisfy $|\chi_\psi(\bm{u})|=1$, and $\psi$ is an eigenvector of the corresponding Heisenberg-Weyl operators.
Choosing their phases so that $\psi$ has eigenvalue one gives an abelian subgroup of $\mathcal{D}_L$ with $d^L$ elements and no nontrivial scalar multiple of the identity~\cite{gross2006hudsons}, leading to  $\psi\in\mathrm{Stab}_L$.

\subsection{Bosons and fermions}\label{app:properties_BF}

Here we prove Theorem~\ref{thm:MRE-properties} for bosons and fermions, with an arbitrary replica number $n\ge 2$ and an arbitrary orthogonal rotation $O$ whose first row has at least two nonzero entries.

For notational simplicity, throughout this section and Appendix~\ref{app:monotonicity_theorems}, we use the rescaled fermionic characteristic polynomial $\chi_X(\bm{u})=\tr[X\exp(\bm{r}^{\rm T}\bm{u})]$, where the displacement in this definition is $D_{\mathrm f}(2\bm{u})$.
The reason is that, while $D_{\mathrm f}(\bm{u})$ takes a form parallel to the bosonic one and is convenient for discussions based on coherent states~\cite{cahill1969ordered,cahill1999density}, fermionic linear optics, which is the main ingredient of these appendices, is often formulated in terms of $D_{\mathrm f}(2\bm{u})$~\cite{bravyi2004lagrangian}.
With this definition, a fermionic Gaussian state with covariance matrix $\sigma$ is expressed as
\begin{equation}
  \chi_{\rho_G}(\bm{u}) = \exp(\frac{i}{2}\bm{u}^{\rm T}\sigma \bm{u}),
  \label{eq:chi_gaussian_fermion_app}
\end{equation}
and the closed form~\eqref{eq:purity_fermion_chi} of the convolution purity obtained from the convolution-multiplication duality reads
\begin{equation}
  \nu_{n}(\rho_{1},\ldots,\rho_{n})=2^{-L}\norm{\prod_{r=1}^{n}\chi_{\rho_{r}}\!\qty(\frac{\bm{u}}{\sqrt{n}})}^{2}.
  \label{eq:purity_fermion_app}
\end{equation}

We establish the basic properties listed in Theorem~\ref{thm:MRE-properties}, namely faithfulness, invariance under Gaussian unitaries, and additivity.
While we focus on pure states in the main text, for the sake of generality, we here introduce a slightly extended version of the MRE so that it can also be applied to mixed states, following Refs.~\cite{lyu2024fermionic,bu2025efficient, hahn2025measuring}.
To this end, we replace the second-order R\'enyi entanglement entropy, which the MRE uses to quantify the correlation between replica 1 and the other replicas, by a general measure of correlation.
We denote such a measure between $A$ and $B$ by $C_{A,B}(\rho_{AB})$, and the reduced density operators by $\rho_{A(B)} = \tr_{B(A)} \rho_{AB}$.
We require $C_{A,B}$ to satisfy the following properties:
\begin{enumerate}
  \item $C_{A,B}(\rho_{AB})\geq0$, with equality holding if and only if $\rho_{AB}=\rho_A\otimes \rho_B$.
  \item $C_{A,B}((\Phi_A\otimes \Phi_B) (\rho_{AB})) \leq C_{A,B}(\rho_{AB})$, where $\Phi_A$ and $\Phi_B$ are local operations on $A$ and $B$, respectively.
  \item $C_{A,B}(\rho_{AB}\otimes \sigma_{AB}) = C_{A,B}(\rho_{AB}) + C_{A,B}(\sigma_{AB})$.
\end{enumerate}
If $C_{A,B}$ is defined only on pure states, these properties are required only for pure states, and $\Phi_A\otimes\Phi_B$ in the second property is required to map the state to a pure state.
Note that invariance under local unitary transformations follows from the second property, since a unitary is reversible.
Using $C_{A,B}$, we consider the following quantity as a measure of magic:
\begin{equation}
  N_C(\rho) = C_{1, 1^{\rm c}}(U \rho^{\otimes n} U^\dagger).
\end{equation}
For pure states, correlation is equivalent to entanglement, which is measured by the R\'enyi entanglement entropy
\begin{align}
  C_{A,B}(\rho_{AB}) = \frac{1}{1-k} \ln \tr_A[\rho_A^k].
  \label{eq:CAB_Renyi}
\end{align}
Here, $k>0$ and $k\neq1$.
For this entropy, the first and third properties are immediate.
The second holds because, when the output $\ket{\psi'}$ of a pure state $\ket{\psi}$ is pure, the Stinespring dilations of $\Phi_A$ and $\Phi_B$ are local isometries mapping $\ket{\psi}$ to $\ket{\psi'}\otimes\ket{e}$, where $\ket{e}$ is a pure state of the two environments, so that $C_{A,B}(\psi)=C_{A,B}(\psi')+C_{A,B}(e)\geq C_{A,B}(\psi')$.
The special case $k=2$ corresponds to the MRE discussed in the main text up to the constant factor $(n-1)^{-1}$.
For mixed states, we may use the mutual information
\begin{equation}
  C_{A,B}(\rho_{AB}) = S(\rho_A)+S(\rho_B)-S(\rho_{AB}),
\end{equation}
where $S(\rho) = -\tr[\rho \ln \rho]$.
When these entropies diverge for bosons, the mutual information is defined by the relative entropy of $\rho_{AB}$ with respect to $\rho_A\otimes\rho_B$.
For $N_C(\rho)$, we will prove the following proposition:
\begin{proposition}
  \label{thm:NC-properties}
  $N_C(\rho)$ satisfies the following properties:
  \begin{enumerate}
    \item Faithfulness: $N_C(\rho)\geq0$, with equality holding if and only if $\rho$ is a Gaussian state.
    \item Channel monotonicity: $N_C(\Phi(\rho))\leq N_C(\rho)$ under a Gaussian channel $\Phi$, namely a composition of additions of pure Gaussian ancillas, Gaussian unitaries, and partial traces, where $\Phi(\rho)$ is required to be pure if $C_{A,B}$ is defined only on pure states.
    \item Additivity: $N_C(\rho\otimes \sigma) = N_C(\rho)+N_C(\sigma)$.
  \end{enumerate}
\end{proposition}
Then, Theorem~\ref{thm:MRE-properties} for bosons and fermions follows by choosing the R\'enyi entanglement entropy~\eqref{eq:CAB_Renyi} with $k=2$ as $C_{A,B}$ on pure states.
Again, note that the invariance under a Gaussian unitary follows from channel monotonicity.
We also remark that the additivity of the MRE can be translated into the multiplicativity of the convolution purity
\begin{equation}
  \nu_n(\rho\otimes \sigma) = \nu_n(\rho)\nu_n(\sigma),
  \label{eq:nu_multiplicative}
\end{equation}
which we use in Sec.~\ref{app:monotone}.

\subsubsection*{Vanishing $N_C$ implies a Gaussian state}

Here we prove the nontrivial part of the faithfulness statement, namely that a state satisfying $N_C(\rho)=0$ is Gaussian.
To this end we factor out a beam splitter from the convolution unitary, so that one of its two inputs is a replica in the original state $\rho$, uncorrelated with the other replicas.
We then show that the output of this beam splitter is a product state across replica $1$ and the others only if $\rho$ is Gaussian.
We extend the two-replica beam splitter $U(\zeta)$ of Eq.~\eqref{eq:O_case1} to one acting on replicas $r$ and $s$ among the $n$ replicas, denoted $U_{rs}(\zeta)$, which transforms the operators as
\begin{align}
  U_{rs}^{\dag}(\zeta)
  \begin{pmatrix}
    \bm{a}_{r} \\
    \bm{a}_{s}
  \end{pmatrix}
  U_{rs}(\zeta)
  =
  \begin{pmatrix}
    \cos\zeta & -\sin\zeta \\
    \sin\zeta & \cos\zeta
  \end{pmatrix}
  \begin{pmatrix}
    \bm{a}_{r} \\
    \bm{a}_{s}
  \end{pmatrix},
  \label{eq:BS_def}
\end{align}
and acts as the identity on all other replicas.

Since the first row of $O$ has at least two nonzero entries and the input replicas are interchangeable in the self-convolution, we relabel them so that $O_{11}$ and $O_{1n}$ are nonzero, which ensures $0<|O_{11}|<1$ and $0<|O_{1n}|<1$.

\begin{lemma}\label{lem:factorize}
  Any convolution unitary $U$ defined by Eq.~\eqref{eq:U_BF} can be written as
  \begin{align}
    U = U_{2,\ldots,n}\,U_{1n}(\zeta)\,U_{1,\ldots,n-1},
    \label{eq:factorize}
  \end{align}
  where $U_{2,\ldots,n}$ is a Gaussian unitary acting only on replicas $2,\ldots,n$, $U_{1,\ldots,n-1}$ is a Gaussian unitary acting only on replicas $1,\ldots,n-1$, and $\zeta$ satisfies $-\sin\zeta=O_{1n}$.
\end{lemma}

\begin{proof}
  Let $R_{rs}(\zeta)$ be the $n\times n$ matrix representing the beam splitter~\eqref{eq:BS_def} on the replica index, i.e., the matrix whose rows and columns $r,s$ form the block $\left(\begin{smallmatrix}\cos\zeta & -\sin\zeta\\ \sin\zeta & \cos\zeta\end{smallmatrix}\right)$ while the remaining diagonal entries equal one.
  Choose an $n\times n$ orthogonal matrix $V$ with $V_{11}=1$ (so that it only mixes replicas $2,\ldots,n$) such that the $n$th column of $V O$ equals $(-\sin\zeta,0,\ldots,0,\cos\zeta)^{\rm T}$, i.e., $V$ collects all nonzero entries $O_{rn}$ with $2\le r\le n$ into the single element $(VO)_{nn}=\cos\zeta$.
  Since the $n$th column of $V O$ coincides with that of $R_{1n}(\zeta)$, we have $V O=R_{1n}(\zeta)V'$ with an orthogonal matrix $V'$ whose $n$th column is $(0,\ldots,0,1)^{\rm T}$.
  Since $V'$ is real orthogonal, its $n$th row must also be $(0,\ldots,0,1)$, which shows that the unitary represented by $V'$ acts as the identity on replica $n$.
  Therefore $O=V^{\rm T} R_{1n}(\zeta)V'$, which lifts to the operator factorization~\eqref{eq:factorize} with $U_{2,\ldots,n}$ and $U_{1,\ldots,n-1}$ represented by $V^{\rm T}$ and $V'$, respectively.
\end{proof}

Next we establish two facts about beam splitters acting on a state.
We call a beam splitter $U_{rs}(\zeta)$ \emph{nontrivial} when it cannot be written as a tensor product of unitaries acting separately on the two replicas, up to a swap of the replicas, i.e., when $\cos\zeta\sin\zeta\neq0$.

\begin{lemma}\label{lem:BS}
  For an input $\rho_r\otimes\sigma_s$, the state obtained after applying a nontrivial beam splitter $U_{rs}(\zeta)$ is again a product state only if both $\rho$ and $\sigma$ are Gaussian.
\end{lemma}

\begin{proof}
  For bosons, this is proven in Theorem 7 of Ref.~\cite{cuesta2020stable}.
  For fermions, we show that $\ln\chi_{\rho}$ and $\ln\chi_{\sigma}$ are at most quadratic.
  Let $\bm{v}=\cos\zeta\,\bm{u}_{r}+\sin\zeta\,\bm{u}_{s}$ and $\bm{w}=-\sin\zeta\,\bm{u}_{r}+\cos\zeta\,\bm{u}_{s}$.
  By the transformation~\eqref{eq:conv_char_transform} of the characteristic function, an output product state $\tau_{r}\otimes\tau_{s}$ satisfies $\chi_{\rho}(\bm{v})\,\chi_{\sigma}(\bm{w})=\chi_{\tau_{r}}(\bm{u}_{r})\,\chi_{\tau_{s}}(\bm{u}_{s})$.
  The characteristic functions are even Grassmann polynomials that equal one at $\bm{u}=0$.
  Their logarithms, defined by the Taylor series around one, are therefore polynomials.
  The logarithm of a product is the sum of the logarithms, which gives
  \begin{align}
    \ln\chi_{\rho}(\bm{v})+\ln\chi_{\sigma}(\bm{w})=\ln\chi_{\tau_{r}}(\bm{u}_{r})+\ln\chi_{\tau_{s}}(\bm{u}_{s}).
  \end{align}
  Let $\partial_{l}$ be the derivative with respect to the $l$th component of the argument.
  Differentiating with respect to the $l$th component of $\bm{u}_{r}$ and the $l'$th component of $\bm{u}_{s}$ removes the right-hand side and gives
  \begin{align}
    (\partial_{l'}\partial_{l}\ln\chi_{\rho})(\bm{v})=(\partial_{l'}\partial_{l}\ln\chi_{\sigma})(\bm{w}),
  \end{align}
  since $\cos\zeta\sin\zeta\neq0$.
  Since $\bm{v}$ and $\bm{w}$ are independent variables, both sides are constant.
  Hence $\ln\chi_{\rho}$ and $\ln\chi_{\sigma}$ are at most quadratic, so both states are Gaussian by the form~\eqref{eq:chi_gaussian_fermion_app} of the characteristic function.
\end{proof}

Lemma~\ref{lem:BS} does not apply directly to the factorization~\eqref{eq:factorize}, because $U_{1,\ldots,n-1}$ may correlate replica $1$ with replicas $2,\ldots,n-1$ before the beam splitter acts.
The following lemma extends it to this case.

\begin{lemma}\label{lem:tripartite}
  Consider a tripartite system $ABC$, where $A$ and $C$ are single replicas while $B$ may be empty or contain several replicas.
  Suppose the input is a product state $\rho_{AB}\otimes\rho_{C}$ and a nontrivial beam splitter $U_{AC}(\zeta)$ acts between $A$ and $C$.
  If the output is a product state across $A$ and $BC$, then $\rho_{C}$ is Gaussian.
\end{lemma}

\begin{proof}
  The output state is $\rho_{\rm out}=(U_{AC}(\zeta)\otimes I_{B})(\rho_{AB}\otimes\rho_{C})(U_{AC}^{\dag}(\zeta)\otimes I_{B})$.
  On the one hand, tracing out $B$ yields
  \begin{align}
    \tr_{B}[\rho_{\rm out}] = U_{AC}(\zeta)(\rho_{A}\otimes\rho_{C})U_{AC}^{\dag}(\zeta),
  \end{align}
  where $\rho_{A}=\tr_{B}[\rho_{AB}]$ and we used that the beam splitter acts only on $AC$.
  On the other hand, the assumption $\rho_{\rm out}=\tau_{A}\otimes\tau_{BC}$ shows that $\tr_{B}[\rho_{\rm out}]=\tau_{A}\otimes\tau_{C}$ with $\tau_{C}=\tr_{B}[\tau_{BC}]$.
  Comparing these expressions, we find that applying the nontrivial beam splitter $U_{AC}(\zeta)$ to the product state $\rho_{A}\otimes\rho_{C}$ again yields a product state $\tau_A \otimes \tau_C$.
  Lemma~\ref{lem:BS} then shows that $\rho_{C}$ is a Gaussian state.
\end{proof}

\begin{lemma}\label{thm:faithful}
  $N_C(\rho)=0$ implies that $\rho$ is Gaussian.
\end{lemma}

\begin{proof}
  Consider the factorization of Lemma~\ref{lem:factorize}.
  Since the last unitary $U_{2,\ldots,n}$ generates no correlation between replica $1$ and replicas $2,\ldots,n$, $N_C(\rho)=0$ means that after applying $U_{1,\ldots,n-1}$ and then $U_{1n}(\zeta)$ to $\rho^{\otimes n}$, replica $1$ and replicas $2,\ldots,n$ form a product state.
  Right after $U_{1,\ldots,n-1}$, replica $n$ remains in its original state $\rho$ and uncorrelated with replicas $1,\ldots,n-1$.
  Moreover, the assumption $0<|O_{1n}|<1$ makes the beam splitter $U_{1n}(\zeta)$ nontrivial.
  Applying Lemma~\ref{lem:tripartite} with $A=1$, $B=\{2,\ldots,n-1\}$, and $C=n$ then shows that $\rho$ is Gaussian.
\end{proof}

\subsubsection*{Commutation of the convolution unitary with Gaussian unitaries}

We next show that the convolution unitary commutes with Gaussian unitaries acting on each replica, up to a replica-dependent change of the displacement for bosons.
This will be used below to prove the invariance of $N_C$ under Gaussian unitaries.
To treat bosons and fermions together, we write the action of a Gaussian unitary on the operators $\bm{r}$ as $U_{G}^{\dag}\bm{r}U_{G}=S\bm{r}+\bm{d}$, which is the affine symplectic transformation~\eqref{eq:symplectic} for bosons and the orthogonal transformation~\eqref{eq:FGU_action} with $\bm{d}=\bm{0}$ for fermions.
We denote the vector $\bm{r}$ of replica $s$ by $\bm{r}_{s}$.

\begin{lemma}\label{lem:commute}
  The convolution unitary $U$ and a Gaussian unitary $U_{G}$ satisfy
  \begin{align}
    U\,U_{G}^{\otimes n}\,U^{\dag}\propto\prod_{s=1}^{n}U_{G,s},
    \label{eq:commute}
  \end{align}
  where $U_{G,s}$ is the Gaussian unitary acting only on replica $s$ as $U_{G,s}^{\dag}\bm{r}_{s}U_{G,s}=S\bm{r}_{s}+\sum_{r=1}^{n}O_{sr}\bm{d}$.
\end{lemma}

\begin{proof}
  Since $O$ is real and orthogonal, the action~\eqref{eq:U_BF} of the convolution unitary on the annihilation operators gives
  \begin{align}
    U^{\dag}\bm{r}_{s}U=\sum_{r=1}^{n}O_{sr}\bm{r}_{r},
    \quad
    U\bm{r}_{r}U^{\dag}=\sum_{s'=1}^{n}O_{s'r}\bm{r}_{s'}.
  \end{align}
  Using these relations and the action of $U_{G}$ on each replica, we obtain
  \begin{align}
    U(U_{G}^{\dag})^{\otimes n}U^{\dag}\,\bm{r}_{s}\,U\,U_{G}^{\otimes n}U^{\dag}
     & =U\qty[\sum_{r=1}^{n}O_{sr}\qty(S\bm{r}_{r}+\bm{d})]U^{\dag}\nonumber \\
     & =S\bm{r}_{s}+\sum_{r=1}^{n}O_{sr}\bm{d},
  \end{align}
  where the last equality follows from $\sum_{r}O_{sr}O_{s'r}=\delta_{ss'}$.
  This is the action of $\prod_{s}U_{G,s}$ on $\bm{r}_{s}$, and two unitaries acting identically on all $\bm{r}_{s}$ differ only by a phase.
\end{proof}

\subsubsection*{Proofs of the basic properties}

We now assemble the results above to prove Proposition~\ref{thm:NC-properties}.

\begin{proof}
  \textit{Faithfulness.---}
  We first note that $N_C(\rho)\geq0$ follows from the first property of $C_{A,B}$.
  Next we show that $N_C(\rho_G)=0$ for Gaussian states $\rho_G$.
  Because the bosonic and fermionic Gaussian characteristic functions~\eqref{eq:chi_gaussian_boson} and \eqref{eq:chi_gaussian_fermion_app} are exponentials of polynomials of degree at most two in $\bm{u}$, the orthogonal transformation $\bm{u}_{r}\mapsto\sum_{s}O_{sr}\bm{u}_{s}$ induced by the convolution unitary in Eq.~\eqref{eq:conv_char_transform} keeps the exponent diagonal in the replica index.
  The output characteristic function therefore factorizes over replicas, so the convolved state is a tensor product with no inter-replica correlation, which means $N_C(\rho_G)=0$.
  Conversely, $N_C(\rho)=0$ implies that $\rho$ is Gaussian by Lemma~\ref{thm:faithful}.
  This establishes faithfulness.

  \textit{Additivity.---}
  Let $\rho_{A}$ and $\rho_{B}$ be states on systems $A$ and $B$.
  Since the convolution unitary acts uniformly on all sites with the same replica rotation $O$, it factorizes across the two systems, and the displacement label of the joint system splits as $\vec{\bm{u}}=(\vec{\bm{u}}_{A},\vec{\bm{u}}_{B})$, with $\vec{\bm{u}}'=(\vec{\bm{u}}_{A}',\vec{\bm{u}}_{B}')$ under the replica rotation $\vec{\bm{u}}'=(O^{\rm T}\otimes I_{2L})\vec{\bm{u}}$.
  By the transformation $\chi_{U\rho_{\rm tot}U^{\dag}}(\vec{\bm{u}})=\chi_{\rho_{\rm tot}}(\vec{\bm{u}}')$, the characteristic function of the convolved product state factorizes,
  \begin{align}
     & \chi_{U(\rho_{A}\otimes\rho_{B})^{\otimes n}U^{\dag}}(\vec{\bm{u}}) \nonumber          \\
     & \quad= \chi_{(\rho_{A}\otimes\rho_{B})^{\otimes n}}(\vec{\bm{u}}')
    = \prod_{r=1}^{n}\chi_{\rho_{A}}(\bm{u}_{A,r}')\,\chi_{\rho_{B}}(\bm{u}_{B,r}') \nonumber \\
     & \quad= \chi_{U\rho_{A}^{\otimes n}U^{\dag}}(\vec{\bm{u}}_{A})\,
    \chi_{U\rho_{B}^{\otimes n}U^{\dag}}(\vec{\bm{u}}_{B}),
  \end{align}
  where $\bm{u}_{A,r}$ ($\bm{u}_{B,r}$) denotes the label of system $A$ ($B$) in replica $r$.
  The convolved product state thus factorizes as
  \begin{align}
     & U(\rho_{A}\otimes\rho_{B})^{\otimes n}U^{\dag}= \qty(U\rho_{A}^{\otimes n}U^{\dag})
    \otimes\qty(U\rho_{B}^{\otimes n}U^{\dag}).
  \end{align}
  Then, using the additivity of the measure of correlation $C$, we have
  \begin{align}
    N_C(\rho_{A}\otimes\rho_{B}) & = C_{1, 1^{\rm c}}(U(\rho_{A}\otimes\rho_{B})^{\otimes n}U^{\dag}) \nonumber                                   \\
                                 & = C_{1, 1^{\rm c}}\qty(U\rho_{A}^{\otimes n}U^{\dag}) + C_{1, 1^{\rm c}}\qty(U\rho_{B}^{\otimes n}U^{\dag})\nonumber \\
                                 & = N_C(\rho_{A}) + N_C(\rho_{B}).
  \end{align}

  \textit{Monotonicity under a Gaussian channel.---}
  We prove monotonicity of the correlation-based measure $N_C$ under Gaussian channels.
  By moving the additions of ancillas to the beginning and the partial traces to the end, we write a Gaussian channel from an $L$-mode system to an $M$-mode output system $S$ as
  \begin{equation}
    \Phi(\rho) = \tr_E[U_G(\rho\otimes\ket{\psi_G}\bra{\psi_G}^{\otimes N}) U_G^\dagger],
  \end{equation}
  where $\ket{\psi_G}\bra{\psi_G}^{\otimes N}$ is an $N$-mode Gaussian ancilla, $U_{G}\in\mathcal{G}_{L+N}$ is a Gaussian unitary, and $\tr_E$ traces out the environment $E$ consisting of the $L+N-M$ discarded modes, keeping the output system $S$.
  It therefore suffices to prove monotonicity under (1) adding the ancilla $\ket{\psi_G}\bra{\psi_G}^{\otimes N}$, (2) applying a Gaussian unitary $U_{G}$, and (3) tracing out the subsystem.
  (1) Invariance under adding the ancilla $\ket{\psi_G}\bra{\psi_G}^{\otimes N}$ follows from additivity and faithfulness:
  \begin{align}
    N_C(\rho \otimes \ket{\psi_G}\bra{\psi_G}^{\otimes N}) & = N_C(\rho) + N_C(\ket{\psi_G}\bra{\psi_G}^{\otimes N})\nonumber \\
                                                           & = N_C(\rho).
  \end{align}
  (2) Invariance under a Gaussian unitary follows because, by Lemma~\ref{lem:commute}, $U(U_{G}\rho U_{G}^{\dag})^{\otimes n}U^{\dag}$ is obtained from $U\rho^{\otimes n}U^{\dag}$ by the unitaries $U_{G,s}$ acting separately on the replicas, which leave $C_{1, 1^{\rm c}}$ invariant.
  (3) To prove monotonicity under tracing out a subsystem, we consider a system composed of a retained part $S$ and a discarded part $E$.
  We accordingly write $\vec{\bm{u}}=(\vec{\bm{u}}_{S},\vec{\bm{u}}_{E})$.
  Applying the convolution unitary to an input state $\sigma = \tr_E \rho$ yields
  \begin{align}
    \chi_{U \sigma^{\otimes n}U^\dag} (\vec{\bm{u}}_S) & = \chi_{\sigma^{\otimes n}}(\vec{\bm{u}}_S')\nonumber           \\
                                                       & = \chi_{\rho^{\otimes n}} (\vec{\bm{u}}_S', \vec{\bm{u}}_E'=0).
  \end{align}
  On the other hand, tracing out $E$ in each replica after the convolution unitary yields
  \begin{align}
    \chi_{\tr_E[U \rho^{\otimes n} U^\dag]} (\vec{\bm{u}}_S) = \chi_{\rho^{\otimes n}} (\vec{\bm{u}}_S', \vec{\bm{u}}_E'=0).
  \end{align}
  Therefore, tracing out a subsystem commutes with the convolution unitary.
  Since tracing out a subsystem in each replica does not generate correlation, $N_C$ is monotone under partial trace.
  Combining (1), (2), and (3), we obtain $N_C(\Phi(\rho))\leq N_C(\rho)$.
\end{proof}

We remark that, because the R\'enyi entanglement entropy is a genuine correlation measure for pure states, the MRE \(M_n\) discussed in the main text does not increase under a pure-to-pure Gaussian channel, but it can increase under a Gaussian channel if mixed outputs are allowed.

\section{Monotonicity theorems}\label{app:monotonicity_theorems}

In this appendix, we prove the monotonicity theorems presented in Sec.~\ref{sec:mono}.
In Sec.~\ref{app:monotone}, we prove Theorems~\ref{thm:qudit-protocol-monotonicity}, \ref{thm:boson-protocol-monotonicity}, and \ref{thm:fermion-protocol-monotonicity} on monotonicity under stabilizer and Gaussian protocols involving measurements and feedforward.
In Sec.~\ref{app:convolution_uniqueness}, we consider general convolution unitaries for qudits and bosons and prove Theorems~\ref{thm:qudit-convolution-uniqueness} and \ref{thm:boson-convolution-uniqueness}.

\subsection{Monotonicity under stabilizer and Gaussian protocols}\label{app:monotone}

Here we prove Theorems~\ref{thm:qudit-protocol-monotonicity}, \ref{thm:boson-protocol-monotonicity}, and \ref{thm:fermion-protocol-monotonicity}, namely monotonicity of the MRE and strong monotonicity of the linear MRE under stabilizer and Gaussian protocols.
These properties follow once we show that the convolution purity does not decrease on average under the free protocols defined in Sec.~\ref{sec:mono}, for qudits with every integer $n\geq2$ coprime to $d$, for bosons with every integer $n\geq2$, and for fermions with $n=2$.
Throughout this subsection, \(\nu_n\) denotes the convolution purity constructed from the convolution unitary~\eqref{eq:qudit_nd_convolution_unitary} for qudits and from the Helmert matrix~\eqref{eq:O_case2} for bosons and fermions.
Since the MRE does not depend on which convolution unitary satisfying the duality is chosen, this choice involves no loss of generality.

As explained in Sec.~\ref{sec:mono}, the key step is to show that a single measurement does not decrease $\nu_n$ on average, which we establish as Lemma~\ref{lem:qudit_measurement} for qudits, Lemma~\ref{lem:homodyne} for bosons, and Lemma~\ref{lem:measurement} for fermions below.
We first prove the purity inequalities~\eqref{eq:qudit_purity_mono}, \eqref{eq:boson_purity_mono}, and \eqref{eq:fermion_purity_mono} assuming these lemmas, then derive Theorems~\ref{thm:qudit-protocol-monotonicity}, \ref{thm:boson-protocol-monotonicity}, and \ref{thm:fermion-protocol-monotonicity}, and finally prove the lemmas.
The proofs of Lemmas~\ref{lem:qudit_measurement} and \ref{lem:homodyne} follow the argument outlined in Sec.~\ref{sec:mono}, while the proof of Lemma~\ref{lem:measurement} is a direct calculation for $n=2$.

\begin{lemma}
  \label{thm:qudit_adaptive}
  Let $n\geq2$ be an integer coprime to $d$, and let \(\mathcal E\) be a stabilizer protocol that maps a pure input $|\psi\rangle$ to pure outputs $|\varphi_\lambda\rangle$ with probabilities $p_\lambda$.
  Then the convolution purity satisfies
  \begin{equation}
    \nu_n(\psi)\leq\sum_\lambda p_\lambda\nu_n(\varphi_\lambda).
    \label{eq:qudit_purity_mono}
  \end{equation}
\end{lemma}

\begin{lemma}
  \label{thm:boson_adaptive}
  Let $n\geq 2$ be an integer, and let \(\mathcal E\) be a bosonic Gaussian protocol that maps a pure input $|\psi\rangle$ to pure outputs $|\varphi_\lambda\rangle$ with probabilities $p_\lambda$.
  Then the convolution purity satisfies
  \begin{equation}
    \nu_n(\psi)
    \leq \int \; d\lambda \,p_\lambda\nu_n(\varphi_\lambda).
    \label{eq:boson_purity_mono}
  \end{equation}
\end{lemma}

\begin{lemma}
  \label{thm:fermion_adaptive}
  Let $\ket{\psi}$ be a normalized pure state of definite fermion parity, and let $\mathcal{E}$ be a fermionic Gaussian protocol that maps $|\psi\rangle$ to pure outputs $|\varphi_\lambda\rangle$ with probabilities $p_\lambda$.
  Then the convolution purity satisfies
  \begin{equation}
    \nu_{2}(\psi)\leq\sum_\lambda p_\lambda\nu_{2}(\varphi_\lambda).
    \label{eq:fermion_purity_mono}
  \end{equation}
\end{lemma}

\begin{proof}[Proof of Lemmas~\ref{thm:qudit_adaptive}, \ref{thm:boson_adaptive}, and \ref{thm:fermion_adaptive}]
  We resolve all measurement outcomes and classical random variables of the protocol into branches labeled by $\lambda$, replacing every partial trace by retaining the traced subsystem as an environment register.
  Along each branch, the total state then remains pure.
  For fermions, the total state of each branch has definite fermion parity at every stage.

  With this dilation, the only nontrivial element of the protocol is a measurement.
  Free unitaries and additions of free ancillas preserve $\nu_n$ by the invariance, additivity, and faithfulness in Theorem~\ref{thm:MRE-properties}.
  A partial trace does nothing to the dilated global state, because the discarded subsystem is merely relabeled as part of the environment.
  Along each branch, an operation conditioned on a measurement outcome or a classical random variable is one of the elements (i)--(iv) in Sec.~\ref{sec:mono}.

  Lemmas~\ref{lem:qudit_measurement}, \ref{lem:homodyne}, and \ref{lem:measurement} apply to the computational-basis, homodyne, and one-mode occupation measurements, respectively, with the unmeasured register including every other physical and environment degree of freedom.
  We note that a simultaneous measurement of commuting occupation numbers can be treated as a sequence of one-mode measurements.
  For qudits and fermions, a measured register kept after the measurement is in a computational-basis or occupation state, which is a free pure state and does not change $\nu_n$.
  For bosons, we treat a homodyne measurement as removing the measured modes from the output, because a position eigenstate is not normalizable.
  Applying these lemmas at every measurement and averaging over all outcomes, we obtain
  \begin{equation}
    \nu_n(\psi)
    \leq
    \sum_\lambda p_\lambda\nu_n(\Omega_\lambda),
    \label{eq:branch_average}
  \end{equation}
  where $|\Omega_\lambda\rangle$ is the final pure state of branch $\lambda$ on the enlarged Hilbert space, and the sum over $\lambda$ is understood as an integral over continuous outcomes for bosons.

  Since the reduced state of $|\Omega_\lambda\rangle$ on the physical Hilbert space is the pure output $|\varphi_\lambda\rangle$, there exists a normalized pure register state $|\xi_\lambda\rangle$ such that
  \begin{equation}
    |\Omega_\lambda\rangle
    =
    |\varphi_\lambda\rangle\otimes|\xi_\lambda\rangle,
  \end{equation}
  where the tensor product is understood as the graded one for fermions.
  In the fermionic case, since $|\Omega_{\lambda}\rangle$ has definite parity, both factors also have definite parities.
  The multiplicativity of $\nu_n$ and the purity bound
  $\nu_n(\xi_\lambda)\leq1$ then imply
  \begin{equation}
    \nu_n(\Omega_\lambda)
    =
    \nu_n(\varphi_\lambda)\nu_n(\xi_\lambda)
    \leq
    \nu_n(\varphi_\lambda),
  \end{equation}
  which proves Eqs.~\eqref{eq:qudit_purity_mono}, \eqref{eq:boson_purity_mono}, and \eqref{eq:fermion_purity_mono}.
\end{proof}

Theorems~\ref{thm:qudit-protocol-monotonicity}, \ref{thm:boson-protocol-monotonicity}, and \ref{thm:fermion-protocol-monotonicity} in the main text follow from Lemmas~\ref{thm:qudit_adaptive}, \ref{thm:boson_adaptive}, and \ref{thm:fermion_adaptive}.
First, strong monotonicity of the linear MRE is a restatement of the purity inequalities.
Since $M_n^{\rm lin}=1-\nu_n$ and the probabilities $p_\lambda$ sum to unity, Eqs.~\eqref{eq:qudit_purity_mono}, \eqref{eq:boson_purity_mono}, and \eqref{eq:fermion_purity_mono} are equivalent to
\begin{equation}
  \sum_\lambda p_\lambda M_n^{\rm lin}(\varphi_\lambda)\leq M_n^{\rm lin}(\psi).
\end{equation}
Second, to show monotonicity of the MRE, let ${\mathcal E}$ be a free protocol that maps a pure input $|\psi\rangle$ deterministically to a pure output $|\varphi\rangle$, i.e., ${\mathcal E}(\dyad{\psi})=\dyad{\varphi}$.
Since a pure state is not a nontrivial mixture, every branch $\lambda$ results in the same output $|\varphi\rangle$, and the purity inequalities reduce to
\begin{equation}
  \nu_n(\psi)\leq\nu_n(\varphi).
\end{equation}
Since $M_n=\ln(\nu_n)/(1-n)$ is a monotonically decreasing function of $\nu_n$ for $n\geq2$, this is equivalent to $M_n(\varphi)\leq M_n(\psi)$, which is the claimed monotonicity of the MRE.
We remark that strong monotonicity of the MRE does not follow from this argument, since the logarithm is concave and Jensen's inequality goes in the opposite direction.

\subsubsection*{Qudits}
To complete the proof for qudits, we now show Lemma~\ref{lem:qudit_measurement}.

\begin{lemma}
  \label{lem:qudit_measurement}
  Consider a normalized pure state with the decomposition
  \begin{equation}
    \ket{\Psi}_{RQ}=\sum_{x\in\mathbb{Z}_d}\sqrt{p(x)}\,\ket{\eta_{x}}_R\otimes\ket{x}_Q,
    \label{eq:qudit_measurement_decomp}
  \end{equation}
  where $\ket{x}_Q$ is the computational basis of a qudit $Q$, $\ket{\eta_{x}}_R$ is a normalized state of the other qudits $R$, and $p(x)$ is a probability distribution.
  Then, for every integer $n\geq2$ coprime to $d$, the computational-basis measurement of $Q$ does not decrease the convolution purity on average:
  \begin{equation}
    \nu_n(\Psi)\leq\sum_{x\in\mathbb{Z}_d}p(x)\,\nu_n(\eta_{x}).
    \label{eq:qudit_measurement_average}
  \end{equation}
\end{lemma}

\begin{proof}
  Apply the convolution unitary to the $n$ replicas.
  Since the convolution unitary~\eqref{eq:qudit_nd_convolution_unitary} acts on each qudit separately, the convolved state is
  \begin{equation}
    |\vec{\Psi}\,'\rangle
    =(U_R\otimes U_Q)|\vec{\Psi}\rangle,
    \qquad
    |\vec{\Psi}\rangle=|\Psi\rangle^{\otimes n},
  \end{equation}
  where $U_R$ and $U_Q$ act on the $R$ and $Q$ registers of the $n$ replicas.
  We label the output replicas of the registers by $R_r$ and $Q_r$ with $r=1,2,\ldots,n$, and denote the collections of the output replicas $2,\ldots,n$ by $R_{1^{\rm c}}$ and $Q_{1^{\rm c}}$.
  By definition, $\nu_n(\Psi)$ is the purity of the reduced state of $|\vec{\Psi}\,'\rangle$ on the first output replica $R_1Q_1$.
  We bound this purity in two conditioning steps by the average of the convolution purity $\nu_n(\eta_{x_1},\ldots,\eta_{x_n})$ of the conditional states, and then reduce the average to the right-hand side of Eq.~\eqref{eq:qudit_measurement_average} by H\"older's and Jensen's inequalities.

  We first write $|\vec{\Psi}\,'\rangle$ in the computational basis of the output $Q$ registers.
  On the $n$ replicas of $Q$, the convolution unitary permutes the computational basis as
  \begin{align}
    U_Q\ket{\vec{x}}_Q=\ket{A\vec{x}}_Q,
  \end{align}
  where $\vec{x}=(x_1,\ldots,x_n)$ and $A$ is the invertible linear map on $\mathbb{Z}_d^{n}$ given by
  \begin{align}
    A\vec{x}=\qty(n^{-1}\sum_{s=1}^{n}x_s,\,n^{-1}(x_1-x_2),\ldots,n^{-1}(x_1-x_n)).
  \end{align}
  Applying the decomposition~\eqref{eq:qudit_measurement_decomp} to each replica, we obtain
  \begin{equation}
    |\vec{\Psi}\,'\rangle
    =\sum_{\vec{x}'}\sqrt{w(\vec{x}')}\,
    |\vec{\eta}\,'_{\vec{x}'}\rangle_{R}\otimes|\vec{x}'\rangle_{Q},
    \label{eq:Psi_rotated}
  \end{equation}
  where
  \begin{equation}
    w(\vec{x}')=\prod_{r=1}^{n}p(x_{r}),
    \qquad
    |\vec{\eta}\,'_{\vec{x}'}\rangle=U_{R}\bigotimes_{r=1}^{n}|\eta_{x_{r}}\rangle,
    \label{eq:w_eta_def}
  \end{equation}
  with $\vec{x}=A^{-1}\vec{x}'$ understood on the right-hand sides.

  Before deriving the bound on $\nu_n(\Psi)$, we recall the convexity of the purity.
  For density operators $\rho_{\lambda}$ mixed with probabilities $p_{\lambda}$, the convexity reads
  \begin{equation}
    \tr\!\left[\left(\sum_{\lambda}p_{\lambda}\rho_{\lambda}\right)^{2}\right]
    \leq
    \sum_{\lambda}p_{\lambda}\tr(\rho_{\lambda}^{2}).
    \label{eq:purity_convexity}
  \end{equation}
  This inequality follows from
  \begin{align}
     & \sum_{\lambda}p_{\lambda}\tr(\rho_{\lambda}^{2})
    -\tr\!\left[\left(\sum_{\lambda}p_{\lambda}\rho_{\lambda}\right)^{2}\right]\nonumber \\
     & \quad=\frac{1}{2}\sum_{\lambda,\lambda'}p_{\lambda}\,p_{\lambda'}\,
    \tr\!\left[\left(\rho_{\lambda}-\rho_{\lambda'}\right)^{2}\right]
    \geq0.
  \end{align}

  We now bound the purity of the reduced state of $|\vec{\Psi}\,'\rangle$ on $R_1Q_1$ in two conditioning steps.
  The first conditioning step is on the variables $\vec{x}'_{1^{\rm c}}\equiv(x'_{2},\ldots,x'_{n})$ of the registers $Q_{1^{\rm c}}$.
  Tracing out $Q_{1^{\rm c}}$ eliminates the components of Eq.~\eqref{eq:Psi_rotated} that are off-diagonal in these variables, so the reduced state on $R_1Q_1$ is the mixture
  \begin{align}
    \tr_{R_{1^{\rm c}}Q_{1^{\rm c}}}\dyad{\vec{\Psi}\,'}
     & =\sum_{\vec{x}'_{1^{\rm c}}}
    w_{1^{\rm c}}(\vec{x}'_{1^{\rm c}})\nonumber \\
     & \qquad\times
    \tr_{R_{1^{\rm c}}}\dyad{\vec{\Psi}\,'_{\vec{x}'_{1^{\rm c}}}},
    \label{eq:step1_mixture}
  \end{align}
  where $w_{1^{\rm c}}(\vec{x}'_{1^{\rm c}})=\sum_{x'_{1}}w(\vec{x}')$ is the marginal probability and
  \begin{align}
    |\vec{\Psi}\,'_{\vec{x}'_{1^{\rm c}}}\rangle
    =\frac{1}{\sqrt{w_{1^{\rm c}}(\vec{x}'_{1^{\rm c}})}}
    \sum_{x'_{1}}\sqrt{w(\vec{x}')}
    |\vec{\eta}\,'_{\vec{x}'}\rangle\otimes|x'_{1}\rangle_{Q_{1}}
    \label{eq:step1_conditional}
  \end{align}
  is the normalized pure state on $R_{1}\cdots R_{n}Q_{1}$ conditioned on $\vec{x}'_{1^{\rm c}}$.
  Equation~\eqref{eq:purity_convexity} then gives
  \begin{align}
    \nu_{n}(\Psi)
    \leq
    \sum_{\vec{x}'_{1^{\rm c}}}
    w_{1^{\rm c}}(\vec{x}'_{1^{\rm c}})
    \tr\!\left[\left(\tr_{R_{1^{\rm c}}}\dyad{\vec{\Psi}\,'_{\vec{x}'_{1^{\rm c}}}}\right)^{2}\right].
    \label{eq:step1_bound}
  \end{align}

  The second conditioning step is on the remaining variable $x'_{1}$.
  Since the conditional state~\eqref{eq:step1_conditional} is pure, its reductions on $R_1Q_1$ and $R_{1^{\rm c}}$ have the same purity, and we bound the latter.
  Tracing out $R_{1}Q_{1}$ in Eq.~\eqref{eq:step1_conditional} leaves the mixture
  \begin{align}
    \tr_{R_{1}Q_{1}}\dyad{\vec{\Psi}\,'_{\vec{x}'_{1^{\rm c}}}}
    =\sum_{x'_{1}}
    \frac{w(\vec{x}')}{w_{1^{\rm c}}(\vec{x}'_{1^{\rm c}})}
    \tr_{R_{1}}\dyad{\vec{\eta}\,'_{\vec{x}'}},
    \label{eq:step2_mixture}
  \end{align}
  so that Eq.~\eqref{eq:purity_convexity} gives
  \begin{align}
     & \tr\!\left[\left(\tr_{R_{1^{\rm c}}}\dyad{\vec{\Psi}\,'_{\vec{x}'_{1^{\rm c}}}}\right)^{2}\right]\nonumber    \\
     & \quad=\tr\!\left[\left(\tr_{R_{1}Q_{1}}\dyad{\vec{\Psi}\,'_{\vec{x}'_{1^{\rm c}}}}\right)^{2}\right]\nonumber \\
     & \quad\leq
    \sum_{x'_{1}}
    \frac{w(\vec{x}')}{w_{1^{\rm c}}(\vec{x}'_{1^{\rm c}})}\,
    \tr\!\left[\left(\tr_{R_{1}}\dyad{\vec{\eta}\,'_{\vec{x}'}}\right)^{2}\right]\nonumber                           \\
     & \quad=\sum_{x'_{1}}
    \frac{w(\vec{x}')}{w_{1^{\rm c}}(\vec{x}'_{1^{\rm c}})}\,\nu_{n}(\eta_{x_{1}},\ldots,\eta_{x_{n}}).
    \label{eq:step2_bound}
  \end{align}

  Finally, we assemble these steps to prove Lemma~\ref{lem:qudit_measurement}.
  Using Eqs.~\eqref{eq:step1_bound} and \eqref{eq:step2_bound} and changing the summation variables from $\vec{x}'$ to $\vec{x}=A^{-1}\vec{x}'$, we obtain
  \begin{equation}
    \nu_n(\Psi)
    \leq
    \sum_{x_1,\ldots,x_n}\prod_{r=1}^{n}p(x_r)\,
    \nu_n(\eta_{x_1},\ldots,\eta_{x_n}).
    \label{eq:cross_average}
  \end{equation}
  On the other hand, since Eq.~\eqref{eq:duality_qudit_general} evaluates the characteristic functions of all inputs at the same point $\bm{u}'$, H\"older's inequality applied to $\nu_n(\eta_{x_1},\ldots,\eta_{x_n})=d^{-L}\sum_{\bm{u}}\prod_{r}|\chi_{\eta_{x_r}}(\bm{u}')|^{2}$ gives
  \begin{equation}
    \nu_n(\eta_{x_1},\ldots,\eta_{x_n})
    \leq
    \prod_{r=1}^{n}\nu_n(\eta_{x_r})^{1/n}.
    \label{eq:cross_holder}
  \end{equation}
  Combining Eqs.~\eqref{eq:cross_average} and \eqref{eq:cross_holder}, we find that the sum factorizes over the replicas, and Jensen's inequality gives
  \begin{align}
    \nu_n(\Psi)
     & \leq
    \left[
      \sum_{x}p(x)\,
      \nu_n(\eta_{x})^{1/n}
      \right]^n
    \nonumber \\
     & \leq
    \sum_{x}p(x)\,\nu_n(\eta_{x}),
  \end{align}
  which proves Eq.~\eqref{eq:qudit_measurement_average}.
\end{proof}

A measurement of any other Heisenberg-Weyl operator, including one with degenerate eigenvalues for a composite $d$, is realized by a $\ket{0}$ ancilla, Clifford unitaries, and a computational-basis measurement of the ancilla, and is therefore covered by Lemma~\ref{lem:qudit_measurement}.

\subsubsection*{Bosons}
To complete the proof for bosons, we now show Lemma~\ref{lem:homodyne}, which is the counterpart of Lemma~\ref{lem:qudit_measurement} for a homodyne measurement.

\begin{lemma}
  \label{lem:homodyne}
  Consider a normalized pure state with the decomposition
  \begin{equation}
    |\Psi\rangle_{RQ}
    =\int d\bm{x}\,\sqrt{p(\bm{x})}\,
    |\eta_{\bm{x}}\rangle_R\otimes|\bm{x}\rangle_Q,
    \label{eq:homodyne_decomp}
  \end{equation}
  where \(|\eta_{\bm{x}}\rangle_R\) is a normalized conditional pure state of the unmeasured register \(R\), \(p(\bm{x})\) is a probability density, and \(|\bm{x}\rangle_Q\) is the homodyne eigenbasis of the \(m\)-mode measured register \(Q\).
  Then, for every integer \(n\geq2\), the homodyne measurement of $Q$ does not decrease the convolution purity on average:
  \begin{equation}
    \nu_n(\Psi)
    \leq
    \int d\bm{x}\,p(\bm{x})\,\nu_n(\eta_{\bm{x}}).
    \label{eq:homodyne_average}
  \end{equation}
\end{lemma}

\begin{proof}
  The proof follows that of Lemma~\ref{lem:qudit_measurement}, with the sums over $\mathbb{Z}_d$ replaced by integrals over the homodyne variables.
  The convolution unitary factorizes as $U_R\otimes U_Q$, where $U_R$ and $U_Q$ act on the $R$ and $Q$ registers of the $n$ replicas and are the Gaussian unitaries given by the Helmert matrix.
  Let $\vec{\bm{x}}=(\bm{x}_{1},\ldots,\bm{x}_{n})$ collect the homodyne variables of the $n$ input $Q$ registers, and let $\vec{\bm{x}}'=(O_{H}\otimes I_{m})\vec{\bm{x}}$ be their replica rotation, where $O_{H}$ is the Helmert matrix~\eqref{eq:O_case2} and $I_{m}$ acts on the mode index within each register.
  Since Eq.~\eqref{eq:U_BF} rotates the annihilation operators, and hence the quadratures of the $Q$ registers, by $O_{H}$, the homodyne eigenstates transform as $U_{Q}|\vec{\bm{x}}\rangle_{Q}=|\vec{\bm{x}}'\rangle_{Q}$, up to a phase that we absorb into $|\vec{\eta}\,'_{\vec{\bm{x}}'}\rangle$.
  This rotation replaces the permutation $A$ in Eq.~\eqref{eq:Psi_rotated}, and since its Jacobian is unity, the probability density remains $w(\vec{\bm{x}}')=\prod_{r}p(\bm{x}_{r})$.
  With the convexity~\eqref{eq:purity_convexity} for a probability density, the two conditioning steps then give
  \begin{equation}
    \nu_n(\Psi)
    \leq
    \int\prod_{r=1}^n
    \left[d\bm{x}_r\,p(\bm{x}_r)\right]\,
    \nu_n(\eta_{\bm{x}_1},\ldots,\eta_{\bm{x}_n}).
  \end{equation}
  Applying H\"older's inequality to Eq.~\eqref{eq:purity_boson_chi} gives Eq.~\eqref{eq:cross_holder} with $x_r$ replaced by $\bm{x}_r$, and Jensen's inequality then gives Eq.~\eqref{eq:homodyne_average} as for qudits.
\end{proof}

\subsubsection*{Fermions}

We next prove Lemma~\ref{lem:measurement}, which is essential for proving the fermionic monotonicity theorem.
Since the argument for qudits and bosons does not apply to fermions, as explained in Sec.~\ref{sec:mono}, we fix the conditional states, write $\nu_2$ of the state before the measurement as a quartic polynomial in the outcome probability, and show that this polynomial is convex.
To this end, we first introduce the notation used in this part.
As in Appendix~\ref{app:properties_BF}, we use the rescaled characteristic polynomial $\chi_X(\bm{u})=\tr[X\exp(\bm{r}^{\rm T}\bm{u})]$, with Grassmann variables ordered to the right of an operator $X$. 
We use the Grassmann monomials $\bm{u}_{A}$, the Majorana monomials $\Gamma_{A}$, and the coefficients $\kappa_{X}(A)$ of Eq.~\eqref{eq:chi_fermion_expansion}, and introduce the rescaled polynomial
\begin{align}
  \tilde{\chi}_{X}(\bm{u}) & \equiv\chi_{X}(\bm{u}/\sqrt{2})\nonumber    \\
                           & =\sum_{A}2^{-|A|/2}\kappa_{X}(A)\bm{u}_{A}.
  \label{eq:chi_rescaled}
\end{align}
The coefficient norm introduced in Sec.~\ref{subsec:MRE_BF} comes from the inner product $\braket{f}{g}=\sum_{A}f_{A}^{*}g_{A}$ of Grassmann polynomials $f=\sum_{A}f_{A}\bm{u}_{A}$ and $g=\sum_{A}g_{A}\bm{u}_{A}$.
The orthogonality of the Majorana monomials converts it into a Hilbert-Schmidt one,
\begin{equation}
  \braket{\chi_{X}}{\chi_{Y}}=2^{L}\tr\!\left[X^{\dagger}Y\right].
  \label{eq:HS_pairing}
\end{equation}
Finally, we write $\Pi_{s}$ for the fermion parity of replica $s$, and $p(\cdot)\in\{0,1\}$ for the parity of a definite-parity state or operator.
Fermionic modes carried by different replicas are combined by the graded tensor product
\begin{equation}
  (X\gtensor Y)(\ket{c}\gtensor\ket{d})=(-1)^{p(Y)p(c)}X\ket{c}\gtensor Y\ket{d},
  \label{eq:graded_tensor}
\end{equation}
to which all mode orderings, swaps, and partial traces refer.
Setting $X=\ket{a}\!\bra{c}$ and $Y=\ket{b}\!\bra{d}$ in Eq.~\eqref{eq:graded_tensor}, we obtain the sign rule
\begin{equation}
  \ket{a,b}\!\bra{c,d}=(-1)^{p(c)\qty[p(b)+p(d)]}\qty(\ket{a}\!\bra{c})\gtensor\qty(\ket{b}\!\bra{d})
  \label{eq:sign_rule}
\end{equation}
for any states $a,b,c,d$ of definite parity, with $\ket{a,b}=\ket{a}\gtensor\ket{b}$.

\begin{lemma} \label{lem:chi_product} For operators $X,Y$
  with definite parities,
  \begin{equation}
    \tilde{\chi}_{X}\tilde{\chi}_{Y}
    =\chi_{\tr_{2}\left[U(X\gtensor Y)U^{\dagger}\right]}.
    \label{eq:chi_product}
  \end{equation}
\end{lemma}
\begin{proof}
  Expanding $\tilde{\chi}_{X}$ and $\tilde{\chi}_{Y}$ as in Eq.~\eqref{eq:chi_rescaled}, with the index sets denoted by $P$ and $Q$, we write the left-hand side of Eq.~\eqref{eq:chi_product} as
  \begin{align}
    \tilde{\chi}_{X}\tilde{\chi}_{Y}
     & =\sum_{A}2^{-|A|/2}\sum_{P\sqcup Q=A}\nonumber                   \\
     & \quad\times\varepsilon(P,Q)\kappa_{X}(P)\kappa_{Y}(Q)\bm{u}_{A},
  \end{align}
  where $\varepsilon(P,Q)$ is the sign factor defined as $\bm{u}_{P}\bm{u}_{Q}=\varepsilon(P,Q)\bm{u}_{A}$ and $\sqcup$ is the disjoint union.
  On the other hand, the coefficient of $\bm{u}_{A}$ in the right-hand side of Eq.~\eqref{eq:chi_product} is the Majorana coefficient
  \begin{align}
     & \tr\!\left[(\Gamma_{A}^{\dagger}\gtensor I)U(X\gtensor Y)U^{\dagger}\right]\nonumber                                           \\
     & \quad=\tr\!\left[U^{\dagger}(\Gamma_{A}^{\dagger}\gtensor I)U(X\gtensor Y)\right]\nonumber                                     \\
     & \quad=\tr\!\left[\left(\prod_{k=1}^{|A|}\frac{r_{1,l_{k}}+r_{2,l_{k}}}{\sqrt{2}}\right)^{\dagger}(X\gtensor Y)\right]\nonumber \\
     & \quad=\sum_{P\sqcup Q=A}\varepsilon(P,Q)2^{-|A|/2}\kappa_{X}(P)\kappa_{Y}(Q),
  \end{align}
  where $A=\{l_{1}<\cdots<l_{|A|}\}$ and $r_{s,j}$ is the $j$th Majorana component of replica $s$.
  In the last equality, the sign from factorizing the trace over the two replicas cancels because $\kappa_{X}(P)$ vanishes unless $|P|\equiv p(X)\pmod{2}$.
  This coincides with the coefficient obtained above.
\end{proof}

We next derive identities that determine the quartic polynomial and express the second differences of its coefficients as squared norms.
To this end, let $\ket{\psi_{+}}$ be an $L$-mode normalized even-parity state and $\ket{\psi_{-}}$ an $L$-mode normalized odd-parity state.
We abbreviate the Grassmann polynomials
\begin{equation}
  \tilde{\chi}_{st}=\tilde{\chi}_{\ket{\psi_{s}}\!\bra{\psi_{t}}},\qquad s,t\in\{+,-\},
  \label{eq:chi_st_def}
\end{equation}
and also define the nonnegative numbers
\begin{equation}
  n_{st,s't'}\equiv\norm{\tilde{\chi}_{st}\tilde{\chi}_{s't'}}^{2},\quad s,t,s',t'\in\{+,-\}.
  \label{eq:n_st_def}
\end{equation}
We prove two technical lemmas.

\begin{lemma} \label{lem:chi_pairing} The following identities hold:
  \begin{align}
    n_{+-,-+}                                                          & =n_{++,--}, \quad \braket{\tilde{\chi}_{++}\tilde{\chi}_{--}}{\tilde{\chi}_{+-}\tilde{\chi}_{-+}} =0,\label{eq:chi_pairing1} \\
    \braket{\tilde{\chi}_{++}^{2}}{\tilde{\chi}_{++}\tilde{\chi}_{--}} & =\braket{\tilde{\chi}_{++}^{2}}{\tilde{\chi}_{+-}\tilde{\chi}_{-+}}=n_{++,+-},\label{eq:chi_pairing2}                        \\
    \braket{\tilde{\chi}_{--}^{2}}{\tilde{\chi}_{++}\tilde{\chi}_{--}} & =-\braket{\tilde{\chi}_{--}^{2}}{\tilde{\chi}_{+-}\tilde{\chi}_{-+}}=n_{--,+-}.\label{eq:chi_pairing3}
  \end{align}
\end{lemma}
\begin{proof}
  To prove these identities, we first define the four two-replica beam-splitter
  outputs:
  \begin{equation}
    \ket{\Xi_{st}}=U\ket{\psi_{s},\psi_{t}},\qquad s,t\in\{+,-\}.
    \label{eq:Xi_def}
  \end{equation}
  Let $S$ be the graded swap of the two replicas, which is defined as
  \begin{equation}
    S\ket{c,d}=(-1)^{p(c)p(d)}\ket{d,c}.
  \end{equation}
  Conjugation by the convolution unitary turns the swap into a replica parity,
  \begin{equation}
    USU^{\dagger}=U^{\dagger}SU=\Pi_{2}.
    \label{eq:swap_to_parity}
  \end{equation}
  Then, these relations give
  \begin{align}
    \Pi_{2}\ket{\Xi_{++}} & =\ket{\Xi_{++}},\quad
    \Pi_{2}\ket{\Xi_{--}}=-\ket{\Xi_{--}},\nonumber \\
    \ket{\Xi_{-+}}        & =\Pi_{2}\ket{\Xi_{+-}},
    \label{eq:Xi_parity}                            \\
    S\ket{\Xi_{++}}       & =\ket{\Xi_{++}},\quad
    S\ket{\Xi_{--}}=-\ket{\Xi_{--}},\nonumber       \\
    S\ket{\Xi_{+-}}       & =-\ket{\Xi_{+-}}.
    \label{eq:Xi_swap}
  \end{align}
  Next, let us decompose $\ket{\Xi_{+-}}$ according to the parity of replica $2$:
  \begin{equation}
    \ket{\Xi_{+-}^{\pm}}=\frac{I\pm\Pi_{2}}{2}\ket{\Xi_{+-}},
    \label{eq:Xi_pm_def}
  \end{equation}
  and define the (unnormalized) reduced density operators on the retained replica 1:
  \begin{align}
    \rho_{++}       & =\tr_{2}\dyad{\Xi_{++}},          &
    \rho_{--}       & =\tr_{2}\dyad{\Xi_{--}},\nonumber   \\
    \rho_{+-}^{\pm} & =\tr_{2}\dyad{\Xi_{+-}^{\pm}}.
    \label{eq:rho_def}
  \end{align}
  The operators $\rho_{++},\rho_{+-}^{-}$ are supported on the even-parity sector, while $\rho_{--},\rho_{+-}^{+}$ are supported on the odd-parity sector.
  From Lemma~\ref{lem:chi_product} and the sign rule~\eqref{eq:sign_rule}, we have
  \begin{equation}
    \tilde{\chi}_{st}\tilde{\chi}_{s't'}
    =(-1)^{p(\psi_{t})\qty[p(\psi_{s'})+p(\psi_{t'})]}\chi_{\tr_{2}\dyad{\Xi_{ss'}}{\Xi_{tt'}}},
    \label{eq:chi_product_general}
  \end{equation}
  leading to
  \begin{equation}
    \begin{aligned}
      \tilde{\chi}_{++}^{2}              & =\chi_{\rho_{++}},\quad
      \tilde{\chi}_{--}^{2}=\chi_{\rho_{--}},                                            \\
      \tilde{\chi}_{++}\tilde{\chi}_{--} & =\chi_{\rho_{+-}^{+}}+\chi_{\rho_{+-}^{-}},   \\
      \tilde{\chi}_{+-}\tilde{\chi}_{-+} & =\chi_{\rho_{+-}^{-}}-\chi_{\rho_{+-}^{+}},   \\
      \tilde{\chi}_{++}\tilde{\chi}_{+-} & =\chi_{\tr_{2}\dyad{\Xi_{++}}{\Xi_{+-}^{+}}}, \\
      \tilde{\chi}_{--}\tilde{\chi}_{+-} & =\chi_{\tr_{2}\dyad{\Xi_{+-}^{-}}{\Xi_{--}}}.
    \end{aligned}
    \label{eq:chi_product_list}
  \end{equation}
  To compute the norms of these characteristic functions, we also use
  \begin{equation}
    \rho_{+-}^{+}\rho_{+-}^{-}=0.
    \label{eq:rho_orthogonal}
  \end{equation}
  Moreover, $S\Pi_{2}=\Pi_{1}S$ and $\Pi_{1}\Pi_{2}\ket{\Xi_{+-}}=-\ket{\Xi_{+-}}$, so Eq.~\eqref{eq:Xi_swap} implies
  \begin{equation}
    \label{eq:Xi_pm_swap}
    S\ket{\Xi_{+-}^{+}}=-\ket{\Xi_{+-}^{-}},\qquad S\ket{\Xi_{+-}^{-}}=-\ket{\Xi_{+-}^{+}}.
  \end{equation}
  Therefore $\Xi_{+-}^{+}$ and $\Xi_{+-}^{-}$ have the same Schmidt coefficients,
  which implies
  \begin{equation}
    \tr[(\rho_{+-}^{+})^{2}]=\tr[(\rho_{+-}^{-})^{2}].
    \label{eq:rho_equal_purity}
  \end{equation}
  Using Eqs.~\eqref{eq:HS_pairing}, \eqref{eq:rho_orthogonal}, and \eqref{eq:rho_equal_purity}, we have
  \begin{align}
    \norm{\tilde{\chi}_{++}\tilde{\chi}_{--}}^{2} & = \norm{\tilde{\chi}_{+-}\tilde{\chi}_{-+}}^{2} \nonumber               \\
                                                  & =2^{L}\qty[\tr((\rho_{+-}^{+})^{2})+\tr((\rho_{+-}^{-})^{2})],\nonumber \\
    \braket{\tilde{\chi}_{++}\tilde{\chi}_{--}}{\tilde{\chi}_{+-}\tilde{\chi}_{-+}}
                                                  & =2^{L}\qty[\tr((\rho_{+-}^{-})^{2})-\tr((\rho_{+-}^{+})^{2})]\nonumber  \\
                                                  & =0,
  \end{align}
  which proves Eq.~\eqref{eq:chi_pairing1}.

  Next, we prove Eqs.~\eqref{eq:chi_pairing2} and \eqref{eq:chi_pairing3} using the identity
  \begin{equation}
    \norm{\tr_{2}\dyad{\Xi}{\Xi'}}^{2}=\tr\!\left[(\tr_{1}\dyad{\Xi})(\tr_{1}\dyad{\Xi'})\right],
    \label{eq:partial_trace_HS}
  \end{equation}
  which holds for arbitrary two-replica states $\ket{\Xi}$ and $\ket{\Xi'}$.
  Here, we apply it to the pairs $(\Xi_{++},\Xi_{+-}^{+})$ and $(\Xi_{+-}^{-},\Xi_{--})$.
  The swap relations~\eqref{eq:Xi_swap} and \eqref{eq:Xi_pm_swap} convert the partial trace over replica $1$ into the partial trace over replica $2$, giving
  \begin{align}
    2^{-L}n_{++,+-} & =\tr(\rho_{++}\rho_{+-}^{-}),\nonumber \\
    2^{-L}n_{--,+-} & =\tr(\rho_{--}\rho_{+-}^{+}).
    \label{eq:n_reduction}
  \end{align}
  From the supports given above, we have $\rho_{++}\rho_{+-}^{+}=0$
  and $\rho_{--}\rho_{+-}^{-}=0$, leading to
  \begin{align}
    \braket{\tilde{\chi}_{++}^{2}}{\tilde{\chi}_{++}\tilde{\chi}_{--}} & =2^{L}\tr[\rho_{++}(\rho_{+-}^{+}+\rho_{+-}^{-})]=n_{++,+-},\nonumber \\
    \braket{\tilde{\chi}_{++}^{2}}{\tilde{\chi}_{+-}\tilde{\chi}_{-+}} & =2^{L}\tr[\rho_{++}(\rho_{+-}^{-}-\rho_{+-}^{+})]=n_{++,+-},\nonumber \\
    \braket{\tilde{\chi}_{--}^{2}}{\tilde{\chi}_{++}\tilde{\chi}_{--}} & =2^{L}\tr[\rho_{--}(\rho_{+-}^{+}+\rho_{+-}^{-})]=n_{--,+-},\nonumber \\
    \braket{\tilde{\chi}_{--}^{2}}{\tilde{\chi}_{+-}\tilde{\chi}_{-+}} & =2^{L}\tr[\rho_{--}(\rho_{+-}^{-}-\rho_{+-}^{+})]=-n_{--,+-}.
  \end{align}
  This proves Eqs.~\eqref{eq:chi_pairing2} and \eqref{eq:chi_pairing3}, which completes the proof of Lemma~\ref{lem:chi_pairing}.
\end{proof}

\begin{lemma} \label{lem:chi_pairing_mixed} The following identity holds:
  \begin{equation}
    \braket{\tilde{\chi}_{++}\tilde{\chi}_{+-}}{\tilde{\chi}_{--}\tilde{\chi}_{+-}}=n_{++,--}.
    \label{eq:chi_pairing_mixed}
  \end{equation}
\end{lemma}
\begin{proof}
  By the pairing~\eqref{eq:HS_pairing} and the products~\eqref{eq:chi_product_list}, the left-hand side of Eq.~\eqref{eq:chi_pairing_mixed} is $2^{L}$ times the Hilbert-Schmidt product of $\tr_{2}\dyad{\Xi_{++}}{\Xi_{+-}^{+}}$ and $\tr_{2}\dyad{\Xi_{+-}^{-}}{\Xi_{--}}$.
  The right-hand side is $2^{L}$ times the squared Hilbert-Schmidt norm of the operator whose characteristic function is $\tilde{\chi}_{++}\tilde{\chi}_{--}$, for which Eq.~\eqref{eq:chi_product_general} gives
  \begin{equation}
    \tilde{\chi}_{++}\tilde{\chi}_{--}=\chi_{\tr_{2}\dyad{\Xi_{+-}}}.
    \label{eq:chi_pp_mm}
  \end{equation}
  To compare the two sides, we write them as matrix elements of a single operator on four replicas.
  Let us choose occupation bases and write
  \begin{align}
    \Xi^{st}_{ab} & =\braket{a,b}{\Xi_{st}}.
  \end{align}
  We write $\ket{r,e,r',e'}$ for the occupation basis of the four replicas, where $r$ and $r'$ label the states of the retained replicas $1$ and $3$, and $e$ and $e'$ those of the traced replicas $2$ and $4$.
  Direct expansion gives
  \begin{equation}
    \begin{aligned}
       & \tr\!\left[\qty(\tr_{2}\dyad{\Xi_{++}}{\Xi_{+-}^{+}})^{\dagger}\tr_{2}\dyad{\Xi_{+-}^{-}}{\Xi_{--}}\right] \\
       & \quad=\sum_{r,e,r',e'}\Xi^{+-}_{re}\qty(\Xi^{++}_{r'e})^{*}\Xi^{+-}_{r'e'}\qty(\Xi^{--}_{re'})^{*},        \\
       & \tr\!\left[\qty(\tr_{2}\dyad{\Xi_{+-}})^{2}\right]                                                         \\
       & \quad=\sum_{r,e,r',e'}\Xi^{+-}_{re}\Xi^{+-}_{r'e'}\qty(\Xi^{+-}_{re'})^{*}\qty(\Xi^{+-}_{r'e})^{*}.
    \end{aligned}
    \label{eq:index_sums}
  \end{equation}
  In the same basis,
  \begin{equation}
    \begin{aligned}
      \ket{\Xi_{+-},\Xi_{+-}}               & =\sum_{r,e,r',e'}\Xi^{+-}_{re}\Xi^{+-}_{r'e'}\ket{r,e,r',e'}, \\
      \braket{\Xi_{--},\Xi_{++}}{r,e',r',e} & =\qty(\Xi^{--}_{re'})^{*}\qty(\Xi^{++}_{r'e})^{*},            \\
      \braket{\Xi_{+-},\Xi_{+-}}{r,e',r',e} & =\qty(\Xi^{+-}_{re'})^{*}\qty(\Xi^{+-}_{r'e})^{*},
    \end{aligned}
    \label{eq:four_replica_components}
  \end{equation}
  so the bra labels in Eq.~\eqref{eq:index_sums} are $(r,e',r',e)$ and the ket labels are $(r,e,r',e')$.

  We next show that $W=\Pi_{4}S_{24}$ satisfies
  \begin{equation}
    W\ket{r,e,r',e'}=\ket{r,e',r',e}
    \label{eq:W_action}
  \end{equation}
  for every occupation basis state $\ket{r,e,r',e'}$ appearing in $\ket{\Xi_{+-},\Xi_{+-}}$.
  The graded swap $S_{24}$ of replicas $2$ and $4$ exchanges these two labels.
  Decomposing the nonadjacent swap $S_{24}$ into adjacent graded swaps shows that the exponent of its sign is
  \begin{equation}
    p(e)p(e')+[p(e)+p(e')]\,p(r').
  \end{equation}
  After the swap, $\Pi_4$ contributes $p(e)$.
  Since the second $\Xi_{+-}$ has odd total parity, $p(r')\equiv1+p(e')\pmod{2}$.
  Hence the exponent of the total sign is
  \begin{equation}
    p(e)p(e')+[p(e)+p(e')][1+p(e')]+p(e)\equiv0\pmod{2},
    \label{eq:swap_sign}
  \end{equation}
  so that Eq.~\eqref{eq:W_action} follows.

  With Eqs.~\eqref{eq:four_replica_components} and \eqref{eq:W_action}, Eq.~\eqref{eq:index_sums} becomes
  \begin{equation}
    \begin{aligned}
       & \tr\!\left[\qty(\tr_{2}\dyad{\Xi_{++}}{\Xi_{+-}^{+}})^{\dagger}\tr_{2}\dyad{\Xi_{+-}^{-}}{\Xi_{--}}\right] \\
       & \quad=\mel{\Xi_{--},\Xi_{++}}{W}{\Xi_{+-},\Xi_{+-}},                                                       \\
       & \tr\!\left[\qty(\tr_{2}\dyad{\Xi_{+-}})^{2}\right]                                                         \\
       & \quad=\mel{\Xi_{+-},\Xi_{+-}}{W}{\Xi_{+-},\Xi_{+-}}.
    \end{aligned}
    \label{eq:trace_contraction}
  \end{equation}
  Here we define
  \begin{equation}
    U_{1234}:=U_{12}U_{34},\quad
    W':=U_{1234}^{\dagger}WU_{1234},
  \end{equation}
  where $U_{12}$ and $U_{34}$ are copies of $U$, and
  \begin{equation}
    \ket{\tau}=\ket{\psi_+,\psi_-,\psi_+,\psi_-},
    \qquad
    \ket{\sigma}=\ket{\psi_-,\psi_-,\psi_+,\psi_+}.
  \end{equation}
  Since $U_{1234}\ket{\tau}=\ket{\Xi_{+-},\Xi_{+-}}$ and $U_{1234}\ket{\sigma}=\ket{\Xi_{--},\Xi_{++}}$, Eqs.~\eqref{eq:HS_pairing}, \eqref{eq:chi_pp_mm}, and \eqref{eq:trace_contraction} yield
  \begin{align}
    2^{-L}\braket{\tilde{\chi}_{++}\tilde{\chi}_{+-}}{\tilde{\chi}_{--}\tilde{\chi}_{+-}}
     & =\mel{\sigma}{W'}{\tau},\nonumber \\
    2^{-L}n_{++,--}
     & =\mel{\tau}{W'}{\tau}.
    \label{eq:W_matrix_elements}
  \end{align}
  Therefore, it remains to prove $\mel{\sigma}{W'}{\tau}=\mel{\tau}{W'}{\tau}$.

  To proceed, let $O$, $O_{1234}$, $P_{4}$, and $T_{ij}$ be the orthogonal matrices that $U$, $U_{1234}$, $\Pi_{4}$, and $S_{ij}$ induce on the replica index of the annihilation operators.
  Their explicit forms are
  \begin{align}
    O      & =\frac{1}{\sqrt{2}}\begin{pmatrix}1 & 1 \\ 1 & -1\end{pmatrix},\quad
    O_{1234}=O\oplus O,\nonumber                                                  \\
    P_{4}  & =\operatorname{diag}(1,1,1,-1),\nonumber                             \\
    T_{14} & =\begin{pmatrix}
                0 & 0 & 0 & 1 \\
                0 & 1 & 0 & 0 \\
                0 & 0 & 1 & 0 \\
                1 & 0 & 0 & 0
              \end{pmatrix},\quad
    T_{24}   =\begin{pmatrix}
                1 & 0 & 0 & 0 \\
                0 & 0 & 0 & 1 \\
                0 & 0 & 1 & 0 \\
                0 & 1 & 0 & 0
              \end{pmatrix}.
  \end{align}
  Then, $W'$ transforms the annihilation operators as
  \begin{equation}
    (W')^{\dagger}
    \begin{pmatrix} \bm{a}_1 \\ \bm{a}_2 \\ \bm{a}_3 \\ \bm{a}_4\end{pmatrix}
    W'
    =M
    \begin{pmatrix}\bm{a}_1 \\ \bm{a}_2 \\ \bm{a}_3 \\ \bm{a}_4\end{pmatrix},
  \end{equation}
  where
  \begin{align}
    M & =O_{1234}^{\rm T}P_{4}T_{24}O_{1234}\nonumber \\
      & =\frac{1}{2}\begin{pmatrix}
                      1  & 1  & 1  & -1 \\
                      1  & 1  & -1 & 1  \\
                      -1 & 1  & 1  & 1  \\
                      1  & -1 & 1  & 1
                    \end{pmatrix}.
  \end{align}
  Direct multiplication gives
  \begin{equation}
    T_{14}M=M T_{24}
    =\frac{1}{2}
    \begin{pmatrix}
      1  & -1 & 1  & 1  \\
      1  & 1  & -1 & 1  \\
      -1 & 1  & 1  & 1  \\
      1  & 1  & 1  & -1
    \end{pmatrix},
  \end{equation}
  so that $S_{14}W'$ and $W'S_{24}$ agree up to a phase.
  Both leave the vacuum invariant, which fixes the phase to unity,
  \begin{equation}
    S_{14}W'=W'S_{24}.
    \label{eq:swap_commute}
  \end{equation}

  Finally, the action of the swaps on $\ket{\tau}$, whose parity string is $(0,1,0,1)$, reads
  \begin{equation}
    S_{14}\ket{\tau}=-\ket{\sigma},\qquad
    S_{24}\ket{\tau}=-\ket{\tau}.
    \label{eq:swap_action}
  \end{equation}
  Using Eqs.~\eqref{eq:swap_commute} and \eqref{eq:swap_action}, we have
  \begin{align}
    \mel{\sigma}{W'}{\tau}
     & =-\mel{\tau}{S_{14}W'}{\tau}\nonumber \\
     & =-\mel{\tau}{W'S_{24}}{\tau}
    =\mel{\tau}{W'}{\tau}.
  \end{align}
  Together with Eq.~\eqref{eq:W_matrix_elements}, this proves Eq.~\eqref{eq:chi_pairing_mixed}.
\end{proof}

\begin{lemma} \label{lem:measurement} Let $\ket{\Psi}$ be a
  normalized pure state of definite parity. Place the measured mode
  last in the fermionic ordering and write
  \begin{equation}
    \ket{\Psi}=\sqrt{p}\,\ket{\eta_{0}}\gtensor\ket{0}+\sqrt{q}\,\ket{\eta_{1}}\gtensor\ket{1},\quad0\leq q=1-p\leq1,
    \label{eq:Psi_split}
  \end{equation}
  where the conditional vectors are normalized. Then
  \begin{equation}
    \nu_{2}(\Psi)\leq p\nu_{2}(\eta_{0})+q\nu_{2}(\eta_{1}).
    \label{eq:measurement_ineq}
  \end{equation}
\end{lemma}
\begin{proof}
  We first consider the case where $\ket{\Psi}$ is even.
  Then $\ket{\eta_{0}}$ is even and $\ket{\eta_{1}}$ is odd.
  We therefore identify $\ket{\psi_{+}}=\ket{\eta_{0}}$ and $\ket{\psi_{-}}=\ket{\eta_{1}}$.
  Let $u_{1},u_{2}$ be the real Grassmann variables associated with the measured mode. We have
  \begin{equation}
    \begin{aligned}
      \tilde{\chi}_{\ket{0}\!\bra{0}} & =1+\frac{i}{2}u_{1}u_{2},       &
      \tilde{\chi}_{\ket{0}\!\bra{1}} & =\frac{u_{1}-iu_{2}}{\sqrt{2}},   \\
      \tilde{\chi}_{\ket{1}\!\bra{0}} & =\frac{u_{1}+iu_{2}}{\sqrt{2}}, &
      \tilde{\chi}_{\ket{1}\!\bra{1}} & =1-\frac{i}{2}u_{1}u_{2}.
    \end{aligned}
    \label{eq:chi_one_mode}
  \end{equation}
  Using Eq.~\eqref{eq:chi_one_mode}, we have
  \begin{equation}
    \begin{aligned}
      \tilde{\chi}_{\dyad{\Psi}}
       & =p\tilde{\chi}_{++}+q\tilde{\chi}_{--}                                                                         \\
       & \quad+\frac{i}{2}\qty(p\tilde{\chi}_{++}-q\tilde{\chi}_{--})u_{1}u_{2}                                         \\
       & \quad+\sqrt{\frac{pq}{2}}[(\tilde{\chi}_{-+}-\tilde{\chi}_{+-})u_{1}                                           \\
                 & \qquad\qquad\quad+i(\tilde{\chi}_{+-}+\tilde{\chi}_{-+})u_{2}],
    \end{aligned}
    \label{eq:chi_Psi}
  \end{equation}
  where we used the sign rule~\eqref{eq:sign_rule}.
  Since $\tilde{\chi}_{+-},\tilde{\chi}_{-+}$ are odd, we have
  \begin{equation}
    \tilde{\chi}_{+-}^{2}=\tilde{\chi}_{-+}^{2}=0,\quad
    \tilde{\chi}_{-+}\tilde{\chi}_{+-}=-\tilde{\chi}_{+-}\tilde{\chi}_{-+}.
    \label{eq:chi_odd_relations}
  \end{equation}
  Squaring Eq.~\eqref{eq:chi_Psi}, using Eq.~\eqref{eq:chi_odd_relations},
  gives
  \begin{equation}
    \begin{aligned}
      \tilde{\chi}_{\dyad{\Psi}}^{2}
       & =\qty(p\tilde{\chi}_{++}+q\tilde{\chi}_{--})^{2}                                                                      \\
       & \quad+\sqrt{2pq}\,\qty(p\tilde{\chi}_{++}+q\tilde{\chi}_{--})                                                         \\
       & \qquad\quad\times[(\tilde{\chi}_{-+}-\tilde{\chi}_{+-})u_{1}                                             \\
                      & \qquad\qquad\qquad+i(\tilde{\chi}_{+-}+\tilde{\chi}_{-+})u_{2}]              \\
       & \quad+i\qty(p^{2}\tilde{\chi}_{++}^{2}-q^{2}\tilde{\chi}_{--}^{2}+2pq\,\tilde{\chi}_{+-}\tilde{\chi}_{-+})u_{1}u_{2},
    \end{aligned}
    \label{eq:chi_Psi_square}
  \end{equation}
  and Eq.~\eqref{eq:purity_fermion_app} at $n=2$ for the $L+1$ modes leads to
  \begin{equation}
    \begin{aligned}
      2^{L+1}\nu_{2}(\Psi)
       & =\norm{\qty(p\tilde{\chi}_{++}+q\tilde{\chi}_{--})^{2}}^{2}                                                     \\
       & \quad+8pq\,\norm{\qty(p\tilde{\chi}_{++}+q\tilde{\chi}_{--})\tilde{\chi}_{+-}}^{2}                              \\
       & \quad+\norm{p^{2}\tilde{\chi}_{++}^{2}-q^{2}\tilde{\chi}_{--}^{2}+2pq\,\tilde{\chi}_{+-}\tilde{\chi}_{-+}}^{2},
    \end{aligned}
    \label{eq:nu_three_terms}
  \end{equation}
  where we used
  \begin{equation}
    \norm{\qty(p\tilde{\chi}_{++}+q\tilde{\chi}_{--})\tilde{\chi}_{+-}}^{2} = \norm{\qty(p\tilde{\chi}_{++}+q\tilde{\chi}_{--})\tilde{\chi}_{-+}}^{2}.
  \end{equation}
  This holds because the product formula~\eqref{eq:chi_product_general} makes the two polynomials the characteristic functions of operators that are Hermitian conjugates of each other, and the pairing~\eqref{eq:HS_pairing} then gives them the same coefficient norm.
  Using Lemmas~\ref{lem:chi_pairing} and \ref{lem:chi_pairing_mixed}, we get
  \begin{align}
     & \norm{\qty(p\tilde{\chi}_{++}+q\tilde{\chi}_{--})^{2}}^{2}\nonumber                                             \\
     & \quad=p^{4}n_{++,++}+q^{4}n_{--,--}+4p^{3}qn_{++,+-}\nonumber                                                   \\
     & \qquad+4pq^{3}n_{--,+-}+4p^{2}q^{2}n_{++,--}\nonumber                                                           \\
     & \qquad+2p^{2}q^{2}\Real\braket{\tilde{\chi}_{++}^{2}}{\tilde{\chi}_{--}^{2}},\nonumber                          \\
     & 8pq\norm{\qty(p\tilde{\chi}_{++}+q\tilde{\chi}_{--})\tilde{\chi}_{+-}}^{2}\nonumber                             \\
     & \quad=8p^{3}qn_{++,+-}+8pq^{3}n_{--,+-}+16p^{2}q^{2}n_{++,--},\nonumber                                         \\
     & \norm{p^{2}\tilde{\chi}_{++}^{2}-q^{2}\tilde{\chi}_{--}^{2}+2pq\tilde{\chi}_{+-}\tilde{\chi}_{-+}}^{2}\nonumber \\
     & \quad=p^{4}n_{++,++}+q^{4}n_{--,--}+4p^{3}qn_{++,+-}\nonumber                                                   \\
     & \qquad+4pq^{3}n_{--,+-}+4p^{2}q^{2}n_{++,--}\nonumber                                                           \\
     & \qquad-2p^{2}q^{2}\Real\braket{\tilde{\chi}_{++}^{2}}{\tilde{\chi}_{--}^{2}}.
    \label{eq:three_norms}
  \end{align}
  Summing up, we have
  \begin{align}
    2^{L+1}\nu_{2}(\Psi)
     & =2n_{++,++}p^{4}+16n_{++,+-}p^{3}q\nonumber            \\
     & \quad+24n_{++,--}p^{2}q^{2}+16n_{--,+-}pq^{3}\nonumber \\
     & \quad+2n_{--,--}q^{4}\nonumber                         \\
     & =\sum_{k=0}^{4}\binom{4}{k}c_{k}p^{k}q^{4-k},
    \label{eq:nu_quartic}
  \end{align}
  where
  \begin{equation}
    \begin{aligned}
      c_{0} & =2n_{--,--}, & c_{1} & =4n_{--,+-}, \\
      c_{2} & =4n_{++,--}, & c_{3} & =4n_{++,+-}, \\
      c_{4} & =2n_{++,++}.
    \end{aligned}
    \label{eq:ck_def}
  \end{equation}
  Then, using Lemmas~\ref{lem:chi_pairing} and \ref{lem:chi_pairing_mixed}, we prove the following relations:
  \begin{align}
    c_{2}-2c_{1}+c_{0} & =2\norm{\tilde{\chi}_{--}^{2}-\tilde{\chi}_{++}\tilde{\chi}_{--}+\tilde{\chi}_{+-}\tilde{\chi}_{-+}}^{2},
    \label{eq:second_difference0}                                                                                                  \\
    c_{3}-2c_{2}+c_{1} & =4\norm{\tilde{\chi}_{++}\tilde{\chi}_{+-}-\tilde{\chi}_{--}\tilde{\chi}_{+-}}^{2},
    \label{eq:second_difference1}                                                                                                  \\
    c_{4}-2c_{3}+c_{2} & =2\norm{\tilde{\chi}_{++}^{2}-\tilde{\chi}_{++}\tilde{\chi}_{--}-\tilde{\chi}_{+-}\tilde{\chi}_{-+}}^{2}.
    \label{eq:second_difference2}
  \end{align}
  Since the right-hand sides of these equalities are nonnegative, we prove the convexity
  \begin{equation}
    \begin{aligned}
       & \frac{d^{2}}{dp^{2}}\left[2^{L+1}\nu_{2}(\Psi)\right]                        \\
       & \quad=12\sum_{k=0}^{2}\binom{2}{k}(c_{k+2}-2c_{k+1}+c_{k})p^{k}q^{2-k}\geq0.
    \end{aligned}
    \label{eq:nu_convex}
  \end{equation}
  Since Eqs.~\eqref{eq:purity_fermion_app}, \eqref{eq:nu_quartic}, and \eqref{eq:ck_def} give $\nu_{2}(\Psi)=\nu_{2}(\eta_{0})$ at $p=1$ and $\nu_{2}(\Psi)=\nu_{2}(\eta_{1})$ at $p=0$, this convexity gives
  \begin{equation}
    \nu_{2}(\Psi)\leq p\nu_{2}(\eta_{0})+q\nu_{2}(\eta_{1}).
  \end{equation}

  Now suppose that $\ket{\Psi}$ is odd.
  Then $\ket{\eta_{0}}$ is odd and $\ket{\eta_{1}}$ is even.
  Let us introduce one ancilla mode $a$ in the occupied state $\ket{1}_{a}$, placed before all unmeasured modes, and consider the state $\ket{1}_{a}\gtensor\ket{\Psi}$.
  The enlarged state has even total parity. Its two conditional states
  are
  \begin{equation}
    \ket{1}_{a}\gtensor\ket{\eta_{0}}\quad\text{and}\quad\ket{1}_{a}\gtensor\ket{\eta_{1}},
  \end{equation}
  with even and odd parity, respectively. The result for even $\ket{\Psi}$ thus
  gives
  \begin{equation}
    \nu_{2}(\Psi)\leq p\nu_{2}(\eta_{0})+q\nu_{2}(\eta_{1}),
  \end{equation}
  where we used the multiplicativity and $\nu_{2}(\ket{1}_{a})=1$.
  This proves Eq.~\eqref{eq:measurement_ineq} for odd $\ket{\Psi}$ as well.
\end{proof}

\subsection{Nonmonotonicity for general convolution unitaries}\label{app:convolution_uniqueness}

Here we prove Theorems~\ref{thm:qudit-convolution-uniqueness} and \ref{thm:boson-convolution-uniqueness}.
The conditions in these theorems are sufficient because the MRE then reduces to the SRE or to the bosonic MRE, which are monotones by Theorems~\ref{thm:qudit-protocol-monotonicity} and \ref{thm:boson-protocol-monotonicity}.
Since strong monotonicity of $M_n^{\rm lin}=1-\nu_n$ implies monotonicity of $M_n$, it remains to show that $M_n$ is not a monotone if the conditions are violated.

To prove necessity, we recall the input~\eqref{eq:protocol_input} of the protocol in Sec.~\ref{sec:mono},
\begin{equation*}
  \ket{\Psi}=\frac{1}{\sqrt{n}}\sum_{r=1}^{n}\qty(U_{F,r}\ket{\varphi})^{\otimes\ell}\otimes\ket{r}_Q,\qquad \ell\geq1.
\end{equation*}
Here $\ket{\varphi}$ is a pure state, the $U_{F,r}$ are free unitaries, and a free measurement of the ancilla $Q$ distinguishes the states $\ket{r}_Q$.
For outcome $r$, applying $(U_{F,r}^{\dagger})^{\otimes\ell}$ and discarding $Q$ give $\ket{\varphi}^{\otimes\ell}$ deterministically.

To show that this protocol increases $M_n$, we establish two inequalities.
The first is the lower bound $\nu_n(\Psi)\geq C\,\nu_n(U_{F,1}\varphi,\ldots,U_{F,n}\varphi)^{\ell}$ with a constant $C>0$ independent of $\ell$.
The second is the increase of the convolution purity~\eqref{eq:purity_increase}, $\nu_n(U_{F,1}\varphi,\ldots,U_{F,n}\varphi)>\nu_n(\varphi)$.
For large $\ell$, the second inequality raised to the power $\ell$ overcomes the constant $C$, so that these two inequalities give $\nu_n(\Psi)>\nu_n(\varphi)^{\ell}=\nu_n(\varphi^{\otimes\ell})$, where the equality is the multiplicativity of $\nu_n$.
Since $M_n$ is a monotonically decreasing function of $\nu_n$, this inequality shows that the protocol increases $M_n$.
We prove the first inequality as Lemma~\ref{lem:feedforward} for qudits and bosons alike, and the second inequality as Lemma~\ref{lem:qudit_purity_increase} for qudits and Lemma~\ref{lem:boson_purity_increase} for bosons.

We take the states $\ket{r}_Q$ so that they have nonnegative coefficients and a free measurement of $Q$ determines $r$.
For qudits, they are distinct computational basis states of $\lceil\log_dn\rceil$ qudits, and the free measurement is a computational-basis measurement.
For bosons, they are one-mode states whose wave functions are nonnegative and have disjoint supports, and the free measurement is a homodyne measurement.

\begin{lemma}
  \label{lem:feedforward}
  Let $\ket{\varphi}$ be a pure state and $U_{F,1},\ldots,U_{F,n}$ be free unitaries.
  Assume that all states $U_{F,r}\ket{\varphi}$ and $\ket{r}_Q$ have nonnegative amplitudes in the computational basis for qudits or the position representation for bosons.
  Let $\nu_n$ be the convolution purity of a convolution unitary that acts separately and identically on each site or mode and maps the basis states of the $n$ replicas to basis states up to a common overall phase.
  Then the input~\eqref{eq:protocol_input} satisfies $\nu_n(\Psi)\geq C\,\nu_n(U_{F,1}\varphi,\ldots,U_{F,n}\varphi)^{\ell}$ for some constant $C>0$ independent of $\ell$.
\end{lemma}

The proof uses nonnegative amplitudes to obtain a lower bound by retaining the term in which input replica $r$ carries $(U_{F,r}\ket{\varphi})^{\otimes\ell}$.

\begin{proof}
  Expanding the $n$ replicas of the input~\eqref{eq:protocol_input} gives
  \begin{equation}
    \ket{\Psi}^{\otimes n}=n^{-n/2}\sum_{\lambda}\ket{\Phi_{\lambda}},
  \end{equation}
  where $\lambda=(\lambda_1,\ldots,\lambda_n)$ runs over $\{1,\ldots,n\}^n$ and
  \begin{equation}
    \ket{\Phi_{\lambda}}=\bigotimes_{r=1}^{n}\qty[\qty(U_{F,\lambda_r}\ket{\varphi})^{\otimes\ell}\otimes\ket{\lambda_r}_Q].
  \end{equation}
  Every $\ket{\Phi_\lambda}$ has nonnegative coefficients, and so does its image under the convolution unitary after the constant phase is removed.
  We denote the coefficients of this image by $\Phi'_\lambda(y,z)$, where $y$ and $z$ label the basis states of the output replica $1$ and of the other output replicas.
  The matrix elements of the convolved state then satisfy
  \begin{align}
    \mel{y}{\mathcal{C}_n(\Psi)}{y'}
     & =n^{-n}\sum_{\lambda,\lambda'}\sum_{z}\Phi'_{\lambda}(y,z)\,\Phi'_{\lambda'}(y',z)\nonumber \\
     & \geq n^{-n}\sum_{z}\Phi'_{\lambda_0}(y,z)\,\Phi'_{\lambda_0}(y',z)\geq0,
    \label{eq:element_bound}
  \end{align}
  where $\lambda_0=(1,2,\ldots,n)$ and the sum over $z$ is an integral for bosons.
  The middle expression is the matrix element of $n^{-n}$ times the convolution of the inputs $(U_{F,r}\varphi)^{\otimes\ell}\otimes\dyad{r}_Q$ in replicas $r=1,\ldots,n$.
  The convolution unitary acts on each site or mode separately, so that $\nu_n$ is multiplicative also for different inputs, generalizing Eq.~\eqref{eq:nu_multiplicative}.
  Since the purity is the sum of the squared matrix elements, Eq.~\eqref{eq:element_bound} gives
  \begin{align}
    \nu_n(\Psi)
     & \geq n^{-2n}\,\nu_n(\ket{1}_Q,\ldots,\ket{n}_Q)\nonumber       \\
     & \quad\times\nu_n(U_{F,1}\varphi,\ldots,U_{F,n}\varphi)^{\ell},
  \end{align}
  which is the first inequality in Eq.~\eqref{eq:protocol_chain} with $C=n^{-2n}\nu_n(\ket{1}_Q,\ldots,\ket{n}_Q)>0$.
\end{proof}

\subsubsection*{Qudits}

Since Eq.~\eqref{eq:nu_general_qudit} depends only on the weights $t_r$, we can choose the convolution unitary to permute the computational basis, as required in Lemma~\ref{lem:feedforward}.
Since $d$ is prime and the weights are nonzero, $t_1$ is invertible in $\mathbb{Z}_d$, and we consider the map
\begin{equation}
  \bigotimes_{r=1}^{n}\ket{\bm{q}_r}\mapsto\ket{\frac{\sum_{r}t_r\bm{q}_r}{\sum_{s}t_s^{2}}}\bigotimes_{r=2}^{n}\ket{\bm{q}_r-\frac{t_r}{t_1}\bm{q}_1}.
  \label{eq:qudit_permutation}
\end{equation}
This map is invertible because $\sum_st_s^2\not\equiv0$ by Eq.~\eqref{eq:t_sum}.
Conjugation by this permutation maps the operators $X_1(\bm{q})$ and $Z_1(\bm{p})$ of Sec.~\ref{subsec:MRE_qubit} to $\bigotimes_rX_r(t_r\bm{q})$ and $\bigotimes_rZ_r(t_r\bm{p}/\sum_st_s^2)$.
Equation~\eqref{eq:T_def} thus holds with $T_r=t_r\,\mathrm{diag}(1,1/\sum_st_s^2)$.
Changing the summation variable as in Appendix~\ref{app:properties_qudit}, we find that the convolution purity of this permutation is also given by Eq.~\eqref{eq:nu_general_qudit}.

\begin{lemma}
  \label{lem:qudit_purity_increase}
  Let $d$ be an odd prime, and suppose that the $t_r$ are not all equal up to sign.
  Let $\ket{\varphi}=(\ket{1,0}+\ket{0,1})/\sqrt{2}$ be a two-qudit state.
  Let $U_{C,r}$ be the Clifford unitary that multiplies all computational-basis labels by $t_r^{-1}$.
  Then Eq.~\eqref{eq:purity_increase} holds with $U_{F,r}=U_{C,r}$, and every $U_{C,r}\ket{\varphi}$ has nonnegative coefficients in the computational basis.
\end{lemma}

\begin{proof}
  The second statement holds because $U_{C,r}\ket{\varphi}=(\ket{t_r^{-1},0}+\ket{0,t_r^{-1}})/\sqrt{2}$.
  To show Eq.~\eqref{eq:purity_increase}, we evaluate both sides by Eq.~\eqref{eq:nu_general_qudit}.
  Since $U_{C,r}^{\dagger}XU_{C,r}=X^{t_r}$ and $U_{C,r}^{\dagger}ZU_{C,r}=Z^{1/t_r}$ for the shift and clock operators~\eqref{eq:qudit_XZ}, we have
  \begin{equation}
    \qty|\chi_{U_{C,r}\rho U_{C,r}^{\dagger}}(t_r\bm{q},t_r\bm{p})|=\qty|\chi_{\rho}(t_r^{2}\bm{q},\bm{p})|.
    \label{eq:qudit_rescale}
  \end{equation}
  The characteristic function of $\ket{\varphi}$ has the magnitude
  \begin{equation}
    \qty|\chi_{\varphi}(\bm{q},\bm{p})|=
    \begin{cases}
      \qty|\omega_d^{p_1}+\omega_d^{p_2}|/2, & \bm{q}=\bm{0},    \\
      1/2,                                   & \bm{q}=\pm(1,-1), \\
      0,                                     & \text{otherwise}.
    \end{cases}
  \end{equation}
  For identical inputs, a term with $\bm{q}\neq\bm{0}$ in Eq.~\eqref{eq:nu_general_qudit} requires $t_r\bm{q}=\pm(1,-1)$ for all $r$, which would make the $t_r$ equal up to sign.
  With $|\omega_d^{a}+\omega_d^{b}|^2/4=\cos^2[\pi(a-b)/d]$, the terms with $\bm{q}=\bm{0}$ give
  \begin{equation}
    \nu_n(\varphi)=\frac{1}{d}\sum_{p\in\mathbb{Z}_d}\prod_{r=1}^{n}\cos^{2}\frac{\pi t_rp}{d},
  \end{equation}
  where $p$ stands for $p_1-p_2$.
  For the inputs $U_{C,r}\ket{\varphi}$, Eq.~\eqref{eq:qudit_rescale} removes $t_r$ from $\bm{p}$.
  Discarding the nonnegative terms with $\bm{q}\neq\bm{0}$ then gives
  \begin{equation}
    \nu_n(U_{C,1}\varphi,\ldots,U_{C,n}\varphi)\geq\frac{1}{d}\sum_{p\in\mathbb{Z}_d}\cos^{2n}\frac{\pi p}{d}.
  \end{equation}
  Since $p\mapsto t_rp$ is a bijection of $\mathbb{Z}_d$, H\"older's inequality gives
  \begin{align}
    \frac{1}{d}\sum_{p}\prod_{r=1}^{n}\cos^{2}\frac{\pi t_rp}{d}
     & \leq\prod_{r=1}^{n}\qty(\frac{1}{d}\sum_{p}\cos^{2n}\frac{\pi t_rp}{d})^{1/n}\nonumber \\
     & =\frac{1}{d}\sum_{p}\cos^{2n}\frac{\pi p}{d}.
  \end{align}
  Since $\cos(\pi x/d)\neq0$ for odd $d$, equality requires $\cos^2(\pi t_rp/d)$ to be independent of $r$ for every $p$.
  At $p=1$, this means $\cos(2\pi t_r/d)=\cos(2\pi t_s/d)$, that is, $t_r\equiv\pm t_s$, which contradicts the assumption.
  The inequality is therefore strict, which proves Eq.~\eqref{eq:purity_increase}.
\end{proof}

We now prove Theorem~\ref{thm:qudit-convolution-uniqueness}.
The permutation~\eqref{eq:qudit_permutation} maps computational basis states to computational basis states and has the convolution purity~\eqref{eq:nu_general_qudit}.
For odd $d$, Lemmas~\ref{lem:feedforward} and \ref{lem:qudit_purity_increase} therefore show that $M_n$ is not a monotone unless the $t_r$ are equal up to sign.
For $d=2$, the $t_r\in\mathbb{Z}_2$ are all equal to one, so that the condition in the theorem always holds.
If the $t_r$ are equal up to sign, then $t_r=\pm1$, because $T_r=t_rT_s$ in Appendix~\ref{app:properties_qudit} gives $t_s=1$.
Since $|\chi_\rho(-\bm{u})|=|\chi_\rho(\bm{u})|$, Eq.~\eqref{eq:nu_general_qudit} then gives $\nu_n(\rho)=d^{-L}\sum_{\bm{u}}|\chi_\rho(\bm{u})|^{2n}$.
Comparing this with the qudit SRE~\eqref{eq:qudit_SRE}, we obtain
\begin{equation}
  M_n=\mathrm{SRE}_n,\qquad M_n^{\rm lin}=\mathrm{SRE}_n^{\rm lin}
\end{equation}
for pure states.
Here, Eq.~\eqref{eq:t_sum} becomes $n\det T_s\equiv1$, so that $n$ is coprime to $d$.
Theorem~\ref{thm:qudit-protocol-monotonicity} then shows that $M_n$ is a monotone and $M_n^{\rm lin}$ is a strong monotone.
This proves Theorem~\ref{thm:qudit-convolution-uniqueness}.

\subsubsection*{Bosons}

The convolution unitary maps position eigenstates to position eigenstates.
Since Eq.~\eqref{eq:U_BF} rotates both the position and the momentum operators of the replicas by $O$, it maps the position eigenstate $\ket{\vec{\bm{x}}}$ of the $n$ replicas to $\ket{(O\otimes I_L)\vec{\bm{x}}}$ up to a constant phase.

\begin{lemma}
  \label{lem:boson_purity_increase}
  Suppose that the weights $O_{1r}$ are not all equal up to sign.
  Then there exist a one-mode pure state $\ket{\varphi}$ and Gaussian unitaries $U_{G,1},\ldots,U_{G,n}$ such that Eq.~\eqref{eq:purity_increase} holds with $U_{F,r}=U_{G,r}$ and every $U_{G,r}\ket{\varphi}$ has a positive wave function.
\end{lemma}

\begin{proof}
  Since Eq.~\eqref{eq:nu_general_boson} depends only on $|O_{1r}|$, we may take all $O_{1r}$ to be positive.
  Let $U_{G,r}$ implement the squeezing $\psi(x)\mapsto O_{1r}^{1/2}\psi(O_{1r}x)$.
  The identity $U_{G,r}^{\dagger}D_{\rm b}(q,p)U_{G,r}=D_{\rm b}(O_{1r}q,p/O_{1r})$ then gives
  \begin{equation}
    \qty|\chi_{U_{G,r}\rho U_{G,r}^{\dagger}}(O_{1r}q,O_{1r}p)|=\qty|\chi_{\rho}(O_{1r}^{2}q,p)|.
    \label{eq:boson_rescale}
  \end{equation}
  We construct $\ket{\varphi}$ as a superposition of $N$ coherent states,
  \begin{equation}
    \ket{\varphi}=\frac{1}{\sqrt{N_A}}\sum_{i=1}^{N}D_{\rm b}(Ax_i,0)\ket{0},
  \end{equation}
  where the $x_i\in\mathbb{R}$ are distinct, $A>0$, and $N_A=\sum_{i,j}e^{-A^{2}(x_i-x_j)^{2}/4}$.
  The wave functions of $\ket{\varphi}$ and $U_{G,r}\ket{\varphi}$ are sums of positive Gaussian functions.
  To prove Eq.~\eqref{eq:purity_increase}, we use Eq.~\eqref{eq:nu_general_boson} to compare the two convolution purities as $A\to\infty$ at fixed $N$.
  Using $D_{\rm b}(a,0)^{\dagger}D_{\rm b}(q,p)D_{\rm b}(a,0)=e^{ipa}D_{\rm b}(q,p)$ and the vacuum characteristic function $\chi_0(q,p)=e^{-(q^2+p^2)/4}$, we obtain
  \begin{equation}
    \chi_{\varphi}(q,p)=\frac{N}{N_A}e^{-(q^{2}+p^{2})/4}f_A(p)+R_A(q,p).
  \end{equation}
  Here, $f_A(p)=N^{-1}\sum_i e^{ipAx_i}$, and $R_A$ collects the terms with $i\neq j$ and satisfies
  \begin{equation}
    |R_A(q,p)|\leq\frac{1}{N_A}\sum_{i\neq j}e^{-\{[q-A(x_j-x_i)]^{2}+p^{2}\}/4}.
  \end{equation}

  We first consider identical inputs.
  For every pair of sequences $(i_r)_r$ and $(\bar{i}_r)_r$ in $\{1,\ldots,N\}^n$, we impose two conditions on the $x_i$.
  Condition (a) requires that the ratios $(x_{\bar{i}_r}-x_{i_r})/O_{1r}$ are not all equal whenever $i_r\neq\bar{i}_r$ for every $r$.
  Condition (b) requires that $\sum_rO_{1r}x_{i_r}=\sum_rO_{1r}x_{\bar{i}_r}$ implies $\sum_{r:i_r=i}O_{1r}=\sum_{r:\bar{i}_r=i}O_{1r}$ for every $i$.
  Both conditions can be satisfied, while keeping the $x_i$ distinct, by avoiding finitely many proper linear subspaces of $\mathbb{R}^N$.
  For condition (a), properness follows by setting $x_{\bar{i}_1}=1$ and all other $x_i=0$.
  The resulting ratios lie in $\{0,\pm1/O_{1r}\}$, with the first equal to $1/O_{11}$.
  They could all coincide only if all $O_{1r}$ were equal, contrary to the assumption.
  With this choice of distinct $x_i$ satisfying condition (a), every term containing $R_A$ in the expansion of $\prod_r|\chi_\varphi(O_{1r}q,O_{1r}p)|^2$ has Gaussian factors in $q$ whose centers are separated by distances proportional to $A$.
  The integrals of these terms therefore vanish as $A\to\infty$.
  Using $\sum_rO_{1r}^2=1$ and $N_A\to N$, the remaining term gives
  \begin{align}
    \lim_{A\to\infty}\nu_n(\varphi)
     & =\lim_{A\to\infty}\int\!\!\frac{dq\,dp}{2\pi}e^{-(q^{2}+p^{2})/2}\prod_{r=1}^{n}\qty|f_A(O_{1r}p)|^{2}\nonumber \\
     & =\lim_{A\to\infty}\frac{1}{N^{2n}}\sum_{(i_r),(\bar{i}_r)}e^{-A^{2}[\sum_rO_{1r}(x_{i_r}-x_{\bar{i}_r})]^{2}/2}.
  \end{align}
  This limit is $N^{-2n}$ times the number of pairs satisfying $\sum_rO_{1r}x_{i_r}=\sum_rO_{1r}x_{\bar{i}_r}$.
  By condition (b) and the positivity of the weights, the two sequences contain the same set of labels, so at most $n$ distinct labels occur.
  If exactly $n$ occur, $(\bar{i}_r)_r$ is a permutation of $(i_r)_r$ within groups of replicas with equal $O_{1r}$.
  Let $m_\gamma$ denote the sizes of these groups.
  Pairs with fewer distinct labels contribute $O(N^{n-1})$, giving
  \begin{equation}
    \lim_{A\to\infty}\nu_n(\varphi)\leq\frac{1}{N^{2n}}\qty[\prod_\gamma m_\gamma!\,N^{n}+O(N^{n-1})].
  \end{equation}

  We next consider the squeezed inputs $U_{G,r}\ket{\varphi}$.
  By Eq.~\eqref{eq:boson_rescale}, the factors in Eq.~\eqref{eq:nu_general_boson} become $|\chi_\varphi(O_{1r}^2q,p)|^2$.
  Restricting $q$ to a fixed bounded interval gives a lower bound because the integrand is nonnegative.
  On this interval, $R_A(O_{1r}^2q,p)$ vanishes uniformly as $A\to\infty$.
  Taking $A\to\infty$ and then enlarging the interval to the whole real line, we obtain
  \begin{align}
     & \liminf_{A\to\infty}\nu_n(U_{G,1}\varphi,\ldots,U_{G,n}\varphi)\nonumber \\
     & \quad\geq\lim_{A\to\infty}\int\!\!\frac{dq\,dp}{2\pi}e^{-(\sum_rO_{1r}^4q^2+np^2)/2}\qty|f_A(p)|^{2n}\nonumber \\
     & \quad=\frac{1}{\sqrt{n\sum_rO_{1r}^{4}}}\nonumber \\
     & \qquad\times\lim_{A\to\infty}\frac{1}{N^{2n}}\sum_{(i_r),(\bar{i}_r)}e^{-A^{2}[\sum_r(x_{i_r}-x_{\bar{i}_r})]^{2}/(2n)}.
  \end{align}
  The surviving pairs include the $n!$ permutations of every ordered tuple with $n$ distinct entries.
  Hence
  \begin{align}
     & \liminf_{A\to\infty}\nu_n(U_{G,1}\varphi,\ldots,U_{G,n}\varphi)\nonumber \\
     & \quad\geq\frac{n!\,N(N-1)\cdots(N-n+1)}{N^{2n}\sqrt{n\sum_rO_{1r}^{4}}}.
  \end{align}
  As $N\to\infty$, the ratio of this lower bound to the upper bound for identical inputs approaches
  \begin{equation}
    \frac{1}{\sqrt{n\sum_rO_{1r}^{4}}}\,\frac{n!}{\prod_\gamma m_\gamma!}>\frac{n}{\sqrt{n}}>1.
  \end{equation}
  Here, we used $\sum_rO_{1r}^4<(\sum_rO_{1r}^2)^2=1$ and $\prod_\gamma m_\gamma!\leq(n-1)!$, the latter following because the weights are not all equal.
  We can therefore choose $N$ so that the lower bound exceeds the upper bound.
  For this fixed $N$, Eq.~\eqref{eq:purity_increase} holds for sufficiently large finite $A$.
\end{proof}

We now prove Theorem~\ref{thm:boson-convolution-uniqueness}.
The rotation of position eigenstates described above and the positive wave functions in Lemma~\ref{lem:boson_purity_increase} satisfy the assumptions of Lemma~\ref{lem:feedforward}.
These lemmas exclude monotonicity of $M_n$ unless all $|O_{1r}|$ are equal.
Conversely, if they are equal, orthogonality gives $|O_{1r}|=1/\sqrt{n}$.
Using $|\chi_\rho(-\bm{u})|=|\chi_\rho(\bm{u})|$, Eq.~\eqref{eq:nu_general_boson} then reduces to the bosonic convolution purity~\eqref{eq:purity_boson_chi}.
Thus, $M_n$ and $M_n^{\rm lin}$ coincide with the bosonic MRE and its linear version.
By Theorem~\ref{thm:boson-protocol-monotonicity}, they are a monotone and a strong monotone, respectively, proving Theorem~\ref{thm:boson-convolution-uniqueness}.

\section{Proof of the exponential lower bound for the fermionic Gaussian rank}\label{app:rank_proof}

In this appendix, we prove Theorem~\ref{thm:gaussian-rank}.
The proof builds on an observation of Refs.~\cite{reardon-smith2026exponential,wei2026robust} about the Majorana monomials $\Gamma_{A}$ of Eq.~\eqref{eq:chi_fermion_expansion}.
For a pure Gaussian state $\ket{\varphi}$, the $\binom{2L}{p}$ vectors $\Gamma_{A}\ket{\varphi}$ with $|A|=p$ span a space whose dimension is exponentially smaller than $\binom{2L}{p}$ when $p$ is proportional to $L$.
Therefore, if $\ket{\psi}$ is close to a superposition of a small number of Gaussian states, the vectors $\Gamma_{A}\ket{\psi}$ lie almost entirely in a space of low dimension.
These references excluded this possibility for tensor products of non-Gaussian states by exploiting their product structure.

Here we apply this argument to an arbitrary normalized pure state $\ket{\psi}$ of $L$ fermionic modes.
The central idea is to bound the following quantity
\begin{equation}
  \sum_{|A|=|B|=p}|\mel{\psi}{\Gamma_{A}^{\dagger}\Gamma_{B}}{\psi}|^{2}
  \label{eq:rank_sum}
\end{equation}
from both sides.
Since unit vectors lying almost entirely in a low-dimensional space have large inner products, this sum is large when $\chi_{\delta}(\psi)$ is small, which we show in Lemma~\ref{lem:rank_lower}.
Lemma~\ref{lem:rank_upper} shows that this sum is small when $M_{2}(\psi)$ is large.
Combining the two lemmas, $\chi_{\delta}(\psi)$ cannot be small when $M_{2}(\psi)$ is large, which yields the theorem.

Throughout this appendix, $\mathcal{C}_{2}$, $\nu_{2}$, and $M_{2}$ are constructed from the balanced beam splitter $U=U(\pi/4)$ of Eq.~\eqref{eq:O_case1}, which gives the fermionic MRE with $n=2$.
We also denote $\rho=\dyad{\psi}$.

\begin{lemma}
  \label{lem:rank_lower}
  For any integer $0\leq p\leq L$ and any $0\leq\delta<1$,
  \begin{equation}
    \sum_{|A|=|B|=p}|\mel{\psi}{\Gamma_{A}^{\dagger}\Gamma_{B}}{\psi}|^{2}
    \geq\frac{(1-\delta^{2})^{2}}{\chi_{\delta}(\psi)\,d_{p}}\binom{2L}{p}^{2},
    \label{eq:rank_lower}
  \end{equation}
  where $d_{p}=\sum_{q=0}^{\lfloor p/2\rfloor}\binom{L}{p-2q}$.
\end{lemma}

\begin{proof}
  Following Refs.~\cite{reardon-smith2026exponential,wei2026robust}, we first show that the vectors $\Gamma_{A}\ket{\varphi}$ with $|A|=p$ span a space of dimension at most $d_{p}$ for any pure Gaussian state $\ket{\varphi}$.
  Since a pure Gaussian state is obtained from a Fock state by a Gaussian unitary, we derive this bound for a Fock state and then show that it holds for a general pure Gaussian state.
  For a Fock state $\ket{\bm{m}}$, the operators $\gamma_{j}$ and $\bar{\gamma}_{j}$ of Eq.~\eqref{eq:Majorana} both flip the occupation number of mode $j$ up to a phase, whereas their product $\gamma_{j}\bar{\gamma}_{j}$ leaves it unchanged.
  Therefore, if $\Gamma_{A}$ contains both $\gamma_{j}$ and $\bar{\gamma}_{j}$ for exactly $q$ modes, then $\Gamma_{A}\ket{\bm{m}}$ is, up to a phase, a Fock state that differs from $\ket{\bm{m}}$ in $p-2q$ modes.
  There are $\binom{L}{p-2q}$ such Fock states for each $q$, so that the vectors $\Gamma_{A}\ket{\bm{m}}$ span a space of dimension at most
  \begin{equation}
    \sum_{q=0}^{\lfloor p/2\rfloor}\binom{L}{p-2q}=d_{p}.
  \end{equation}

  We now show that the same bound holds for a general pure Gaussian state $\ket{\varphi}=U_{G}\ket{\bm{m}}$, where $\ket{\bm{m}}$ is the vacuum for even parity and a Fock state with one occupied mode for odd parity.
  The Gaussian unitary acts as $U_{G}^{\dagger}\bm{r}U_{G}=O\bm{r}$ with a real orthogonal matrix $O$.
  Since the components of $O\bm{r}$ anticommute with each other, expanding the product $U_{G}^{\dagger}\Gamma_{A}U_{G}$ gives
  \begin{equation}
    U_{G}^{\dagger}\Gamma_{A}U_{G}=\sum_{|B|=p}\det(O_{AB})\,\Gamma_{B},
  \end{equation}
  where $O_{AB}$ is the submatrix of $O$ with rows in $A$ and columns in $B$.
  Hence, every $\Gamma_{A}\ket{\varphi}=U_{G}(U_{G}^{\dagger}\Gamma_{A}U_{G})\ket{\bm{m}}$ lies in the image under $U_{G}$ of the space spanned by the vectors $\Gamma_{B}\ket{\bm{m}}$ with $|B|=p$.
  This image has dimension at most $d_{p}$.

  We next show that the vectors $\Gamma_{A}\ket{\psi}$ lie almost entirely in a space of dimension at most $\chi_{\delta}(\psi)d_{p}$.
  Let $N=\chi_{\delta}(\psi)$.
  By the definitions of the Gaussian rank~\eqref{eq:Gaussian_rank} and the $\delta$-approximate Gaussian rank~\eqref{eq:approx_Gaussian_rank}, there is a vector $\ket{\psi'}=\sum_{j=1}^{N}c_{j}\ket{\varphi_{j}}$ with pure Gaussian states $\ket{\varphi_{j}}$ and $\norm{\ket{\psi}-\ket{\psi'}}\leq\delta$.
  Let $P$ be the orthogonal projector onto the space spanned by all $\Gamma_{A}\ket{\varphi_{j}}$ with $|A|=p$ and $1\leq j\leq N$.
  Moreover, since the trace of $P$ is the dimension of this space, we have
  \begin{equation}
    \tr P\leq Nd_{p}.\label{eq:dim_bound}
  \end{equation}
  Since $(I-P)\Gamma_{A}\ket{\psi'}=0$ and $\Gamma_{A}$ is unitary, we have
  \begin{equation}
    \norm{(I-P)\Gamma_{A}\ket{\psi}}=\norm{(I-P)\Gamma_{A}(\ket{\psi}-\ket{\psi'})}\leq\delta.
  \end{equation}
  Because $\Gamma_{A}\ket{\psi}$ is a unit vector, this means
  \begin{equation}
    \mel{\psi}{\Gamma_{A}^{\dagger}P\Gamma_{A}}{\psi}\geq1-\delta^{2},\label{eq:proj_bound}
  \end{equation}
  which shows that the vectors $\Gamma_{A}\ket{\psi}$ lie almost entirely in the range of $P$, whose dimension is at most $\chi_{\delta}(\psi)d_{p}$.

  Finally, we sum the inequality~\eqref{eq:proj_bound} over the $\binom{2L}{p}$ subsets $A$ and apply the Cauchy-Schwarz inequality for the Hilbert-Schmidt inner product to obtain
  \begin{align}
    (1-\delta^{2})\binom{2L}{p}
     & \leq\tr\!\qty[P\sum_{|A|=p}\Gamma_{A}\rho\Gamma_{A}^{\dagger}]\nonumber                           \\
     & \leq(\tr P)^{1/2}\qty(\tr\!\qty[\qty(\sum_{|A|=p}\Gamma_{A}\rho\Gamma_{A}^{\dagger})^{2}])^{1/2}.
  \end{align}
  Expanding the square, we have
  \begin{equation}
    \tr\!\qty[\qty(\sum_{|A|=p}\Gamma_{A}\rho\Gamma_{A}^{\dagger})^{2}]=\sum_{|A|=|B|=p}|\mel{\psi}{\Gamma_{A}^{\dagger}\Gamma_{B}}{\psi}|^{2}.
  \end{equation}
  Using Eq.~\eqref{eq:dim_bound} and squaring both sides, we obtain the lower bound~\eqref{eq:rank_lower}.
\end{proof}

\begin{lemma}
  \label{lem:rank_upper}
  Let $x_{*}=(\sqrt{2}-1)/(1+\sqrt{2\sqrt{2}-2})=0.2168\ldots$, which is the solution of $2x_{*}/(1+x_{*}^{2})=\tan(\pi/8)$ in $0<x_{*}<1$.
  Then, for any integer $0\leq p\leq L$,
  \begin{equation}
    \sum_{|A|=|B|=p}|\mel{\psi}{\Gamma_{A}^{\dagger}\Gamma_{B}}{\psi}|^{2}
    \leq\frac{(1+6x_{*}^{2}+x_{*}^{4})^{L}}{x_{*}^{2p}}\,e^{-M_{2}(\psi)/4}.
    \label{eq:rank_upper}
  \end{equation}
\end{lemma}

\begin{proof}
  We first bound the left-hand side by a sum of $|\mel{\psi}{\Gamma_{C}}{\psi}|^{2}$ over all subsets $C$.
  Since the Majorana operators common to $A$ and $B$ square to the identity, we have $\Gamma_{A}^{\dagger}\Gamma_{B}=\pm\Gamma_{A\triangle B}$.
  Here $A\triangle B$ is the set of indices contained in exactly one of $A$ and $B$.
  For any $x>0$, adding the nonnegative terms with other values of $|A|$ and $|B|$ gives
  \begin{align}
     & x^{2p}\sum_{|A|=|B|=p}|\mel{\psi}{\Gamma_{A}^{\dagger}\Gamma_{B}}{\psi}|^{2}\nonumber         \\
     & \quad\leq\sum_{A,B}x^{|A|+|B|}|\mel{\psi}{\Gamma_{A\triangle B}}{\psi}|^{2}\nonumber          \\
     & \quad=(1+x^{2})^{2L}\sum_{C}\qty(\frac{2x}{1+x^{2}})^{|C|}|\mel{\psi}{\Gamma_{C}}{\psi}|^{2}.
    \label{eq:rank_weight}
  \end{align}
  The last line collects the pairs $(A,B)$ with $A\triangle B=C$.
  An index in $C$ belongs to exactly one of $A$ and $B$ and contributes the factor $2x$.
  An index outside $C$ belongs to both or neither and contributes the factor $1+x^{2}$.

  We next show that the overlap between $\ket{\psi,\psi}$ and $U\ket{\psi,\psi}$ is a sum of the same form.
  We note that a related quantity has recently been introduced as the replica twist entropy in Ref.~\cite{sun2026characterizing}.
  By Eq.~\eqref{eq:O_case1} with $\zeta=\pi/4$, the unitary $U$ acts on the Majorana operators $r_{s,l}$ of replica $s$ as $U^{\dagger}r_{1,l}U=(r_{1,l}-r_{2,l})/\sqrt{2}$ and $U^{\dagger}r_{2,l}U=(r_{1,l}+r_{2,l})/\sqrt{2}$.
  This action determines $U$ up to a phase, which does not affect $\mathcal{C}_{2}$.
  Using $(r_{1,l}r_{2,l})^{2}=-I$, we fix the phase by
  \begin{align}
    U & =\prod_{l=1}^{2L}\exp\qty(-\frac{\pi}{8}r_{1,l}r_{2,l})\nonumber            \\
      & =\prod_{l=1}^{2L}\qty(\cos\frac{\pi}{8}-\sin\frac{\pi}{8}\,r_{1,l}r_{2,l}).
  \end{align}
  We expand this product by choosing the second term for $l$ in a subset $C$.
  Reordering the Majorana operators of the two replicas gives $\prod_{l\in C}r_{1,l}r_{2,l}=\Gamma_{C}^{\dagger}\gtensor\Gamma_{C}$.
  For even $|C|$, the graded tensor product~\eqref{eq:graded_tensor} shows that the expectation value of this operator in $\ket{\psi,\psi}$ is $\mel{\psi}{\Gamma_{C}^{\dagger}}{\psi}\mel{\psi}{\Gamma_{C}}{\psi}=|\mel{\psi}{\Gamma_{C}}{\psi}|^{2}$.
  For odd $|C|$, this expectation value and $\mel{\psi}{\Gamma_{C}}{\psi}$ both vanish because $\ket{\psi}$ has definite parity.
  We thus obtain
  \begin{equation}
    \mel{\psi,\psi}{U}{\psi,\psi}=\cos^{2L}\frac{\pi}{8}\sum_{C}\tan^{|C|}\frac{\pi}{8}\,\qty|\mel{\psi}{\Gamma_{C}}{\psi}|^{2}.
  \end{equation}
  Setting $x=x_{*}$ in the bound~\eqref{eq:rank_weight} makes the two sums identical.
  Together with $(1+x_{*}^{2})^{2}/\cos^{2}(\pi/8)=(1+x_{*}^{2})^{2}+4x_{*}^{2}=1+6x_{*}^{2}+x_{*}^{4}$, this gives
  \begin{align}
     & \sum_{|A|=|B|=p}|\mel{\psi}{\Gamma_{A}^{\dagger}\Gamma_{B}}{\psi}|^{2}\nonumber        \\
     & \quad\leq\frac{(1+6x_{*}^{2}+x_{*}^{4})^{L}}{x_{*}^{2p}}\mel{\psi,\psi}{U}{\psi,\psi}.
  \end{align}

  It remains to bound this overlap by the convolution purity.
  The sign rule~\eqref{eq:sign_rule} gives $\ket{\psi,\psi}\!\bra{\psi,\psi}=\rho\gtensor\rho$, and $\rho\gtensor(I-\rho)$ is positive semidefinite.
  Since $\mathcal{C}_{2}(\rho)$ is the reduced density operator of $U\ket{\psi,\psi}$ on replica $1$, we have
  \begin{align}
    \mel{\psi,\psi}{U}{\psi,\psi}^{2}
     & =\mel{\psi,\psi}{U^{\dagger}(\rho\gtensor\rho)U}{\psi,\psi}\nonumber  \\
     & \leq\mel{\psi,\psi}{U^{\dagger}(\rho\gtensor I)U}{\psi,\psi}\nonumber \\
     & =\tr\!\qty[\rho\,\mathcal{C}_{2}(\rho)]\nonumber                      \\
     & \leq\qty(\tr\rho^{2}\,\tr\mathcal{C}_{2}(\rho)^{2})^{1/2}
    =\nu_{2}(\psi)^{1/2},
  \end{align}
  where the last inequality is the Cauchy-Schwarz inequality.
  The overlap is thus at most $\nu_{2}(\psi)^{1/4}=e^{-M_{2}(\psi)/4}$, and we obtain the upper bound~\eqref{eq:rank_upper}.
\end{proof}

\begin{proof}[Proof of Theorem~\ref{thm:gaussian-rank}]
  The lower bound~\eqref{eq:rank_lower} and the upper bound~\eqref{eq:rank_upper} on the same sum give, for any integer $0\leq p\leq L$,
  \begin{equation}
    \chi_{\delta}(\psi)\geq\frac{(1-\delta^{2})^{2}x_{*}^{2p}}{d_{p}(1+6x_{*}^{2}+x_{*}^{4})^{L}}\binom{2L}{p}^{2}e^{M_{2}(\psi)/4}.
  \end{equation}
  To obtain the density form, we set $p=\lfloor\alpha L\rfloor$ with a constant $0<\alpha<1/2$.
  With the binary entropy $h(\beta)=-\beta\ln\beta-(1-\beta)\ln(1-\beta)$, we have $\ln\binom{2L}{p}=2Lh(\alpha/2)+O(\ln L)$.
  Since $\binom{L}{p}$ is the largest term of $d_{p}$, we also have $\ln d_{p}=Lh(\alpha)+O(\ln L)$.
  For a sequence $\ket{\psi_{L}}$ and a fixed $\delta<1$, we thus obtain
  \begin{align}
    \frac{1}{L}\ln\chi_{\delta}(\psi_{L})
     & \geq\frac{M_{2}(\psi_{L})}{4L}+4h\qty(\frac{\alpha}{2})-h(\alpha)+2\alpha\ln x_{*}\nonumber \\
     & \quad-\ln(1+6x_{*}^{2}+x_{*}^{4})+O\qty(\frac{\ln L}{L}).
  \end{align}
  The right-hand side is largest at the solution of $\alpha(1-\alpha)=x_{*}^{2}(2-\alpha)^{2}$ in $0<\alpha<1/2$, which is
  \begin{equation}
    \alpha_{*}=\frac{1+4x_{*}^{2}-\sqrt{1-8x_{*}^{2}}}{2(1+x_{*}^{2})}=0.1902\ldots.
  \end{equation}
  At $\alpha=\alpha_{*}$, this equation eliminates $\ln x_{*}$, and the terms on the right-hand side other than $M_{2}(\psi_{L})/(4L)$ and $O(\ln L/L)$ reduce to $-c_{*}$ with
  \begin{equation}
    c_{*}=\ln\frac{(1-\alpha_{*}/2)^{4}(1+6x_{*}^{2}+x_{*}^{4})}{1-\alpha_{*}}=0.0615\ldots.
  \end{equation}
  Taking the limit inferior of both sides gives the bound~\eqref{eq:rank_bound}.
\end{proof}

\section{Boundary condition imposed by the maximally entangled state}\label{app:bc_MES}

In Sec.~\ref{subsec:path_integral} we expressed the Liouville-space
partition function as a path integral on the imaginary-time interval
$\tau\in[0,\beta/2]$, with $\sket{I}$ at both ends. 
In this appendix, we determine the boundary condition that $\sket{I}$ imposes on the field
configurations.
The bosonic and fermionic cases are treated
simultaneously through the parameter $\eta$ of Eq.~\eqref{eq:CCR_CAR_prelim},
with $\eta=+1$ for bosons and $\eta=-1$ for fermions.

\subsection{Liouville coherent states}
\label{subsec:Liouville_coherent}
To begin with, we define the coherent states for both bosons and fermions as the displacement of the vacuum,
\begin{align}
  \ket{\bm{\alpha}} = D(\bm{\alpha})\ket{0}^{\otimes L}.
  \label{eq:coh_BF}
\end{align}
Here, $D(\bm{\alpha})=\exp[\sum_{i}(a_{i}^{\dag}\alpha_{i}-\alpha_{i}^{*}a_{i})]$ labels the displacement operators of Sec.~\ref{subsec:boson-fermion} by a complex vector $\bm{\alpha}$ for bosons and by a complex Grassmann vector $\bm{\alpha}$ for fermions.
It equals $D_{\mathrm{b}}(\bm{u})$ with $\bm{\alpha}=(\bm{q}+i\bm{p})/\sqrt{2}$ and $D_{\mathrm{f}}(\bm{u})$ with $\bm{\alpha}=(\bm{q}-i\bm{p})/2$.
The inner product between coherent states is
\begin{align}
  \braket{\bm{\alpha}_{2}}{\bm{\alpha}_{1}}
  = \exp\!\qty[
            -\tfrac{1}{2}\bm{\alpha}_{2}^{*}\!\cdot\!\bm{\alpha}_{2}
            -\tfrac{1}{2}\bm{\alpha}_{1}^{*}\!\cdot\!\bm{\alpha}_{1}
            +\bm{\alpha}_{2}^{*}\!\cdot\!\bm{\alpha}_{1}
          ],
  \label{eq:coherent_overlap}
\end{align}
where $\bm{\alpha}^{*}\!\cdot\!\bm{\beta}\equiv\sum_{i}\alpha_{i}^{*}\beta_{i}$.
The completeness relation reads
$\int d\mu(\bm{\alpha})\ket{\bm{\alpha}}\bra{\bm{\alpha}}=I$,
with the canonical measure $d\mu(\bm{\alpha})=\pi^{-L}d^{2L}\bm{\alpha}$
for bosons and the Berezin measure $d\mu(\bm{\alpha})=d\bm{\alpha}$
with
$d\bm{\alpha}=d\alpha_{L}^{*}\,d\alpha_{L}\cdots
  d\alpha_{1}^{*}\,d\alpha_{1}$ for fermions.

We extend this construction to $\bar{\mathscr{H}}$ by introducing a
second copy of coherent-state labels
$\tilde{\bm{\alpha}}=(\tilde\alpha_{1},\ldots,\tilde\alpha_{L})^{\rm T}$
for the auxiliary sector and defining the Liouville coherent state
$\sket{\bm{\alpha},\tilde{\bm{\alpha}}}$ by the eigenvalue equations
\begin{align}
  a_{i}\sket{\bm{\alpha},\tilde{\bm{\alpha}}}
   & = \alpha_{i}\sket{\bm{\alpha},\tilde{\bm{\alpha}}},
  \nonumber                                                    \\
  \tilde{a}_{i}\sket{\bm{\alpha},\tilde{\bm{\alpha}}}
   & = \tilde\alpha_{i}\sket{\bm{\alpha},\tilde{\bm{\alpha}}},
  \label{eq:Liouville_CS_eigen}
\end{align}
with the normalization fixed by
$\sbk{\bm{\alpha},\tilde{\bm{\alpha}}}{\bm{\alpha},\tilde{\bm{\alpha}}}=1$.
A direct calculation, using
Eqs.~\eqref{eq:superop_a}--\eqref{eq:superop_atilde}, gives the inner
product
\begin{align}
   & \sbk{\bm{\alpha}_{2},\tilde{\bm{\alpha}}_{2}}{\bm{\alpha}_{1},\tilde{\bm{\alpha}}_{1}}
  \nonumber                                                                                 \\
   & = \exp\!\Bigl[
               -\tfrac{1}{2}|\bm{\alpha}_{2}|^{2}
               -\tfrac{1}{2}|\bm{\alpha}_{1}|^{2}
               +\bm{\alpha}_{2}^{*}\!\cdot\!\bm{\alpha}_{1}
               \nonumber\\
               & \quad
               -\tfrac{1}{2}|\tilde{\bm{\alpha}}_{2}|^{2}
               -\tfrac{1}{2}|\tilde{\bm{\alpha}}_{1}|^{2}
               +\tilde{\bm{\alpha}}_{2}^{*}\!\cdot\!\tilde{\bm{\alpha}}_{1}
               \Bigr]
  \label{eq:Liouville_CS_inner}
\end{align}
and the corresponding completeness relation
\begin{align}
  \int\!d\mu(\bm{\alpha})\,d\mu(\tilde{\bm{\alpha}})\,
  \sket{\bm{\alpha},\tilde{\bm{\alpha}}}
  \sbra{\bm{\alpha},\tilde{\bm{\alpha}}}
  = I_{\bar{\mathscr{H}}}.
  \label{eq:Liouville_CS_completeness}
\end{align}
Here $|\bm{\alpha}|^{2}\equiv\bm{\alpha}^{*}\!\cdot\!\bm{\alpha}$.

The overlaps with the MES~\eqref{eq:MES} are obtained from the
eigenvalue properties~\eqref{eq:Liouville_CS_eigen},
\begin{align}
  \sbk{\bm{\alpha},\tilde{\bm{\alpha}}}{I}
   & = \exp\!\qty[
               -\tfrac{1}{2}|\bm{\alpha}|^{2}
               -\tfrac{1}{2}|\tilde{\bm{\alpha}}|^{2}
               +\bm{\alpha}^{*}\!\cdot\!\tilde{\bm{\alpha}}^{*}
             ],
  \label{eq:CS_I_right}                                        \\
  \sbk{I}{\bm{\alpha},\tilde{\bm{\alpha}}}
   & = \exp\!\qty[
               -\tfrac{1}{2}|\bm{\alpha}|^{2}
               -\tfrac{1}{2}|\tilde{\bm{\alpha}}|^{2}
               + \eta\,\bm{\alpha}\!\cdot\!\tilde{\bm{\alpha}}
             ].
  \label{eq:CS_I_left}
\end{align}
The sign $\eta$ in Eq.~\eqref{eq:CS_I_left} originates from the
Grassmann reversal that takes
$\bm{\alpha}^{*}\!\cdot\!\tilde{\bm{\alpha}}^{*}$ to its complex
conjugate; it leads to the antiperiodic boundary condition for
fermions.

\subsection{Trotter decomposition}
\label{subsec:Trotter}

We discretize the imaginary-time interval $[0,\beta/2]$ into $N$ slices
of width $\varepsilon=\beta/(2N)$ and insert the
completeness relation~\eqref{eq:Liouville_CS_completeness} between
successive Trotter factors. Writing $H$ as a normal-ordered
function of the elementary operators,
$H = H(\bm a^{\dag},\bm a)$, the standard
coherent-state matrix element gives
\begin{align}
   & \sbra{\bm{\alpha}_{k},\tilde{\bm{\alpha}}_{k}}\,
  e^{-\varepsilon H}\otimes e^{-\varepsilon \tilde{H}}\,
  \sket{\bm{\alpha}_{k-1},\tilde{\bm{\alpha}}_{k-1}}
  \nonumber                                                                  \\
   & = e^{-\varepsilon[
             H(\bm{\alpha}_{k}^{*},\bm{\alpha}_{k-1})
             +\tilde{H}(\tilde{\bm{\alpha}}_{k}^{*},\tilde{\bm{\alpha}}_{k-1})
           ]}\,
  \sbk{\bm{\alpha}_{k},\tilde{\bm{\alpha}}_{k}}{\bm{\alpha}_{k-1},\tilde{\bm{\alpha}}_{k-1}}.
  \label{eq:Trotter_step}
\end{align}
In Eq.~\eqref{eq:Trotter_step}, $H$ and $\tilde{H}$ denote operators on the left-hand side and their normal-ordered coherent-state symbols on the right-hand side, respectively.
For example, $H(\bm{\alpha}_{k}^{*},\bm{\alpha}_{k-1})$ is obtained from $H(\bm a^{\dag},\bm a)$ by replacing $\bm a^{\dag}\to\bm\alpha_k^*$ and $\bm a\to\bm\alpha_{k-1}$.

Combining Eqs.~\eqref{eq:Liouville_CS_inner},
\eqref{eq:CS_I_right}, \eqref{eq:CS_I_left}, and \eqref{eq:Trotter_step}, the discretized
partition function takes the form
\begin{align}
  Z = \int\!\prod_{k=0}^{N}d\mu(\bm{\alpha}_{k})\,d\mu(\tilde{\bm{\alpha}}_{k})\;
  e^{-S},
  \label{eq:Z_discrete}
\end{align}
with the action
\begin{align}
  S
   & =- \eta\,\bm{\alpha}_{N}\!\cdot\!\tilde{\bm{\alpha}}_{N}
  -\bm{\alpha}_{0}^{*}\!\cdot\!\tilde{\bm{\alpha}}_{0}^{*}
  \nonumber                                                   \\
   & \quad
  -\sum_{k=1}^{N}\qty(
  \bm{\alpha}_{k}^{*}\!\cdot\!\bm{\alpha}_{k-1}
  -\bm{\alpha}_{k}^{*}\!\cdot\!\bm{\alpha}_{k})
  +\bm{\alpha}_{0}^{*}\!\cdot\!\bm{\alpha}_{0}
  \nonumber                                                   \\
   & \quad
  -\sum_{k=1}^{N}\qty(
  \tilde{\bm{\alpha}}_{k}^{*}\!\cdot\!\tilde{\bm{\alpha}}_{k-1}
  -\tilde{\bm{\alpha}}_{k}^{*}\!\cdot\!\tilde{\bm{\alpha}}_{k})
  +\tilde{\bm{\alpha}}_{0}^{*}\!\cdot\!\tilde{\bm{\alpha}}_{0}
  \nonumber                                                   \\
   & \quad
  +\varepsilon\sum_{k=1}^{N}\qty[
                              H(\bm{\alpha}_{k}^{*},\bm{\alpha}_{k-1})
                              +\tilde{H}(\tilde{\bm{\alpha}}_{k}^{*},\tilde{\bm{\alpha}}_{k-1})
                            ].
  \label{eq:S_unfolded}
\end{align}
The first two terms come from
$\sbk{I}{\bm{\alpha}_{N},\tilde{\bm{\alpha}}_{N}}$ and
$\sbk{\bm{\alpha}_{0},\tilde{\bm{\alpha}}_{0}}{I}$, respectively, and
encode the boundary conditions imposed by the MES.

\subsection{Unfolding to a single time axis}
\label{subsec:folding}

To compare Eq.~\eqref{eq:S_unfolded} with the conventional
single-line path integral, we relabel the auxiliary fields as a
continuation of the original ones along an extended time axis
$k=0,1,\ldots,2N+1$. For both statistics, define
\begin{align}
  (\bm{\alpha}_{N+k},\bm{\alpha}_{N+k}^{*})
  \equiv (\eta\,\tilde{\bm{\alpha}}_{N+1-k}^{*},\,\tilde{\bm{\alpha}}_{N+1-k}),
  \label{eq:fold_relabel}
\end{align}
for $k=1,2,\ldots,N+1$. The sign $\eta$ is the unique choice that
matches the form of the bilinear $\bm{\alpha}^{*}\!\cdot\!\bm{\alpha}$
appearing in $S$: under
Eq.~\eqref{eq:fold_relabel}, the auxiliary blocks of $S$ become
\begin{align}
   & -\sum_{k=1}^{N}\qty(
  \tilde{\bm{\alpha}}_{k}^{*}\!\cdot\!\tilde{\bm{\alpha}}_{k-1}
  -\tilde{\bm{\alpha}}_{k}^{*}\!\cdot\!\tilde{\bm{\alpha}}_{k})
  +\tilde{\bm{\alpha}}_{0}^{*}\!\cdot\!\tilde{\bm{\alpha}}_{0}
  \nonumber                 \\
   & = -\sum_{k=1}^{N}\qty(
  \bm{\alpha}_{N+k+1}^{*}\!\cdot\!\bm{\alpha}_{N+k}
  -\bm{\alpha}_{N+k}^{*}\!\cdot\!\bm{\alpha}_{N+k})
  \nonumber                 \\
   & \quad
  +\eta\,\bm{\alpha}_{2N+1}\!\cdot\!\bm{\alpha}_{2N+1}^{*},
  \label{eq:aux_relabeled}
\end{align}
while the terms coupling $\bm{\alpha}$ and $\tilde{\bm{\alpha}}$ transform as
\begin{align}
  \eta\,\bm{\alpha}_{N}\!\cdot\!\tilde{\bm{\alpha}}_{N}
   & = \bm{\alpha}_{N+1}^{*}\!\cdot\!\bm{\alpha}_{N},
  \\
  \bm{\alpha}_{0}^{*}\!\cdot\!\tilde{\bm{\alpha}}_{0}^{*}
   & = \eta\,\bm{\alpha}_{0}^{*}\!\cdot\!\bm{\alpha}_{2N+1}.
  \label{eq:seam_relabeled}
\end{align}
We further assume that the auxiliary Hamiltonian symbol satisfies
$\tilde{H}(\eta\,\bm{\alpha}_{N+k},\bm{\alpha}_{N+k+1}^{*})
  =H(\bm{\alpha}_{N+k+1}^{*},\bm{\alpha}_{N+k})$ as a c-number identity
(this holds true for, e.g., a time-reversal-symmetric system $H^{*}=H$).
Substituting Eqs.~\eqref{eq:aux_relabeled}--\eqref{eq:seam_relabeled}
into Eq.~\eqref{eq:S_unfolded} and reorganizing,
\begin{align}
  S
   & = -\sum_{k=1}^{2N+1}\qty(
  \bm{\alpha}_{k}^{*}\!\cdot\!\bm{\alpha}_{k-1}
  -\bm{\alpha}_{k}^{*}\!\cdot\!\bm{\alpha}_{k})
  -\eta\,\bm{\alpha}_{0}^{*}\!\cdot\!\bm{\alpha}_{2N+1}
  +\bm{\alpha}_{0}^{*}\!\cdot\!\bm{\alpha}_{0}
  \nonumber                    \\
   & \quad
  +\varepsilon\sum_{k=1}^{N}\qty[
                              H(\bm{\alpha}_{k}^{*},\bm{\alpha}_{k-1})
                              +H(\bm{\alpha}_{N+k+1}^{*},\bm{\alpha}_{N+k})
                            ].
  \label{eq:S_folded}
\end{align}

The right-hand side of Eq.~\eqref{eq:S_folded} coincides exactly with
the action obtained by directly Trotterizing the original-space
representation $Z=\tr[(e^{-\varepsilon H})^{N}I(e^{-\varepsilon H})^{N}]$
in coherent states $\ket{\bm{\alpha}_{k}}$, provided one imposes the
boundary identification
\begin{align}
  \bm{\alpha}_{2N+1} = \eta\,\bm{\alpha}_{0}.
  \label{eq:PBC_APBC}
\end{align}
With Eq.~\eqref{eq:PBC_APBC} the boundary term
$-\eta\,\bm{\alpha}_{0}^{*}\!\cdot\!\bm{\alpha}_{2N+1}
  =-\bm{\alpha}_{0}^{*}\!\cdot\!\bm{\alpha}_{0}$ cancels against the
remaining $+\bm{\alpha}_{0}^{*}\!\cdot\!\bm{\alpha}_{0}$.
Equation~\eqref{eq:PBC_APBC} reproduces the periodic boundary
condition for bosons ($\eta=+1$) and the antiperiodic one for
fermions ($\eta=-1$).

\subsection{Continuum limit and the sewing condition}
\label{subsec:sewing}

After taking the limit $N\to\infty$ with $k\varepsilon\to\tau\in[0,\beta]$,
Eqs.~\eqref{eq:S_folded} and~\eqref{eq:PBC_APBC} become the continuum
action and (anti)periodic boundary condition of the ordinary
single-replica partition function on $\tau\in[0,\beta]$, while the
relabeling~\eqref{eq:fold_relabel} encodes how the auxiliary fields
$(\tilde{\bm{\alpha}},\tilde{\bm{\alpha}}^{*})(\tau)$ are stitched to
the original ones $(\bm{\alpha},\bm{\alpha}^{*})(\tau)$.
Concretely, $\tau\in[0,\beta/2]$ corresponds to the original branch
and $\tau\in[\beta/2,\beta]$ to the auxiliary branch, with
\begin{align}
  (\bm{\alpha},\bm{\alpha}^{*})(\tau)
  = (\eta\,\tilde{\bm{\alpha}}^{*},\,\tilde{\bm{\alpha}})(\beta-\tau),
  \quad \tau\in[\beta/2,\beta].
  \label{eq:fold_continuum}
\end{align}
The lower endpoint $\tau=0$ carries the ket $\sket{I}$, while the upper endpoint $\tau=\beta/2$ carries the bra $\sbra{I}$.
Continuity of the unfolded fields at $\tau=\beta/2$, together with the thermal identification at $\tau=0\sim\beta$, gives
\begin{align}
  \sket{I}\!:\;
   & (\bm{\alpha},\bm{\alpha}^{*})(0)
  = (\tilde{\bm{\alpha}}^{*},\,\eta\,\tilde{\bm{\alpha}})(0),
  \nonumber                                 \\
  \sbra{I}\!:\;
   & (\bm{\alpha},\bm{\alpha}^{*})(\beta/2)
  = (\eta\,\tilde{\bm{\alpha}}^{*},\,\tilde{\bm{\alpha}})(\beta/2).
  \label{eq:sewing}
\end{align}

For bosons ($\eta=+1$) the two conditions
in Eq.~\eqref{eq:sewing} coincide,
\begin{align}
  (\bm{\alpha},\bm{\alpha}^{*}) = (\tilde{\bm{\alpha}}^{*},\,\tilde{\bm{\alpha}}),
  \label{eq:sewing_boson}
\end{align}
which is equivalent to $\bm{q}-\tilde{\bm{q}}=\bm{0}$ and $\bm{p}+\tilde{\bm{p}}=\bm{0}$ for $\bm{\alpha}=(\bm{q}+i\bm{p})/\sqrt{2}$ and $\tilde{\bm{\alpha}}=(\tilde{\bm{q}}+i\tilde{\bm{p}})/\sqrt{2}$.
The MES therefore implements a mixed Dirichlet-Neumann boundary
condition that identifies the original coordinate with the auxiliary
coordinate while reversing the auxiliary momentum.

For fermions ($\eta=-1$), $\bm\alpha$ and $\bm\alpha^*$ are independent Grassmann integration variables.
The two conditions in Eq.~\eqref{eq:sewing} differ by a sign,
\begin{align}
  \sket{I}\!:\;
   & (\bm{\alpha},\bm{\alpha}^{*})
  = (\tilde{\bm{\alpha}}^{*},\,-\tilde{\bm{\alpha}}),
  \nonumber                        \\
  \sbra{I}\!:\;
   & (\bm{\alpha},\bm{\alpha}^{*})
  = (-\tilde{\bm{\alpha}}^{*},\,\tilde{\bm{\alpha}}),
  \label{eq:sewing_fermion}
\end{align}
i.e., the sewings imposed by $\sket{I}$ and $\sbra{I}$ are
  {distinct}. Although $\sket{I}\sbra{I}$ remains a (nontrivial)
projector onto the maximally entangled subspace up to normalization, the field
configurations stabilized by it just before and just after the
projection insertion are not the same.
This may be regarded as a manifestation of the
antiperiodic boundary condition~\eqref{eq:PBC_APBC} or the canonical anticommutation relations.

\section{Bulk-induced boundary RG flows}\label{app:bulk_boundary}

To examine the effect of the exactly marginal bulk perturbation ${\cal S}_{\lambda}$
on the conformal invariance of the boundary state $\sket{\mathcal{J}}$, we need to determine whether
this bulk perturbation, when brought close to the boundary, can
induce any relevant or marginal boundary operator in the boundary RG sense. To this end, we check whether the mismatch boundary operator
\begin{equation}
  \delta T(x)\equiv\lim_{y\to0}\left(T(x+iy)-\bar{T}(x-iy)\right)
\end{equation}
vanishes for every operator ${\cal X}$ inserted away from $x$:
\begin{equation}
  \langle\delta T(x)\mathcal{X}\rangle_\lambda=0,
\end{equation}
where $T(z),{\bar T}({\bar z})$ are the stress tensors of the unperturbed theory ${\cal S}_0$, the expectation value within the full theory ${\cal S}={\cal S}_{0}+{\cal S}_{\lambda}$ is denoted by
\begin{equation}
  \langle\cdots\rangle_\lambda \equiv\frac{1}{Z_{\mathcal{J}\mathcal{J}}}\langle e^{-{\cal S}}\cdots\rangle
\end{equation}
with $Z_{\mathcal{J}\mathcal{J}}=\langle e^{-{\cal S}}\rangle$,
and the geometry is the cylinder with the boundary states $\sket{\mathcal{J}}$ at both ends.
The mismatch term can be expanded as
\begin{widetext}
  \begin{align}
    \langle\delta T(x){\cal X}\rangle_\lambda & =\frac{\langle e^{-{\cal S}_{\lambda}}\delta T(x){\cal X}\rangle_{0}}{\langle e^{-{\cal S}_{\lambda}}\rangle_{0}}=\sum_{p=1}^{\infty}\frac{\lambda^{p}}{p!}\int_{{\cal D}_{\epsilon}}\prod_{i=1}^{p}d^{2}w_{i}\langle\prod_{i=1}^{p}{\cal O}(w_{i},\bar{w}_{i})\delta T(x){\cal X}\rangle_{,{\rm c}}\label{eq:bb-1} \\
                                              & =\sum_{p=1}^{\infty}\frac{-i\pi\lambda^{p}}{(p-1)!}\int dx'\,\epsilon\partial_{x'}\delta(x'-x)\int_{{\cal D}_{\epsilon}}\prod_{i=1}^{p-1}d^{2}w_{i}\langle{\cal O}(x'+i\epsilon/2,x'-i\epsilon/2)\prod_{i=1}^{p-1}{\cal O}(w_{i},\bar{w}_{i}){\cal X}\rangle_{,{\rm c}},
  \end{align}
\end{widetext}
where we define
\begin{equation}
  {\cal D}_{\epsilon}:\;{\rm Im}w_{i}>\epsilon/2,\;\;|w_{i}-w_{j}|>\epsilon,\;\;|w_{i}-(x'+i\epsilon/2)|>\epsilon.
\end{equation}
Here, $\epsilon$ is the length cutoff, and the unperturbed expectation
value is defined as
\begin{equation}
  \langle\cdots\rangle_{0}\equiv\frac{1}{Z_{\mathcal{J}\mathcal{J}}^{(0)}}\langle e^{-{\cal S}_{0}}\cdots\rangle
\end{equation}
with $Z_{\mathcal{J}\mathcal{J}}^{(0)}=\langle e^{-{\cal S}_{0}}\rangle$,
and the last line can be shown by using the Ward identity:
\begin{align}
  T(z){\cal O}(w,\bar{w})             & \simeq\frac{{\cal O}(w,\bar{w})}{(z-w)^{2}}+\frac{\partial_{w}{\cal O}(w,\bar{w})}{z-w},                               \\
  \bar{T}(\bar{z}){\cal O}(w,\bar{w}) & \simeq\frac{{\cal O}(w,\bar{w})}{(\bar{z}-\bar{w})^{2}}+\frac{\partial_{\bar{w}}{\cal O}(w,\bar{w})}{\bar{z}-\bar{w}}.
\end{align}
The leading UV singularity comes from the contribution due to the
bulk-boundary OPE of the full short-distance cluster where all the
insertion points approach the boundary point $x'$:
\begin{equation}
  \ w_{i}=x'+\epsilon\zeta_{i},\;\;\bar{w}_{i}=x'+\epsilon\bar{\zeta}_{i}
\end{equation}
with $\epsilon\to0$ and $i=1,2,\ldots,p-1$.
We may then expand as
\begin{align}
   & {\cal O}(x'+i\epsilon/2,x'-i\epsilon/2)\prod_{i=1}^{p-1}{\cal O}(w_{i},\bar{w}_{i})\nonumber \\
   & \simeq\sum_{a}\epsilon^{h_{a}-2p}C_{a}^{(p)}(\boldsymbol{\zeta})\psi_{a}(x').\label{eq:ope}
\end{align}
Using $\prod_{i=1}^{p-1}d^{2}w_{i}=\epsilon^{2p-2}d^{2p-2}\zeta$
and defining
\begin{equation}
  B_{a}^{(p)}\equiv\frac{1}{(p-1)!}\int_{{\cal D}_{\zeta}}d^{2p-2}\zeta\;C_{a}^{(p)}(\boldsymbol{\zeta})
\end{equation}
with the domain
\begin{align}
  {\cal D}_{\zeta}:\; & {\rm Im}\zeta_{i}>1/2,\;\;\;\;\;|\zeta_{i}-\zeta_{j}|>1,\nonumber \\
                      & |\zeta_{i}-i/2|>1,\;\;\;|\zeta_{i}|<\ell(\epsilon)/\epsilon,
\end{align}
we have
\begin{equation}
  \langle\delta T(x){\cal X}\rangle_\lambda\sim\sum_{p=1}^{\infty}i\pi\lambda^{p}\sum_{a}\epsilon^{h_{a}-1}B_{a}^{(p)}\langle\partial_{x}\psi_{a}(x){\cal X}\rangle_{,{\rm c}}.
  \label{eq:mismatch}
\end{equation}
We note that $\ell(\epsilon)$ must be chosen such that the local OPE can
be justified, i.e., $\epsilon\ll \ell(\epsilon)\ll \ell^{*}$ with $\ell^{*}$ being,
e.g., a characteristic length scale for the distance between $x$
and the insertion point of ${\cal X}$, while it should also satisfy $\ell(\epsilon)\to0$ in the limit $\epsilon\to0$; a practical choice
for this is $\ell(\epsilon)=\sqrt{\epsilon \ell^{*}}$. Also, $B_{a}^{(p)}$
is the effective coefficient characterizing the strength of the boundary
operator $\psi_{a}$ in the $p$th order perturbative correction.
If all the boundary operators are irrelevant, i.e., $h_{a}>1$ $\forall a$,
the mismatch term vanishes in $\epsilon\to0$ and thus the boundary
state $\sket{\mathcal{J}}$ remains conformally invariant.
We note that the identity channel should be treated separately. Although it has $h_{\boldsymbol{1}}=0$, it does not contribute to the stress-tensor mismatch in Eq.~\eqref{eq:mismatch} because $\partial_x \boldsymbol{1}=0$. It can, however, renormalize the $g$ factor and hence the universal contribution $s_n$ to the MRE.

We can also consider varying a cutoff scale $\epsilon\sim\Lambda^{-1}_{\text{c}}$
as $\Lambda_{\text{c}}\to\Lambda_{\text{c}}(1-dl)$. Following calculations similar to those above, integrating out the momentum shell $dl$
induces the boundary action
\begin{equation}
  d\delta{\cal S}_{\partial}=dl\sum_{a}C_{a}\Lambda_{\text{c}}^{1-h_{a}}\int dx\;\psi_{a}(x),\;\;C_{a}\equiv-\sum_{p}\lambda^{p}B_{a}^{(p)}.
\end{equation}
Comparing it with the parametrization in Eq. (\ref{eq:boundary_action_general}),
we have
\begin{equation}
  \frac{d(\mu_{a}(l)\Lambda_{\text{c}}^{1-h_{a}})}{dl}=C_{a}\Lambda_{\text{c}}^{1-h_{a}},
\end{equation}
which gives the boundary RG equation (\ref{eq:boundary_RG_general}).

\section{Gluing matrix and compactification lattice}\label{app:compact}
We provide here the expressions of the gluing orthogonal matrix $G$ and the compactification lattice $\Lambda$. To this end, we note that the boundary conditions~(\ref{eq:J_gluing_O}) are satisfied provided that the fields are compactified as
\begin{equation}
  \boldsymbol{\Phi}\sim\boldsymbol{\Phi}+2\pi\boldsymbol{T},\;\;
  \boldsymbol{\Theta}\sim\boldsymbol{\Theta}+2\pi\boldsymbol{W},
\end{equation}
with the vectors of the winding numbers satisfying
\begin{equation}\label{eq:tw}
  G\boldsymbol{T}=\boldsymbol{T},\;\;G\boldsymbol{W}=-\boldsymbol{W}.
\end{equation}
The orthogonal matrix is then given by
\begin{equation}\label{eq:Omatrix}
  G=\frac{1}{2}\left(\begin{array}{cccccccc}
    -\gamma & 0       & \gamma  & 0       & 1       & -\mu    & 1       & \mu     \\
    0       & \gamma  & 0       & -\gamma & -\mu    & 1       & \mu     & 1       \\
    \gamma  & 0       & -\gamma & 0       & 1       & \mu     & 1       & -\mu    \\
    0       & -\gamma & 0       & \gamma  & \mu     & 1       & -\mu    & 1       \\
    1       & -\mu    & 1       & \mu     & -\gamma & 0       & \gamma  & 0       \\
    -\mu    & 1       & \mu     & 1       & 0       & \gamma  & 0       & -\gamma \\
    1       & \mu     & 1       & -\mu    & \gamma  & 0       & -\gamma & 0       \\
    \mu     & 1       & -\mu    & 1       & 0       & -\gamma & 0       & \gamma
  \end{array}\right),
\end{equation}
where $\mu=2\sqrt{K}/(K+1)$.
The compactification lattice is defined by combining Eqs.~\eqref{eq:compact_cond} and \eqref{eq:tw}.
Specifically, we obtain
\begin{align}
  {\cal T} & =\left\{ \boldsymbol{T}|\boldsymbol{T}=\sum_{i=1}^{4}n_{i}\boldsymbol{t}_{i},\;n_{i}\in\mathbb{Z}\right\}, \\
  {\cal W} & =\left\{ \boldsymbol{W}|\boldsymbol{W}=\sum_{i=1}^{4}m_{i}\boldsymbol{w}_{i},\;m_{i}\in\mathbb{Z}\right\},
\end{align}
which are spanned by the following primitive vectors:
\begin{align}
  \boldsymbol{t}_{1} & =\frac{1}{2\sqrt{2}}\left(\frac{1}{\sqrt{K}},1,0,0,0,0,\frac{1}{\sqrt{K}},1\right)^{{\rm T}}   \\
  \boldsymbol{t}_{2} & =\frac{1}{2\sqrt{2}}\left(0,0,\frac{1}{\sqrt{K}},1,\frac{1}{\sqrt{K}},1,0,0\right)^{{\rm T}}   \\
  \boldsymbol{t}_{3} & =\frac{1}{2\sqrt{2}}\left(\frac{1}{\sqrt{K}},-1,0,0,\frac{1}{\sqrt{K}},-1,0,0\right)^{{\rm T}} \\
  \boldsymbol{t}_{4} & =\frac{1}{2\sqrt{2}}\left(0,0,\frac{1}{\sqrt{K}},-1,0,0,\frac{1}{\sqrt{K}},-1\right)^{{\rm T}}
\end{align}
and
\begin{align}
  \boldsymbol{w}_{1} & =\frac{1}{2\sqrt{2}}\left(\sqrt{K},1,0,0,0,0,-\sqrt{K},-1\right)^{{\rm T}}
  \\
  \boldsymbol{w}_{2} & =\frac{1}{2\sqrt{2}}\left(0,0,-\sqrt{K},-1,\sqrt{K},1,0,0\right)^{{\rm T}}
  \\
  \boldsymbol{w}_{3} & =\frac{1}{2\sqrt{2}}\left(\sqrt{K},-1,0,0,-\sqrt{K},1,0,0\right)^{{\rm T}}
  \\
  \boldsymbol{w}_{4} & =\frac{1}{2\sqrt{2}}\left(0,0,\sqrt{K},-1,0,0,-\sqrt{K},1\right)^{{\rm T}}.
\end{align}
In addition, the selection rule~\eqref{eq:TLL_compactification} leads to the constraint
\begin{equation}
  n_{i}\equiv m_{i}\pmod{2}.
\end{equation}
Taken together, the composite, total winding number vectors $\boldsymbol{\lambda}=(\boldsymbol{T},\boldsymbol{W})^{{\rm T}}$ must belong to the following compactification lattice $\Lambda$:
\begin{equation}
  \Lambda=\left\{ {\cal T}\oplus{\cal W}|n_{i}\equiv m_{i}\pmod{2}\right\}.
\end{equation}
Its unit-cell volume can be readily obtained as
\begin{align}
  \mathcal{V}_{\Lambda}
  = \frac{1}{1-\gamma^{2}}.
\end{align}

\section{Perturbative calculations}\label{app:perturbation}
We here derive Eq.~\eqref{eq:deltaZ_second_order_TLL} by unfolding the $\mathcal{J}\mathcal{J}$ cylinder to a torus and evaluating the current correlator of $\sum_{m}J_{mR}^{x}J_{mL}^{x}$.
The unfolding replaces the left-moving current at the boundary by an image right-moving current with charge rotated by $G$.
The charge-neutrality condition then forces the two perturbing operators to have the same sector index, and the four-point function reduces to the product of two chiral current two-point functions on the torus.

To begin with, we note that it suffices to keep the $\lambda^2$ term while setting $K=1$ within the expectation value to obtain the leading $O(\gamma^2)$ contribution and write the correction term~\eqref{eq:delta_Z} by
\begin{widetext}
  \begin{align}
    \frac{\delta Z_{\mathcal{J}\mathcal{J}}}{Z_{\mathcal{J}\mathcal{J}}^{(0)}}=64\pi^4\lambda^{2}\int_{0}^{\beta/2}d\tau_{1}\int_{0}^{\tau_{1}}d\tau_{2}\int_{0}^{L}dx_{1}dx_{2}\;C^{(4)}_{J}(x_1,\tau_1;x_2,\tau_2)+O(\gamma^{4}),
    \label{appeq:Z2} \\
    C^{(4)}_{J}(x_1,\tau_1;x_2,\tau_2)\equiv \sum_{m,n=1}^{4}\langle J_{mL}^{x}(x_{1},\tau_{1})J_{mR}^{x}(x_{1},\tau_{1})J_{nL}^{x}(x_{2},\tau_{2})J_{nR}^{x}(x_{2},\tau_{2})\rangle_{{\rm cyl}}.
  \end{align}
\end{widetext}
Here, the expectation value is evaluated on the cylinder geometry of the unperturbed $K=1$ Euclidean action with the boundary condition $\mathcal{J}$ at $\tau=0,\beta/2$,
\begin{align}
  {\cal S}_{0}^{(K=1)}             & =\frac{1}{4\pi}\int d\tau dx\Bigl[i\left(\partial_{x}\boldsymbol{\Phi}_{R}\cdot\partial_{\tau}\boldsymbol{\Phi}_{R}-\partial_{x}\boldsymbol{\Phi}_{L}\cdot\partial_{\tau}\boldsymbol{\Phi}_{L}\right)\nonumber\\
                                                                   &\quad+\left(\partial_{x}\boldsymbol{\Phi}_{R}\right)^{2}+\left(\partial_{x}\boldsymbol{\Phi}_{L}\right)^{2}\Bigr], \\
  \langle\cdots\rangle_{{\rm cyl}} & \equiv\frac{1}{Z_{\mathcal{J}\mathcal{J}}^{(0,K=1)}}\langle e^{-{\cal S}_{0}^{(K=1)}}\cdots\rangle,
\end{align}
where
$Z_{\mathcal{J}\mathcal{J}}^{(0,K=1)}=\langle e^{-{\cal S}_{0}^{(K=1)}}\rangle$.
Since the same boundary conditions are
imposed at both ends, we can unfold the theory so that the cylinder
doubles to a torus of Euclidean-time period $\beta$:
\begin{equation}
  {\cal D}=\{(x,\tau)|x\sim x+L,\;\tau\sim\tau+\beta\},
\end{equation}
on which the field is defined as
\begin{equation}
  \boldsymbol{\Phi}(x,\tau)=\begin{cases}
    \boldsymbol{\Phi}_{R}(x,\tau)             & 0\leq\tau\leq\beta/2  \\
    G^{{\rm T}}\boldsymbol{\Phi}_{L}(x,-\tau) & -\beta/2\leq\tau\leq0
  \end{cases},
\end{equation}
where we note that the orthogonal matrix at $K=1$ reduces to (see Eq.~\eqref{eq:Omatrix})
\begin{equation}
  G=\frac{1}{2}\left(\begin{array}{cccccccc}
    0  & 0  & 0  & 0  & 1  & -1 & 1  & 1  \\
    0  & 0  & 0  & 0  & -1 & 1  & 1  & 1  \\
    0  & 0  & 0  & 0  & 1  & 1  & 1  & -1 \\
    0  & 0  & 0  & 0  & 1  & 1  & -1 & 1  \\
    1  & -1 & 1  & 1  & 0  & 0  & 0  & 0  \\
    -1 & 1  & 1  & 1  & 0  & 0  & 0  & 0  \\
    1  & 1  & 1  & -1 & 0  & 0  & 0  & 0  \\
    1  & 1  & -1 & 1  & 0  & 0  & 0  & 0
  \end{array}\right).
\end{equation}
After unfolding, this defines the 8-component holomorphic field
on the torus of size $L\times\beta$. The current operator can be
expressed by using the vertex operator as
\begin{align}
  J_{mR}^{x}(x,\tau)          & =\frac{1}{2 L}\sum_{\sigma=\pm1}V_{\sigma\boldsymbol{e}_{m}}(z),\;\;z=x+i\tau, \\
  V_{\sigma\boldsymbol{v}}(z) & \equiv e^{i\sigma\sqrt{2}\boldsymbol{v}\cdot\boldsymbol{\Phi}(z)},
\end{align}
where $\boldsymbol{e}_{m}\in\mathbb{R}^{8}$ is the unit vector of
the mode $(m,-)$. The effect of the defect is encoded by the imaging
rule, where the left-moving vertex operator is represented by the
image charge of the right-moving vertex:
\begin{align}
  J_{mL}^{x}(x,\tau)         & =\frac{1}{2 L}\sum_{\sigma=\pm1}V_{\sigma\tilde{\boldsymbol{e}}_{m}}(\bar{z}),\;\;\bar{z}=x-i\tau, \\
  \tilde{\boldsymbol{e}}_{m} & \equiv G^{\rm T}\boldsymbol{e}_{m}.
\end{align}
The correlation function we have to evaluate is
\begin{align}
   & C_J^{(4)}(z_1,\bar{z}_1,z_2,\bar{z}_2)
  =\frac{1}{(2 L)^{4}}\times\nonumber                                                                                                                                                                                                            \\
   & \sum_{\sigma_{1,2,3,4}=\pm1}\langle V_{\sigma_{1}\boldsymbol{e}_{m}}(z_{1})V_{\sigma_{2}\tilde{\boldsymbol{e}}_{m}}(\bar{z}_{1})V_{\sigma_{3}\boldsymbol{e}_{n}}(z_{2})V_{\sigma_{4}\tilde{\boldsymbol{e}}_{n}}(\bar{z}_{2})\rangle_{{\cal D}}
  \label{appeq:JxJx}
\end{align}
where $z_{1,2}$ ($\bar{z}_{1,2}$) runs over the upper (lower) half plane of the torus.

We now observe that the vertex correlator can be nonzero only if the
zero-mode part satisfies the charge-neutrality condition:
\begin{equation}
  \sigma_{1}\boldsymbol{e}_{m}+\sigma_{2}\tilde{\boldsymbol{e}}_{m}+\sigma_{3}\boldsymbol{e}_{n}+\sigma_{4}\tilde{\boldsymbol{e}}_{n}=0.
\end{equation}
One can easily check that this condition can be met only if
\begin{equation}
  \sigma_{1}=-\sigma_{3},\;\;\sigma_{2}=-\sigma_{4},\;\;m=n.
\end{equation}
Since the diagonal elements of $G$ vanish, we also have
\begin{equation}
  \boldsymbol{e}_{m}\cdot\tilde{\boldsymbol{e}}_{m}=0.
\end{equation}
This means that the bosonic-field components $\boldsymbol{e}_{m}\cdot\boldsymbol{\Phi}$
and $\tilde{\boldsymbol{e}}_{m}\cdot\boldsymbol{\Phi}$ are
independent, which allows us to factorize the correlation function:
\begin{align}
   & C_J^{(4)}(z_1,\bar{z}_1,z_2,\bar{z}_2)
  =\frac{1}{(2 L)^{4}}\times\nonumber    \\
   & \!\!
  \sum_{\sigma=\pm1}\langle V_{\sigma\boldsymbol{e}_{m}}(z_{1})V_{-\sigma\boldsymbol{e}_{m}}(z_{2})\rangle_{{\cal D}}\!\sum_{\tilde{\sigma}=\pm1}\langle V_{\tilde{\sigma}\tilde{\boldsymbol{e}}_{m}}(\bar{z}_{1})V_{-\tilde{\sigma}\tilde{\boldsymbol{e}}_{m}}(\bar{z}_{2})\rangle_{{\cal D}}.
\end{align}
Now each expectation value in the right-hand side is a standard chiral torus two-point function, i.e.,
a trace over the unfolded chiral Hilbert space. We can thus rewrite the
right-moving vertex as
\begin{align}
   & \frac{1}{(2 L)^{2}}\sum_{\sigma=\pm1}\langle V_{\sigma\boldsymbol{e}_{m}}(z_{1})V_{-\sigma\boldsymbol{e}_{m}}(z_{2})\rangle_{{\cal D}} \\
   & =\langle J_{mR}^{x}(z_{1})J_{mR}^{x}(z_{2})\rangle_{{\cal D}}\label{appeq:Jx}                                                             \\
   & =\langle J_{mR}^{z}(z_{1})J_{mR}^{z}(z_{2})\rangle_{{\cal D}}\label{appeq:Jz}                                                             \\
   & \equiv C_J^{(2)}(z_{1},z_{2}),
\end{align}
where we use the ${\rm SU(2)}$ symmetry of the $K=1$ chiral theory
of the doubled holomorphic field $\boldsymbol{\Phi}$ to rotate $J^{x}$
into $J^{z}$ from Eq.~\eqref{appeq:Jx} to Eq.~\eqref{appeq:Jz} \footnote{We note that originally the boundary $\sket{\mathcal{J}}$
  does not satisfy this symmetry, but once we factorize the correlation
  function into a product of the two-point correlation functions of
  the independent single-component doubled holomorphic fields ($\boldsymbol{e}_{m}\cdot\boldsymbol{\Phi}$
  and $\tilde{\boldsymbol{e}}_{m}\cdot\boldsymbol{\Phi}$), we can use
  the ${\rm SU(2)}$ symmetry of the $K=1$ chiral Hamiltonian of $\boldsymbol{\Phi}$
  for each of the factors.}. We here also define the correlation function
$C_J^{(2)}(z_{1},z_{2})$ of the current density $J^{z}$ of the standard
single component $K=1$ chiral boson theory on the torus, which is independent of a choice of the component of $\boldsymbol{\Phi}$.
Specifically, its expression can be given by
\begin{align}
  C_J^{(2)}(z_{1},z_{2}) & =\frac{1}{2L^{2}}\frac{\sum_{m\in\mathbb{Z}}m^{2}q^{m^{2}/2}}{\sum_{m\in\mathbb{Z}}q^{m^{2}/2}}\nonumber \\
                         & +\frac{1}{2L^{2}}\sum_{n=1}^{\infty}n\left[\frac{e^{-inu}}{1-q^{n}}+\frac{q^{n}e^{inu}}{1-q^{n}}\right],
\end{align}
where the first (second) term is the zero-mode (oscillator) contribution, and we define $u=2\pi(z_{1}-z_{2})/L$.
Similarly, we can also obtain
\begin{equation}
  \frac{1}{(2 L)^{2}}\sum_{\tilde{\sigma}=\pm1}\langle V_{\tilde{\sigma}\tilde{\boldsymbol{e}}_{m}}(\bar{z}_{1})V_{-\tilde{\sigma}\tilde{\boldsymbol{e}}_{m}}(\bar{z}_{2})\rangle_{{\cal D}}=C_J^{(2)}(\bar{z}_{1},\bar{z}_{2}).
\end{equation}
Equation~\eqref{appeq:JxJx} thus simplifies to $C_J^{(4)}=4\left|C_J^{(2)}\right|^{2}$, where the factor of $4$ comes from the summation over replica indices with the constraint $m=n$.
Substituting this into Eq.~\eqref{appeq:Z2}, we obtain
\begin{equation}
  \frac{\delta Z_{\mathcal{J}\mathcal{J}}}{Z_{\mathcal{J}\mathcal{J}}^{(0)}}\simeq256\pi^{4}\lambda^{2}\!\!\int_{0}^{\beta/2}d\tau_{1}\int_{0}^{\tau_{1}}d\tau_{2}\int_{0}^{L}dx_{1}dx_{2}\,\left|C_J^{(2)}\right|^{2}.
\end{equation}
To perform the integration over the upper half plane of the torus, we note that the cross term between the zero modes and oscillator modes in $\left|C_J^{(2)}\right|^{2}$
vanishes. The zero-mode contribution does not depend on $x,\tau$, and it simply gives
\begin{align}
   & 256\pi^{4}\lambda^{2}\times\frac{\beta^{2}L^{2}}{8}\times\left(\frac{1}{2L^{2}}\frac{\sum_{m\in\mathbb{Z}}m^{2}q^{m^{2}/2}}{\sum_{m\in\mathbb{Z}}q^{m^{2}/2}}\right)^{2}\nonumber \\
   & =\pi^{2}\lambda^{2}\times\frac{8\pi^{2}\beta^{2}}{L^{2}}\left(2q\frac{d}{dq}\ln\vartheta_{3}(q)\right)^{2}.
\end{align}
The oscillator part gives the universal part
\begin{equation}
  256\pi^{4}\lambda^{2}\times\left[\frac{\beta}{32\pi L}\sum_{n=1}^{\infty}n\frac{1+q^{n}}{1-q^{n}}-\frac{1}{64\pi^2}\sum_{n=1}^{\infty}1\right],
\end{equation}
which can be evaluated as
\begin{align}
&\sum_{n=1}^{\infty}1\to-\frac{1}{2},\\
 & \sum_{n=1}^{\infty}n\frac{1+q^{n}}{1-q^{n}}=\sum_{n=1}^{\infty}n+2\sum_{n=1}^{\infty}\frac{nq^{n}}{1-q^{n}}\to-\frac{1}{12}E_{2}(q),
\end{align}
where we use the zeta regularization.
Combining these results with $\lambda=\frac{K-1/K}{4\pi}\simeq\frac{\gamma}{\pi}$,
we finally obtain Eq.~(\ref{eq:deltaZ_second_order_TLL}).

\section{Extension to arbitrary rotation angle}\label{app:general_angle}

The perturbative analysis can be extended to the case of a general mixing angle $\zeta$. Under the unitary, the current densities rotate as
\begin{equation}
  J_{\chi}^{z}\to\cos\left(2\zeta\right)J_{\chi}^{z}+\sin\left(2\zeta\right)J_{\chi}^{x}.
\end{equation}
At leading order in $\gamma$, the effects of this change are (1) changing the value of the TLL parameter $K_{-}$ for the Gaussian part of the antisymmetric sector from $1$ to $K_{\zeta}$ specified below and (2) changing the perturbation term $H_{\Delta}$ such that
it now includes contributions from $J^{z}J^{x}$ cross term in addition
to $J^{x}J^{x}$ term considered above. More explicitly, regarding (1), the unitary leads to the change
\begin{equation}
  K_-=1\to K_{\zeta}=\frac{K+\frac{1}{K}+\left(K-\frac{1}{K}\right)\cos^{2}2\zeta}{\sqrt{\left(K+\frac{1}{K}\right)^{2}-\left(K-\frac{1}{K}\right)^{2}\cos^{4}2\zeta}},
\end{equation}
which modifies the unit-cell volume of the compactification lattice as
\begin{equation}
  {\cal V}_{\Lambda}\to\left[\frac{1}{2}\left(\sqrt{\frac{K}{K_{\zeta}}}+\sqrt{\frac{K_{\zeta}}{K}}\right)\right]^{2}\simeq1+\gamma^{2}\sin^{4}2\zeta.
\end{equation}
Regarding (2), the contribution from the $J^{x}J^{x}$ coupling can
be obtained by simply replacing the perturbation coefficient as
\begin{equation}
  \lambda\to\lambda\sin^{2}\left(2\zeta\right).
\end{equation}
The contribution from the $J^{z}J^{x}$ cross term gives the same
contribution as in Eq.~(\ref{eq:deltaZ_second_order_TLL}), up to the overall constant factor $2\sin^{2}(2\zeta)\cos^{2}(2\zeta)$.
The corrected ordered integral has no finite open-channel correction at this order. Cardy's consistency conditon thus fixes the Gaussian-reference coefficient $g_{0,{\cal J}}$ by
\begin{equation}
  \frac{g_{0,{\cal J}}(\zeta)^2}{1+\gamma^2\sin^4(2\zeta)}
  =1+O(\gamma^4).
\end{equation}
The finite closed-channel overlap correction instead gives
\begin{align}
  \frac{g_{\cal J}(\zeta)^2}{g_{0,{\cal J}}(\zeta)^2}
  &=1+2\gamma^2\sin^4(2\zeta)\nonumber\\
  &\quad+4\gamma^2\sin^2(2\zeta)\cos^2(2\zeta)+O(\gamma^4).
\end{align}
Combining the two coefficients yields the overlap
\begin{align}
  g_{{\cal J}}(\zeta)=1+\gamma^2\left[2\sin^2(2\zeta)-\frac12\sin^4(2\zeta)\right]+O(\gamma^4),
\end{align}
which gives the UV value in Eq.~\eqref{eq:gGamma_angle_TLL}.

\section{Current-algebra derivation of the fourth-order correction and angular flow}
\label{app:current_corrections}

We here work in the rotated-boundary picture and keep the bulk as eight independent TLLs.
All fermion bilinears are normal ordered with respect to the neutral $K=1$ Neveu-Schwarz vacuum.
To include the currents connecting the two purity branches, we collect the (nontilde) fields in the order $(1,2,1',2')$, with an analogous grouping for the tilde fields.
With a four-component matrix $X$, we define
\begin{align}
  J_{X,\chi}(\ell)
         & =\int_0^L dx\,e^{\mp ik_\ell x}
  :\boldsymbol a_\chi^\dagger(x)X\boldsymbol a_\chi(x):,
  \nonumber                                             \\
  k_\ell & =\frac{2\pi\ell}{L},\qquad \ell\in\mathbb Z,
  \label{eq:TLL_current_modes}
\end{align}
where the upper and lower signs refer to $\chi=R,L$, respectively.
The normalization of Eq.~\eqref{eq:fermion_currents} gives, within each chirality and four-component group,
\begin{align}
  [J_X(\ell),J_Y(p)]
                     & =J_{[X,Y]}(\ell+p)
  +\ell\,\operatorname{tr}(XY)\delta_{\ell+p,0},
  \label{eq:TLL_current_algebra}            \\
  J_X(\ell)|0\rangle & =0\qquad(\ell\geq0).
  \label{eq:TLL_current_vacuum}
\end{align}
We also introduce  the fermionic sewing matrix $\mathsf V(\zeta)$ obtained from Eqs.~\eqref{eq:J_TLL_MES} and \eqref{eq:general_angle_U_TLL}.
It reads
\begin{equation}
  \mathsf V(\zeta)=
  \begin{pmatrix}
    \sin^2\zeta        & \tfrac12\sin2\zeta & \cos^2\zeta         & -\tfrac12\sin2\zeta \\
    \tfrac12\sin2\zeta & \cos^2\zeta        & -\tfrac12\sin2\zeta & \sin^2\zeta         \\
    -\cos^2\zeta       & \tfrac12\sin2\zeta & \sin^2\zeta         & \tfrac12\sin2\zeta  \\
    \tfrac12\sin2\zeta & -\sin^2\zeta       & \tfrac12\sin2\zeta  & \cos^2\zeta
  \end{pmatrix},
  \label{eq:TLL_sewing_matrix_angle}
\end{equation}
and its value at the balanced angle $\zeta=\pi/4$ is
\begin{equation}
  \mathsf V_0\equiv\mathsf V(\pi/4)
  =\frac12
  \begin{pmatrix}
    1  & 1  & 1  & -1 \\
    1  & 1  & -1 & 1  \\
    -1 & 1  & 1  & 1  \\
    1  & -1 & 1  & 1
  \end{pmatrix},
  \quad \mathsf V_0\mathsf V_0^{\rm T}=1.
  \label{eq:TLL_current_sewing}
\end{equation}
For simplicity, we write $\Pi_a=|a\rangle\langle a|$ with $|a\rangle$ the $a$th unit vector of the four-component space, $j_{a,\chi}=J_{\Pi_a,\chi}$, and $\widehat j_{c,\chi}=J_{\widehat\Pi_c,\chi}$, where $\widehat\Pi_c=\mathsf V_0\Pi_c\mathsf V_0^{\rm T}$.

For each sheet, the normalized TLL ground state can be expressed by using the currents as follows:
\begin{equation}
  \begin{aligned}
     & \left[\prod_{\ell>0}
         (1-\gamma^2e^{-\epsilon k_\ell})^{1/2}\right]       \\
     & \quad\times\exp\!\left[
                          \gamma\sum_{\ell>0}
                          \frac{e^{-\epsilon k_\ell/2}}{\ell}
                          j_L(-\ell)j_R(-\ell)
                          \right]|0\rangle,
  \end{aligned}
  \label{eq:TLL_current_ground_state}
\end{equation}
where $\epsilon=\Lambda_{\rm UV}^{-1}$ and $j_\chi$ is the single-sheet density current in the Fourier convention of Eq.~\eqref{eq:TLL_current_modes}.
This expression can be obtained by the Bogoliubov diagonalization of Eq.~\eqref{eq:TLL_unrotated_pm}.

At fourth order, the chiral amplitude that we need to evaluate can be obtained by commuting positive modes to the right. Specifically, for positive $\ell_i$ satisfying $\ell_1+\ell_2=\ell_3+\ell_4$, we have
\begin{equation}
  \begin{aligned}
     & \langle0|j_a(\ell_1)j_b(\ell_2)
    \widehat j_c(-\ell_3)\widehat j_d(-\ell_4)|0\rangle \\
     & \quad=\frac{\ell_1\ell_2}{16}
    \left(\delta_{\ell_1\ell_3}\delta_{\ell_2\ell_4}
    +\delta_{\ell_1\ell_4}\delta_{\ell_2\ell_3}\right)  \\
     & \qquad-\min_i\ell_i\,
    \operatorname{tr}\!\left([\Pi_a,\widehat\Pi_c][\Pi_b,\widehat\Pi_d]\right),
  \end{aligned}
  \label{eq:TLL_four_current_amplitude}
\end{equation}
where we abbreviate the chirality subindex $\chi$ for notational simplicity.
The amplitude vanishes unless the level constraint is satisfied.
The first term contains the two-point contractions, while the second is the connected contribution from the commutator part of the current algebra.
Its matrix coefficient sums to zero over $a,b,c,d$, whereas the sum of its squares is $6$.
Since the two chiral sewings give mutually conjugate amplitudes, the latter contributes with the expansion factor $1/(2!2!)$.
Including the ground-state normalization and taking the logarithm gives
\begin{equation}
  \begin{aligned}
    \ln\nu_2(\epsilon)
    ={} & -3\gamma^2\sum_{\ell>0}e^{-\epsilon k_\ell} \\
        & +\frac{3\gamma^4}{2}
    \Biggl[
      \sum_{\substack{\ell_i>0\\\ell_1+\ell_2=\ell_3+\ell_4}}
      \frac{e^{-\epsilon k_{\ell_1+\ell_2}}(\min_i\ell_i)^2}
      {\ell_1\ell_2\ell_3\ell_4}\\
      &\qquad-\sum_{\ell>0}e^{-2\epsilon k_\ell}
      \Biggr]+O(\gamma^6).
  \end{aligned}
  \label{eq:TLL_log_overlap_quartic}
\end{equation}
We note that odd orders vanish as consistent with the duality discussed in Sec.~\ref{subsec:fermion_TLL_gfactor}.

To extract the boundary entropy $s$, we retain the constant term as $L\to\infty$ at fixed $\epsilon$.
Each single-mode sum in Eq.~\eqref{eq:TLL_log_overlap_quartic} has constant term $-1/2$.
Meanwhile, the constrained four-mode sum reduces to
\begin{equation}
  \frac{1}{1-e^{-\epsilon k_1}}
  \sum_{j=1}^{\infty}
  \frac{e^{-2\epsilon k_j}}{j^2(2j-1)^2},
  \label{eq:TLL_quartic_mode_sum}
\end{equation}
whose constant term is
\begin{align}
  \sum_{j=1}^{\infty}
  \frac{\tfrac12-2j}{j^2(2j-1)^2}
   & =\frac12\sum_{j=1}^{\infty}
  \left[\frac1{j^2}-\frac4{(2j-1)^2}\right]
  \nonumber                      \\
   & =-\frac{\pi^2}{6}.
  \label{eq:TLL_quartic_finite_part}
\end{align}
Using Eq.~\eqref{eq:Mn_BCFT}, we have Eq.~\eqref{eq:s2_TLL_quartic}.

We next consider the angular boundary RG flow using the current algebra.
Differentiating the rotated boundary state with respect to $\zeta$, with the normalization of Eq.~\eqref{eq:current_rotation_unitary}, we can identify the boundary perturbation by
\begin{equation}
  \delta\mathcal{S}_\partial
  =\delta\zeta\int dx\,\psi_\zeta(x),
  \label{eq:TLL_angle_boundary_action}
\end{equation}
where we define the boundary operator
\begin{equation}
  \psi_\zeta
  =2i\sum_{m=1}^{4}
  \left(J^y_{mR}+J^y_{mL}\right)
  \Big|_{\zeta=\pi/4}.
  \label{eq:TLL_angle_boundary_current}
\end{equation}
Since an angular displacement rotates the sewing matrix as $\mathsf V(\pi/4+\delta\zeta)=(1+\delta\zeta\,\mathsf V_0'\mathsf V_0^{\rm T})\mathsf V_0$, where the prime denotes the derivative with respect to $\zeta$ at $\pi/4$, $\psi_\zeta$ is the fermion bilinear with the matrix $\mathsf V_0'\mathsf V_0^{\rm T}$ in the notation of Eq.~\eqref{eq:TLL_current_modes}.

Since $\psi_\zeta$ has scaling dimension $1$ at $K=1$, the flow of $\delta\zeta$ is determined by the source term for $\psi_\zeta$ in the boundary RG equation~\eqref{eq:boundary_RG_general}.
This source term arises only from the chirality-mixing part of the bulk Hamiltonian, because the chirality-diagonal part only renormalizes the velocity of the free fermions.
The chirality-mixing part of the unrotated bulk Hamiltonian is
\begin{equation}
  -4\pi^2\lambda\int dx
  \sum_{m=1}^{4}\sum_{r=1}^{2}
  j_{mr,R}(x)j_{mr,L}(x),
  \label{eq:TLL_bulk_density_mixing}
\end{equation}
where $j_{mr,\chi}(x)=:a^\dagger_{mr,\chi}(x)a_{mr,\chi}(x):$ is the density in replica $r$ of sector $m$, and $\lambda$ is defined in Eq.~\eqref{eq:TLL_HDelta_bosonized}.
Consider the boundary at $\tau=0$ and the bulk at $\tau<0$, and use $z=x+iy$ with $y=-v\tau>0$.
In local coordinates, Eq.~\eqref{eq:TLL_current_algebra} becomes
\begin{equation}
  J_X(z)J_Y(w)
  \sim\frac{\operatorname{tr}(XY)}{(z-w)^2}
  +\frac{J_{[X,Y]}(w)}{z-w},
  \label{eq:TLL_current_OPE}
\end{equation}
where $J_X(z)=2\pi i:\boldsymbol a_L^\dagger X\boldsymbol a_L:$, with the opposite sign for the right-moving current.
After sewing, the currents within a single bulk insertion belong to independent four-component groups, and there is no first-order source for a boundary current.
At second order, a double pole in one group and a single pole in the other generate a current whose matrix coefficient is
\begin{equation}
  \begin{aligned}
     & \sum_{a,c}\mathsf V_{ac}(\zeta)^2
    \left[\Pi_a,\mathsf V(\zeta)\Pi_c\mathsf V(\zeta)^{\rm T}\right] \\
     & \qquad=\delta\zeta\,\mathsf V'_0\mathsf V_0^{\rm T}
    +O\!\left((\delta\zeta)^2\right).
  \end{aligned}
  \label{eq:TLL_angle_current_coefficient}
\end{equation}
The zeroth-order term vanishes because $\mathsf V_{0,ac}^2=1/4$ and the projectors sum to the identity.
The generated current is therefore $\delta\zeta\,\psi_\zeta$.
Performing the coordinate integrations around the boundary yields
\begin{equation}
  d(\delta\mathcal{S}_\partial)
  \simeq\left(\frac{\pi\lambda}{v}\right)^2
  \delta\zeta\,dl\int dx\,\psi_\zeta(x).
  \label{eq:TLL_angle_induced_action}
\end{equation}
Since $\pi\lambda/v=\gamma/(1-\gamma^2)$, we finally have Eq.~\eqref{eq:TLL_angle_RG}.

We next show that $\zeta=0, \pi/4$ are the only fixed points for $0\leq\zeta\leq\pi/4$ and that the flow from any other angle ends at $\zeta=0$.
To follow the flow beyond linear order in $\delta\zeta$, we use Eq.~\eqref{eq:TLL_angle_current_coefficient} without the expansion.
Together with the coefficient of Eq.~\eqref{eq:TLL_angle_induced_action}, the generated current then rotates a general sewing matrix $\mathsf V$ at order $\gamma^2$ as
\begin{equation}
  \frac{d\mathsf V}{dl}
  =\gamma^2\Bigl\{\sum_{a,c}\mathsf V_{ac}^2
  \bigl[\Pi_a,\mathsf V\Pi_c\mathsf V^{\rm T}\bigr]\Bigr\}\mathsf V.
  \label{eq:TLL_sewing_flow}
\end{equation}
Starting from $\mathsf V(\zeta)$, $\mathsf V$ keeps the form
\begin{equation}
  \mathsf V=
  \begin{pmatrix}
    \mathsf V_{11}  & \mathsf V_{12}  & \mathsf V_{13}  & -\mathsf V_{12} \\
    \mathsf V_{12}  & \mathsf V_{13}  & -\mathsf V_{12} & \mathsf V_{11}  \\
    -\mathsf V_{13} & \mathsf V_{12}  & \mathsf V_{11}  & \mathsf V_{12}  \\
    \mathsf V_{12}  & -\mathsf V_{11} & \mathsf V_{12}  & \mathsf V_{13}
  \end{pmatrix},
  \label{eq:TLL_sewing_matrix_form}
\end{equation}
i.e., the form of Eq.~\eqref{eq:TLL_sewing_matrix_angle} with $\sin^2\zeta$, $\tfrac12\sin2\zeta$, and $\cos^2\zeta$ replaced by $\mathsf V_{11}$, $\mathsf V_{12}$, and $\mathsf V_{13}$, which may depend independently on $\zeta$.
These entries obey
\begin{equation}
  \frac{d\mathsf V_{1b}}{dl}
  =2\gamma^2\bigl[\mathsf V_{1b}^3-\bigl(\mathsf V_{11}^4+2\mathsf V_{12}^4+\mathsf V_{13}^4\bigr)\mathsf V_{1b}\bigr],
  \qquad b=1,2,3.
  \label{eq:TLL_sewing_flow_entry}
\end{equation}
For $0<\zeta<\pi/4$, the initial entries satisfy $0<\mathsf V_{11}<\mathsf V_{12}<\mathsf V_{13}$, and Eq.~\eqref{eq:TLL_sewing_flow_entry} gives $d\ln(\mathsf V_{1b}/\mathsf V_{13})/dl=2\gamma^2(\mathsf V_{1b}^2-\mathsf V_{13}^2)<0$ for $b=1,2$, so that no such angle is a fixed point.
Since the normalization $\mathsf V_{11}^2+2\mathsf V_{12}^2+\mathsf V_{13}^2=1$ gives $\mathsf V_{13}^2\geq1/4$, $(\mathsf V_{11},\mathsf V_{12},\mathsf V_{13})$ flows to $(0,0,1)$, i.e., $\mathsf V(\zeta)$ flows to $\mathsf V(0)$.

\section{Boundary RG flow for the spinless-fermion chain at \texorpdfstring{$K>3$}{K>3}}
\label{app:largeK}

We here provide the details of the analysis in Sec.~\ref{subsec:fermion_TLL_largeK}.
The first four parts work in the rotated-bulk picture with $K_{+}=K$ and $K_{-}=1$, and the last two in the rotated-boundary picture with $K_{+}=K_{-}=K$.

\subsubsection*{Boundary operators on $\mathcal{J}$}

We first show that, for $K>1$, the boundary operators with $h<1$ on $\mathcal{J}$ are built only from $\theta_{1}-\theta_{1'}$ and $\theta_{2}-\theta_{2'}$.
A boundary operator generated at any order carries the charge
\begin{equation}
  \bm{Q}=\sum_{m}\left(t_{m}\bm{e}_{m}+u_{m}\tilde{\bm{e}}_{m}\right),
  \qquad
  \tilde{\bm{e}}_{m}=G^{\rm T}\bm{e}_{m},
\end{equation}
with the integers in Eq.~\eqref{eq:boundary_charge_TLL}.
Since $\{\bm{e}_{m}\}$ and $\{\tilde{\bm{e}}_{m}\}$ are orthonormal, its dimension~\eqref{eq:boundary_dimension_TLL} is
\begin{equation}
  h=\lvert\bm{t}\rvert^{2}+\lvert\bm{u}\rvert^{2}+2\bm{t}^{\rm T}M\bm{u},
\end{equation}
where Eq.~\eqref{eq:Omatrix} gives $M_{mn}=\bm{e}_{m}\cdot G\bm{e}_{n}$ as
\begin{equation}
  M=\frac{1}{2}
  \begin{pmatrix}
    \gamma  & -\gamma & 1       & 1       \\
    -\gamma & \gamma  & 1       & 1       \\
    1       & 1       & \gamma  & -\gamma \\
    1       & 1       & -\gamma & \gamma
  \end{pmatrix}.
\end{equation}
We expand $\bm{t}$ in the eigenvectors of $M$ as
\begin{align}
  A_{t} & =t_{1}-t_{2},
        &
  B_{t} & =t_{3}-t_{4},
  \nonumber                         \\
  C_{t} & =t_{1}+t_{2}-t_{3}-t_{4},
        &
  D_{t} & =t_{1}+t_{2}+t_{3}+t_{4},
\end{align}
whose eigenvalues are $\gamma$, $\gamma$, $-1$, and $1$, and expand $\bm{u}$ in the same way.
The dimension then becomes
\begin{align}
  h={} & (1+\gamma)\left(a^{2}+b^{2}\right)+(1-\gamma)\left(a'^{2}+b'^{2}\right)
  \nonumber                                                                                \\
       & +\frac{1}{4}\left(C_{t}-C_{u}\right)^{2}+\frac{1}{4}\left(D_{t}+D_{u}\right)^{2},
  \label{appeq:largeK_h}
\end{align}
where
\begin{align}
  a  & =\frac{A_{t}+A_{u}}{2},
     &
  a' & =\frac{A_{t}-A_{u}}{2},
  \nonumber                    \\
  b  & =\frac{B_{t}+B_{u}}{2},
     &
  b' & =\frac{B_{t}-B_{u}}{2}.
\end{align}
The selection rule $t_{m}\equiv u_{m}\pmod{2}$ makes these four numbers integers and $C_{t}-C_{u}$ and $D_{t}+D_{u}$ even.
For $K>1$, where $0<\gamma<1$, the condition $h<1$ thus requires
\begin{equation}
  a=b=0,
  \qquad
  C_{t}=C_{u},
  \qquad
  D_{t}=-D_{u},
\end{equation}
which leaves $h=(1-\gamma)(A_{t}^{2}+B_{t}^{2})$ with $A_{t}\equiv B_{t}\pmod{2}$.
The charges with $h<1$ are therefore
\begin{align}
  \bm{Q}     & =n_{1}\bm{q}_{1}+n_{2}\bm{q}_{2},
  \qquad
  n_{1},n_{2}\in\mathbb{Z},
  \nonumber                                       \\
  \bm{q}_{1} & =(I_{8}-G)(\bm{e}_{1}+\bm{e}_{3}),
  \quad
  \bm{q}_{2}=(I_{8}-G)(\bm{e}_{2}+\bm{e}_{3}),
\end{align}
where $I_{8}$ is the $8\times8$ identity matrix, and $\bm{q}_{1}$ and $\bm{q}_{2}$ satisfy the above conditions with $(A_{t},B_{t})=(1,1)$ and $(-1,1)$.
Their dimensions are
\begin{equation}
  h=\left(n_{1}^{2}+n_{2}^{2}\right)h_{\min},
  \qquad
  h_{\min}=2(1-\gamma)=\frac{4}{K+1}.
\end{equation}
A single term $\cos(2\sqrt{2}\theta_{m,-})$ near the boundary carries the charges $\pm(\bm{e}_{m}-\tilde{\bm{e}}_{m})$, so that $\bm{q}_{1}$ arises from the sectors $m=1,3$ and $\bm{q}_{2}$ from $m=2,3$.
By Eq.~\eqref{eq:J_theta_TLL}, the corresponding exponents become at the boundary
\begin{align}
  2\sqrt{2}\left(\theta_{1,-}+\theta_{3,-}\right)
   & =-2\sqrt{2}\left(\theta_{2,-}+\theta_{4,-}\right)
  =2\left(\theta_{1}-\theta_{1'}\right),
  \nonumber                                            \\
  2\sqrt{2}\left(\theta_{2,-}+\theta_{3,-}\right)
   & =-2\sqrt{2}\left(\theta_{1,-}+\theta_{4,-}\right)
  =2\left(\theta_{2}-\theta_{2'}\right).
  \label{eq:TLL_largeK_pair}
\end{align}
Hence $\bm{q}_{1}$ and $\bm{q}_{2}$ correspond to $\cos2(\theta_{1}-\theta_{1'})$ and $\cos2(\theta_{2}-\theta_{2'})$, which also arise from the sectors $m=2,4$ and $m=1,4$.
For larger $K$, the operators with $n_{1}^{2}+n_{2}^{2}\geq2$ also become relevant, but they do not change the results below.

\subsubsection*{Signs of the second-order couplings}

We next show that the two cosines enter the boundary action with the negative sign in Eq.~\eqref{eq:TLL_largeK_step1}.
At second order in $\lambda$, the weight $e^{-\mathcal{S}_{\lambda}}$ contains the products of pairs of vertex operators in $\mathcal{O}$, such as
\begin{equation}
  \lambda^{2}\int d^{2}z_{1}d^{2}z_{2}\,
  e^{-2\sqrt{2}i\theta_{1,-}(z_{1})}e^{-2\sqrt{2}i\theta_{3,-}(z_{2})},
\end{equation}
which comes from $\cos(2\sqrt{2}\theta_{1,-})$ and $\cos(2\sqrt{2}\theta_{3,-})$ in $\mathcal{O}$.
As $z_{1},z_{2}\to x$, the OPE of these vertex operators gives $e^{-2i(\theta_{1}-\theta_{1'})(x)}$ by Eq.~\eqref{eq:TLL_largeK_pair}, with a coefficient that is a power of the distances and hence positive.
The same holds for the sectors $m=2,4$, for the conjugate operators, and for $\cos2(\theta_{2}-\theta_{2'})$.
The second-order term thus adds
\begin{equation}
  c\,\lambda^{2}\int dx\left[\cos2\left(\theta_{1}-\theta_{1'}\right)+\cos2\left(\theta_{2}-\theta_{2'}\right)\right]
\end{equation}
with a positive constant $c$ to the exponent of $e^{-\mathcal{S}}$, which gives Eq.~\eqref{eq:TLL_largeK_step1}.

\subsubsection*{Boundary operators after the first step}

We now show that, after the first step, the only nonidentity boundary operator with $h\leq1$ is $\psi$ in Eq.~\eqref{eq:TLL_largeK_psi}.
On $\mathcal{J}$, the Dirichlet directions form the eigenspace of $G$ with the eigenvalue $-1$.
This eigenspace has one direction for each MES in Eq.~\eqref{eq:J_TLL_MES}, and is spanned by $\bm{q}_{1}$, $\bm{q}_{2}$, $\bm{d}$, and $\bm{c}$ defined by
\begin{align}
  \bm{d}\cdot\bm{\Phi} & =\frac{1}{2}\sum_{m}\varepsilon_{m}\Phi_{m,+},
  \qquad
  \bm{c}\cdot\bm{\Phi} =\frac{1}{2}\sum_{m}\varepsilon_{m}\Phi_{m,-},
    \nonumber\\
  \varepsilon_{m}&=(1,1,-1,-1).
\end{align}
Indeed, $G\bm{d}=-\bm{d}$ and $G\bm{c}=-\bm{c}$, while $G^{2}=I_{8}$ gives $G\bm{q}_{a}=-\bm{q}_{a}$ and $\bm{q}_{a}\cdot\bm{d}=\bm{q}_{a}\cdot\bm{c}=0$.
The first step turns $\bm{q}_{1}$ and $\bm{q}_{2}$ into Neumann directions, and the gluing matrix becomes
\begin{equation}
  G_{1}=I_{8}-2\bm{d}\bm{d}^{\rm T}-2\bm{c}\bm{c}^{\rm T},
  \qquad
  G_{1}\bm{e}_{m}=\bm{e}_{m}-\varepsilon_{m}\bm{c}.
\end{equation}
Replacing $G$ by $G_{1}$, the charge of a boundary operator becomes
\begin{equation}
  \bm{Q}=\sum_{m}\left(t_{m}+u_{m}\right)\bm{e}_{m}-\Bigl(\sum_{m}\varepsilon_{m}u_{m}\Bigr)\bm{c},
\end{equation}
and its dimension is
\begin{equation}
  h=\left(\bm{c}\cdot\bm{Q}\right)^{2}+\frac{1}{4}\sum_{m<n}\left[\varepsilon_{m}\left(t_{m}+u_{m}\right)-\varepsilon_{n}\left(t_{n}+u_{n}\right)\right]^{2}.
\end{equation}
Here $\bm{c}\cdot\bm{Q}=\frac{1}{2}\sum_{m}\varepsilon_{m}(t_{m}-u_{m})$ is an integer, and each bracket is even.
If the four values $\varepsilon_{m}(t_{m}+u_{m})$ are not all equal, at least three brackets are nonzero and $h\geq3$.
The nonidentity vertex operators with $h\leq1$ therefore have
\begin{equation}
  \bm{Q}=\pm\bm{c},
  \qquad
  h=1.
\end{equation}
A single term $\cos(2\sqrt{2}\theta_{m,-})$ carries the charges $\pm(\bm{e}_{m}-G_{1}\bm{e}_{m})=\pm\varepsilon_{m}\bm{c}$.
Equation~\eqref{eq:TLL_largeK_theta_step1} gives at the boundary
\begin{equation}
  2\sqrt{2}\,\varepsilon_{m}\theta_{m,-}
  =2\left(\theta_{1}-\theta_{2}\right),
  \label{eq:TLL_largeK_after1}
\end{equation}
so that these operators form $\psi$.
The charge-neutral current with $h=1$ is forbidden by the symmetry under $\bm{\Phi}_{\chi}\to-\bm{\Phi}_{\chi}$, which the cosines in Eq.~\eqref{eq:TLL_largeK_step1} also respect.

\subsubsection*{Flow along the marginal family}

We derive the factor $\cos(\pi\mu_{\psi})$ in the RG equation~\eqref{eq:TLL_largeK_flow} from the coefficient of $\psi$ in the bulk-boundary OPE~\eqref{eq:ope} with $p=1$ on the state with coupling $\mu_{\psi}$,
\begin{equation}
  \mathcal{O}(x,y)\simeq B_{\psi}^{(1)}(\mu_{\psi})\,(2y)^{-1}\psi(x)+\cdots.
  \label{appeq:largeK_OPE}
\end{equation}

We first determine the image of a right-moving exponential on the state with coupling $\mu_{\psi}$.
The chiral boson $\bm{c}\cdot\bm{\Phi}_{\chi}$ has the $\mathrm{SU}(2)_{1}$ current
\begin{equation}
  J^{x}_{\bm{c},\chi}=\cos\left(\sqrt{2}\,\bm{c}\cdot\bm{\Phi}_{\chi}\right).
\end{equation}
After the first step, $\bm{c}\cdot\bm{\Phi}=0$ at the boundary, and Eq.~\eqref{eq:TLL_largeK_after1} with $K_{-}=1$ gives
\begin{equation}
  \bm{c}\cdot\bm{\Phi}_{L}
  =-\bm{c}\cdot\bm{\Phi}_{R}
  =\bm{c}\cdot\bm{\Theta}
  =\sqrt{2}\left(\theta_{1}-\theta_{2}\right),
\end{equation}
so that $\psi=J^{x}_{\bm{c},L}$ at the boundary.
At $\mu_{\psi}=0$, the relation $\bm{c}\cdot\bm{\Phi}_{L}=-\bm{c}\cdot\bm{\Phi}_{R}$ replaces a right-moving exponential at $z=x+iy$ by the left-moving exponential with the opposite exponent at $z^{*}=x-iy$.
The boundary action~\eqref{eq:TLL_largeK_action} rotates this image by $e^{-2\pi i\mu_{\psi}J^{x}_{\bm{c},0}}$~\cite{callan1994exact}, where $J^{x}_{\bm{c},0}$ is the zero mode of $J^{x}_{\bm{c},L}$.
On the two exponentials $e^{\pm i\bm{c}\cdot\bm{\Phi}_{L}/\sqrt{2}}$, the zero mode acts as $\sigma^{x}/2$, so that
\begin{equation}
  e^{\mp\frac{i}{\sqrt{2}}\bm{c}\cdot\bm{\Phi}_{L}}
  \to\cos(\pi\mu_{\psi})\,e^{\mp\frac{i}{\sqrt{2}}\bm{c}\cdot\bm{\Phi}_{L}}
  -i\sin(\pi\mu_{\psi})\,e^{\pm\frac{i}{\sqrt{2}}\bm{c}\cdot\bm{\Phi}_{L}}.
  \label{appeq:largeK_rotation}
\end{equation}

We next evaluate the coefficient in Eq.~\eqref{appeq:largeK_OPE}.
Since the rotation acts only on the direction of $\bm{c}$, we decompose the exponent of a single term as
\begin{equation}
  2\sqrt{2}\theta_{m,-}
  =\frac{\varepsilon_{m}}{\sqrt{2}}\left(\bm{c}\cdot\bm{\Phi}_{L}-\bm{c}\cdot\bm{\Phi}_{R}\right)
  +2\sqrt{2}\left(\bm{e}_{m}-\frac{\varepsilon_{m}}{2}\bm{c}\right)\cdot\bm{\Theta}.
\end{equation}
The second term lies in the Neumann eigenspace of $G_{1}$ and contributes a factor independent of $\mu_{\psi}$.
In the first term, the right-moving factor at $z=x+iy$ is replaced by the rotated image at $z^{*}=x-iy$, which gives
\begin{align}
   & e^{\mp\frac{i\varepsilon_{m}}{\sqrt{2}}\bm{c}\cdot\bm{\Phi}_{L}}(z)
  \Bigl[\cos(\pi\mu_{\psi})\,e^{\mp\frac{i\varepsilon_{m}}{\sqrt{2}}\bm{c}\cdot\bm{\Phi}_{L}}
  \nonumber                                                                       \\
   & \qquad-i\sin(\pi\mu_{\psi})\,e^{\pm\frac{i\varepsilon_{m}}{\sqrt{2}}\bm{c}\cdot\bm{\Phi}_{L}}\Bigr](z^{*}).
\end{align}
With the $\cos(\pi\mu_{\psi})$ term, the operator product gives $e^{\mp i\sqrt{2}\varepsilon_{m}\bm{c}\cdot\bm{\Phi}_{L}}$ in $\psi$, while with the $\sin(\pi\mu_{\psi})$ term it gives only charge-neutral operators.
Only the $\cos(\pi\mu_{\psi})$ term thus contributes to $\psi$, which gives Eq.~\eqref{appeq:largeK_OPE} with $B_{\psi}^{(1)}(\mu_{\psi})=B_{\psi}^{(1)}\cos(\pi\mu_{\psi})$.
The charge-neutral current with $h=1$ is excluded by the symmetry under $\bm{\Phi}_{\chi}\to-\bm{\Phi}_{\chi}$, which $\psi$ also respects.
To first order in the bulk coupling, the source term for $\psi$ on the state with coupling $\mu_{\psi}$ is therefore
\begin{equation}
  C_{\psi}\cos(\pi\mu_{\psi}),
  \qquad
  C_{\psi}=-\frac{\lambda}{2}B_{\psi}^{(1)},
\end{equation}
where $B_{\psi}^{(1)}>0$ as for the second-order couplings above.
This gives Eq.~\eqref{eq:TLL_largeK_flow}.
Since $\lambda>0$ for $K>1$, we obtain $C_{\psi}<0$.

\subsubsection*{$g$ factor}

We evaluate the $g$ factor of the endpoint at $\mu_{\psi}=-1/2$ from the unit-cell volume $\mathcal{V}$ of the compactification lattice.
Since the bulk consists of eight decoupled TLLs with $K_{\pm}=K$, we work with the fields $\Phi_{s}=\phi_{s}/\sqrt{K}$ and $\Theta_{s}=\sqrt{K}\theta_{s}$ on the individual sheets $s$.
Equation~\eqref{eq:compact_cond} then gives the zero-mode vectors in Eq.~\eqref{eq:tw} as
\begin{equation}
  T_{s}=\frac{Q_{s}}{2\sqrt{K}},
  \qquad
  W_{s}=\frac{\sqrt{K}J_{s}}{2}.
\end{equation}
We introduce the 8-component vector $\hat{\bm{\varepsilon}}$ with $\hat{\varepsilon}_{s}=1$ on the untilded sheets and $\hat{\varepsilon}_{s}=-1$ on the tilde sheets, so that $\hat{\varepsilon}_{s}=\varepsilon_{m}$ on the two sheets $s$ of sector $m$.
At the endpoint, the boundary values~\eqref{eq:TLL_largeK_theta_step1} with $\theta_{1}=\theta_{2}$ fix the seven combinations of $\theta_{s}$ orthogonal to $\hat{\bm{\varepsilon}}$.
The boundary still fixes them in the rotated-boundary picture, as shown in Sec.~\ref{subsec:fermion_TLL_largeK}.
The field $\bm{\Phi}$ thus obeys the Dirichlet condition only along $\hat{\bm{\varepsilon}}$.
The conditions~\eqref{eq:tw}, which allow $\bm{T}$ only in the Neumann directions and $\bm{W}$ only in the Dirichlet direction, and the selection rule~\eqref{eq:TLL_compactification} give
\begin{equation}
  \hat{\bm{\varepsilon}}\cdot\bm{Q}=0,
  \qquad
  J_{s}=J_{0}\hat{\varepsilon}_{s},
  \qquad
  Q_{s}=J_{0}+2x_{s},
\end{equation}
with $J_{0},x_{s}\in\mathbb{Z}$ and $\hat{\bm{\varepsilon}}\cdot\bm{x}=0$.
The lattice is therefore
\begin{align}
  (\bm{T},\bm{W}) & =\left(\frac{\bm{x}}{\sqrt{K}},\bm{0}\right)+J_{0}\bm{v}_{0},
  \nonumber                                                                       \\
  \bm{v}_{0}      & =\frac{1}{2\sqrt{K}}\left((1,\ldots,1),K\hat{\bm{\varepsilon}}\right).
\end{align}

The unit-cell volume $\mathcal{V}$ is the square root of the Gram determinant of the primitive vectors, which are $\bm{v}_{0}$ and seven primitive vectors of the lattice of $(\bm{x}/\sqrt{K},\bm{0})$.
The $\bm{T}$ component of $\bm{v}_{0}$ lies in the span of the seven vectors, because $(1,\ldots,1)$ is orthogonal to $\hat{\bm{\varepsilon}}$.
Subtracting this component from $\bm{v}_{0}$ leaves the volume unchanged and gives $(\bm{0},\sqrt{K}\hat{\bm{\varepsilon}}/2)$, which is orthogonal to the seven vectors.
The Gram matrix thus becomes block diagonal, with the determinant $8K^{-7}$ for the seven vectors and the squared length $2K$ of $(\bm{0},\sqrt{K}\hat{\bm{\varepsilon}}/2)$.
Hence
\begin{equation}
  \mathcal{V}=\sqrt{8K^{-7}\times2K}=4K^{-3}.
\end{equation}
Since $g=\sqrt{\mathcal{V}}$ as in Eq.~\eqref{eq:ZJJ_poisson}, we obtain $g=2K^{-3/2}$ for the state at $\mu_{\psi}=-1/2$.

\subsubsection*{Stability of the fixed point}

The stability of the state at $\mu_{\psi}=-1/2$ follows from a conservation law.
The bulk perturbation transfers fermions only between the two replicas within each sector, while $\mathcal{J}$ conserves the difference of the particle numbers on the two sheets of each MES in Eq.~\eqref{eq:J_TLL_MES}.
The only combination conserved by both is
\begin{equation}
  \left(N_{1}+N_{2}+N_{1'}+N_{2'}\right)-\left(N_{\tilde{1}}+N_{\tilde{2}}+N_{\tilde{1}'}+N_{\tilde{2}'}\right),
  \label{eq:TLL_largeK_charge}
\end{equation}
where $N_{s}$ denotes the particle number on sheet $s$.
The Dirichlet condition along $\hat{\bm{\varepsilon}}$ expresses this conservation at the boundary.
The boundary operators at the fixed point are labeled by the dual-lattice vectors $\bm{\lambda}^{*}=(\bm{\lambda}^{*}_{T},\bm{\lambda}^{*}_{W})$ and have dimensions $h=\bm{\lambda}^{*2}/2$, as in Eq.~\eqref{eq:ZJJ_poisson}.
The dual lattice consists of the vectors $\bm{\lambda}^{*}$ in the space of the compactification lattice with $\bm{\lambda}^{*}_{T}\cdot\bm{T}+\bm{\lambda}^{*}_{W}\cdot\bm{W}\in\mathbb{Z}$ for all its vectors $(\bm{T},\bm{W})$.
Operators with $\bm{\lambda}^{*}_{W}\neq0$ contain $\hat{\bm{\varepsilon}}\cdot\bm{\Theta}$ and change Eq.~\eqref{eq:TLL_largeK_charge}, so that they are not generated.
For $\bm{\lambda}^{*}_{W}=0$, the condition $\bm{\lambda}^{*}_{T}\cdot\bm{T}\in\mathbb{Z}$ for the lattice vectors with $J_{0}=0$ and $J_{0}=1$ gives
\begin{equation}
  \bm{\lambda}^{*}_{T}=\sqrt{K}P\bm{y},
  \qquad
  \bm{y}\in\mathbb{Z}^{8},
  \qquad
  \sum_{s}y_{s}\in2\mathbb{Z},
\end{equation}
where $P$ is the projector onto the subspace orthogonal to $\hat{\bm{\varepsilon}}$.
The dimension is
\begin{equation}
  h=\frac{K}{16}\left[8\lvert\bm{y}\rvert^{2}-(\hat{\bm{\varepsilon}}\cdot\bm{y})^{2}\right].
\end{equation}
The bracket and the condition on $\bm{y}$ are invariant under $\bm{y}\to\bm{y}-k\hat{\bm{\varepsilon}}$ with $k\in\mathbb{Z}$, which shifts the even integer $\hat{\bm{\varepsilon}}\cdot\bm{y}$ by $-8k$.
We can thus take $\hat{\bm{\varepsilon}}\cdot\bm{y}\in\{-2,0,2,4\}$, so that $\bm{y}\neq\bm{0}$ for the nonidentity operators.
The inequalities $\lvert\bm{y}\rvert^{2}\geq2$, which follows from $\sum_{s}y_{s}\in2\mathbb{Z}$, and $\lvert\bm{y}\rvert^{2}\geq\lvert\hat{\bm{\varepsilon}}\cdot\bm{y}\rvert$ then show that the bracket is at least $12$.
We thus obtain
\begin{equation}
  h\geq\frac{3K}{4}.
\end{equation}
The current channel with $h=1$ is forbidden by the symmetry under $\bm{\Phi}_{\chi}\to-\bm{\Phi}_{\chi}$, as in Sec.~\ref{subsec:fermion_TLL_stability}.

\section{Boundary RG flow for the XXZ chain}
\label{app:XXZ}

We here provide the details of the analysis in Sec.~\ref{sec:spin_XXZ}.

\subsubsection*{Boundary operators at $K=1$}

We show that $\psi$ in Eq.~\eqref{eq:XXZ_psi} is the only nonidentity boundary operator with $h\leq1$ generated on the state~\eqref{eq:XXZ_K1}.
We first confirm the last relation of Eq.~\eqref{eq:XXZ_dictionary}.
The two Ising CFTs of the sheets $r$ and $\tilde{r}$ have two scalar operators of dimension one, $\varepsilon_{F,r}$ and $\varepsilon_{F,\tilde{r}}$.
On the orbifold side, a vertex operator with momentum $p$ and winding $w$ has the dimension $p^{2}/4+w^{2}$, and the twisted sector does not contain dimension one.
The scalar operators of dimension one are therefore $\cos2\varphi^{F}_{r}$ and $\cos2\vartheta^{F}_{r}$, which gives the relation.

A boundary operator generated at any order carries the charge of chain $F$ given by an integer combination of the charges in Eq.~\eqref{eq:XXZ_charges},
\begin{equation}
  \sqrt{2}\bm{q}=t\bm{c}+2\bm{y},
  \qquad
  h=\frac{1}{2}\sum_{r}\qty(\frac{t}{2}+2y_{r})^{2},
\end{equation}
with $t\in\mathbb{Z}$ and $\bm{y}\in\mathbb{Z}^{4}$.
The four components $t/2+2y_{r}$ are equal modulo $2$.
For odd $t$, they are half-integers, so that $h\geq1/2$ with equality only for $\sqrt{2}\bm{q}=\pm\bm{c}$, and $h\geq3/2$ otherwise.
For even $t$, they are integers of the same parity, so that $h\geq2$ for $\bm{q}\neq0$.
The nonidentity vertex operators with $h\leq1$ of the two chains are thus $e^{\pm i\bm{c}\cdot\bm{\varphi}^{\mathrm{A}}}$, $e^{\pm i\bm{c}\cdot\bm{\varphi}^{\mathrm{B}}}$, and their products.

Among these operators, we show that only $\psi$ is generated.
Since the scalar product $\sqrt{2}\bm{q}\cdot\bm{c}$ is $\pm1$ for $e^{\pm2i\varphi^{F}_{r}}$ and $0$ for $e^{\pm2i\vartheta^{F}_{r}}$, each insertion of $\mathcal{O}$ changes it by $\pm1$ in both chains or leaves it unchanged in both chains, so that its parity is the same in the two chains.
Therefore, only the products of $e^{\pm i\bm{c}\cdot\bm{\varphi}^{\mathrm{A}}}$ and $e^{\pm i\bm{c}\cdot\bm{\varphi}^{\mathrm{B}}}$ are allowed.
To single out $\psi$, we use the identification $\bm{\varphi}^{F}\to-\bm{\varphi}^{F}$ of each orbifold, which leaves the bulk perturbation and the boundary state unchanged.
Among the products, only $\psi$ is invariant under this identification.
This also excludes the charge-neutral operators $\partial_{x}\varphi^{F}_{r}$ and $\partial_{x}\vartheta^{F}_{r}$ with $h=1$, which are odd.

\subsubsection*{Flow along the marginal family}

We derive the factor $\cos(\pi\mu_{\psi}/2)$ in the RG equation~\eqref{eq:XXZ_flow}.
The two bosons $\bm{c}\cdot(\bm{\varphi}^{\mathrm{A}}_{\chi}\pm\bm{\varphi}^{\mathrm{B}}_{\chi})/\sqrt{2}$ have the $\mathrm{SU}(2)_{1}$ currents
\begin{equation}
  J^{x}_{\pm,\chi}=\cos\bm{c}\cdot\qty(\bm{\varphi}^{\mathrm{A}}_{\chi}\pm\bm{\varphi}^{\mathrm{B}}_{\chi}),
\end{equation}
and $\psi=\frac{1}{2}(J^{x}_{+,L}+J^{x}_{-,L})$ at the boundary, where the Neumann condition gives $\bm{c}\cdot\bm{\varphi}^{F}_{L}=\bm{c}\cdot\bm{\varphi}^{F}_{R}=\bm{c}\cdot\bm{\varphi}^{F}$.
The identification $\bm{\varphi}^{\mathrm{B}}\to-\bm{\varphi}^{\mathrm{B}}$ exchanges the two currents and leaves the bulk perturbation and the state with coupling $\mu_{\psi}$ unchanged.
The two currents therefore keep the same coupling $\mu_{\psi}/2$ along the flow.
The coupling $\mu_{\psi}$ rotates the image of a right-moving operator by the zero modes of $J^{x}_{\pm,L}$ as in Eq.~\eqref{appeq:largeK_rotation} with $\mu_{\psi}$ replaced by $\mu_{\psi}/2$ and $\sqrt{2}\bm{c}\cdot\bm{\Phi}_{L}$ replaced by $\bm{c}\cdot(\bm{\varphi}^{\mathrm{A}}_{L}\pm\bm{\varphi}^{\mathrm{B}}_{L})$.
For $\bm{\varphi}^{\mathrm{A}}_{L}+\bm{\varphi}^{\mathrm{B}}_{L}$, this gives
\begin{align}
  e^{\mp\frac{i}{2}\bm{c}\cdot\qty(\bm{\varphi}^{\mathrm{A}}_{L}+\bm{\varphi}^{\mathrm{B}}_{L})}
  \to{} & \cos\frac{\pi\mu_{\psi}}{2}\,e^{\mp\frac{i}{2}\bm{c}\cdot\qty(\bm{\varphi}^{\mathrm{A}}_{L}+\bm{\varphi}^{\mathrm{B}}_{L})}
  \nonumber                                                                                                                              \\
        & -i\sin\frac{\pi\mu_{\psi}}{2}\,e^{\pm\frac{i}{2}\bm{c}\cdot\qty(\bm{\varphi}^{\mathrm{A}}_{L}+\bm{\varphi}^{\mathrm{B}}_{L})},
\end{align}
and the relation for $\bm{\varphi}^{\mathrm{A}}_{L}-\bm{\varphi}^{\mathrm{B}}_{L}$ follows from $\bm{\varphi}^{\mathrm{B}}_{L}\to-\bm{\varphi}^{\mathrm{B}}_{L}$.

To obtain the coefficient of $\psi$, we decompose the exponent of the first term of Eq.~\eqref{eq:XXZ_O_bosons} as
\begin{align}
  2\qty(\varphi^{\mathrm{A}}_{r}\pm\varphi^{\mathrm{B}}_{r})
  ={} & \frac{1}{2}\bm{c}\cdot\qty(\bm{\varphi}^{\mathrm{A}}_{R}\pm\bm{\varphi}^{\mathrm{B}}_{R})
  +\frac{1}{2}\bm{c}\cdot\qty(\bm{\varphi}^{\mathrm{A}}_{L}\pm\bm{\varphi}^{\mathrm{B}}_{L})
  \nonumber                                                                                                   \\
      & +2\qty(\bm{e}_{r}-\frac{1}{2}\bm{c})\cdot\qty(\bm{\varphi}^{\mathrm{A}}\pm\bm{\varphi}^{\mathrm{B}}).
\end{align}
The last term lies in the Dirichlet eigenspace of $G$ and contributes a factor independent of $\mu_{\psi}$ and $r$.
The first two terms have the $\mathrm{SU}(2)_{1}$ exponentials of dimension $1/4$.
The right-moving factor at $z=x+iy$ is replaced by the rotated image at $z^{*}=x-iy$ with the same charge at $\mu_{\psi}=0$.
The operator product is then the one for the fermionic MRE in Appendix~\ref{app:largeK} with the above replacements.
Only the $\cos(\pi\mu_{\psi}/2)$ term of the image gives the exponentials in $J^{x}_{\pm,L}$, so that the coefficient of $J^{x}_{\pm,L}$ is proportional to $\cos(\pi\mu_{\psi}/2)$.
The charge-neutral operators with $h=1$ are excluded by the identification $\bm{\varphi}^{F}\to-\bm{\varphi}^{F}$, as shown above.
The term $C_{a}(\lambda)$ for $\psi$ on the state with coupling $\mu_{\psi}$ is thus
\begin{equation}
  C_{\psi}\cos\frac{\pi\mu_{\psi}}{2},
  \qquad
  C_{\psi}=-\frac{\lambda}{2} B_{\psi}^{(1)},
\end{equation}
where $B_{\psi}^{(1)}>0$ because Eq.~\eqref{eq:XXZ_dictionary} gives a positive constant in the first term of Eq.~\eqref{eq:XXZ_O_bosons} and each value of $r$ contributes a positive prefactor as in Appendix~\ref{app:largeK}.
Since $\lambda\propto1-K$, $C_{\psi}$ has the sign of $K-1$.

\subsubsection*{Stability of the fixed point}

We show that the endpoint reached in Sec.~\ref{subsec:XXZ_g} is stable.
It suffices that the generated nonidentity operators have $h\geq1$ at $K=1$ and that their dimensions increase on the side where the endpoint is reached.
At $K=1$, the endpoints are obtained from the state~\eqref{eq:XXZ_K1} by the $\mathrm{SU}(2)$ rotations induced by the coupling $\mu_{\psi}$.
These rotations map each boundary operator to one with the same dimension~\cite{callan1994exact,gaberdiel2001conformal}.
Since $\psi$ has the same parity of $\sqrt{2}\bm{q}\cdot\bm{c}$ in the two chains, the rotations do not change whether an operator is generated.
The generated nonidentity operators at the endpoints therefore have $h\geq1$ at $K=1$, as on the state~\eqref{eq:XXZ_K1}.

To follow the dimensions for $K\neq1$, we use the bosons $\phi_{r}$ in Eq.~\eqref{eq:XXZ_dictionary_TLL}.
At the endpoint, all $\phi_{r}$ obey the Neumann condition for $K>1$ and the Dirichlet condition for $K<1$, as shown in Sec.~\ref{subsec:XXZ_g}.
Since the action~\eqref{eq:XXZ_action} depends on $K$ only through the coefficient $1/K$, the dimensions of the vertex operators scale as $K$ at the Neumann endpoint and as $K^{-1}$ at the Dirichlet endpoint.
With $h\geq1$ at $K=1$, these operators are therefore irrelevant at the endpoint reached on each side.
The currents $\partial\phi_{r}$ keep $h=1$ and generate only exactly marginal deformations that do not change $g$.

\bibliography{biblio}
\end{document}